\documentclass[10pt,aps,pra,superscriptaddress,nofootinbib,longbibliography,table,xcdraw]{revtex4-2}
\usepackage[utf8]{inputenc}
\usepackage[english]{babel}
\usepackage[T1]{fontenc}
\usepackage{amsmath}
\usepackage{amsfonts}
\usepackage{dsfont}
\usepackage{amssymb}
\usepackage{tikz}
\usepackage{courier}
\usepackage{listings}
\usepackage{natbib}
\usepackage{enumerate}
\usepackage{physics}
\usepackage{mdframed}
\usepackage{tcolorbox}
\tcbuselibrary{breakable}
\usepackage{graphicx}
\newcommand{\caphead}[1]{\textbf{#1}}
\usepackage[linktocpage=true, colorlinks=true, linkcolor=blue, urlcolor=blue, citecolor=blue, hyperfootnotes=true]{hyperref}
\usepackage{cleveref}
\usepackage{csquotes}
\usepackage[normalem]{ulem}
\usepackage{multirow}
\usepackage{tikz}
\def\checkmark{\tikz\fill[scale=0.4](0,.35) -- (.25,0) -- (1,.7) -- (.25,.15) -- cycle;} 
\usetikzlibrary{shapes.geometric}

\usepackage{bm}
\input{preamble}
\usetikzlibrary{shapes.multipart}

\tikzstyle{bwSpider}=[
       rectangle split,
       rectangle split parts=2,
       rectangle split part fill={black,white},
 minimum size=3.6 mm, inner sep=-2mm, draw=black,scale=0.5,rounded corners=0.8 mm
       ]
 \tikzstyle{wbSpider}=[
       rectangle split,
       rectangle split parts=2,
       rectangle split part fill={white,black},
 minimum size=3.6 mm, inner sep=-2mm, draw=black,scale=0.5,rounded corners=0.8 mm
       ]
\tikzstyle{qWire}=[line width = 1pt, color=quantumviolet]
\tikzstyle{hWire}=[line width = 1pt, color=higherorderred]
\tikzstyle{sWire}=[line width = 1pt, color=higherordergreen]
\tikzstyle{cWire}=[color=quantumgray,thin]
\tikzstyle{rWire}=[color=black,thin]

\tikzstyle{env}=[copoint,regular polygon rotate=0,minimum width=0.2cm, fill=black]

\tikzstyle{probs}=[shape=semicircle,fill=white,draw=black,shape border rotate=180,minimum width=1.2cm]

\tikzstyle{every picture}=[baseline=-0.25em,scale=.5]
\tikzstyle{dotpic}=[] 
\tikzstyle{diredges}=[every to/.style={diredge}]
\tikzstyle{math matrix}=[matrix of math nodes,left delimiter=(,right delimiter=),inner sep=2pt,column sep=1em,row sep=0.5em,nodes={inner sep=0pt},text height=1.5ex, text depth=0.25ex]

\tikzstyle{inline text}=[text height=1.5ex, text depth=0.25ex,yshift=0.5mm]
\tikzstyle{label}=[font=\scriptsize,text height=1.5ex, text depth=0.25ex,yshift=0.0mm]
\tikzstyle{left label}=[label,anchor=east,xshift=1.5mm]
\tikzstyle{right label}=[label,anchor=west,xshift=-1mm]

\tikzstyle{braceedge}=[decorate,decoration={brace,amplitude=2mm,raise=-1mm}]
\tikzstyle{small braceedge}=[decorate,decoration={brace,amplitude=1mm,raise=-1mm}]

\tikzstyle{doubled}=[line width=1.6pt] 
\tikzstyle{boldedge}=[doubled,shorten <=-0.17mm,shorten >=-0.17mm]
\tikzstyle{boldedgegray}=[doubled,gray,shorten <=-0.17mm,shorten >=-0.17mm]
\tikzstyle{singleedgegray}=[gray]

\tikzstyle{semidoubled}=[line width=1.4pt] 
\tikzstyle{semiboldedgegray}=[semidoubled,gray,shorten <=-0.17mm,shorten >=-0.17mm]

\tikzstyle{boxedge}=[semiboldedgegray]

\tikzstyle{boldedgedashed}=[very thick,dashed,shorten <=-0.17mm,shorten >=-0.17mm]
\tikzstyle{vboldedgedashed}=[doubled,dashed,shorten <=-0.17mm,shorten >=-0.17mm]
\tikzstyle{left hook arrow}=[left hook-latex]
\tikzstyle{right hook arrow}=[right hook-latex]
\tikzstyle{sembracket}=[line width=0.5pt,shorten <=-0.07mm,shorten >=-0.07mm]

\tikzstyle{causal edge}=[->,thick,gray]
\tikzstyle{causal nondir}=[thick,gray]
\tikzstyle{timeline}=[thick,gray, dashed]

\tikzstyle{cedge}=[<->,thick,gray!70!white]

\tikzstyle{empty diagram}=[draw=gray!40!white,dashed,shape=rectangle,minimum width=1cm,minimum height=1cm]
\tikzstyle{empty circle}=[draw=black,dashed,shape=circle,minimum width=.2cm,minimum height=.2cm, inner sep=1pt]
\tikzstyle{empty diagram small}=[draw=gray!50!white,dashed,shape=rectangle,minimum width=0.4cm,minimum height=0.3cm]

\tikzstyle{dot}=[inner sep=0mm,minimum width=2mm,minimum height=2mm,draw,shape=circle]
\tikzstyle{proj}=[trapezium, trapezium angle=67.5, draw, inner sep=1pt, outer sep=0pt, minimum height=6pt, minimum width=2pt,rotate=0,trapezium stretches=true]
\tikzstyle{inc}=[trapezium, trapezium angle=-67.5, draw, inner sep=1pt, outer sep=0pt, minimum height=6pt, minimum width=2pt,rotate=0,trapezium stretches=true]
\tikzstyle{coproj}=[trapezium, trapezium angle=-67.5, draw, inner sep=1pt, outer sep=0pt, minimum height=6pt, minimum width=2pt,rotate=0,trapezium stretches=true]
\tikzstyle{coinc}=[trapezium, trapezium angle=67.5, draw, inner sep=1pt, outer sep=0pt, minimum height=6pt, minimum width=2pt,rotate=0,trapezium stretches=true]
\tikzstyle{leak}=[white dot, shape=regular polygon, minimum size=3.3 mm, regular polygon sides=3, outer sep=-0.2mm, regular polygon rotate=270]
\tikzstyle{wide proj}=[draw,fill=white,chamfered rectangle,chamfered rectangle angle=30, minimum width=15mm,minimum height=1mm,scale=0.5, outer sep=-0.2mm]
\tikzstyle{very wide proj}=[draw,fill=white,chamfered rectangle,chamfered rectangle angle=30, minimum width=25mm,minimum height=1mm,scale=0.5, outer sep=-0.2mm]
\tikzstyle{very very wide proj}=[draw,fill=white,chamfered rectangle,chamfered rectangle angle=30, minimum width=35mm,minimum height=1mm,scale=0.5, outer sep=-0.2mm]
\tikzstyle{preleak}=[proj]

\tikzstyle{split proj out}=[regular polygon,regular polygon sides=3,draw,scale=0.75,inner sep=-0.5pt,minimum width=3.3mm,fill=white,regular polygon rotate=180]
\tikzstyle{split proj in}=[regular polygon,regular polygon sides=3,draw,scale=0.75,inner sep=-0.5pt,minimum width=3.3mm,fill=white]
\tikzstyle{Vleak}=[white dot, shape=regular polygon, minimum size=3.3 mm, regular polygon sides=3, outer sep=-0.2mm, regular polygon rotate=90]
\tikzstyle{dleak}=[white dot, line width=1.6pt, shape=regular polygon, minimum size=3.3 mm, regular polygon sides=3, outer sep=-0.2mm, regular polygon rotate=270]

\tikzstyle{Wsquare}=[white dot, shape=regular polygon, minimum size=3.3 mm, regular polygon sides=3, outer sep=-0.2mm]
\tikzstyle{Wsquareadj}=[white dot, shape=regular polygon, minimum size=3.3 mm, regular polygon sides=3, outer sep=-0.2mm, regular polygon rotate=180]
\tikzstyle{ddot}=[inner sep=0mm, doubled, minimum width=2.5mm,minimum height=2.5mm,draw,shape=circle]

\tikzstyle{black dot}=[dot,fill=black]
\tikzstyle{white dot}=[dot,fill=white,,text depth=-0.2mm]
\tikzstyle{white Wsquare}=[Wsquare,fill=gray,,text depth=-0.2mm]
\tikzstyle{white Wsquareadj}=[Wsquareadj,fill=white,,text depth=-0.2mm]
\tikzstyle{green dot}=[white dot] 
\tikzstyle{gray dot}=[dot,fill=gray!40!white,,text depth=-0.2mm]
\tikzstyle{red dot}=[gray dot] 

\tikzstyle{black ddot}=[ddot,fill=black]
\tikzstyle{white ddot}=[ddot,fill=white]
\tikzstyle{gray ddot}=[ddot,fill=gray!40!white]

\tikzstyle{gray edge}=[gray!60!white]

\tikzstyle{small dot}=[inner sep=0.5mm,minimum width=0pt,minimum height=0pt,draw,shape=circle]

\tikzstyle{small black dot}=[small dot,fill=black]
\tikzstyle{small white dot}=[small dot,fill=white]
\tikzstyle{small gray dot}=[small dot,fill=gray!40!white]

\tikzstyle{causal dot}=[inner sep=0.4mm,minimum width=0pt,minimum height=0pt,draw=white,shape=circle,fill=gray!40!white]

\tikzstyle{phase dimensions}=[minimum size=5mm,font=\footnotesize,rectangle,rounded corners=2.5mm,inner sep=0.2mm,outer sep=-2mm]
\tikzstyle{dphase dimensions}=[minimum size=5mm,font=\footnotesize,rectangle,rounded corners=2.5mm,inner sep=0.2mm,outer sep=-2mm]

\tikzstyle{white phase dot}=[dot,fill=white,phase dimensions]
\tikzstyle{white phase ddot}=[ddot,fill=white,dphase dimensions]

\tikzstyle{white rect ddot}=[draw=black,fill=white,doubled,minimum size=5mm,font=\footnotesize,rectangle,rounded corners=2.5mm,inner sep=0.2mm]
\tikzstyle{gray rect ddot}=[draw=black,fill=gray!40!white,doubled,minimum size=6mm,font=\footnotesize,rectangle,rounded corners=3mm]

\tikzstyle{gray phase dot}=[dot,fill=gray!40!white,phase dimensions]
\tikzstyle{gray phase ddot}=[ddot,fill=gray!40!white,dphase dimensions]
\tikzstyle{grey phase dot}=[gray phase dot]
\tikzstyle{grey phase ddot}=[gray phase ddot]

\tikzstyle{small phase dimensions}=[minimum size=4mm,font=\tiny,rectangle,rounded corners=2mm,inner sep=0.2mm,outer sep=-2mm]
\tikzstyle{small dphase dimensions}=[minimum size=4mm,font=\tiny,rectangle,rounded corners=2mm,inner sep=0.2mm,outer sep=-2mm]

\tikzstyle{small gray phase dot}=[dot,fill=gray!40!white,small phase dimensions]
\tikzstyle{small gray phase ddot}=[ddot,fill=gray!40!white,small dphase dimensions]

\tikzstyle{small map}=[draw,shape=rectangle,minimum height=4mm,minimum width=4mm,fill=white]
\tikzstyle{math map}=[draw,shape=rectangle,minimum height=4mm,minimum width=4mm,fill=black, font=\color{white}]

\tikzstyle{cnot}=[fill=white,shape=circle,inner sep=-1.4pt]

\tikzstyle{asym hadamard}=[fill=white,draw,shape=NEbox,inner sep=0.6mm,font=\footnotesize,minimum height=4mm]
\tikzstyle{asym hadamard conj}=[fill=white,draw,shape=NWbox,inner sep=0.6mm,font=\footnotesize,minimum height=4mm]
\tikzstyle{asym hadamard dag}=[fill=white,draw,shape=SEbox,inner sep=0.6mm,font=\footnotesize,minimum height=4mm]

\tikzstyle{hadamard}=[fill=white,draw,inner sep=0.6mm,font=\footnotesize,minimum height=4mm,minimum width=4mm]
\tikzstyle{small hadamard}=[fill=white,draw,inner sep=0.6mm,minimum height=1.5mm,minimum width=1.5mm]
\tikzstyle{small hadamard rotate}=[small hadamard,rotate=45]
\tikzstyle{dhadamard}=[hadamard,doubled]
\tikzstyle{small dhadamard}=[small hadamard,doubled]
\tikzstyle{small dhadamard rotate}=[small hadamard rotate,doubled]
\tikzstyle{antipode}=[white dot,inner sep=0.3mm,font=\footnotesize]

\tikzstyle{scalar}=[diamond,draw,inner sep=0.5pt,font=\small]
\tikzstyle{dscalar}=[diamond,doubled, draw,inner sep=0.5pt,font=\small]

\tikzstyle{small box}=[rectangle,inline text,fill=white,draw,minimum height=0.5mm,yshift=-0.5mm,minimum width=5mm,font=\small]
\tikzstyle{small gray box}=[small box,fill=gray!30]
\tikzstyle{medium box}=[rectangle,inline text,fill=white,draw,minimum height=5mm,yshift=-0.5mm,minimum width=10mm,font=\small]
\tikzstyle{square box}=[small box] 
\tikzstyle{medium gray box}=[small box,fill=gray!30]
\tikzstyle{semilarge box}=[rectangle,inline text,fill=white,draw,minimum height=5mm,yshift=-0.5mm,minimum width=12.5mm,font=\small]
\tikzstyle{large box}=[rectangle,inline text,fill=white,draw,minimum height=5mm,yshift=-0.5mm,minimum width=15mm,font=\small]
\tikzstyle{large gray box}=[small box,fill=gray!30]

\tikzstyle{Bayes box}=[rectangle,fill=black,draw, minimum height=3mm, minimum width=3mm]

\tikzstyle{gray square point}=[small box,fill=gray!50]

\tikzstyle{dphase box white}=[dhadamard]
\tikzstyle{dphase box gray}=[dhadamard,fill=gray!50!white]
\tikzstyle{phase box white}=[hadamard]
\tikzstyle{phase box gray}=[hadamard,fill=gray!50!white]

\tikzstyle{point}=[regular polygon,regular polygon sides=3,draw,scale=0.75,inner sep=-0.5pt,minimum width=9mm,fill=white,regular polygon rotate=180]
\tikzstyle{point nosep}=[regular polygon,regular polygon sides=3,draw,scale=0.75,inner sep=-2pt,minimum width=9mm,fill=white,regular polygon rotate=180]
\tikzstyle{copoint}=[regular polygon,regular polygon sides=3,draw,scale=0.75,inner sep=-0.5pt,minimum width=9mm,fill=white]
\tikzstyle{dpoint}=[point,doubled]
\tikzstyle{dcopoint}=[copoint,doubled]

\tikzstyle{pointgrow}=[shape=cornerpoint,kpoint common,scale=0.75,inner sep=3pt]
\tikzstyle{pointgrow dag}=[shape=cornercopoint,kpoint common,scale=0.75,inner sep=3pt]

\tikzstyle{wide copoint}=[fill=white,draw,shape=isosceles triangle,shape border rotate=90,isosceles triangle stretches=true,inner sep=0pt,minimum width=1.5cm,minimum height=6.12mm]
\tikzstyle{wide point}=[fill=white,draw,shape=isosceles triangle,shape border rotate=-90,isosceles triangle stretches=true,inner sep=0pt,minimum width=1.5cm,minimum height=6.12mm,yshift=-0.0mm]
\tikzstyle{wide point plus}=[fill=white,draw,shape=isosceles triangle,shape border rotate=-90,isosceles triangle stretches=true,inner sep=0pt,minimum width=1.74cm,minimum height=7mm,yshift=-0.0mm]

\tikzstyle{wide dpoint}=[fill=white,doubled,draw,shape=isosceles triangle,shape border rotate=-90,isosceles triangle stretches=true,inner sep=0pt,minimum width=1.5cm,minimum height=6.12mm,yshift=-0.0mm]

\tikzstyle{tinypoint}=[regular polygon,regular polygon sides=3,draw,scale=0.55,inner sep=-0.15pt,minimum width=6mm,fill=white,regular polygon rotate=180]

\tikzstyle{white point}=[point]
\tikzstyle{white dpoint}=[dpoint]
\tikzstyle{green point}=[white point] 
\tikzstyle{white copoint}=[copoint]
\tikzstyle{gray point}=[point,fill=gray!40!white]
\tikzstyle{gray dpoint}=[gray point,doubled]
\tikzstyle{red point}=[gray point] 
\tikzstyle{gray copoint}=[copoint,fill=gray!40!white]
\tikzstyle{gray dcopoint}=[gray copoint,doubled]

\tikzstyle{white point guide}=[regular polygon,regular polygon sides=3,font=\scriptsize,draw,scale=0.65,inner sep=-0.5pt,minimum width=9mm,fill=white,regular polygon rotate=180]

\tikzstyle{black point}=[point,fill=black,font=\color{white}]
\tikzstyle{black copoint}=[copoint,fill=black,font=\color{white}]

\tikzstyle{tiny gray point}=[tinypoint,fill=gray!40!white]

\tikzstyle{diredge}=[->]
\tikzstyle{ddiredge}=[<->]
\tikzstyle{rdiredge}=[<-]
\tikzstyle{thickdiredge}=[->, very thick]
\tikzstyle{pointer edge}=[->,very thick,gray]
\tikzstyle{pointer edge part}=[very thick,gray]
\tikzstyle{dashed edge}=[dashed]
\tikzstyle{thick dashed edge}=[very thick,dashed]
\tikzstyle{thick gray dashed edge}=[thick dashed edge,gray!40]
\tikzstyle{thick map edge}=[very thick,|->]

\makeatletter
\newcommand{\boxshape}[3]{%
\pgfdeclareshape{#1}{
\inheritsavedanchors[from=rectangle] 
\inheritanchorborder[from=rectangle]
\inheritanchor[from=rectangle]{center}
\inheritanchor[from=rectangle]{north}
\inheritanchor[from=rectangle]{south}
\inheritanchor[from=rectangle]{west}
\inheritanchor[from=rectangle]{east}
\backgroundpath{
\southwest \pgf@xa=\pgf@x \pgf@ya=\pgf@y
\northeast \pgf@xb=\pgf@x \pgf@yb=\pgf@y

\@tempdima=#2
\@tempdimb=#3

\pgfpathmoveto{\pgfpoint{\pgf@xa - 5pt + \@tempdima}{\pgf@ya}}
\pgfpathlineto{\pgfpoint{\pgf@xa - 5pt - \@tempdima}{\pgf@yb}}
\pgfpathlineto{\pgfpoint{\pgf@xb + 5pt + \@tempdimb}{\pgf@yb}}
\pgfpathlineto{\pgfpoint{\pgf@xb + 5pt - \@tempdimb}{\pgf@ya}}
\pgfpathlineto{\pgfpoint{\pgf@xa - 5pt + \@tempdima}{\pgf@ya}}
\pgfpathclose
}
}}

\boxshape{NEbox}{0pt}{3pt}
\boxshape{SEbox}{0pt}{-3pt}
\boxshape{NWbox}{5pt}{0pt}
\boxshape{SWbox}{-5pt}{0pt}
\boxshape{EBox}{-3pt}{3pt}
\boxshape{WBox}{3pt}{-3pt}
\makeatother

\tikzstyle{cloud}=[shape=cloud,draw,minimum width=1.5cm,minimum height=1.5cm]

\tikzstyle{map}=[draw,shape=NEbox,inner sep=1pt,minimum height=4mm,fill=white]
\tikzstyle{dashedmap}=[draw,dashed,shape=NEbox,inner sep=2pt,minimum height=6mm,fill=white]
\tikzstyle{mapdag}=[draw,shape=SEbox,inner sep=1pt,minimum height=4mm,fill=white]
\tikzstyle{mapadj}=[draw,shape=SEbox,inner sep=2pt,minimum height=6mm,fill=white]
\tikzstyle{maptrans}=[draw,shape=SWbox,inner sep=2pt,minimum height=6mm,fill=white]
\tikzstyle{mapconj}=[draw,shape=NWbox,inner sep=2pt,minimum height=6mm,fill=white]

\tikzstyle{medium map}=[draw,shape=NEbox,inner sep=2pt,minimum height=6mm,fill=white,minimum width=7mm]
\tikzstyle{medium map dag}=[draw,shape=SEbox,inner sep=2pt,minimum height=6mm,fill=white,minimum width=7mm]
\tikzstyle{medium map adj}=[draw,shape=SEbox,inner sep=2pt,minimum height=6mm,fill=white,minimum width=7mm]
\tikzstyle{medium map trans}=[draw,shape=SWbox,inner sep=2pt,minimum height=6mm,fill=white,minimum width=7mm]
\tikzstyle{medium map conj}=[draw,shape=NWbox,inner sep=2pt,minimum height=6mm,fill=white,minimum width=7mm]
\tikzstyle{semilarge map}=[draw,shape=NEbox,inner sep=2pt,minimum height=6mm,fill=white,minimum width=9.5mm]
\tikzstyle{semilarge map trans}=[draw,shape=SWbox,inner sep=2pt,minimum height=6mm,fill=white,minimum width=9.5mm]
\tikzstyle{semilarge map adj}=[draw,shape=SEbox,inner sep=2pt,minimum height=6mm,fill=white,minimum width=9.5mm]
\tikzstyle{semilarge map dag}=[draw,shape=SEbox,inner sep=2pt,minimum height=6mm,fill=white,minimum width=9.5mm]
\tikzstyle{semilarge map conj}=[draw,shape=NWbox,inner sep=2pt,minimum height=6mm,fill=white,minimum width=9.5mm]
\tikzstyle{large map}=[draw,shape=NEbox,inner sep=2pt,minimum height=6mm,fill=white,minimum width=12mm]
\tikzstyle{large map conj}=[draw,shape=NWbox,inner sep=2pt,minimum height=6mm,fill=white,minimum width=12mm]
\tikzstyle{very large map}=[draw,shape=NEbox,inner sep=2pt,minimum height=6mm,fill=white,minimum width=17mm]

\tikzstyle{medium dmap}=[draw,doubled,shape=NEbox,inner sep=2pt,minimum height=6mm,fill=white,minimum width=7mm]
\tikzstyle{medium dmap dag}=[draw,doubled,shape=SEbox,inner sep=2pt,minimum height=6mm,fill=white,minimum width=7mm]
\tikzstyle{medium dmap adj}=[draw,doubled,shape=SEbox,inner sep=2pt,minimum height=6mm,fill=white,minimum width=7mm]
\tikzstyle{medium dmap trans}=[draw,doubled,shape=SWbox,inner sep=2pt,minimum height=6mm,fill=white,minimum width=7mm]
\tikzstyle{medium dmap conj}=[draw,doubled,shape=NWbox,inner sep=2pt,minimum height=6mm,fill=white,minimum width=7mm]
\tikzstyle{semilarge dmap}=[draw,doubled,shape=NEbox,inner sep=2pt,minimum height=6mm,fill=white,minimum width=9.5mm]
\tikzstyle{semilarge dmap trans}=[draw,doubled,shape=SWbox,inner sep=2pt,minimum height=6mm,fill=white,minimum width=9.5mm]
\tikzstyle{semilarge dmap adj}=[draw,doubled,shape=SEbox,inner sep=2pt,minimum height=6mm,fill=white,minimum width=9.5mm]
\tikzstyle{semilarge dmap dag}=[draw,doubled,shape=SEbox,inner sep=2pt,minimum height=6mm,fill=white,minimum width=9.5mm]
\tikzstyle{semilarge dmap conj}=[draw,doubled,shape=NWbox,inner sep=2pt,minimum height=6mm,fill=white,minimum width=9.5mm]
\tikzstyle{large dmap}=[draw,doubled,shape=NEbox,inner sep=2pt,minimum height=6mm,fill=white,minimum width=12mm]
\tikzstyle{large dmap conj}=[draw,doubled,shape=NWbox,inner sep=2pt,minimum height=6mm,fill=white,minimum width=12mm]
\tikzstyle{large dmap trans}=[draw,doubled,shape=SWbox,inner sep=2pt,minimum height=6mm,fill=white,minimum width=12mm]
\tikzstyle{large dmap adj}=[draw,doubled,shape=SEbox,inner sep=2pt,minimum height=6mm,fill=white,minimum width=12mm]
\tikzstyle{large dmap dag}=[draw,doubled,shape=SEbox,inner sep=2pt,minimum height=6mm,fill=white,minimum width=12mm]
\tikzstyle{very large dmap}=[draw,doubled,shape=NEbox,inner sep=2pt,minimum height=6mm,fill=white,minimum width=19.5mm]

\tikzstyle{muxbox}=[draw,shape=rectangle,minimum height=3mm,minimum width=3mm,fill=white]
\tikzstyle{dmuxbox}=[muxbox,doubled]

\tikzstyle{box}=[draw,shape=rectangle,inner sep=2pt,minimum height=6mm,minimum width=6mm,fill=white]
\tikzstyle{dbox}=[draw,doubled,shape=rectangle,inner sep=2pt,minimum height=6mm,minimum width=6mm,fill=white]
\tikzstyle{dmap}=[draw,doubled,shape=NEbox,inner sep=2pt,minimum height=6mm,fill=white]
\tikzstyle{dmapdag}=[draw,doubled,shape=SEbox,inner sep=2pt,minimum height=6mm,fill=white]
\tikzstyle{dmapadj}=[draw,doubled,shape=SEbox,inner sep=2pt,minimum height=6mm,fill=white]
\tikzstyle{dmaptrans}=[draw,doubled,shape=SWbox,inner sep=2pt,minimum height=6mm,fill=white]
\tikzstyle{dmapconj}=[draw,doubled,shape=NWbox,inner sep=2pt,minimum height=6mm,fill=white]

\tikzstyle{ddmap}=[draw,doubled,dashed,shape=NEbox,inner sep=2pt,minimum height=6mm,fill=white]
\tikzstyle{ddmapdag}=[draw,doubled,dashed,shape=SEbox,inner sep=2pt,minimum height=6mm,fill=white]
\tikzstyle{ddmapadj}=[draw,doubled,dashed,shape=SEbox,inner sep=2pt,minimum height=6mm,fill=white]
\tikzstyle{ddmaptrans}=[draw,doubled,dashed,shape=SWbox,inner sep=2pt,minimum height=6mm,fill=white]
\tikzstyle{ddmapconj}=[draw,doubled,dashed,shape=NWbox,inner sep=2pt,minimum height=6mm,fill=white]

\boxshape{sNEbox}{0pt}{3pt}
\boxshape{sSEbox}{0pt}{-3pt}
\boxshape{sNWbox}{3pt}{0pt}
\boxshape{sSWbox}{-3pt}{0pt}
\tikzstyle{smap}=[draw,shape=sNEbox,fill=white]
\tikzstyle{smapdag}=[draw,shape=sSEbox,fill=white]
\tikzstyle{smapadj}=[draw,shape=sSEbox,fill=white]
\tikzstyle{smaptrans}=[draw,shape=sSWbox,fill=white]
\tikzstyle{smapconj}=[draw,shape=sNWbox,fill=white]

\tikzstyle{dsmap}=[draw,dashed,shape=sNEbox,fill=white]
\tikzstyle{dsmapdag}=[draw,dashed,shape=sSEbox,fill=white]
\tikzstyle{dsmaptrans}=[draw,dashed,shape=sSWbox,fill=white]
\tikzstyle{dsmapconj}=[draw,dashed,shape=sNWbox,fill=white]

\boxshape{mNEbox}{0pt}{10pt}
\boxshape{mSEbox}{0pt}{-10pt}
\boxshape{mNWbox}{10pt}{0pt}
\boxshape{mSWbox}{-10pt}{0pt}
\tikzstyle{mmap}=[draw,shape=mNEbox]
\tikzstyle{mmapdag}=[draw,shape=mSEbox]
\tikzstyle{mmaptrans}=[draw,shape=mSWbox]
\tikzstyle{mmapconj}=[draw,shape=mNWbox]

\tikzstyle{mmapgray}=[draw,fill=gray!40!white,shape=mNEbox]
\tikzstyle{smapgray}=[draw,fill=gray!40!white,shape=sNEbox]

\makeatletter

\pgfdeclareshape{cornerpoint}{
\inheritsavedanchors[from=rectangle] 
\inheritanchorborder[from=rectangle]
\inheritanchor[from=rectangle]{center}
\inheritanchor[from=rectangle]{north}
\inheritanchor[from=rectangle]{south}
\inheritanchor[from=rectangle]{west}
\inheritanchor[from=rectangle]{east}
\backgroundpath{
\southwest \pgf@xa=\pgf@x \pgf@ya=\pgf@y
\northeast \pgf@xb=\pgf@x \pgf@yb=\pgf@y

\pgfmathsetmacro{\pgf@shorten@left}{\pgfkeysvalueof{/tikz/shorten left}}
\pgfmathsetmacro{\pgf@shorten@right}{\pgfkeysvalueof{/tikz/shorten right}}

\pgfpathmoveto{\pgfpoint{0.5 * (\pgf@xa + \pgf@xb)}{\pgf@ya - 5pt}}
\pgfpathlineto{\pgfpoint{\pgf@xa - 8pt + \pgf@shorten@left}{\pgf@yb - 1.5 * \pgf@shorten@left}}
\pgfpathlineto{\pgfpoint{\pgf@xa - 8pt + \pgf@shorten@left}{\pgf@yb}}
\pgfpathlineto{\pgfpoint{\pgf@xb + 8pt - \pgf@shorten@right}{\pgf@yb}}
\pgfpathlineto{\pgfpoint{\pgf@xb + 8pt - \pgf@shorten@right}{\pgf@yb - 1.5 * \pgf@shorten@right}}
\pgfpathclose
}
}

\pgfdeclareshape{cornercopoint}{
\inheritsavedanchors[from=rectangle] 
\inheritanchorborder[from=rectangle]
\inheritanchor[from=rectangle]{center}
\inheritanchor[from=rectangle]{north}
\inheritanchor[from=rectangle]{south}
\inheritanchor[from=rectangle]{west}
\inheritanchor[from=rectangle]{east}
\backgroundpath{
\southwest \pgf@xa=\pgf@x \pgf@ya=\pgf@y
\northeast \pgf@xb=\pgf@x \pgf@yb=\pgf@y

\pgfmathsetmacro{\pgf@shorten@left}{\pgfkeysvalueof{/tikz/shorten left}}
\pgfmathsetmacro{\pgf@shorten@right}{\pgfkeysvalueof{/tikz/shorten right}}

\pgfpathmoveto{\pgfpoint{0.5 * (\pgf@xa + \pgf@xb)}{\pgf@yb + 5pt}}
\pgfpathlineto{\pgfpoint{\pgf@xa - 8pt + \pgf@shorten@left}{\pgf@ya + 1.5 * \pgf@shorten@left}}
\pgfpathlineto{\pgfpoint{\pgf@xa - 8pt + \pgf@shorten@left}{\pgf@ya}}
\pgfpathlineto{\pgfpoint{\pgf@xb + 8pt - \pgf@shorten@right}{\pgf@ya}}
\pgfpathlineto{\pgfpoint{\pgf@xb + 8pt - \pgf@shorten@right}{\pgf@ya + 1.5 * \pgf@shorten@right}}
\pgfpathclose
}
}

\makeatother

\pgfkeyssetvalue{/tikz/shorten left}{0pt}
\pgfkeyssetvalue{/tikz/shorten right}{0pt}

\tikzstyle{kpoint common}=[draw,fill=white,inner sep=1pt,minimum height=4mm]
\tikzstyle{kpoint sc}=[shape=cornerpoint,kpoint common]
\tikzstyle{kpoint adjoint sc}=[shape=cornercopoint,kpoint common]
\tikzstyle{kpoint}=[shape=cornerpoint,shorten left=5pt,kpoint common]
\tikzstyle{kpoint adjoint}=[shape=cornercopoint,shorten left=5pt,kpoint common]
\tikzstyle{kpoint conjugate}=[shape=cornerpoint,shorten right=5pt,kpoint common]
\tikzstyle{kpoint transpose}=[shape=cornercopoint,shorten right=5pt,kpoint common]
\tikzstyle{kpoint symm}=[shape=cornerpoint,shorten left=5pt,shorten right=5pt,kpoint common]

\tikzstyle{wide kpoint sc}=[shape=cornerpoint,kpoint common, minimum width=1 cm]
\tikzstyle{wide kpointdag sc}=[shape=cornercopoint,kpoint common, minimum width=1 cm]

\tikzstyle{black kpoint}=[shape=cornerpoint,shorten left=5pt,kpoint common,fill=black,font=\color{white}]

\tikzstyle{black kpoint sm}=[shape=cornerpoint,shorten left=5pt,kpoint common,fill=black,font=\color{white},scale=0.75]

\tikzstyle{black kpoint adjoint}=[shape=cornercopoint,shorten left=5pt,kpoint common,fill=black,font=\color{white}]
\tikzstyle{black kpointadj}=[shape=cornercopoint,shorten left=5pt,kpoint common,fill=black,font=\color{white}]

\tikzstyle{black kpointadj sm}=[shape=cornercopoint,shorten left=5pt,kpoint common,fill=black,font=\color{white},scale=0.75]

\tikzstyle{black dkpoint}=[shape=cornerpoint,shorten left=5pt,kpoint common,fill=black, doubled,font=\color{white}]
\tikzstyle{black dkpoint adjoint}=[shape=cornercopoint,shorten left=5pt,kpoint common,fill=black, doubled,font=\color{white}]
\tikzstyle{black dkpointadj}=[shape=cornercopoint,shorten left=5pt,kpoint common,fill=black, doubled,font=\color{white}]

\tikzstyle{black dkpoint sm}=[shape=cornerpoint,shorten left=5pt,kpoint common,fill=black, doubled,font=\color{white},scale=0.75]
\tikzstyle{black dkpointadj sm}=[shape=cornercopoint,shorten left=5pt,kpoint common,fill=black, doubled,font=\color{white},scale=0.75]

\tikzstyle{kpointdag}=[kpoint adjoint]
\tikzstyle{kpointadj}=[kpoint adjoint]
\tikzstyle{kpointconj}=[kpoint conjugate]
\tikzstyle{kpointtrans}=[kpoint transpose]

\tikzstyle{big kpoint}=[kpoint, minimum width=1.2 cm, minimum height=8mm, inner sep=4pt, text depth=3mm]

\tikzstyle{wide kpoint}=[kpoint, minimum width=1 cm, inner sep=2pt]
\tikzstyle{wide kpointdag}=[kpointdag, minimum width=1 cm, inner sep=2pt]
\tikzstyle{wide kpointconj}=[kpointconj, minimum width=1 cm, inner sep=2pt]
\tikzstyle{wide kpointtrans}=[kpointtrans, minimum width=1 cm, inner sep=2pt]

\tikzstyle{wider kpoint}=[kpoint, minimum width=1.25 cm, inner sep=2pt]
\tikzstyle{wider kpointdag}=[kpointdag, minimum width=1.25 cm, inner sep=2pt]
\tikzstyle{wider kpointconj}=[kpointconj, minimum width=1.25 cm, inner sep=2pt]
\tikzstyle{wider kpointtrans}=[kpointtrans, minimum width=1.25 cm, inner sep=2pt]

\tikzstyle{gray kpoint}=[kpoint,fill=gray!50!white]
\tikzstyle{gray kpointdag}=[kpointdag,fill=gray!50!white]
\tikzstyle{gray kpointadj}=[kpointadj,fill=gray!50!white]
\tikzstyle{gray kpointconj}=[kpointconj,fill=gray!50!white]
\tikzstyle{gray kpointtrans}=[kpointtrans,fill=gray!50!white]

\tikzstyle{gray dkpoint}=[kpoint,fill=gray!50!white,doubled]
\tikzstyle{gray dkpointdag}=[kpointdag,fill=gray!50!white,doubled]
\tikzstyle{gray dkpointadj}=[kpointadj,fill=gray!50!white,doubled]
\tikzstyle{gray dkpointconj}=[kpointconj,fill=gray!50!white,doubled]
\tikzstyle{gray dkpointtrans}=[kpointtrans,fill=gray!50!white,doubled]

\tikzstyle{white label}=[draw,fill=white,rectangle,inner sep=0.7 mm]
\tikzstyle{gray label}=[draw,fill=gray!50!white,rectangle,inner sep=0.7 mm]
\tikzstyle{black label}=[draw,fill=black,rectangle,inner sep=0.7 mm]

\tikzstyle{dkpoint}=[kpoint,doubled]
\tikzstyle{wide dkpoint}=[wide kpoint,doubled]
\tikzstyle{dkpointdag}=[kpoint adjoint,doubled]
\tikzstyle{wide dkpointdag}=[wide kpointdag,doubled]
\tikzstyle{dkcopoint}=[kpoint adjoint,doubled]
\tikzstyle{dkpointadj}=[kpoint adjoint,doubled]
\tikzstyle{dkpointconj}=[kpoint conjugate,doubled]
\tikzstyle{dkpointtrans}=[kpoint transpose,doubled]

\tikzstyle{kscalar}=[kpoint common, shape=EBox, inner xsep=-1pt, inner ysep=3pt,font=\small]
\tikzstyle{kscalarconj}=[kpoint common, shape=WBox, inner xsep=-1pt, inner ysep=3pt,font=\small]

\tikzstyle{spekpoint}=[kpoint sc,minimum height=5mm,inner sep=3pt]
\tikzstyle{spekcopoint}=[kpoint adjoint sc,minimum height=5mm,inner sep=3pt]

\tikzstyle{dspekpoint}=[spekpoint,doubled]
\tikzstyle{dspekcopoint}=[spekcopoint,doubled]

 \tikzstyle{upground}=[circuit ee IEC,thick,ground,rotate=90,scale=2.5]
 \tikzstyle{downground}=[circuit ee IEC,thick,ground,rotate=-90,scale=2.5]
 \tikzstyle{bigground}=[regular polygon,regular polygon sides=3,draw=gray,scale=0.50,inner sep=-0.5pt,minimum width=10mm,fill=gray]

\tikzstyle{arrs}=[-latex,font=\small,auto]
\tikzstyle{arrow plain}=[arrs]
\tikzstyle{arrow dashed}=[dashed,arrs]
\tikzstyle{arrow bold}=[very thick,arrs]
\tikzstyle{arrow hide}=[draw=white!0,-]
\tikzstyle{arrow reverse}=[latex-]
\tikzstyle{cdnode}=[]

\newcommand{\Csa}[1]{\mathbf{H}_{#1}(\mathbb{C})}
\newcommand{\Csap}[1]{\mathbf{H}_{#1}^{+}(\mathbb{C})}

\newcommand{\m}[1]{\mathcal{#1}}

\DeclareMathOperator{\conv}{conv}
\DeclareMathOperator{\diag}{diag}
\DeclareMathOperator{\linspan}{span}

\usepackage{amsthm}
\usepackage{mathtools}
\theoremstyle{plain}
\newtheorem{theorem}{Theorem}
\newtheorem{lemma}[theorem]{Lemma}
\newtheorem{corollary}[theorem]{Corollary}
\newtheorem{definition}{Definition}

\newtheorem{example}{Example}

\newtheorem{conjecture}{Conjecture}
\newtheorem*{remark*}{Remark}
\newtheorem*{theorem*}{Theorem}
\newtheorem{proposition}{Proposition}

\newcommand{\red}{\color{red}}
\newcommand{\david}{\color{blue}}
\newcommand{\baldi}{\color{magenta}}
\newcommand{\baldiComment}[1]{{\color{magenta}[Baldi: #1]}}
\newcommand{\marco}[1]{{\color{teal}#1}}
\newcommand{\john}{\color{red}}
\newcommand{\jnote}[1]{{\color{red}[JHS: #1]}}
\newcommand{\blk}{\color{black}}
\newcommand{\blu}{\color{blue}}
\newcommand{\bel}{\color{orange}}

\usepackage{soul}
\usepackage{setspace}
\usepackage{framed}
    \definecolor{shadecolor}{gray}{0.8}
\newcommand{\rdbaldi}[1]{{\color{purple}{[#1]}}}

\begin{document}

\title{Tomographically-nonlocal entanglement}

\author{Roberto D. Baldij\~ao}
\email{rbaldijao@perimeterinstitute.ca}
\affiliation{Perimeter Institute for Theoretical Physics, 31 Caroline Street North, Waterloo, Ontario Canada N2L 2Y5}
\affiliation{International Centre for Theory of Quantum Technologies, Uniwersytet Gdański, ul.~Jana Bażyńskiego 1A, 80-309 Gdańsk, Poland}
\author{Marco Erba}
\affiliation{International Centre for Theory of Quantum Technologies, Uniwersytet Gdański, ul.~Jana Bażyńskiego 1A, 80-309 Gdańsk, Poland}
\author{David Schmid}
\thanks{These authors share last authorship}
\affiliation{Perimeter Institute for Theoretical Physics, 31 Caroline Street North, Waterloo, Ontario Canada N2L 2Y5}
\author{John H. Selby}
\thanks{These authors share last authorship}
\affiliation{International Centre for Theory of Quantum Technologies, Uniwersytet Gdański, ul.~Jana Bażyńskiego 1A, 80-309 Gdańsk, Poland}
\affiliation{Theoretical Sciences Visiting Program, Okinawa Institute of Science and Technology Graduate University, Onna, 904-0495, Japan}
\author{Ana Bel\'{e}n Sainz}
\thanks{These authors share last authorship}
\affiliation{International Centre for Theory of Quantum Technologies, Uniwersytet Gdański, ul.~Jana Bażyńskiego 1A, 80-309 Gdańsk, Poland}
\affiliation{Theoretical Sciences Visiting Program, Okinawa Institute of Science and Technology Graduate University, Onna, 904-0495, Japan}
\affiliation{Basic Research Community for Physics e.V., Germany}

\begin{abstract}
Entanglement is a central and subtle feature of quantum theory, whose structure and operational behavior can change dramatically when additional physical constraints, such as symmetries or superselection rules, are imposed. Such constraints can give rise to striking and counter-intuitive phenomena, including local broadcasting of entangled states and failures of entanglement monogamy. These effects naturally arise in tomographically nonlocal theories (like real quantum theory, twirled worlds, or fermionic quantum theory), where composite systems possess holistic degrees of freedom that are inaccessible to local measurements. In this work, we study entanglement in such theories within the framework of generalized probabilistic theories. We show that the failure of tomographic locality leads to two qualitatively distinct forms of entanglement, which we term \emph{tomographically-local} entanglement and \emph{tomographically-nonlocal} entanglement. We analyze the operational consequences of this distinction, proving that tomographically-nonlocal entanglement is useless for Bell nonlocality, steering, and teleportation, but sufficient for dense coding and perfectly secure data hiding. This framework clarifies the origin of several previously puzzling features of entanglement that arise when tomographic locality fails, as can happen even in quantum theory when one considers fermions or fundamental superselection rules.

\end{abstract}

\maketitle

\tableofcontents

\section{Introduction}

The framework of generalized probabilistic theories (GPTs)~\cite{barrett_informationGPTs_2006,mullerGPTnotes,Plavala_2023_GPTsIntro, hardy2011reformulatingreconstructingquantumtheory,hardy2001quantum,Chiribella_QuantumFromPrinciples2016} gives a rigorous approach to studying physical theories in terms of their operational features---how systems are prepared, transformed, and measured, and how probabilities are assigned to outcomes in circuits composed of such processes. GPTs encompass both classical and quantum theories while abstracting away from specific mathematical structures such as Hilbert spaces or phase spaces. This flexible framework allows us to compare different physical theories on common ground, to explore structural alternatives to quantum theory, and to better understand which features of physical theories are genuinely quantum. By doing so, the GPT framework provides insights into the foundations of physics~\cite{Mazurek_2021,Schmid2024structuretheorem}, sheds light on the nature of information and computation~\cite{barrett_informationGPTs_2006,Mueller_2012}, and helps clarify which physical principles might explain why our world is quantum rather than something else~\cite{hardy2011reformulatingreconstructingquantumtheory,Chiribella_2011InfoDerivationQT}.

Many prior works have focused on GPTs satisfying the property of \emph{tomographic locality}—the property that one can tomographically characterize every multipartite process in the GPT via local states and effects. This is a strong assumption that is often made to simplify the framework and facilitate the derivation of results. Well-known GPTs satisfying this property include classical theory (and hence macroscopically realist theories~\cite{Schmid2024reviewreformulation}), Spekkens toy theory~\cite{toytheory}, the stabilizer subtheory~\cite{gottesman1998heisenbergrepresentationquantumcomputers}, Boxworld~\cite{barrett_informationGPTs_2006}, and—perhaps most famously—unrestricted (or \emph{bare}) quantum theory. By unrestricted (or bare) quantum theory, we mean the operational theory of quantum systems without superselection rules or symmetry-induced restrictions. However, once such restrictions are imposed, tomographic locality generally fails~\cite{Dariano2014feynmanproblem}. This is the case, for example, in real quantum theory~\cite{Chiribella_QuantumFromPrinciples2016,Wooter_1990} (RQT) and in fermionic quantum theory~\cite{darianoFermionic2014,Dariano2014feynmanproblem}, both of which arise as restricted versions of bare quantum theory. In this sense, if superselection rules are fundamental, understanding the resulting failure of tomographic locality becomes essential. Indeed, there has been growing interest in understanding tomographically-nonlocal theories more generally. Some examples of tomographically-nonlocal theories include twirled worlds~\cite{centeno2024twirledworldssymmetryinducedfailures}, swirled worlds~\cite{ying2025quantumtheoryneedscomplex}, quaternionic quantum theory~\cite{hardy2001quantum}, modifications of classical theories~\cite{PhysRevA.101.042118,d2020classicality,scandolo2019information,Chiribella2024,Rolino_MOPTs_2025,soltani2025decouplinglocalclassicalityclassical} and of quantum theory~\cite{erba2025compositionrulequantumsystems}, and a beyond-quantum theory constructed from Euclidean Jordan algebras~\cite{barnum2020composites}.

Tomographically-nonlocal theories exhibit a wide array of distinctive features, even in comparison to unrestricted quantum theory. In such theories, composite systems possess \emph{holistic} degrees of freedom. Operationally, failures of local tomography imply the existence of distinct multipartite states (and, more generally, processes) that cannot be distinguished by any choice of local measurements (dual processes).

There are several motivations for studying such tomographically-nonlocal (TNL) theories. First, there is no consensus on whether the fundamental theory describing our world is tomographically local. In particular, even within quantum theory, there remains debate about whether certain superselection rules are fundamental~\cite{Wick_SymmetrisFundamental,giulini2009superselectionrulesFundamental}. As discussed above, if such restrictions are fundamental, then the resulting restricted quantum theory will generally violate tomographic locality~\cite{centeno2024twirledworldssymmetryinducedfailures}. Second, even if constraints such as superselection rules are merely effective or emergent, they can nonetheless lead to effective failures of tomographic locality (and hence effective holism), making it important to understand the operational consequences of such constraints. Third, tomographically-nonlocal theories serve as valuable \emph{foil theories} for tomographically-local ones, allowing us to sharpen our understanding of which operational and informational features genuinely rely on tomographic locality~\cite{Chiribella_QuantumFromPrinciples2016,toytheory}.

In this work, we study how the tomographic nonlocality of a theory affects the forms of entanglement that can exist within it. In tomographically-nonlocal theories, entanglement behaves differently from what is familiar from unrestricted quantum theory, exhibiting features that are often considered counter-intuitive both mathematically and operationally. For example, in some TNL theories (including RQT and fermionic quantum theory) there exist maximally entangled states that are non-monogamous and locally broadcastable~\cite{PianiNoLocalBroadcasting_2008,Weilenmann_2025,Dariano2014feynmanproblem}. There are also maximally entangled states that do not violate any Bell inequality, and some tomographically-nonlocal theories display entanglement despite the fact that every system has a state space identical to that of a classical system~\cite{PhysRevA.101.042118}.

We show that the failure of tomographic locality naturally leads to two qualitatively distinct forms of entanglement. Concretely, we distinguish between \emph{tomographically-local entanglement} and \emph{tomographically-nonlocal entanglement}. Tomographically-nonlocal entanglement is only possible in tomographically-nonlocal theories, while tomographically-local entanglement is possible in both tomographically-local and tomographically-nonlocal theories. Both forms may be present within a single state. This distinction allows us to separate the part of entanglement that coincides with that found in unrestricted quantum theory (i.e., the part accessible to local tomography) from the genuinely holistic contribution that resides entirely in the tomographically-nonlocal sector.

To study these ideas, we introduce a set of mathematical tools for analyzing tomographic locality more broadly. In particular, we introduce projectors that single out the holistic and non-holistic degrees of freedom of composite systems, and show how these projectors can be used to decompose states and effects into locally tomographic and purely holistic components. This provides a systematic way to identify which part of a state’s entanglement is accessible to local tomography, and to isolate the degrees of freedom responsible for tomographic nonlocality.

We then analyze when and how tomographically-nonlocal entanglement is useful for information-processing tasks. On the one hand, we prove that tomographically-nonlocal entanglement is \emph{entirely useless} for generating nonclassical correlations in Bell scenarios~\cite{Bell1964,Bellreview2014}, for steering~\cite{steeringwiseman}, and for teleportation~\cite{PhysRevLett.119.110501}. This is in stark contrast to tomographically-local entanglement, which is useful for all three tasks in unrestricted quantum theory. On the other hand, we show that tomographically-nonlocal entanglement \emph{is} sufficient for dense coding and data hiding in a broad class of GPTs, including the canonical example of real quantum theory. In particular, within RQT, we show that tomographically-nonlocal entanglement enables dense coding and data hiding even when one demands that secrets be encoded and decoded locally, and we discuss how these results generalize beyond RQT.

These distinct notions of entanglement clarify several previously puzzling observations regarding entanglement in real quantum theory. For example, it has been noted that \emph{some}, but not \emph{all}, entangled states in real quantum theory can be locally broadcast~\cite{Weilenmann_2025}, and similarly that \emph{some}, but not \emph{all}, can be infinitely shared (i.e., violate entanglement monogamy)~\cite{Wootters_2010}. We show that, at least in the two-rebits case, the states exhibiting these features are precisely those that lack tomographically-local entanglement—the form of entanglement that exists in unrestricted quantum theory—and are instead entangled entirely in the tomographically-nonlocal sector. We illustrate these ideas with concrete examples, primarily in rebit theory.

The mathematical tools introduced here provide a first step toward a systematic study of this under-appreciated class of generalized probabilistic theories. As emphasized above, this framework is useful not only within the methodology of foil theories, but also for understanding situations in which superselection rules or symmetry constraints lead quantum theory to be genuinely or effectively tomographically nonlocal.

\section{Preliminaries}

We give a brief introduction to the framework of GPTs, primarily to set up notation. A reader with no  background in \blk GPTs is encouraged to read Refs.~\cite{barrett_informationGPTs_2006,mullerGPTnotes} for a more comprehensive introduction. This work also focuses on state and effect spaces rather than fully compositional theories (including transformations), and focuses on full GPT state and effect spaces rather than on fragments thereof~\cite{PhysRevA.107.062203}.

We will consider the GPT description of a subsystem $\mathcal{A}$, a subsystem $\mathcal{B}$, and their composite $\mathcal{AB}$. Each system is characterized by a state space and an effect space, defined with respect to a real- and finite-dimensional vector space.\footnote{In fact, it is possible to show that one can start with the state and effect spaces and obtain the vector space from a universal construction, which is  a consequence of the possibility of combining them to get probabilities (see, eg, Ref.~\cite[\S1.1.2\&1.7.1]{lami2018nonclassicalcorrelationsquantummechanics}, and Refs \cite{Schmid2024structuretheorem,erba2026categorical}).} For instance, for subsystem $\mathcal{A}$, we associate a vector space $A$. Then, states are vectors in $A$, that together form a state space $S_{\m{A}}$ contained in $A$, $S_\mathcal{A}\subset A$. The effect space $E_{\mathcal A}$ are vectors that live in the dual of $A$, $E_{\m{A}}\subset A^*$. They must assign probabilities when acting on states, so we require $0\leq e(\omega)\leq 1$ for all $e\in E_{\mathcal{A}}$ and $\omega\in S_{\mathcal{A}}$. Both $S_{\mathcal{A}}$ and $E_{\mathcal A}$ must carry some extra structure. First, both $S_{\mathcal A}$ and $E_{\mathcal A}$ are convex sets and contain the zero vector of their respective vector spaces. Secondly, $E_{\mathcal{A}}$ has a unique {\em unit} effect $u^{\m A}\in E_{\m A}$ such that, for all $e\in E_{\m A}$,
$u^{\m A}\geq e$. (Here, $e\geq f$ iff there exists $g\in E^{\m A}$ such that $e=f+g$. This captures the fact that $e(\omega)\geq f(\omega), \forall \omega\in S_{\m A}$.) The element $u^{\m A}\in E_{\m A}$ gives us the definition of normalized states and of measurements: 

    \begin{itemize}
        \item Normalized states are those $\omega^{\m{A}}\in S_{\m A}$ that satisfy $u^{\m A}(\omega)=1$; we denote the convex set of normalized states by $\Omega^{\m A}$. The uniqueness of $u^{\m A}$ implies that every state in $S_{\m A}$ can be written as $p\omega$ for some $\omega\in\Omega^{\m A}$ and $p\in[0,1]$~\cite{Chiribella_2010} . Therefore,  the full state space $S_{\m A}$, that contains both normalized and subnormalized states, is given by $\mathcal{S}^\mathcal{A}=\mathsf{ConvHull}[0,\Omega^\mathcal{A}]$.
        \item Measurements are collections $\{e_i\}_i$ such that $\sum_i e_i = u^{\m A}$; we assume that all effects are part of a measurement; that is, for all valid effect $e\in E_{\m A}$, $(u^{\m A}-e)\in E_{\m A}$.
    \end{itemize}

    Finally, a GPT system must obey \emph{tomography}, i.e.: given any pair of distinct states $s,s'$, there exists an effect $e$ which assign different probabilities to them, $e(s)\neq e(s')$; and  given any pair of distinct effects $e,e'$, there exists a state $s$ which assign different probabilities to them, $e(s)\neq e'(s)$.

The theory is said to obey the non-restriction hypothesis iff $E^{\m A}$ contains all elements $e\in A^*$ such that $0\leq e(\omega)\leq 1$ for all $\omega\in\Omega$.  In this manuscript, we will not assume that theories satisfy this property, and we will mention it explicitly if ever relevant.
We will follow the standard convention that ${\rm Span}(E_{\m{A}})=A^*$ (which applies even if the theory does not satisfy no-restriction), although this can be relaxed~\cite{schmid2024shadowssubsystemsgeneralizedprobabilistic}.

In this work, we are going to use a diagrammatic notation, in which we represent a state $\omega$ and an effect $e$ of system $\m A$ (respectively) as 

\begin{equation}
\label{eq:ExampleStatesEffects}
\vcenter{\hbox{\begin{tikzpicture}
  \begin{pgfonlayer}{nodelayer}

    \node [style=none] (L0) at (-2.2,0) {};      
    \node [style=none] (L1) at (-1.0,0) {};      
    \node [style=none] (L2) at (-1.6,-0.9) {};   

    \node [style=none] (Lomega) at (-1.6,-0.35) {$\omega$};

    \node [style=none] (LTwL) at (-1.66,0) {};
    \node [style=none] (LTwR) at (-1.54,0) {};

    \node [style=none] (LA) at (-1.30,0.45) {$\scriptstyle \m{A}$};

    \node [style=none] (R0) at (1.0,0) {};       
    \node [style=none] (R1) at (2.2,0) {};       
    \node [style=none] (R2) at (1.6,0.9) {};     

    \node [style=none] (Re) at (1.6,0.45) {$e$};

    \node [style=none] (RTwL) at (1.54,0) {};
    \node [style=none] (RTwR) at (1.66,0) {};

    \node [style=none] (RA) at (1.90,-0.45) {$\scriptstyle \m{A}$};

  \end{pgfonlayer}

  \begin{pgfonlayer}{edgelayer}

    \draw (L0.center) to (L1.center);
    \draw (L1.center) to (L2.center);
    \draw (L2.center) to (L0.center);

    \draw[line width=0.6pt] (LTwL.center) to +(0,0.9);
    \draw[line width=0.6pt] (LTwR.center) to +(0,0.9);

    \draw (R0.center) to (R1.center);
    \draw (R1.center) to (R2.center);
    \draw (R2.center) to (R0.center);

    \draw[line width=0.6pt] (RTwL.center) to +(0,-0.9);
    \draw[line width=0.6pt] (RTwR.center) to +(0,-0.9);

  \end{pgfonlayer}
\end{tikzpicture}
}}
\end{equation}

As a familiar example of a GPT system, consider a quantum system $\mathcal{Q}$ with $d$ levels (and no superselection rules). It admits a GPT description in which the set of normalized states $\Omega_{\m Q}$ consists of $d\times d$ density matrices, i.e., positive semidefinite operators with unit trace. The full state space $\mathcal{S}^{\m Q}$ is then given by the positive semidefinite operators with trace less than or equal to~$1$. Effects are linear functionals of the form ${\rm Tr}[M\,\cdot]$, where $M$ is a POVM element, and the unit effect is $u^{\m Q}:={\rm Tr}[\mathds{1}\,\cdot]$. One can directly verify that these objects satisfy the GPT requirements above: the state and effect spaces are convex, $\mathcal{S}^{\m Q}=\mathsf{ConvHull}[0,\Omega^{\m Q}]$, the unit effect is unique, and the complement of any effect is again a valid effect.

It is worth highlighting that, while quantum theory is usually formulated in terms of a $d$-dimensional complex Hilbert space $\m H_d$, the GPT description does not make direct reference to this structure. Instead, the relevant GPT vector space ${A}^{\m{Q}}$ is (isomorphic to) the $d^2$-dimensional real vector space $\mathbf{H}_d$ of Hermitian operators acting on $\m H_d$. Thus, the abstract GPT vector space $A:=\mathsf{Span}[S_{\m A}]$ should be understood as analogous to $\mathbf{H}_d$, rather than to the underlying Hilbert space itself as in general GPTs, no underlying Hilbert space is assumed or required.

In summary, the GPT description of a system $\mathcal{A}$ is a tuple $\left(A,S_{\mathcal A},E_{\mathcal A},u^{\mathcal A}\right)$. We will often refer to subsystems and their composite, say
subsystem $\mathcal{B}$ and composite system $\mathcal{AB}$, and those are associated with analogous tuples. However, the tuples for the composite system $\mathcal{AB}$ must be related to those of the constituent subsystems $\mathcal{A}$ and $\mathcal{B}$, as we now describe.

\subsection{Composite GPT systems}

Consider the composite system formed by a pair of subsystems $\mathcal{A}$ and $\mathcal{B}$. One can describe the composite as one global system $\mathcal{AB}$. Since $\mathcal{AB}$ is itself a system, it is defined by a tuple $(AB,S_{\m{AB}},E_{\m{AB}},u^{\m{AB}})$. For $\mathcal{AB}$ to be a composite of $\mathcal{A}$ and $\mathcal{B}$, the elements of this tuple must satisfy some conditions (relative to the tuples defining $\mathcal{A}$ and $\mathcal{B}$).
For instance, one can always view two independent preparations described by $\omega^{\m A}\in S_{\m A}$ and by $\nu^{\m B}\in S_{\m B}$ as a state in $S_{\m{AB}}$, and similarly, a pair of independent measurement effects should correspond to some effect in $E_{\m{AB}}$.

With this in mind, we can define a bipartite \blk composite system $\mathcal{AB}$ as follows.

\begin{definition}[Composite GPT systems] 
\label{def: CompositionRequirements}
The composition of a pair of (sub)systems $\mathcal{A}$ and $\mathcal{B}$ is defined by a system $\mathcal{AB}=(AB,S_{\m{AB}},E_{\m{AB}},u^{\m{AB}})$, together with two bilinear maps $A\times B\to AB$ and $ A^*\times B^*\to (AB)^*$. We denote the action of these maps on $(s^A,s^B)\in A\times B$ and $(e^{\m{A}},e^{\m{B}})\in A^*\times B^*$ via $s^A\boxtimes s^B$ and $e^{\m{A}} \boxtimes e^{\m{B}}$, respectively. $\mathcal{AB}$ and the linear maps denoted by $\boxtimes$ should satisfy the following:
\begin{enumerate}
    \item Independent states (effects), represented by product states (effects), are valid composite states (effects): \begin{itemize}
        \item $\omega^{\m{A}}\boxtimes\nu^{\m{B}}\in S_{\m{AB}}$ for all $\omega^{\m{A}}\in S_{A}$ and $\nu^{\m{B}}\in S_{B}$,
        \item $e^{\m{A}}\boxtimes f^{\m{B}}\in E_{\m{AB}}$ for all $e^{\m{A}}\in E_{\m{A}}$ and $f^{\m{B}}\in E_{\m{B}}$;
    \end{itemize}
    Moreover, independent experiments lead to independent statistics: $(e^{\m{A}}\boxtimes f^{\m{B}})[\omega^{\m{A}}\boxtimes\nu^{\m{B}}]=e^{\m{A}}(\omega^{\m{A}})f^{\m{B}}(\nu^{\m{B}})$.
    \item The unit effect of the composite is the composite of the local unit effects: $u^{\m{A}}\boxtimes u^B=u^{\m{AB}}$;
    \item Closure under steering (or validity of conditional states/effects): 
    
    \begin{itemize}
        
\item For every local effect $\bar{e}^B\in E_{\m{B}}$ and global state $\omega^{\m{AB}}$, the bilinear map  $E_{\m{B}}\times S_{\m{AB}}\to A$ defined by ${\tilde{\omega}^A:= (\bullet)\boxtimes\bar{e}^B[\omega^{\m{AB}}]}$ outputs a local state of $A$, i.e., $\tilde{\omega}\in S_{A}$. A similar condition is valid with the change $A\leftrightarrow B$;
        \item For every local state $\bar{\nu}^B\in E_{\m{B}}$ and global effect $e^{\m{AB}}$, the bilinear map $S_{B}\times E_{\m{AB}}\to A$ defined by ${\tilde{e}^A:= e^{\m{AB}}[(\bullet\boxtimes \bar{\nu}^B)]}$ outputs a local effect of $A$, i.e., $\tilde{e}^A\in E_{\m{A}}$. A similar condition is valid with the change $A\leftrightarrow B$;
    \end{itemize}
\end{enumerate}

\end{definition}

Diagrammatically, we represent independent state preparations (effects) by 

\begin{equation}
\label{eq:ExampleProductStatesEffects}
\vcenter{\hbox{
\begin{tikzpicture}
  \begin{pgfonlayer}{nodelayer}


    \node [style=none] (LoBL) at (0.70,0.45) {};    
    \node [style=none] (LoBR) at (2.20,0.45) {};    
    \node [style=none] (LoT)  at (1.45,-0.50) {};   
    \node [style=none] (LoLab) at (1.45,0.00) {$\omega$};

    \node [style=none] (LoWL) at (1.38,0.45) {};
    \node [style=none] (LoWR) at (1.52,0.45) {};

    \node [style=none] (LoWireLab) at (1.75,1.00) {$\scriptstyle \m{A}$};

    \node [style=none] (LnBL) at (2.40,0.45) {};
    \node [style=none] (LnBR) at (3.90,0.45) {};
    \node [style=none] (LnT)  at (3.15,-0.50) {};   
    \node [style=none] (LnLab) at (3.15,0.00) {$\nu$};

    \node [style=none] (LnWL) at (3.08,0.45) {};
    \node [style=none] (LnWR) at (3.22,0.45) {};

    \node [style=none] (LnWireLab) at (3.45,1.00) {$\scriptstyle \m{B}$};

    \node [style=none] (LinL) at (4.65,0.00) {$\in$};
    \node [style=none] (Lset) at (6.05,0.00) {$\m{S_{AB}}$};


    \node [style=none] (ReBL) at (8.10,0.60) {};
    \node [style=none] (ReBR) at (9.60,0.60) {};
    \node [style=none] (ReT)  at (8.85,1.55) {};
    \node [style=none] (ReLab) at (8.85,1.05) {$e$};

    \node [style=none] (ReWL) at (8.78,0.60) {};
    \node [style=none] (ReWR) at (8.92,0.60) {};

    \node [style=none] (ReWireLab) at (9.15,0.00) {$\scriptstyle \m{A}$};

    \node [style=none] (RfBL) at (9.80,0.60) {};
    \node [style=none] (RfBR) at (11.30,0.60) {};
    \node [style=none] (RfT)  at (10.55,1.55) {};
    \node [style=none] (RfLab) at (10.55,1.05) {$f$};

    \node [style=none] (RfWL) at (10.48,0.60) {};
    \node [style=none] (RfWR) at (10.62,0.60) {};

    \node [style=none] (RfWireLab) at (10.85,0.00) {$\scriptstyle \m{B}$};

    \node [style=none] (Rin)  at (12.05,0.00) {$\in$};
    \node [style=none] (Rset) at (13.70,0.00) {$E_{\m{AB}}$};

  \end{pgfonlayer}

  \begin{pgfonlayer}{edgelayer}

    \draw (LoBL.center) to (LoBR.center);
    \draw (LoBR.center) to (LoT.center);
    \draw (LoT.center)  to (LoBL.center);

    \draw[line width=0.6pt] (LoWL.center) to +(0,1.10);
    \draw[line width=0.6pt] (LoWR.center) to +(0,1.10);

    \draw (LnBL.center) to (LnBR.center);
    \draw (LnBR.center) to (LnT.center);
    \draw (LnT.center)  to (LnBL.center);

    \draw[line width=0.6pt] (LnWL.center) to +(0,1.10);
    \draw[line width=0.6pt] (LnWR.center) to +(0,1.10);

    \draw (ReBL.center) to (ReBR.center);
    \draw (ReBR.center) to (ReT.center);
    \draw (ReT.center)  to (ReBL.center);

    \draw[line width=0.6pt] (ReWL.center) to +(0,-1.10);
    \draw[line width=0.6pt] (ReWR.center) to +(0,-1.10);

    \draw (RfBL.center) to (RfBR.center);
    \draw (RfBR.center) to (RfT.center);
    \draw (RfT.center)  to (RfBL.center);

    \draw[line width=0.6pt] (RfWL.center) to +(0,-1.10);
    \draw[line width=0.6pt] (RfWR.center) to +(0,-1.10);

 \node[style=none] (,)   at (6.8, -0.26) {,};
    
  \end{pgfonlayer}
\end{tikzpicture},
}}
\end{equation}
and the conditions on composites translate to:

    \begin{enumerate}
        \item Condition on independent states, effects and their probabilities: \begin{equation}
\label{eq:CompositionProductStuff}
\vcenter{\hbox{\begin{tikzpicture}
  \begin{pgfonlayer}{nodelayer}


    \node [style=none] (LoBL) at (0.70,0.45) {};    
    \node [style=none] (LoBR) at (2.20,0.45) {};    
    \node [style=none] (LoT)  at (1.45,-0.50) {};   
    \node [style=none] (LoLab) at (1.45,0.00) {$\omega$};

    \node [style=none] (LoWL) at (1.38,0.45) {};
    \node [style=none] (LoWR) at (1.52,0.45) {};

    \node [style=none] (LoWireLab) at (1.75,1.00) {$\scriptstyle \m{A}$};

    \node [style=none] (LnBL) at (2.40,0.45) {};
    \node [style=none] (LnBR) at (3.90,0.45) {};
    \node [style=none] (LnT)  at (3.15,-0.50) {};   
    \node [style=none] (LnLab) at (3.15,0.00) {$\nu$};

    \node [style=none] (LnWL) at (3.08,0.45) {};
    \node [style=none] (LnWR) at (3.22,0.45) {};

    \node [style=none] (LnWireLab) at (3.45,1.00) {$\scriptstyle \m{B}$};

    \node [style=none] (LinL) at (4.65,0.00) {$\in$};
    \node [style=none] (Lset) at (6.05,0.00) {$\m{S_{AB}}$};


    \node [style=none] (ReBL) at (8.10,0.60) {};
    \node [style=none] (ReBR) at (9.60,0.60) {};
    \node [style=none] (ReT)  at (8.85,1.55) {};
    \node [style=none] (ReLab) at (8.85,1.05) {$e$};

    \node [style=none] (ReWL) at (8.78,0.60) {};
    \node [style=none] (ReWR) at (8.92,0.60) {};

    \node [style=none] (ReWireLab) at (9.15,0.00) {$\scriptstyle \m{A}$};

    \node [style=none] (RfBL) at (9.80,0.60) {};
    \node [style=none] (RfBR) at (11.30,0.60) {};
    \node [style=none] (RfT)  at (10.55,1.55) {};
    \node [style=none] (RfLab) at (10.55,1.05) {$f$};

    \node [style=none] (RfWL) at (10.48,0.60) {};
    \node [style=none] (RfWR) at (10.62,0.60) {};

    \node [style=none] (RfWireLab) at (10.85,0.00) {$\scriptstyle \m{B}$};

    \node [style=none] (Rin)  at (12.05,0.00) {$\in$};
    \node [style=none] (Rset) at (13.70,0.00) {$E_{\m{AB}}$};

  \end{pgfonlayer}

  \begin{pgfonlayer}{edgelayer}

    \draw (LoBL.center) to (LoBR.center);
    \draw (LoBR.center) to (LoT.center);
    \draw (LoT.center)  to (LoBL.center);

    \draw[line width=0.6pt] (LoWL.center) to +(0,1.10);
    \draw[line width=0.6pt] (LoWR.center) to +(0,1.10);

    \draw (LnBL.center) to (LnBR.center);
    \draw (LnBR.center) to (LnT.center);
    \draw (LnT.center)  to (LnBL.center);

    \draw[line width=0.6pt] (LnWL.center) to +(0,1.10);
    \draw[line width=0.6pt] (LnWR.center) to +(0,1.10);

    \draw (ReBL.center) to (ReBR.center);
    \draw (ReBR.center) to (ReT.center);
    \draw (ReT.center)  to (ReBL.center);

    \draw[line width=0.6pt] (ReWL.center) to +(0,-1.10);
    \draw[line width=0.6pt] (ReWR.center) to +(0,-1.10);

    \draw (RfBL.center) to (RfBR.center);
    \draw (RfBR.center) to (RfT.center);
    \draw (RfT.center)  to (RfBL.center);

    \draw[line width=0.6pt] (RfWL.center) to +(0,-1.10);
    \draw[line width=0.6pt] (RfWR.center) to +(0,-1.10);

 \node[style=none] (,)   at (6.8, -0.26) {,};
    
  \end{pgfonlayer}
\end{tikzpicture}, $\,\,$ with$\,\,$\begin{tikzpicture}
  \begin{pgfonlayer}{nodelayer}

    \node[style=none] (eq)   at ( 1.10, 0.55) {$=$};


    \node[style=none] (LeBL) at (-2.20,1.00) {};
    \node[style=none] (LeBR) at (-1.00,1.00) {};
    \node[style=none] (LeT)  at (-1.60,1.95) {};
    \node[style=none] (Le)   at (-1.60,1.40) {$e$};

    \node[style=none] (LwTL) at (-2.20,0.20) {};
    \node[style=none] (LwTR) at (-1.00,0.20) {};
    \node[style=none] (LwB)  at (-1.60,-0.75) {};
    \node[style=none] (Lw)   at (-1.60,-0.20) {$\omega$};

    \node[style=none] (LAuL) at (-1.67,1.00) {};
    \node[style=none] (LAuR) at (-1.53,1.00) {};
    \node[style=none] (LAdL) at (-1.67,0.20) {};
    \node[style=none] (LAdR) at (-1.53,0.20) {};

    \node[style=none] (LA)   at (-1.25,0.60) {$\scriptstyle \m{A}$};

    \node[style=none] (LfBL) at (-1.05,1.00) {};
    \node[style=none] (LfBR) at ( 0.15,1.00) {};
    \node[style=none] (LfT)  at (-0.45,1.95) {};
    \node[style=none] (Lf)   at (-0.45,1.40) {$f$};

    \node[style=none] (LnTL) at (-1.05,0.20) {};
    \node[style=none] (LnTR) at ( 0.15,0.20) {};
    \node[style=none] (LnB)  at (-0.45,-0.75) {};
    \node[style=none] (Ln)   at (-0.45,-0.20) {$\nu$};

    \node[style=none] (LBuL) at (-0.52,1.00) {};
    \node[style=none] (LBuR) at (-0.38,1.00) {};
    \node[style=none] (LBdL) at (-0.52,0.20) {};
    \node[style=none] (LBdR) at (-0.38,0.20) {};

    \node[style=none] (LB)   at (-0.05,0.60) {$\scriptstyle \m{B}$};


    \node[style=none] (ReBL) at (2.00,1.00) {};
    \node[style=none] (ReBR) at (3.20,1.00) {};
    \node[style=none] (ReT)  at (2.60,1.95) {};
    \node[style=none] (Re)   at (2.60,1.40) {$e$};

    \node[style=none] (RwTL) at (2.00,0.20) {};
    \node[style=none] (RwTR) at (3.20,0.20) {};
    \node[style=none] (RwB)  at (2.60,-0.75) {};
    \node[style=none] (Rw)   at (2.60,-0.20) {$\omega$};

    \node[style=none] (RAuL) at (2.53,1.00) {};
    \node[style=none] (RAuR) at (2.67,1.00) {};
    \node[style=none] (RAdL) at (2.53,0.20) {};
    \node[style=none] (RAdR) at (2.67,0.20) {};

    \node[style=none] (RA)   at (2.95,0.60) {$\scriptstyle \m{A}$};

    \node[style=none] (RfBL) at (3.45,1.00) {};
    \node[style=none] (RfBR) at (4.65,1.00) {};
    \node[style=none] (RfT)  at (4.05,1.95) {};
    \node[style=none] (Rf)   at (4.05,1.40) {$f$};

    \node[style=none] (RnuTL) at (3.45,0.20) {};
    \node[style=none] (RnuTR) at (4.65,0.20) {};
    \node[style=none] (RnuB)  at (4.05,-0.75) {};
    \node[style=none] (Rnu)   at (4.05,-0.20) {$\nu$};

    \node[style=none] (RBuL) at (3.98,1.00) {};
    \node[style=none] (RBuR) at (4.12,1.00) {};
    \node[style=none] (RBdL) at (3.98,0.20) {};
    \node[style=none] (RBdR) at (4.12,0.20) {};

    \node[style=none] (RB)   at (4.40,0.60) {$\scriptstyle \m{B}$};

  \end{pgfonlayer}

  \begin{pgfonlayer}{edgelayer}

    \draw (LeBL.center) to (LeBR.center);
    \draw (LeBR.center) to (LeT.center);
    \draw (LeT.center)  to (LeBL.center);
    \draw (LwTL.center) to (LwTR.center);
    \draw (LwTR.center) to (LwB.center);
    \draw (LwB.center)  to (LwTL.center);

    \draw (LfBL.center) to (LfBR.center);
    \draw (LfBR.center) to (LfT.center);
    \draw (LfT.center)  to (LfBL.center);
    \draw (LnTL.center) to (LnTR.center);
    \draw (LnTR.center) to (LnB.center);
    \draw (LnB.center)  to (LnTL.center);

    \draw[line width=0.6pt] (LAuL.center) to (LAdL.center);
    \draw[line width=0.6pt] (LAuR.center) to (LAdR.center);
    \draw[line width=0.6pt] (LBuL.center) to (LBdL.center);
    \draw[line width=0.6pt] (LBuR.center) to (LBdR.center);

    \draw[draw=blue!70!black, dashed, line width=0.8pt, rounded corners=0pt]
      (-2.35,0.90) rectangle (0.30,2.05);
    \draw[draw=blue!70!black, dashed, line width=0.8pt, rounded corners=0pt]
      (-2.35,-0.85) rectangle (0.30,0.35);

    \draw (ReBL.center) to (ReBR.center);
    \draw (ReBR.center) to (ReT.center);
    \draw (ReT.center)  to (ReBL.center);
    \draw (RwTL.center) to (RwTR.center);
    \draw (RwTR.center) to (RwB.center);
    \draw (RwB.center)  to (RwTL.center);

    \draw (RfBL.center) to (RfBR.center);
    \draw (RfBR.center) to (RfT.center);
    \draw (RfT.center)  to (RfBL.center);
    \draw (RnuTL.center) to (RnuTR.center);
    \draw (RnuTR.center) to (RnuB.center);
    \draw (RnuB.center)  to (RnuTL.center);

    \draw[line width=0.6pt] (RAuL.center) to (RAdL.center);
    \draw[line width=0.6pt] (RAuR.center) to (RAdR.center);
    \draw[line width=0.6pt] (RBuL.center) to (RBdL.center);
    \draw[line width=0.6pt] (RBuR.center) to (RBdR.center);

    \draw[draw=orange!80!black, dotted, line width=0.8pt, rounded corners=0.1pt]
      (1.85,-0.85) rectangle (3.34,2.05);
    \draw[draw=orange!80!black, dotted, line width=0.8pt, rounded corners=0.1pt]
      (3.30,-0.85) rectangle (4.81,2.05);

  \end{pgfonlayer}
\end{tikzpicture}
}}
\end{equation}
\item Condition on global unit effect: \begin{equation}
\label{eq:GlobalUnitEffect}
\vcenter{\hbox{\begin{tikzpicture}[baseline={(eq.base)}]
  \begin{pgfonlayer}{nodelayer}


    \node[style=none] (LA_L)  at (-2.10,0.18) {};
    \node[style=none] (LA_R)  at (-1.98,0.18) {};
    \node[style=none] (LA_Lb) at (-2.10,-1.20) {};
    \node[style=none] (LA_Rb) at (-1.98,-1.20) {};

    \node[style=none] (LB_L)  at (-1.20,0.18) {};
    \node[style=none] (LB_R)  at (-1.08,0.18) {};
    \node[style=none] (LB_Lb) at (-1.20,-1.20) {};
    \node[style=none] (LB_Rb) at (-1.08,-1.20) {};

    \node[style=none] (LabA) at (-1.78,-0.55) {$\scriptstyle \m{A}$};
    \node[style=none] (LabB) at (-0.88,-0.55) {$\scriptstyle \m{B}$};

    \node[style=none] (eq) at (0.00,-0.55) {$=$};


    \node[style=none] (RA_L)  at (1.10,0.18) {};
    \node[style=none] (RA_R)  at (1.22,0.18) {};
    \node[style=none] (RA_Lb) at (1.10,-1.20) {};
    \node[style=none] (RA_Rb) at (1.22,-1.20) {};
    \node[style=none] (RLabA) at (1.40,-0.55) {$\scriptstyle \m{A}$};

    \node[style=none] (RB_L)  at (2.15,0.18) {};
    \node[style=none] (RB_R)  at (2.27,0.18) {};
    \node[style=none] (RB_Lb) at (2.15,-1.20) {};
    \node[style=none] (RB_Rb) at (2.27,-1.20) {};
    \node[style=none] (RLabB) at (2.45,-0.55) {$\scriptstyle \m{B}$};

  \end{pgfonlayer}

  \begin{pgfonlayer}{edgelayer}

    \draw[line width=0.6pt] (LA_L.center) to (LA_Lb.center);
    \draw[line width=0.6pt] (LA_R.center) to (LA_Rb.center);

    \draw[line width=0.6pt] (LB_L.center) to (LB_Lb.center);
    \draw[line width=0.6pt] (LB_R.center) to (LB_Rb.center);

    \draw[line width=0.6pt] (-2.35,0.18) -- (-0.83,0.18);
    \draw[line width=0.6pt] (-2.15,0.30) -- (-1.03,0.30);
    \draw[line width=0.6pt] (-1.95,0.42) -- (-1.23,0.42);

    \draw[line width=0.6pt] (RA_L.center) to (RA_Lb.center);
    \draw[line width=0.6pt] (RA_R.center) to (RA_Rb.center);

    \draw[line width=0.6pt] (RB_L.center) to (RB_Lb.center);
    \draw[line width=0.6pt] (RB_R.center) to (RB_Rb.center);

    \draw[line width=0.6pt] (0.80,0.18) -- (1.52,0.18);
    \draw[line width=0.6pt] (0.91,0.30) -- (1.41,0.30);
    \draw[line width=0.6pt] (1.02,0.42) -- (1.30,0.42);

    \draw[line width=0.6pt] (1.85,0.18) -- (2.57,0.18);
    \draw[line width=0.6pt] (1.96,0.30) -- (2.46,0.30);
    \draw[line width=0.6pt] (2.07,0.42) -- (2.35,0.42);

  \end{pgfonlayer}
\end{tikzpicture}
}}
\end{equation}
\item Validity of conditional states and effects: \begin{equation}
\label{eq:ConditionalStatesAndEffects}
\vcenter{\hbox{\begin{tikzpicture}[baseline={(Lin.base)}]
  \begin{pgfonlayer}{nodelayer}


    \node[style=none] (LeBL)  at (-3.80, 1.20) {};
    \node[style=none] (LeBR)  at (-2.00, 1.20) {};
    \node[style=none] (LeT)   at (-2.90, 2.25) {};
    \node[style=none] (LeLab) at (-2.90, 1.70) {$e$};

    \node[style=none] (LAuL)  at (-3.35, 1.20) {};
    \node[style=none] (LAuR)  at (-3.23, 1.20) {};
    \node[style=none] (LAdL)  at (-3.35, 0.05) {};
    \node[style=none] (LAdR)  at (-3.23, 0.05) {};
    \node[style=none] (LALab) at (-3.55, 0.65) {$\scriptstyle \m{A}$};

    \node[style=none] (LBuL)  at (-2.55, 1.20) {};
    \node[style=none] (LBuR)  at (-2.43, 1.20) {};
    \node[style=none] (LBdL)  at (-2.55, 0.55) {};  
    \node[style=none] (LBdR)  at (-2.43, 0.55) {};
    \node[style=none] (LBLab) at (-2.20, 0.85) {$\scriptstyle \m{B}$};

    \node[style=none] (LnuTL)  at (-2.90, 0.55) {};
    \node[style=none] (LnuTR)  at (-2.08, 0.55) {};
    \node[style=none] (LnuB)   at (-2.49,-0.10) {};
    \node[style=none] (LnuLab) at (-2.45, 0.26) {$\nu$};

    \node[style=none] (Lin)  at (-1.10, 0.95) {$\in$};
    \node[style=none] (Lset) at ( 0.35, 0.95) {$E_{\m{A}}$};


    \node[style=none] (RAuL)  at ( 2.05, 1.30) {};
    \node[style=none] (RAuR)  at ( 2.17, 1.30) {};
    \node[style=none] (RAdL)  at ( 2.05, 0.35) {};
    \node[style=none] (RAdR)  at ( 2.17, 0.35) {};
    \node[style=none] (RALab) at ( 1.78, 0.85) {$\scriptstyle \m{A}$};

    \node[style=none] (RBuL)  at ( 3.05, 1.30) {};
    \node[style=none] (RBuR)  at ( 3.17, 1.30) {};
    \node[style=none] (RBdL)  at ( 3.05, 0.35) {};
    \node[style=none] (RBdR)  at ( 3.17, 0.35) {};
    \node[style=none] (RBLab) at ( 3.35, 0.85) {$\scriptstyle \m{B}$};

    \node[style=none] (RheBL)  at ( 2.65, 1.30) {};
    \node[style=none] (RheBR)  at ( 3.55, 1.30) {};
    \node[style=none] (RheT)   at ( 3.10, 2.05) {};
    \node[style=none] (RheLab) at ( 3.10, 1.58) {$\bar e$};

    \node[style=none] (RwTL)  at ( 1.55, 0.35) {};
    \node[style=none] (RwTR)  at ( 3.75, 0.35) {};
    \node[style=none] (RwB)   at ( 2.65,-0.70) {};
    \node[style=none] (RwLab) at ( 2.65,-0.18) {$\omega$};

    \node[style=none] (Rin)  at ( 4.55, 0.95) {$\in$};
    \node[style=none] (Rset) at ( 5.85, 0.95) {$\m{S_{A}}$};
    

   \node[style=none] (,) at ( 1.13, 0.65) {,};

  \end{pgfonlayer}

  \begin{pgfonlayer}{edgelayer}

    \draw (LeBL.center) to (LeBR.center);
    \draw (LeBR.center) to (LeT.center);
    \draw (LeT.center)  to (LeBL.center);

    \draw[line width=0.6pt] (LAuL.center) to (LAdL.center);
    \draw[line width=0.6pt] (LAuR.center) to (LAdR.center);

    \draw (LnuTL.center) to (LnuTR.center);
    \draw (LnuTR.center) to (LnuB.center);
    \draw (LnuB.center)  to (LnuTL.center);

    \draw[line width=0.6pt] (LBuL.center) to (LBdL.center);
    \draw[line width=0.6pt] (LBuR.center) to (LBdR.center);

    \draw (RheBL.center) to (RheBR.center);
    \draw (RheBR.center) to (RheT.center);
    \draw (RheT.center)  to (RheBL.center);

    \draw (RwTL.center) to (RwTR.center);
    \draw (RwTR.center) to (RwB.center);
    \draw (RwB.center)  to (RwTL.center);

    \draw[line width=0.6pt] (RAuL.center) to (RAdL.center);
    \draw[line width=0.6pt] (RAuR.center) to (RAdR.center);

    \draw[line width=0.6pt] (RBuL.center) to (RBdL.center);
    \draw[line width=0.6pt] (RBuR.center) to (RBdR.center);

  \end{pgfonlayer}
\end{tikzpicture}
}}
\end{equation}

    \end{enumerate}

Once we consider composite systems, there is a particularly important feature a theory might exhibit or not: tomographic locality. In tomographically local theories, any process can be distinguished by probing it locally, i.e., processes are separated by the statistics they generate when acting on local preparations and effects. Both classical and quantum theory are tomographically local, and this feature seemed so natural that it took some time for it to be explicitly formalized and treated within the GPT framework. Nonetheless, one can define theories in which it fails. For instance, quantum theory over real Hilbert spaces is not tomographically local~\cite{Wootters_2010,Hardy_2011,Chiribella_QuantumFromPrinciples2016}, the theory that describes fermions also fails to be tomographically local~\cite{darianoFermionic2014} -- as is many theories that emerge by restricting a given system with respect to some fixed symmetries ~\cite{centeno2024twirledworldssymmetryinducedfailures,ying2025quantumtheoryneedscomplex}.  

\subsection{Tomographic locality and its failure}

A theory satisfies tomographic locality if and only if product effects are sufficient for characterizing the states, and product states are sufficient for characterizing the effects. 

\begin{definition}[Tomographic Locality (TL)]
\label{def: TL}
    A composition of subsystems $\mathcal{A}$ and $\mathcal{B}$ into system $\mathcal{AB}$ obeys tomographic locality iff for any pair $\omega^{\m{AB}},\nu^{\m{AB}}\in \mathcal{S}$ \blk 
\begin{equation}
\label{eq:TLStates}
\vcenter{\hbox{
\begin{tikzpicture}[baseline={(iff.base)}]
  \begin{pgfonlayer}{nodelayer}


    \node[style=none] (LeBL) at (-6.45,1.30) {};
    \node[style=none] (LeBR) at (-5.45,1.30) {};
    \node[style=none] (LeT)  at (-5.95,2.25) {};
    \node[style=none] (Le)   at (-5.95,1.72) {$e$};

    \node[style=none] (LfBL) at (-5.25,1.30) {};
    \node[style=none] (LfBR) at (-4.25,1.30) {};
    \node[style=none] (LfT)  at (-4.75,2.25) {};
    \node[style=none] (Lf)   at (-4.78,1.72) {$f$};

    \node[style=none] (LowTL) at (-6.75,0.40) {};
    \node[style=none] (LowTR) at (-4.15,0.40) {};
    \node[style=none] (LowB)  at (-5.45,-1.05) {};
    \node[style=none] (Low)   at (-5.45,-0.25) {$\omega$};

    \node[style=none] (LAu1) at (-6.02,1.30) {};
    \node[style=none] (LAu2) at (-5.88,1.30) {};
    \node[style=none] (LAd1) at (-6.02,0.40) {};
    \node[style=none] (LAd2) at (-5.88,0.40) {};
    \node[style=none] (LAlab) at (-5.65,0.95) {$\scriptstyle \m{A}$};

    \node[style=none] (LBu1) at (-4.82,1.30) {};
    \node[style=none] (LBu2) at (-4.68,1.30) {};
    \node[style=none] (LBd1) at (-4.82,0.40) {};
    \node[style=none] (LBd2) at (-4.68,0.40) {};
    \node[style=none] (LBlab) at (-4.45,0.95) {$\scriptstyle \m{B}$};

    \node[style=none] (eqL) at (-3.78,0.55) {$=$};

    \node[style=none] (MeBL) at (-3.05,1.30) {};
    \node[style=none] (MeBR) at (-2.05,1.30) {};
    \node[style=none] (MeT)  at (-2.55,2.25) {};
    \node[style=none] (Me)   at (-2.55,1.72) {$e$};

    \node[style=none] (MfBL) at (-1.85,1.30) {};
    \node[style=none] (MfBR) at (-0.85,1.30) {};
    \node[style=none] (MfT)  at (-1.35,2.25) {};
    \node[style=none] (Mf)   at (-1.38,1.72) {$f$};

    \node[style=none] (MnuTL) at (-3.35,0.40) {};
    \node[style=none] (MnuTR) at (-0.55,0.40) {};
    \node[style=none] (MnuB)  at (-1.95,-1.05) {};
    \node[style=none] (Mnu)   at (-1.95,-0.25) {$\nu$};

    \node[style=none] (MAu1) at (-2.62,1.30) {};
    \node[style=none] (MAu2) at (-2.48,1.30) {};
    \node[style=none] (MAd1) at (-2.62,0.40) {};
    \node[style=none] (MAd2) at (-2.48,0.40) {};
    \node[style=none] (MAlab) at (-2.25,0.95) {$\scriptstyle \m{A}$};

    \node[style=none] (MBu1) at (-1.42,1.30) {};
    \node[style=none] (MBu2) at (-1.28,1.30) {};
    \node[style=none] (MBd1) at (-1.42,0.40) {};
    \node[style=none] (MBd2) at (-1.28,0.40) {};
    \node[style=none] (MBlab) at (-1.05,0.95) {$\scriptstyle \m{B}$};

    \node[style=none] (forall) at (0.18,0.55) {$\forall$};

    \node[style=none] (SeBL) at (0.55,1.30) {};
    \node[style=none] (SeBR) at (1.65,1.30) {};
    \node[style=none] (SeT)  at (1.10,2.30) {};
    \node[style=none] (Se)   at (1.10,1.72) {$e$};

    \node[style=none] (SfBL) at (1.85,1.30) {};
    \node[style=none] (SfBR) at (2.95,1.30) {};
    \node[style=none] (SfT)  at (2.40,2.30) {};
    \node[style=none] (Sf)   at (2.40,1.72) {$f$};

    \node[style=none] (SAu1) at (1.01,1.30) {};
    \node[style=none] (SAu2) at (1.19,1.30) {};
    \node[style=none] (SAd1) at (1.01,0.00) {};
    \node[style=none] (SAd2) at (1.19,0.00) {};
    \node[style=none] (SAlab) at (1.42,0.70) {$\scriptstyle \m{A}$};

    \node[style=none] (SBu1) at (2.31,1.30) {};
    \node[style=none] (SBu2) at (2.49,1.30) {};
    \node[style=none] (SBd1) at (2.31,0.00) {};
    \node[style=none] (SBd2) at (2.49,0.00) {};
    \node[style=none] (SBlab) at (2.72,0.70) {$\scriptstyle \m{B}$};

    \node[style=none] (iff) at (4.15,0.55) {$\Longleftrightarrow$};


    \node[style=none] (RowTL) at (5.35,0.40) {};
    \node[style=none] (RowTR) at (7.75,0.40) {};
    \node[style=none] (RowB)  at (6.55,-1.05) {};
    \node[style=none] (Row)   at (6.55,-0.25) {$\omega$};

    \node[style=none] (RAu1) at (6.08,0.40) {};
    \node[style=none] (RAu2) at (6.22,0.40) {};
    \node[style=none] (RAt1) at (6.08,1.55) {};
    \node[style=none] (RAt2) at (6.22,1.55) {};
    \node[style=none] (RAlab) at (6.42,1.05) {$\scriptstyle \m{A}$};

    \node[style=none] (RBu1) at (6.88,0.40) {};
    \node[style=none] (RBu2) at (7.02,0.40) {};
    \node[style=none] (RBt1) at (6.88,1.55) {};
    \node[style=none] (RBt2) at (7.02,1.55) {};
    \node[style=none] (RBlab) at (7.22,1.05) {$\scriptstyle \m{B}$};

    \node[style=none] (eqR) at (8.45,0.55) {$=$};

    \node[style=none] (RnuTL) at (9.25,0.40) {};
    \node[style=none] (RnuTR) at (11.65,0.40) {};
    \node[style=none] (RnuB)  at (10.45,-1.05) {};
    \node[style=none] (Rnu)   at (10.45,-0.25) {$\nu$};

    \node[style=none] (NAu1) at (9.98,0.40) {};
    \node[style=none] (NAu2) at (10.12,0.40) {};
    \node[style=none] (NAt1) at (9.98,1.55) {};
    \node[style=none] (NAt2) at (10.12,1.55) {};
    \node[style=none] (NAlab) at (10.32,1.05) {$\scriptstyle \m{A}$};

    \node[style=none] (NBu1) at (10.78,0.40) {};
    \node[style=none] (NBu2) at (10.92,0.40) {};
    \node[style=none] (NBt1) at (10.78,1.55) {};
    \node[style=none] (NBt2) at (10.92,1.55) {};
    \node[style=none] (NBlab) at (11.12,1.05) {$\scriptstyle \m{B}$};

  \end{pgfonlayer}

  \begin{pgfonlayer}{edgelayer}

    \draw (LeBL.center) to (LeBR.center);
    \draw (LeBR.center) to (LeT.center);
    \draw (LeT.center)  to (LeBL.center);

    \draw (LfBL.center) to (LfBR.center);
    \draw (LfBR.center) to (LfT.center);
    \draw (LfT.center)  to (LfBL.center);

    \draw (LowTL.center) to (LowTR.center);
    \draw (LowTR.center) to (LowB.center);
    \draw (LowB.center)  to (LowTL.center);

    \draw[line width=0.6pt] (LAd1.center) to (LAu1.center);
    \draw[line width=0.6pt] (LAd2.center) to (LAu2.center);
    \draw[line width=0.6pt] (LBd1.center) to (LBu1.center);
    \draw[line width=0.6pt] (LBd2.center) to (LBu2.center);

    \draw (MeBL.center) to (MeBR.center);
    \draw (MeBR.center) to (MeT.center);
    \draw (MeT.center)  to (MeBL.center);

    \draw (MfBL.center) to (MfBR.center);
    \draw (MfBR.center) to (MfT.center);
    \draw (MfT.center)  to (MfBL.center);

    \draw (MnuTL.center) to (MnuTR.center);
    \draw (MnuTR.center) to (MnuB.center);
    \draw (MnuB.center)  to (MnuTL.center);

    \draw[line width=0.6pt] (MAd1.center) to (MAu1.center);
    \draw[line width=0.6pt] (MAd2.center) to (MAu2.center);
    \draw[line width=0.6pt] (MBd1.center) to (MBu1.center);
    \draw[line width=0.6pt] (MBd2.center) to (MBu2.center);

    \draw (SeBL.center) to (SeBR.center);
    \draw (SeBR.center) to (SeT.center);
    \draw (SeT.center)  to (SeBL.center);

    \draw (SfBL.center) to (SfBR.center);
    \draw (SfBR.center) to (SfT.center);
    \draw (SfT.center)  to (SfBL.center);

    \draw[line width=0.6pt] (SAu1.center) to (SAd1.center);
    \draw[line width=0.6pt] (SAu2.center) to (SAd2.center);
    \draw[line width=0.6pt] (SBu1.center) to (SBd1.center);
    \draw[line width=0.6pt] (SBu2.center) to (SBd2.center);

    \draw (RowTL.center) to (RowTR.center);
    \draw (RowTR.center) to (RowB.center);
    \draw (RowB.center)  to (RowTL.center);

    \draw[line width=0.6pt] (RAu1.center) to (RAt1.center);
    \draw[line width=0.6pt] (RAu2.center) to (RAt2.center);
    \draw[line width=0.6pt] (RBu1.center) to (RBt1.center);
    \draw[line width=0.6pt] (RBu2.center) to (RBt2.center);

    \draw (RnuTL.center) to (RnuTR.center);
    \draw (RnuTR.center) to (RnuB.center);
    \draw (RnuB.center)  to (RnuTL.center);

    \draw[line width=0.6pt] (NAu1.center) to (NAt1.center);
    \draw[line width=0.6pt] (NAu2.center) to (NAt2.center);
    \draw[line width=0.6pt] (NBu1.center) to (NBt1.center);
    \draw[line width=0.6pt] (NBu2.center) to (NBt2.center);

  \end{pgfonlayer}
\end{tikzpicture}
}}
\end{equation}

and for any $e^{\m{AB}},f^{\m{AB}}\in E_{\m{AB}}$

\begin{equation}
\label{eq:TLEffects}
\vcenter{\hbox{
\begin{tikzpicture}[baseline={(iff.base)}]
  \begin{pgfonlayer}{nodelayer}


    \node[style=none] (LeBL) at (-8.10, 1.40) {};
    \node[style=none] (LeBR) at (-6.10, 1.40) {};
    \node[style=none] (LeT)  at (-7.10, 2.90) {};
    \node[style=none] (Le)   at (-7.10, 2.15) {$e$};

    \node[style=none] (LAu1) at (-7.82, 1.40) {};
    \node[style=none] (LAu2) at (-7.68, 1.40) {};
    \node[style=none] (LAd1) at (-7.82, 0.55) {};
    \node[style=none] (LAd2) at (-7.68, 0.55) {};
    \node[style=none] (LAlab) at (-8.15, 1.00) {$\scriptstyle \m{A}$};

    \node[style=none] (LBu1) at (-6.52, 1.40) {};
    \node[style=none] (LBu2) at (-6.38, 1.40) {};
    \node[style=none] (LBd1) at (-6.52, 0.55) {};
    \node[style=none] (LBd2) at (-6.38, 0.55) {};
    \node[style=none] (LBlab) at (-6.12, 1.00) {$\scriptstyle \m{B}$};

    \node[style=none] (LwTL) at (-8.25, 0.55) {};
    \node[style=none] (LwTR) at (-7.25, 0.55) {};
    \node[style=none] (LwB)  at (-7.75,-0.25) {};
    \node[style=none] (Lw)   at (-7.73, 0.21) {$\omega$};

    \node[style=none] (LnuTL) at (-6.95, 0.55) {};
    \node[style=none] (LnuTR) at (-5.95, 0.55) {};
    \node[style=none] (LnuB)  at (-6.45,-0.25) {};
    \node[style=none] (Lnu)   at (-6.42, 0.21) {$\nu$};

    \node[style=none] (eqL) at (-5.35, 1.00) {$=$};

    \node[style=none] (LePBL) at (-4.70, 1.40) {};
    \node[style=none] (LePBR) at (-2.70, 1.40) {};
    \node[style=none] (LePT)  at (-3.70, 2.90) {};
    \node[style=none] (LeP)   at (-3.70, 2.15) {$e'$};

    \node[style=none] (PAu1) at (-4.42, 1.40) {};
    \node[style=none] (PAu2) at (-4.28, 1.40) {};
    \node[style=none] (PAd1) at (-4.42, 0.55) {};
    \node[style=none] (PAd2) at (-4.28, 0.55) {};
    \node[style=none] (PAlab) at (-4.78, 1.00) {$\scriptstyle \m{A}$};

    \node[style=none] (PBu1) at (-3.12, 1.40) {};
    \node[style=none] (PBu2) at (-2.98, 1.40) {};
    \node[style=none] (PBd1) at (-3.12, 0.55) {};
    \node[style=none] (PBd2) at (-2.98, 0.55) {};
    \node[style=none] (PBlab) at (-2.72, 1.00) {$\scriptstyle \m{B}$};

    \node[style=none] (PwTL) at (-4.85, 0.55) {};
    \node[style=none] (PwTR) at (-3.85, 0.55) {};
    \node[style=none] (PwB)  at (-4.35,-0.25) {};
    \node[style=none] (Pw)   at (-4.32, 0.21) {$\omega$};

    \node[style=none] (PnuTL) at (-3.55, 0.55) {};
    \node[style=none] (PnuTR) at (-2.55, 0.55) {};
    \node[style=none] (PnuB)  at (-3.05,-0.25) {};
    \node[style=none] (Pnu)   at (-3.02, 0.21) {$\nu$};

    \node[style=none] (forall) at (-1.65, 1.00) {$\forall$};

    \node[style=none] (TwTL) at (-1.15, 0.55) {};
    \node[style=none] (TwTR) at (-0.15, 0.55) {};
    \node[style=none] (TwB)  at (-0.65,-0.25) {};
    \node[style=none] (Tw)   at (-0.62, 0.21) {$\omega$};

    \node[style=none] (TnTL) at ( 0.15, 0.55) {};
    \node[style=none] (TnTR) at ( 1.15, 0.55) {};
    \node[style=none] (TnB)  at ( 0.65,-0.25) {};
    \node[style=none] (Tn)   at ( 0.62, 0.21) {$\nu$};

    \node[style=none] (TAu1) at (-0.72, 0.55) {};
    \node[style=none] (TAu2) at (-0.58, 0.55) {};
    \node[style=none] (TAt1) at (-0.72, 1.95) {};
    \node[style=none] (TAt2) at (-0.58, 1.95) {};
    \node[style=none] (TAlab) at (-0.98, 1.00) {$\scriptstyle \m{A}$};

    \node[style=none] (TBu1) at ( 0.58, 0.55) {};
    \node[style=none] (TBu2) at ( 0.72, 0.55) {};
    \node[style=none] (TBt1) at ( 0.58, 1.95) {};
    \node[style=none] (TBt2) at ( 0.72, 1.95) {};
    \node[style=none] (TBlab) at ( 0.98, 1.00) {$\scriptstyle \m{B}$};

    \node[style=none] (iff) at (2.45, 1.00) {$\Longleftrightarrow$};


    \node[style=none] (ReBL) at (3.75, 1.40) {};
    \node[style=none] (ReBR) at (5.75, 1.40) {};
    \node[style=none] (ReT)  at (4.75, 2.90) {};
    \node[style=none] (Re)   at (4.75, 2.15) {$e$};

    \node[style=none] (RAu1) at (4.23, 1.40) {};
    \node[style=none] (RAu2) at (4.37, 1.40) {};
    \node[style=none] (RAd1) at (4.23, 0.05) {};
    \node[style=none] (RAd2) at (4.37, 0.05) {};
    \node[style=none] (RAlab) at (3.98, 1.00) {$\scriptstyle \m{A}$};

    \node[style=none] (RBu1) at (5.13, 1.40) {};
    \node[style=none] (RBu2) at (5.27, 1.40) {};
    \node[style=none] (RBd1) at (5.13, 0.05) {};
    \node[style=none] (RBd2) at (5.27, 0.05) {};
    \node[style=none] (RBlab) at (5.55, 1.00) {$\scriptstyle \m{B}$};

    \node[style=none] (eqR) at (6.55, 1.00) {$=$};

    \node[style=none] (RePBL) at (7.45, 1.40) {};
    \node[style=none] (RePBR) at (9.45, 1.40) {};
    \node[style=none] (RePT)  at (8.45, 2.90) {};
    \node[style=none] (ReP)   at (8.45, 2.15) {$e'$};

    \node[style=none] (NAu1) at (7.93, 1.40) {};
    \node[style=none] (NAu2) at (8.07, 1.40) {};
    \node[style=none] (NAd1) at (7.93, 0.05) {};
    \node[style=none] (NAd2) at (8.07, 0.05) {};
    \node[style=none] (NAlab) at (7.68, 1.00) {$\scriptstyle \m{A}$};

    \node[style=none] (NBu1) at (8.83, 1.40) {};
    \node[style=none] (NBu2) at (8.97, 1.40) {};
    \node[style=none] (NBd1) at (8.83, 0.05) {};
    \node[style=none] (NBd2) at (8.97, 0.05) {};
    \node[style=none] (NBlab) at (9.35, 1.00) {$\scriptstyle \m{B}$};

  \end{pgfonlayer}

  \begin{pgfonlayer}{edgelayer}

    \draw (LeBL.center) to (LeBR.center);
    \draw (LeBR.center) to (LeT.center);
    \draw (LeT.center)  to (LeBL.center);

    \draw[line width=0.6pt] (LAu1.center) to (LAd1.center);
    \draw[line width=0.6pt] (LAu2.center) to (LAd2.center);
    \draw[line width=0.6pt] (LBu1.center) to (LBd1.center);
    \draw[line width=0.6pt] (LBu2.center) to (LBd2.center);

    \draw (LwTL.center) to (LwTR.center);
    \draw (LwTR.center) to (LwB.center);
    \draw (LwB.center)  to (LwTL.center);

    \draw (LnuTL.center) to (LnuTR.center);
    \draw (LnuTR.center) to (LnuB.center);
    \draw (LnuB.center)  to (LnuTL.center);

    \draw (LePBL.center) to (LePBR.center);
    \draw (LePBR.center) to (LePT.center);
    \draw (LePT.center)  to (LePBL.center);

    \draw[line width=0.6pt] (PAu1.center) to (PAd1.center);
    \draw[line width=0.6pt] (PAu2.center) to (PAd2.center);
    \draw[line width=0.6pt] (PBu1.center) to (PBd1.center);
    \draw[line width=0.6pt] (PBu2.center) to (PBd2.center);

    \draw (PwTL.center) to (PwTR.center);
    \draw (PwTR.center) to (PwB.center);
    \draw (PwB.center)  to (PwTL.center);

    \draw (PnuTL.center) to (PnuTR.center);
    \draw (PnuTR.center) to (PnuB.center);
    \draw (PnuB.center)  to (PnuTL.center);

    \draw (TwTL.center) to (TwTR.center);
    \draw (TwTR.center) to (TwB.center);
    \draw (TwB.center)  to (TwTL.center);

    \draw (TnTL.center) to (TnTR.center);
    \draw (TnTR.center) to (TnB.center);
    \draw (TnB.center)  to (TnTL.center);

    \draw[line width=0.6pt] (TAu1.center) to (TAt1.center);
    \draw[line width=0.6pt] (TAu2.center) to (TAt2.center);
    \draw[line width=0.6pt] (TBu1.center) to (TBt1.center);
    \draw[line width=0.6pt] (TBu2.center) to (TBt2.center);

    \draw (ReBL.center) to (ReBR.center);
    \draw (ReBR.center) to (ReT.center);
    \draw (ReT.center)  to (ReBL.center);

    \draw[line width=0.6pt] (RAu1.center) to (RAd1.center);
    \draw[line width=0.6pt] (RAu2.center) to (RAd2.center);
    \draw[line width=0.6pt] (RBu1.center) to (RBd1.center);
    \draw[line width=0.6pt] (RBu2.center) to (RBd2.center);

    \draw (RePBL.center) to (RePBR.center);
    \draw (RePBR.center) to (RePT.center);
    \draw (RePT.center)  to (RePBL.center);

    \draw[line width=0.6pt] (NAu1.center) to (NAd1.center);
    \draw[line width=0.6pt] (NAu2.center) to (NAd2.center);
    \draw[line width=0.6pt] (NBu1.center) to (NBd1.center);
    \draw[line width=0.6pt] (NBu2.center) to (NBd2.center);

  \end{pgfonlayer}
\end{tikzpicture}
}}
\end{equation}

\end{definition}
In other words, a composite system obeys tomographic locality if one can always distinguish composite states (effects) by local effects (resp.~states).\footnote{For general processes, i.e., transformations that might have nontrivial inputs or outputs, the natural definition requires that every process with at least two (nontrivial) inputs and two (nontrivial) outputs can be distinguished when acting on product states \emph{and} product effects~\cite{Schmid2024structuretheorem}. Nevertheless, in finite dimensions, it is possible to show that a theory obeys tomographic locality according to this general definition if and only if  Eq.~\eqref{eq:TLStates} (or, equivalently, Eq.~\eqref{eq:TLEffects}) holds. } Notice that bare quantum theory is indeed tomographically local: any state, even if entangled, can be distinguished by the probabilities it generates for product effects. For example, the Bell state  $\ketbra{\Phi_+}{\Phi_+}$ is the only one which assigns probability $1$ for perfectly correlated outcomes for $XX$ and $ZZ$ measurements.
A simple signature of tomographic locality is given by the dimension of the joint state space~\cite{barrett_informationGPTs_2006}: the composition is tomographically local if and only if ${\rm dim}(AB) = {\rm dim}(A){\rm dim(B)}$. For instance, in the case of unrestricted quantum systems, the space of a bipartite system $\mathcal{AB}$ in which $\mathcal{A}$ is a $d_A$-levels quantum system and $\mathcal{B}$ is a $d_B-$levels quantum system, we have that $\mathsf{dim}(AB)=(d_Ad_B)^2$ which is the same as $\mathsf{dim}(A)\mathsf{dim}(B)=d_A^2d_B^2$.

Examples of tomographically local theories include quantum theory, classical theory, and Spekkens toy theory~\cite{toytheory}. Examples of tomographically nonlocal theories include fermionic quantum theory~\cite{darianoFermionic2014,Dariano2014feynmanproblem}, real quantum theory~\cite{Wootters_2010,Chiribella_QuantumFromPrinciples2016}, twirled worlds~\cite{centeno2024twirledworldssymmetryinducedfailures} and swirled worlds~\cite{ying2025quantumtheoryneedscomplex}, and bilocal classical theory~\cite{d2020classicality}.

\subsection{Entanglement}

An important feature of physical theories is the existence of entanglement, or lack thereof. Entanglement was first noticed in quantum systems {} and for a long time was considered a striking signature of quantumness. In fact, it is a common feature in GPT systems, i.e., GPT systems often exhibit entanglement -- even  in  some theories where every local state space is simplicial~\cite{d2020classicality,PhysRevA.101.042118}. 
From Definition~\ref{def: CompositionRequirements}, we can now define the separable-entangled boundary between states, and between effects, in a given GPT.

\begin{definition}[Entangled states]
\label{def: EntangledStates}
    States $\omega^{\m{AB}}\in \Omega_{\m{AB}}$ of a composite system $\mathcal{AB}$ are {\em separable} if and only if they can be written in the form
    \begin{align}
        \omega^{\m{AB}} = \sum_i p_i \omega^{\m{A}}_i\boxtimes\nu^B_i,
    \end{align}
    with local states $\{\omega^{\m{A}}_i\}_i\subset \Omega_{\m{A}}$ and $\{\nu^B_i\}_i\subset \Omega_{\m{B}}$, and a probability distribution $\{p_i\}_i$.  States $\omega^{\m{AB}}$ that cannot be written in this form are called {\em entangled states}.
\end{definition}
Diagrammatically, normalized separable  states are defined through
\begin{equation}
\vcenter{\hbox{%
\begin{tikzpicture}
	\begin{pgfonlayer}{nodelayer}

		\node [style=none] (TA)  at (-0.6,0.8) {};
		\node [style=none] (TB)  at (0.6,0.8) {};
		\node [style=none] (TA2) at (-0.6,0.0) {};
		\node [style=none] (TB2) at (0.6,0.0) {};

		\node [style=none] (Alabel) at (-0.9,0.45) {$\scriptstyle \m{A}$};
		\node [style=none] (Blabel) at (0.9,0.45) {$\scriptstyle \m{B}$};

		\node [style=none] (TL) at (-0.9,0.0) {};
		\node [style=none] (TR) at (0.9,0.0) {};
		\node [style=none] (TC) at (0,-1.0) {};
		\node [style=none] (W)  at (0,-0.45) {$\omega$};

	\end{pgfonlayer}

	\begin{pgfonlayer}{edgelayer}

		\draw[double] (TA.center) -- (TA2.center);
		\draw[double] (TB.center) -- (TB2.center);

		\draw (TL.center) -- (TR.center);
		\draw (TR.center) -- (TC.center);
		\draw (TC.center) -- (TL.center);

	\end{pgfonlayer}
\end{tikzpicture}%
}}
\in \mathsf{Sep}(\Omega_{\m{AB}})
 \iff \vcenter{\hbox{%
\begin{tikzpicture}
	\begin{pgfonlayer}{nodelayer}

		\node[style=none] (LTA)  at (-0.6,0.9) {};
		\node[style=none] (LTB)  at ( 0.6,0.9) {};
		\node[style=none] (LTA2) at (-0.6,0.0) {};
		\node[style=none] (LTB2) at ( 0.6,0.0) {};

		\node[style=none] (LAlab)  at (-0.95,0.55) {$\scriptstyle \m{A}$};
		\node[style=none] (LBblab) at ( 0.95,0.55) {$\scriptstyle \m{B}$};

		\node[style=none] (LTL) at (-0.9,0.0) {};
		\node[style=none] (LTR) at ( 0.9,0.0) {};
		\node[style=none] (LTC) at ( 0.0,-1.0) {};
		\node[style=none] (LW)  at ( 0.0,-0.45) {$\omega$};

		\node[style=none] (Eq)   at (2.00,-0.10) {$=$};
		\node[style=none] (Sum)  at (3.25,-0.10) {$\sum_i p_i$};

		\def\DW{0.08}

		\node[style=none] (RAwireTopL) at (4.70-\DW,0.9) {};
		\node[style=none] (RAwireTopR) at (4.70+\DW,0.9) {};
		\node[style=none] (RAwireBotL) at (4.70-\DW,0.0) {};
		\node[style=none] (RAwireBotR) at (4.70+\DW,0.0) {};
		\node[style=none] (RAwireLab)  at (4.40,0.55) {$\scriptstyle \m{A}$};

		\node[style=none] (RATL) at (4.17,0.0) {};
		\node[style=none] (RATR) at (5.23,0.0) {};
		\node[style=none] (RATC) at (4.70,-0.95) {};
		\node[style=none] (RAom) at (4.70,-0.34) {$\omega_i$};

		\node[style=none] (RBwireTopL) at (6.05-\DW,0.9) {};
		\node[style=none] (RBwireTopR) at (6.05+\DW,0.9) {};
		\node[style=none] (RBwireBotL) at (6.05-\DW,0.0) {};
		\node[style=none] (RBwireBotR) at (6.05+\DW,0.0) {};
		\node[style=none] (RBwireLab)  at (6.35,0.55) {$\scriptstyle \m{B}$};

		\node[style=none] (RBTL) at (5.52,0.0) {};
		\node[style=none] (RBTR) at (6.58,0.0) {};
		\node[style=none] (RBTC) at (6.05,-0.95) {};
		\node[style=none] (RBom) at (6.05,-0.34) {$\nu_i$};

        \node[style=none] (Punctuation) at (7.0,0.0) {,};
	\end{pgfonlayer}

	\begin{pgfonlayer}{edgelayer}

		\draw[double] (LTA.center) -- (LTA2.center);
		\draw[double] (LTB.center) -- (LTB2.center);

		\draw (LTL.center) -- (LTR.center);
		\draw (LTR.center) -- (LTC.center);
		\draw (LTC.center) -- (LTL.center);

		\draw[line width=0.6pt] (RAwireTopL.center) -- (RAwireBotL.center);
		\draw[line width=0.6pt] (RAwireTopR.center) -- (RAwireBotR.center);
		\draw (RATL.center) -- (RATR.center);
		\draw (RATR.center) -- (RATC.center);
		\draw (RATC.center) -- (RATL.center);

		\draw[line width=0.6pt] (RBwireTopL.center) -- (RBwireBotL.center);
		\draw[line width=0.6pt] (RBwireTopR.center) -- (RBwireBotR.center);
		\draw (RBTL.center) -- (RBTR.center);
		\draw (RBTR.center) -- (RBTC.center);
		\draw (RBTC.center) -- (RBTL.center);

	\end{pgfonlayer}
\end{tikzpicture}%
}}
\end{equation}
where $\{p_i\}$ is a probability distribution, $\mathsf{Sep}[\Omega_{\m{AB}}]$ is the set of separable states and, again, $\{\omega_i^{\m{A}}\}_i\subset\Omega^{\m{A}}$ and $\{\nu_i^{\m{B}}\}\subset \Omega_{\m{B}}$.

Finally,  states are entangled if and only if
\begin{equation}
\label{eq:NotSeparableDiagram}
\vcenter{\hbox{%
\begin{tikzpicture}
	\begin{pgfonlayer}{nodelayer}

		\node [style=none] (TA)  at (-0.6,0.8) {};
		\node [style=none] (TB)  at (0.6,0.8) {};
		\node [style=none] (TA2) at (-0.6,0.0) {};
		\node [style=none] (TB2) at (0.6,0.0) {};

		\node [style=none] (Alabel) at (-0.9,0.45) {$\scriptstyle \m{A}$};
		\node [style=none] (Blabel) at (0.9,0.45) {$\scriptstyle \m{B}$};

		\node [style=none] (TL) at (-0.9,0.0) {};
		\node [style=none] (TR) at (0.9,0.0) {};
		\node [style=none] (TC) at (0,-1.0) {};
		\node [style=none] (W)  at (0,-0.45) {$\omega$};

	\end{pgfonlayer}

	\begin{pgfonlayer}{edgelayer}

		\draw[double] (TA.center) -- (TA2.center);
		\draw[double] (TB.center) -- (TB2.center);

		\draw (TL.center) -- (TR.center);
		\draw (TR.center) -- (TC.center);
		\draw (TC.center) -- (TL.center);

	\end{pgfonlayer}
\end{tikzpicture}%
}}
\notin \mathsf{Sep}(\Omega_{\m{AB}}).
\end{equation}

\begin{definition}[Entangled effects]
\label{def: EntangledEffects}
    Effects $e^{\m{AB}}\in E_{\m{AB}}$ of a composite system $\mathcal{AB}$ are {\em separable} if and only if they can be written in the form
    \begin{align}
        e^{\m{AB}} = \sum_i  e^{\m{A}}_i\boxtimes e^{\m{B}}_i,
    \end{align}
    with local effects $\{e^{\m{A}}_i\}_i\subset E_{\m{A}}$ and $\{e^{\m{B}}_i\}_i\subset E_{\m{B}}$.  Effects $e^{\m{AB}}$ that cannot be written in this form are called {\em entangled effects}.
\end{definition}

Diagrammatically, 
separable effects are defined through
\begin{equation}
\vcenter{\hbox{%
\begin{tikzpicture}
	\begin{pgfonlayer}{nodelayer}

		\node [style=none] (BA)  at (-0.6,-0.8) {};
		\node [style=none] (BB)  at (0.6,-0.8) {};
		\node [style=none] (BA2) at (-0.6,0.0) {};
		\node [style=none] (BB2) at (0.6,0.0) {};

		\node [style=none] (Alabel) at (-0.9,-0.45) {$\scriptstyle \m{A}$};
		\node [style=none] (Blabel) at (0.9,-0.45) {$\scriptstyle \m{B}$};

		\node [style=none] (TL) at (-0.9,0.0) {};
		\node [style=none] (TR) at (0.9,0.0) {};
		\node [style=none] (TC) at (0,1.0) {};
		\node [style=none] (E)  at (0,0.45) {$e$};

	\end{pgfonlayer}

	\begin{pgfonlayer}{edgelayer}

		\draw[double] (BA.center) -- (BA2.center);
		\draw[double] (BB.center) -- (BB2.center);

		\draw (TL.center) -- (TR.center);
		\draw (TR.center) -- (TC.center);
		\draw (TC.center) -- (TL.center);

	\end{pgfonlayer}
\end{tikzpicture}%
}}
\in \mathsf{Sep}(E_{\m{AB}})
 \iff \vcenter{\hbox{%
\begin{tikzpicture}
	\begin{pgfonlayer}{nodelayer}

		\def\DW{0.08}

		\node[style=none] (LTA)  at (-0.6,0.0) {};
		\node[style=none] (LTB)  at ( 0.6,0.0) {};
		\node[style=none] (LTA2) at (-0.6,-0.9) {};
		\node[style=none] (LTB2) at ( 0.6,-0.9) {};

		\node[style=none] (LAlab)  at (-0.95,-0.45) {$\scriptstyle \m{A}$};
		\node[style=none] (LBblab) at ( 0.95,-0.45) {$\scriptstyle \m{B}$};

		\node[style=none] (LTL) at (-0.9,0.0) {};
		\node[style=none] (LTR) at ( 0.9,0.0) {};
		\node[style=none] (LTT) at ( 0.0,1.0) {};
		\node[style=none] (LE)  at ( 0.0,0.45) {$e$};

		\node[style=none] (Eq)   at (2.00,0.10) {$=$};
		\node[style=none] (Sum)  at (3.25,0.10) {$\sum_i $ };


		\node[style=none] (RAwireTopL) at (4.70-\DW,0.0) {};
		\node[style=none] (RAwireTopR) at (4.70+\DW,0.0) {};
		\node[style=none] (RAwireBotL) at (4.70-\DW,-0.9) {};
		\node[style=none] (RAwireBotR) at (4.70+\DW,-0.9) {};
		\node[style=none] (RAwireLab)  at (4.40,-0.45) {$\scriptstyle \m{A}$};

		\node[style=none] (RATL) at (4.17,0.0) {};
		\node[style=none] (RATR) at (5.23,0.0) {};
		\node[style=none] (RATT) at (4.70,0.95) {};
		\node[style=none] (RAe)  at (4.70,0.38) {$e_i$};

		\node[style=none] (RBwireTopL) at (6.05-\DW,0.0) {};
		\node[style=none] (RBwireTopR) at (6.05+\DW,0.0) {};
		\node[style=none] (RBwireBotL) at (6.05-\DW,-0.9) {};
		\node[style=none] (RBwireBotR) at (6.05+\DW,-0.9) {};
		\node[style=none] (RBwireLab)  at (6.35,-0.45) {$\scriptstyle \m{B}$};

		\node[style=none] (RBTL) at (5.52,0.0) {};
		\node[style=none] (RBTR) at (6.58,0.0) {};
		\node[style=none] (RBTT) at (6.05,0.95) {};
		\node[style=none] (RBf)  at (6.05,0.38) {$f_i$};

	\end{pgfonlayer}

	\begin{pgfonlayer}{edgelayer}

		\draw[double] (LTA.center) -- (LTA2.center);
		\draw[double] (LTB.center) -- (LTB2.center);

		\draw (LTL.center) -- (LTR.center);
		\draw (LTR.center) -- (LTT.center);
		\draw (LTT.center) -- (LTL.center);

		\draw[line width=0.6pt] (RAwireTopL.center) -- (RAwireBotL.center);
		\draw[line width=0.6pt] (RAwireTopR.center) -- (RAwireBotR.center);
		\draw (RATL.center) -- (RATR.center);
		\draw (RATR.center) -- (RATT.center);
		\draw (RATT.center) -- (RATL.center);

		\draw[line width=0.6pt] (RBwireTopL.center) -- (RBwireBotL.center);
		\draw[line width=0.6pt] (RBwireTopR.center) -- (RBwireBotR.center);
		\draw (RBTL.center) -- (RBTR.center);
		\draw (RBTR.center) -- (RBTT.center);
		\draw (RBTT.center) -- (RBTL.center);

	\end{pgfonlayer}
\end{tikzpicture}%
}} 
\end{equation}
where $\{e^{\m{A}}_i\}_i\subset E_{\m{A}}$ and $\{e^{\m{B}}_i\}_i\subset E_{\m{B}}$ are local effects, \blk
and  effects are entangled if and only if
\begin{equation}
\label{eq:NotSeparableEffectDiagram}
\vcenter{\hbox{%
\begin{tikzpicture}
	\begin{pgfonlayer}{nodelayer}

		\node [style=none] (BA)  at (-0.6,-0.8) {};
		\node [style=none] (BB)  at (0.6,-0.8) {};
		\node [style=none] (BA2) at (-0.6,0.0) {};
		\node [style=none] (BB2) at (0.6,0.0) {};

		\node [style=none] (Alabel) at (-0.9,-0.45) {$\scriptstyle \m{A}$};
		\node [style=none] (Blabel) at (0.9,-0.45) {$\scriptstyle \m{B}$};

		\node [style=none] (TL) at (-0.9,0.0) {};
		\node [style=none] (TR) at (0.9,0.0) {};
		\node [style=none] (TC) at (0,1.0) {};
		\node [style=none] (E)  at (0,0.45) {$e$};

	\end{pgfonlayer}

	\begin{pgfonlayer}{edgelayer}

		\draw[double] (BA.center) -- (BA2.center);
		\draw[double] (BB.center) -- (BB2.center);

		\draw (TL.center) -- (TR.center);
		\draw (TR.center) -- (TC.center);
		\draw (TC.center) -- (TL.center);

	\end{pgfonlayer}
\end{tikzpicture}%
}}
\notin \mathsf{Sep}(E_{\m{AB}}) \,.
\end{equation}

The above definitions provide a mathematical characterization of entangled states and effects, inspired by the traditional quantum notions. Nevertheless, one could aim for operational notions of separable states and effects, i.e., to consider a particular set of free local operations and the entangled states and effects would be those that cannot be generated through such free operations. For instance, one could consider local operations and classical communication to be free, or local operations and shared randomness instead. Depending on the set of free operations one chooses, different sets of entangled states and effects might arise, and therefore might not coincide to those provided by Defs.~\ref{def: EntangledStates} and ~\ref{def: EntangledEffects}. We suspect that the entangled states defined with respect to LOCC and LOSR might coincide with that of Def.~\ref{def: EntangledStates}, while differences might arise for entangled effects.  In any case, we leave the analysis of these different definitions and their impact for future research, and the natural mathematical generalization provided by Defs.~\ref{def: EntangledEffects} and ~\ref{def: EntangledStates} will be sufficient for our purposes in this work.

One particular aspect that also carries over from the quantum description to the GPT description is that pure states are separable only if they are product states, as is proven\footnote{An analogous statement holds for effects, though with some subtleties: ray-extremal effects are separable if and only if they are product, where ray-extremal means that the effects are extremal in the cone $\mathsf{cone}[E_{\m{AB}}]$. Note that Ref.~\cite{Baldijao_QDarwinisminGPTs_2022} calls ray-extremal states by pure. Additionally, note that Ref.~\cite{Baldijao_QDarwinisminGPTs_2022} assumes that composition preserves purity, i.e.,  $s^{\m{A}}\otimes r^{\m{B}}$ is pure whenever $s^{\m{A}}$ and $r^{\m{B}}$ are pure, which we need not assume in this work. With this extra condition, a separable state is pure if and only if it is product, with the extra implication being trivial.} in Ref.~\cite{Baldijao_QDarwinisminGPTs_2022}. Note that the definition of entanglement is relative to the given theory or the full GPT system.
Indeed, one can define what a subsystem would be in the GPT setting~\cite{schmid2024shadowssubsystemsgeneralizedprobabilistic}, and a state might be entangled for the subsystem and not be entangled for the full system, as in fact happens for real quantum theory systems, which are subsystems of complex quantum systems; there are states that are entangled in the former and separable in the latter.

\section{Tools for exploring the absence of tomographic locality}
In this Section, we propose a few definitions to explore the failure of tomographic locality, which will be needed for refining the notion of entanglement in GPTs.

\subsection{Decomposition of vector spaces and holistic subspaces}

\begin{lemma}[$AB$ contains $A\otimes B$]
\label{lemma:AtBplusC} Given any composite system $\mathcal{AB}$ of two GPT systems $\mathcal{A}$ and $\mathcal{B}$, the span of product states (effects) is isomorphic to the algebraic tensor product $A\otimes B$ (resp., $A^*\otimes B^*$). i.e., the following isomorhpisms hold: 

\begin{equation}
\label{eq:TPStates}
        \mathrm{Span}{\lbrace\omega^{\m{A}}\boxtimes\nu^{\m{B}}\mid\omega^{\m{A}}\in S_{\mathcal{A}},\nu^{\m{B}}\in S_{\mathcal{B}}\rbrace}
        \cong
        A\otimes B
    \end{equation}
    and 
    \begin{equation}
    \label{eq:TPEffects}
        \mathrm{Span}{\lbrace e^{\m{A}}\boxtimes f^{\m{B}}\mid e^{\m{A}}\in E_{\mathcal{A}},f^{\m{B}}\in E_{\mathcal{B}}\rbrace}
        \cong
        A^*\otimes B^*
        ,
    \end{equation}
where $\otimes$ denotes the algebraic tensor product of real vector spaces. 
Moreover, the real vector spaces associated with the composite $\mathcal{AB}$ have the form
\begin{align}
    AB \cong A\otimes B\oplus C
    \label{eq:AssumpComposition}
\end{align} 
and 
\begin{align}
    (AB)^* \cong A^*\otimes B^*\oplus D
    \label{eq:AssumpCompositionE}
\end{align} 
\blk
The composition is tomographically local iff $C=\{0\}$(or, equivalently, iff $D=\{0\}$), so that $\mathsf{Span}[\Omega_{\m{AB}}] \cong A \otimes B$ (equivalently, $\mathsf{Span}[E_{\m{AB}}] \cong A^* \otimes B^*$).
\end{lemma}

\begin{proof}
A proof can be found in Ref.~\cite{Barnum_CathegoriesJordanAlgebras_2020}. We give an adapted proof in  Appendix~\ref{compositionproofs}.
\end{proof}
Eqs.~\eqref{eq:TPStates} and \eqref{eq:TPEffects} suggest the following shorthand notation: $AB_{\otimes}:=\mathrm{Span}{\lbrace\omega^{\m{A}}\boxtimes\nu^{\m{B}}\mid\omega^{\m{A}}\in S_{\mathcal{A}},\nu^{\m{B}}\in S_{\mathcal{B}}\rbrace}$ and $(AB)^*_{\otimes}:=\mathrm{Span}{\lbrace e^{\m{A}}\boxtimes f^{\m{B}}\mid e^{\m{A}}\in E_{\mathcal{A}},f^{\m{B}}\in E_{\mathcal{B}}\rbrace}$, where $AB_\otimes$ is an isomorphic copy of $A\otimes B$ that lives in $AB$ and similarly for $(AB)^*_\otimes$. We call $AB_\otimes$ the tomographically-local subspace of $AB$, and $(AB)^*_\otimes$ the tomographically-local subspace of $(AB)^*$.
Lemma~\ref{lemma:AtBplusC} shows that any composition satisfying the requirements in Def.~\ref{def: CompositionRequirements} necessarily maps the subsystems to another system that contains at least  $AB_{\otimes}$ as a subspace of $AB$ (and $(AB)^*_{\otimes}$ as a subspace of $(AB)^*$). It is only when the composed system obeys tomographic locality that $AB$ is the span of product states, and thus $AB=AB_\otimes$ (the analogous  statement holding for the effect spaces).
This shows that for GPT systems, we should keep a distinction among the two operations $\boxtimes$ and $\otimes$.

Notice that the vector spaces $AB$ and $AB_\otimes$  (similarly for their analogue for effects) are spanned by the states (resp.~effects) of the system, which means they may contain elements that do not correspond to states (resp.~effects). For instance, some vectors in $AB$ may give a negative number when composed with some effect in $E_{AB}$. This fact will be leveraged throughout this manuscript, especially when decomposing states and effects as linear combinations of elements in these vector spaces, yielding decompositions where the individual terms cannot be assumed to be states or effects themselves. \blk 

Note that for a given $A$, $B$, and $V:=A\otimes B\oplus C$, there may be more than one choice of $C$ for which the above decomposition holds, even when the inclusion of $A\otimes B$ in $V$ is already fixed. This is just an instance of the general fact that, for a vector space $V$ and a nontrivial strict subspace $U\subset V$, there are many distinct subspaces $W$ such that $V=U\oplus W$. For instance, consider a three-dimensional vector space $V$. Since it is isomorphic to $\mathbb{R}^3$, we can describe its elements by  $(x,y,z)$ with $x,y,z\in\mathbb{R}$. Now take $U$ to be the xy plane, i.e., the vectors with the form $(x,y,0)$. Any one-dimensional vector space $W$ within $V$ which is not contained in the xy plane satisfies $V=U\oplus W$. For instance,   $W:=\{(0,0,z)|z\in\mathbb{R}\}$ and $W':=\{(a,a,a)|a\in\mathbb{R}\}$ are two such subspaces. The `decompositions' of the composite vector space of $\m{AB}$ expressed in Eqs.~\eqref{eq:AssumpComposition} and ~\eqref{eq:AssumpCompositionE} will be essential from now on. 

Consider a composed system in which $C$ is not trivial (so tomographic locality necessarily fails in the theory describing such a GPT system).
The fact that the composite vector space (the span of the composite states) is larger than $AB_\otimes$ entails that composite systems possess some degrees of freedom that are `global' in that they cannot be fully accessed locally. This idea has been described in the GPT literature in the past~\cite{Sainz_2018,barrett_informationGPTs_2006,Muller_2012, Hardy_2011,hardy2009foliableoperationalstructuresgeneral}. We propose the following definition to capture this idea.
\begin{definition}[Holistic subspaces]
\label{def:HolisticSubspaces}
    Consider a composite system $\mathcal{AB}$ and consider the local effect spaces   
    $E_{\m{A}}$ and $E_{\m{B}}$.
   
    The \emph{holistic-state subspace}, $H_S$ is defined by 
\begin{equation}
\label{eq:HolisticSubspaceDiagrammatic_tensor_corrected}
H_S 
:=\left\{
\vcenter{\hbox{%
\begin{tikzpicture}[baseline={(base.center)}]
  \begin{pgfonlayer}{nodelayer}
    \node[style=none] (base) at (0,0) {};

    \node[style=none] (TA)  at (-0.55,1.05) {};
    \node[style=none] (TB)  at ( 0.55,1.05) {};
    \node[style=none] (TA2) at (-0.55,0.00) {};
    \node[style=none] (TB2) at ( 0.55,0.00) {};

    \node[style=none] (Alab) at (-0.90,0.65) {$\scriptstyle \m{A}$};
    \node[style=none] (Blab) at ( 0.90,0.65) {$\scriptstyle \m{B}$};

    \node[style=none] (TL) at (-0.90,0.00) {};
    \node[style=none] (TR) at ( 0.90,0.00) {};
    \node[style=none] (TC) at ( 0.00,-0.95) {};
    \node[style=none] (Nu) at ( 0.00,-0.45) {$\tilde{h}$};
  \end{pgfonlayer}

  \begin{pgfonlayer}{edgelayer}
    \draw[double] (TA.center) -- (TA2.center);
    \draw[double] (TB.center) -- (TB2.center);

    \draw (TL.center) -- (TR.center);
    \draw (TR.center) -- (TC.center);
    \draw (TC.center) -- (TL.center);
  \end{pgfonlayer}
\end{tikzpicture}%
}}
\in {AB}
\;\Bigm|\;
\vcenter{\hbox{%
\begin{tikzpicture}[baseline={(base.center)}]
  \begin{pgfonlayer}{nodelayer}
    \node[style=none] (base) at (0,0) {};

    \def\yTop{0.90}
    \def\yMid{0.10}
    \def\yBot{-0.85}

    \node[style=none] (eTL) at (-1.00,\yTop) {};
    \node[style=none] (eTR) at (-0.20,\yTop) {};
    \node[style=none] (eTT) at (-0.60,1.65) {};
    \node[style=none] (eLab) at (-0.60,1.25) {$e$};

    \node[style=none] (fTL) at ( 0.20,\yTop) {};
    \node[style=none] (fTR) at ( 1.00,\yTop) {};
    \node[style=none] (fTT) at ( 0.60,1.65) {};
    \node[style=none] (fLab) at ( 0.60,1.25) {$f$};

    \node[style=none] (AwL_top) at (-0.68,\yTop) {};
    \node[style=none] (AwR_top) at (-0.52,\yTop) {};
    \node[style=none] (AwL_mid) at (-0.68,\yMid) {};
    \node[style=none] (AwR_mid) at (-0.52,\yMid) {};
    \node[style=none] (AwLab) at (-0.92,0.55) {$\scriptstyle \m{A}$};

    \node[style=none] (BwL_top) at ( 0.52,\yTop) {};
    \node[style=none] (BwR_top) at ( 0.68,\yTop) {};
    \node[style=none] (BwL_mid) at ( 0.52,\yMid) {};
    \node[style=none] (BwR_mid) at ( 0.68,\yMid) {};
    \node[style=none] (BwLab) at ( 0.92,0.55) {$\scriptstyle \m{B}$};

    \node[style=none] (sTL) at (-1.10,\yMid) {};
    \node[style=none] (sTR) at ( 1.10,\yMid) {};
    \node[style=none] (sTC) at ( 0.00,\yBot) {};
    \node[style=none] (sNu) at ( 0.00,-0.30) {$\tilde{h}$};

  \end{pgfonlayer}

  \begin{pgfonlayer}{edgelayer}
    \draw (eTL.center) -- (eTR.center);
    \draw (eTR.center) -- (eTT.center);
    \draw (eTT.center) -- (eTL.center);

    \draw (fTL.center) -- (fTR.center);
    \draw (fTR.center) -- (fTT.center);
    \draw (fTT.center) -- (fTL.center);

    \draw[line width=0.6pt] (AwL_top.center) -- (AwL_mid.center);
    \draw[line width=0.6pt] (AwR_top.center) -- (AwR_mid.center);

    \draw[line width=0.6pt] (BwL_top.center) -- (BwL_mid.center);
    \draw[line width=0.6pt] (BwR_top.center) -- (BwR_mid.center);

    \draw (sTL.center) -- (sTR.center);
    \draw (sTR.center) -- (sTC.center);
    \draw (sTC.center) -- (sTL.center);

  \end{pgfonlayer}
\end{tikzpicture}%
}}
=0
\;\;\forall\;
 \vcenter{\hbox{%
            \begin{tikzpicture}[baseline={(base.center)}]
            \begin{pgfonlayer}{nodelayer}
            \node[style=none] (base) at (0,0) {};

            \def\yTop{0.00}
            \def\yBot{-0.80}

            \node[style=none] (eTL) at (-1.00,\yTop) {};
            \node[style=none] (eTR) at (-0.20,\yTop) {};
            \node[style=none] (eTT) at (-0.60,0.75) {};
            \node[style=none] (eLab) at (-0.60,0.35) {$e$};

            \node[style=none] (AwL_top) at (-0.68,\yTop) {};
            \node[style=none] (AwR_top) at (-0.52,\yTop) {};
            \node[style=none] (AwL_bot) at (-0.68,\yBot) {};
            \node[style=none] (AwR_bot) at (-0.52,\yBot) {};
            \node[style=none] (AwLab) at (-0.92,-0.40) {$\scriptstyle \m{A}$};

        \end{pgfonlayer}

        \begin{pgfonlayer}{edgelayer}
            \draw (eTL.center) -- (eTR.center);
            \draw (eTR.center) -- (eTT.center);
            \draw (eTT.center) -- (eTL.center);

            \draw[line width=0.6pt] (AwL_top.center) -- (AwL_bot.center);
            \draw[line width=0.6pt] (AwR_top.center) -- (AwR_bot.center);

        \end{pgfonlayer}
        \end{tikzpicture}%
        }}
        \in E_{\m{A}}\,,\,\vcenter{\hbox{%
        \begin{tikzpicture}[baseline={(base.center)}]
        \begin{pgfonlayer}{nodelayer}
            \node[style=none] (base) at (0,0) {};

            \def\yTop{0.00}
            \def\yBot{-0.80}

            \node[style=none] (fTL) at (-1.00,\yTop) {};
            \node[style=none] (fTR) at (-0.20,\yTop) {};
            \node[style=none] (fTT) at (-0.60,0.75) {};
            \node[style=none] (fLab) at (-0.60,0.35) {$f$};

            \node[style=none] (BwL_top) at (-0.68,\yTop) {};
            \node[style=none] (BwR_top) at (-0.52,\yTop) {};
            \node[style=none] (BwL_bot) at (-0.68,\yBot) {};
            \node[style=none] (BwR_bot) at (-0.52,\yBot) {};
            \node[style=none] (bwLab) at (-0.92,-0.40) {$\scriptstyle \m{B}$};

    \end{pgfonlayer}

    \begin{pgfonlayer}{edgelayer}
    \draw (fTL.center) -- (fTR.center);
    \draw (fTR.center) -- (fTT.center);
    \draw (fTT.center) -- (fTL.center);

    \draw[line width=0.6pt] (AwL_top.center) -- (BwL_bot.center);
    \draw[line width=0.6pt] (AwR_top.center) -- (BwR_bot.center);

  \end{pgfonlayer}
    \end{tikzpicture}%
    }}
    \in E_{\m{B}}
\right\}.
\end{equation}

    Similarly, consider the local state spaces $\m{S_A}$ and $\m{S_{B}}$. Then, the \emph{holistic-effect subspace}, $H_E$, is defined by
\begin{equation}
\label{eq:HolisticEffectSubspaceDiagrammatic_tensor}
H_{E}
:=\left\{
\vcenter{\hbox{%
\begin{tikzpicture}[baseline={(base.center)}]
  \begin{pgfonlayer}{nodelayer}
    \node[style=none] (base) at (0,0) {};

    \node[style=none] (hTL) at (-0.90,0.00) {};
    \node[style=none] (hTR) at ( 0.90,0.00) {};
    \node[style=none] (hTT) at ( 0.00,1.00) {};
    \node[style=none] (hLab) at ( 0.00,0.45) {$\tilde{h}$};

    \node[style=none] (AwL_top) at (-0.58,0.00) {};
    \node[style=none] (AwR_top) at (-0.42,0.00) {};
    \node[style=none] (AwL_bot) at (-0.58,-1.10) {};
    \node[style=none] (AwR_bot) at (-0.42,-1.10) {};
    \node[style=none] (Alab) at (-0.90,-0.55) {$\scriptstyle \m{A}$};

    \node[style=none] (BwL_top) at ( 0.42,0.00) {};
    \node[style=none] (BwR_top) at ( 0.58,0.00) {};
    \node[style=none] (BwL_bot) at ( 0.42,-1.10) {};
    \node[style=none] (BwR_bot) at ( 0.58,-1.10) {};
    \node[style=none] (Blab) at ( 0.90,-0.55) {$\scriptstyle \m{B}$};

  \end{pgfonlayer}

  \begin{pgfonlayer}{edgelayer}
    \draw (hTL.center) -- (hTR.center);
    \draw (hTR.center) -- (hTT.center);
    \draw (hTT.center) -- (hTL.center);

    \draw[line width=0.6pt] (AwL_top.center) -- (AwL_bot.center);
    \draw[line width=0.6pt] (AwR_top.center) -- (AwR_bot.center);

    \draw[line width=0.6pt] (BwL_top.center) -- (BwL_bot.center);
    \draw[line width=0.6pt] (BwR_top.center) -- (BwR_bot.center);
  \end{pgfonlayer}
\end{tikzpicture}%
}}
\in ({AB})^{*}
\;\Bigm|\;
\vcenter{\hbox{%
\begin{tikzpicture}[baseline={(base.center)}]
  \begin{pgfonlayer}{nodelayer}
    \node[style=none] (base) at (0,0) {};


    \node[style=none] (hTL) at (-0.90,0.80) {};
    \node[style=none] (hTR) at ( 0.90,0.80) {};
    \node[style=none] (hTT) at ( 0.00,1.80) {};
    \node[style=none] (hLab) at ( 0.00,1.25) {$\tilde{h}$};

    \node[style=none] (AwL_top) at (-0.58,0.80) {};
    \node[style=none] (AwR_top) at (-0.42,0.80) {};
    \node[style=none] (AwL_mid) at (-0.58,0.00) {};
    \node[style=none] (AwR_mid) at (-0.42,0.00) {};
    \node[style=none] (Alab) at (-0.92,0.35) {$\scriptstyle \m{A}$};

    \node[style=none] (BwL_top) at ( 0.42,0.80) {};
    \node[style=none] (BwR_top) at ( 0.58,0.80) {};
    \node[style=none] (BwL_mid) at ( 0.42,0.00) {};
    \node[style=none] (BwR_mid) at ( 0.58,0.00) {};
    \node[style=none] (Blab) at ( 0.92,0.35) {$\scriptstyle \m{B}$};

    \node[style=none] (oTL) at (-1.05,0.00) {};
    \node[style=none] (oTR) at (-0.10,0.00) {};
    \node[style=none] (oTC) at (-0.58,-0.95) {};
    \node[style=none] (oLab) at (-0.58,-0.35) {$\omega_i$};

    \node[style=none] (nTL) at ( 0.10,0.00) {};
    \node[style=none] (nTR) at ( 1.05,0.00) {};
    \node[style=none] (nTC) at ( 0.58,-0.95) {};
    \node[style=none] (nLab) at ( 0.58,-0.35) {$\nu_i$};

  \end{pgfonlayer}

  \begin{pgfonlayer}{edgelayer}
    \draw (hTL.center) -- (hTR.center);
    \draw (hTR.center) -- (hTT.center);
    \draw (hTT.center) -- (hTL.center);

    \draw[line width=0.6pt] (AwL_top.center) -- (AwL_mid.center);
    \draw[line width=0.6pt] (AwR_top.center) -- (AwR_mid.center);

    \draw[line width=0.6pt] (BwL_top.center) -- (BwL_mid.center);
    \draw[line width=0.6pt] (BwR_top.center) -- (BwR_mid.center);

    \draw (oTL.center) -- (oTR.center);
    \draw (oTR.center) -- (oTC.center);
    \draw (oTC.center) -- (oTL.center);

    \draw (nTL.center) -- (nTR.center);
    \draw (nTR.center) -- (nTC.center);
    \draw (nTC.center) -- (nTL.center);
  \end{pgfonlayer}
\end{tikzpicture}%
}}
=0
\;\;\forall\;
\vcenter{\hbox{%
\begin{tikzpicture}[baseline={(base.center)}]
  \begin{pgfonlayer}{nodelayer}
    \node[style=none] (base) at (0,0) {};

    \def\DW{0.08}

    \node[style=none] (wTL) at (-1.05,0.00) {};
    \node[style=none] (wTR) at (-0.10,0.00) {};
    \node[style=none] (wTC) at (-0.58,-0.95) {};
    \node[style=none] (wLab) at (-0.58,-0.35) {$\omega_i$};

    \node[style=none] (AwL_bot) at (-0.66,0.00) {};
    \node[style=none] (AwR_bot) at (-0.50,0.00) {};
    \node[style=none] (AwL_top) at (-0.66,0.95) {};
    \node[style=none] (AwR_top) at (-0.50,0.95) {};
    \node[style=none] (Alab) at (-0.92,0.55) {$\scriptstyle \m{A}$};

  \end{pgfonlayer}

  \begin{pgfonlayer}{edgelayer}
    \draw (wTL.center) -- (wTR.center);
    \draw (wTR.center) -- (wTC.center);
    \draw (wTC.center) -- (wTL.center);

    \draw[line width=0.6pt] (AwL_bot.center) -- (AwL_top.center);
    \draw[line width=0.6pt] (AwR_bot.center) -- (AwR_top.center);

  \end{pgfonlayer}
\end{tikzpicture}%
}}
\in \m{S_{A}}, 
\vcenter{\hbox{%
\begin{tikzpicture}[baseline={(base.center)}]
  \begin{pgfonlayer}{nodelayer}
    \node[style=none] (base) at (0,0) {};

    \def\DW{0.08}

    \node[style=none] (vTL) at (-1.05,0.00) {};
    \node[style=none] (vTR) at (-0.10,0.00) {};
    \node[style=none] (vTC) at (-0.58,-0.95) {};
    \node[style=none] (vLab) at (-0.58,-0.35) {$\nu_i$};

    \node[style=none] (BwL_bot) at (-0.66,0.00) {};
    \node[style=none] (BwR_bot) at (-0.50,0.00) {};
    \node[style=none] (BwL_top) at (-0.66,0.95) {};
    \node[style=none] (BwR_top) at (-0.50,0.95) {};
    \node[style=none] (Blab) at (-0.92,0.55) {$\scriptstyle \m{B}$};

  \end{pgfonlayer}

  \begin{pgfonlayer}{edgelayer}
    \draw (vTL.center) -- (vTR.center);
    \draw (vTR.center) -- (vTC.center);
    \draw (vTC.center) -- (vTL.center);

    \draw[line width=0.6pt] (BwL_bot.center) -- (BwL_top.center);
    \draw[line width=0.6pt] (BwR_bot.center) -- (BwR_top.center);

  \end{pgfonlayer}
\end{tikzpicture}%
}}
\in \m{S_{B}} 
\right\}.
\end{equation}
\end{definition} 

Since $(AB)_{\otimes}^*$ is the space generated by linear combination of product effects, the condition $e^{\m{A}}\boxtimes f^{\m{B}}(\tilde{h})=0$ for all $e^{\m{A}}$ and $f^{\m{B}}$ is equivalent to $e(\tilde{h})=0$ for all $e\in (AB_{\otimes})^*$. Thus, the holistic-state subspace is the set of all and only the vectors of $AB$ that evaluate to $0$ on all elements of $(AB)_{\otimes}^*$.
In linear algebra terms,  $H_S$ is  the annihilator of the subspace $(AB)^*_\otimes\subset(AB)^*$; similarly, $H_E$ is the annihilator of the subspace   $AB_\otimes\subset AB$.  Notice the profound implications of these facts: since $u^{\m{AB}} \in (AB)_{\otimes}^*$, then $u^{\m{AB}}(\tilde{h})=0$. That is, the holistic part of the state of a system never contributes to the norm of the state (it has norm 0), and hence there can never be a state that is purely holistic. This fact will be crucial at various points in the manuscript, and we will revisit it later. \blk

One can show that the holistic subspaces are vector space complements to the tomographically-local part of the state and effect spaces, respectively:
\begin{lemma}
\label{lemma:AtensorB+Holistic}
The real vector space associated with any composite $AB$ of two GPT systems $\mathcal{A}$ and $\mathcal{B}$ has the form
    \begin{align}
        AB = (AB_\otimes)\oplus H_S,
    \end{align}
    and
    \begin{align}
        (AB)^* = (AB)^*_\otimes\oplus H_E.
    \end{align}
\end{lemma}
We give a proof in Appendix~\ref{lemma2proof}. 
Lemma~\ref{lemma:AtensorB+Holistic} is very similar to Lemma~\ref{lemma:AtBplusC}, but is stronger: it
shows that one can find a decomposition as in Lemma~\ref{lemma:AtBplusC} where the subspace $C$ has a particularly useful form, namely, where $C$ is the holistic subspace. (Recall that other choices are possible in any tomographically nonlocal theory.)

Note that the holistic-state subspace is defined by its relationship with the tomographically-local subspace of the effect space (namely by being its annihilator), and similarly the holistic-effect subspace is defined by its relationship with the tomographically-local subspace of the state space. In contrast, Lemma~\ref{lemma:AtensorB+Holistic} establishes a relationship between the holistic-state subspace and the tomographically local subspace of the {\em state} space, and similarly for the effect space. 

Thus we have relationships between the holistic subspaces and the tomographically local subspaces, both within their duals (as from the definition) and within the same vector space (as a consequence of  Lemma~\ref{lemma:AtensorB+Holistic}). It is natural to also ask whether one can find a simple relationship between $H_S$ and $H_E$. In Appendix~\ref{Hrels}, we show that introducing some additional structure (namely, an inner product that properly captures the splitting in Lemma~\ref{lemma:AtensorB+Holistic}), one can find a natural isomorphism between $H_E$ and $H_S^*$. With this additional structure, one can think of those subspaces as dual to one another.

The decomposition in Lemma~\ref{lemma:AtensorB+Holistic} implies that we can always write a vector in $AB$ as
\begin{equation}
\label{eq:DecompositionStateGeneral}
\vcenter{\hbox{
\begin{tikzpicture}
	\begin{pgfonlayer}{nodelayer}

		\node [style=none] (L1) at (-4.8,0) {};
		\node [style=none] (R1) at (-3.0,0) {};
		\node [style=none] (C1) at (-3.9,-1.2) {};
		\node [style=none] (W1) at (-3.9,-0.4) {$\omega$};

		\node [style=none] (T1a) at (-4.4,0) {};
		\node [style=none] (T1b) at (-3.4,0) {};

		\node [style=none] (T1aL) at (-4.44,0) {};
		\node [style=none] (T1aR) at (-4.36,0) {};
		\node [style=none] (T1bL) at (-3.44,0) {};
		\node [style=none] (T1bR) at (-3.36,0) {};

		\node [style=none] (A1) at (-4.15,0.5) {$\scriptstyle \m{A}$};
		\node [style=none] (B1) at (-3.15,0.5) {$\scriptstyle \m{B}$};

		\node [style=none] (eq)  at (-2.1,0) {$=$};
		\node [style=none] (Sij) at (-0.35,0) {$\sum_{ij} r_{ij}$};

		\node [style=none] (L2) at (0.7,0) {};
		\node [style=none] (R2) at (2.2,0) {};
		\node [style=none] (C2) at (1.45,-1.0) {};
		\node [style=none] (W2) at (1.45,-0.33) {$\omega_i$};

		\node [style=none] (T2) at (1.45,0) {};
		\node [style=none] (T2L) at (1.41,0) {};
		\node [style=none] (T2R) at (1.49,0) {};
		\node [style=none] (A2) at (1.70,0.42) {$\scriptstyle \m{A}$};

		\node [style=none] (L3) at (2.4,0) {};
		\node [style=none] (R3) at (3.9,0) {};
		\node [style=none] (C3) at (3.15,-1.0) {};
		\node [style=none] (W3) at (3.15,-0.33) {$\omega_j$};

		\node [style=none] (T3) at (3.15,0) {};
		\node [style=none] (T3L) at (3.11,0) {};
		\node [style=none] (T3R) at (3.19,0) {};
		\node [style=none] (B3) at (3.40,0.42) {$\scriptstyle \m{B}$};

		\node [style=none] (plus) at (4.6,0) {$+$};
		\node [style=none] (St) at (5.7,0) {$\sum_{t} r_{t}$};

		\node [style=none] (L4) at (6.8,0) {};
		\node [style=none] (R4) at (8.9,0) {};
		\node [style=none] (C4) at (7.85,-1.0) {};
		\node [style=none] (W4) at (7.85,-0.35) {$\tilde{t}$};

		\node [style=none] (T4a) at (7.35,0) {};
		\node [style=none] (T4b) at (8.35,0) {};

		\node [style=none] (T4aL) at (7.31,0) {};
		\node [style=none] (T4aR) at (7.39,0) {};
		\node [style=none] (T4bL) at (8.31,0) {};
		\node [style=none] (T4bR) at (8.39,0) {};

		\node [style=none] (A4) at (7.60,0.5) {$\scriptstyle \m{A}$};
		\node [style=none] (B4) at (8.60,0.5) {$\scriptstyle \m{B}$};

	\end{pgfonlayer}

	\begin{pgfonlayer}{edgelayer}

		\draw (L1.center) to (R1.center);
		\draw (R1.center) to (C1.center);
		\draw (C1.center) to (L1.center);

		\draw (L2.center) to (R2.center);
		\draw (R2.center) to (C2.center);
		\draw (C2.center) to (L2.center);

		\draw (L3.center) to (R3.center);
		\draw (R3.center) to (C3.center);
		\draw (C3.center) to (L3.center);

		\draw (L4.center) to (R4.center);
		\draw (R4.center) to (C4.center);
		\draw (C4.center) to (L4.center);

		\draw[line width=0.6pt] (T1aL.center) to +(0,0.9);
		\draw[line width=0.6pt] (T1aR.center) to +(0,0.9);

		\draw[line width=0.6pt] (T1bL.center) to +(0,0.9);
		\draw[line width=0.6pt] (T1bR.center) to +(0,0.9);

		\draw[line width=0.6pt] (T2L.center) to +(0,0.9);
		\draw[line width=0.6pt] (T2R.center) to +(0,0.9);

		\draw[line width=0.6pt] (T3L.center) to +(0,0.9);
		\draw[line width=0.6pt] (T3R.center) to +(0,0.9);

		\draw[line width=0.6pt] (T4aL.center) to +(0,0.9);
		\draw[line width=0.6pt] (T4aR.center) to +(0,0.9);

		\draw[line width=0.6pt] (T4bL.center) to +(0,0.9);
		\draw[line width=0.6pt] (T4bR.center) to +(0,0.9);

	\end{pgfonlayer}
\end{tikzpicture}
}}
\end{equation}
for some real numbers $r_{ij}$ and $r_t$, and similarly for effects.

\subsection{Projections onto the tomographically local and holistic subspaces}

One object that will play a particularly important role in our endeavor is the following operator:

\begin{definition}[Projection onto the tomographically local subspace]
\label{def:PiTL}
Let $\mathcal{AB}$ be a composite system with vector space decomposition
$AB = AB_\otimes \oplus H_S$
where $AB_\otimes$ is the tomographically local subspace and $H_S$ is the
holistic subspace. Writing a generic vector $v \in AB$ as $v=v_{TL}+h$ for some $v_{TL}\in (AB)_\otimes$ and $h\in H_S$, we can define the \emph{projection onto the tomographically local subspace} as the linear map $\Pi_{TL} : AB \to AB$
given by
\begin{align}
    \Pi_{TL}(v_{TL}+h) = v_{TL}.
    \end{align}
\end{definition}
\noindent That is, $\Pi_{\rm TL}$ just kills off the holistic component  of a given vector in $AB$, while preserving vectors in the span of product states. This implies that $\Pi_{\rm TL}^*: (AB)^*\to (AB)^*$, defined by $\Pi_{\rm TL}^*(w)=w\circ \Pi_{\rm TL}$ for any element $w\in (AB)^*$, also kills off the holistic component in the effect space $(AB)^*$ (while preserving vectors in the span of product effects), as we see next.
\begin{proposition}
[Dual of $\Pi_{\rm TL}$] \label{prop:DualPiTL}
The operator $\Pi_{\rm TL}^*: (AB)^*\to (AB)^*$ defined by 
\begin{align}
    \Pi_{\rm TL}^*(w)=w\circ \Pi_{\rm TL}.
\end{align}
is the projector onto the tomographically-local subspace $(AB)^*_\otimes$. 
\end{proposition}
\begin{proof}
    
Since $(AB)^* = (AB)^*_\otimes \oplus H_E$, every $w\in (AB)^*$ decomposes uniquely as $w = w_{\rm TL} + h'$ for some 
$w_{\rm TL}\in (AB)^*_\otimes$ and $h'\in H_E$. 
Then, for any $v=v_{\rm TL}+h \in AB$ we have
\begin{align}
    \Pi_{\rm TL}^*(w)(v) &= (w \circ \Pi_{\rm TL})(v) \\
    &= w(\Pi_{\rm TL}(v)) \\
    &= w(v_{\rm TL}) \\
    &= w_{\rm TL}(v_{\rm TL}) + h'(v_{\rm TL}).
\end{align}
Since $h'$ annihilates $AB_\otimes$, the last term vanishes, and hence
\begin{align}
    \Pi_{\rm TL}^*(w)(v) = w_{\rm TL}(v) \quad \forall v \in AB.
\end{align}
Therefore, $\Pi_{\rm TL}^*$ coincides with the projection $w \mapsto w_{\rm TL}$ 
of $(AB)^*$ onto $(AB)^*_\otimes$ along $H_E$. Equivalently, $\Pi_{\rm TL}^*$ is 
the projection onto the tomographically local subspace in the dual space; i.e., 
it just kills off the holistic component of any functional in $(AB)^*$.\footnote{Note that in this proposition we have proven that the kernel of $\Pi_{\rm TL}^*$ is $H_E$, which is by definition the annihilator of $AB_\otimes$ (recall Def.~\ref{def:HolisticSubspaces}). Moreover, $AB_\otimes$ is the image of $\Pi_{\rm TL}$, as it is the projector onto that subspace. All of this constitute a special case of the general fact that for any linear map $\phi: V\to W$, $\mathsf{ker}[\phi^*]=(\mathsf{Im}[\phi])^0$, where $\mathsf{Im}[\phi]^0\subset W^*$ is the annihilator of $\mathsf{Im}[\phi]\subset W$. }\blk

\end{proof}

$\Pi_{\rm TL}$ has some useful properties we list below. These properties follow directly from the fact that $\Pi_{\rm TL}$ is a projection onto the subspace $AB_\otimes$ (and $\Pi_{\rm TL}^*$ is a projection onto $(AB)^*_\otimes$), and we provide explicit proofs in Appendix~\ref{projprop} (see Prop.~\ref{prop16}). 
\begin{enumerate}
\item $\Pi_{\rm TL}$ acts trivially on $AB_{\otimes}$ and (by pre-composition) on $(AB)^*_\otimes$. Thus, $\Pi_{\rm TL}(\omega^{\m{AB}})=\omega^{\m{AB}}$ for all states ${\omega^{\m{AB}}\in AB_\otimes}$ and $e^{\m{AB}}\circ\Pi_{\rm TL}=e^{\m{AB}}$ for all effects $e\in (AB)^*_\otimes$.

\item  $\Pi_{\rm TL}$ preserves normalization: $u^{\m{AB}}\circ \Pi_{\rm TL} = u^{\m{AB}}$.

\item The only state $\omega^{\m{AB}}\in \Omega_{\m{AB}}$ for which $\Pi_{\rm TL}(\omega^{\m{AB}})=0^{AB}$ is $\omega^{\m{AB}}=0^{AB}$. 

\item If $e^{\m{AB}}\left[\Pi_{\rm TL}\left(\omega^{\m{AB}}\right)\right]\neq e^{\m{AB}}[\omega^{\m{AB}}]$, then both the effect $e^{\m{AB}}$ and the state $\omega^{\m{AB}}$ have holistic components; i.e., $e^{\m{AB}}\not\in (AB)^*_\otimes$ and  $\omega^{\m{AB}}\not\in AB_\otimes$. 
\end{enumerate}
Note that $\Pi_{\rm TL}$ takes all states in $S_{\m{AB}}$ to vectors \blk in $AB_\otimes$, but there is no guarantee that $\Pi_{\rm TL}(\omega)$ is itself a state; in general, it will not be the case. We give one such example in Section~\ref{sec:CompositesRQT}. A similar fact follows for effects. Nevertheless, since separable states and effects always belong to $AB_\otimes$ and $(AB)^*_\otimes$ (resp.), property $1$ immediately implies the following:
\begin{corollary}
    For any composite GPT system $\m{AB}$, the projection $\Pi_{\rm TL}$ acts trivially on separable states and separable effects.
\end{corollary}

It is insightful to write $\Pi_{\rm TL}$ more explicitly, as follows.

\begin{proposition}
\label{propPiTL}
The tomographically local projection for a composite system $\mathcal{AB}$ can be written
in two equivalent ways:
\begin{enumerate}
    \item Local vectors forming basis and dual basis: for any given basis of $A$ (whose elements might not be states), $\{v^A_i\}_i\subset A$, take its dual basis in $A^*$ (which, again, might contain vectors that are not effects), $\{t^{A^*}_{i'}\}\subset A^*$. Similarly, given any basis for $B$, $\{v^B_j\}$ take the dual basis $\{t^{B^*}_{j'}\}$. Then,
\begin{align}
\label{eq:Pitl_DualBases}
    \Pi_{\rm TL} (\bullet) = (v^{A}_i\boxtimes v^{B}_j)\circ  [t^{A^*}_i\boxtimes t^{B^*}_j](\bullet).
\end{align}
Diagrammaticaly,
\begin{equation}
\label{eq:Pi_TL_decomposition}
\vcenter{\hbox{%
\begin{tikzpicture}
  \begin{pgfonlayer}{nodelayer}

    \def\dw{0.08}   
    \def\gap{0.16}  

    \def\xLA{-4.6}
    \def\xLB{-3.2}

    \def\yTop{ 1.5}
    \def\yMidTop{ 0.6}
    \def\yMidBot{-0.6}
    \def\yBot{-1.5}

    \node[style=none] (LTA1) at ({\xLA-\dw}, \yTop) {};
    \node[style=none] (LTA2) at ({\xLA-\dw}, \yMidTop) {};
    \node[style=none] (LBA1) at ({\xLA-\dw}, \yMidBot) {};
    \node[style=none] (LBA2) at ({\xLA-\dw}, \yBot) {};

    \node[style=none] (LTA3) at ({\xLA+\dw}, \yTop) {};
    \node[style=none] (LTA4) at ({\xLA+\dw}, \yMidTop) {};
    \node[style=none] (LBA3) at ({\xLA+\dw}, \yMidBot) {};
    \node[style=none] (LBA4) at ({\xLA+\dw}, \yBot) {};

    \node[style=none] (LTB1) at ({\xLB-\dw}, \yTop) {};
    \node[style=none] (LTB2) at ({\xLB-\dw}, \yMidTop) {};
    \node[style=none] (LBB1) at ({\xLB-\dw}, \yMidBot) {};
    \node[style=none] (LBB2) at ({\xLB-\dw}, \yBot) {};

    \node[style=none] (LTB3) at ({\xLB+\dw}, \yTop) {};
    \node[style=none] (LTB4) at ({\xLB+\dw}, \yMidTop) {};
    \node[style=none] (LBB3) at ({\xLB+\dw}, \yMidBot) {};
    \node[style=none] (LBB4) at ({\xLB+\dw}, \yBot) {};

    \node[style=none] (LLabA_top) at ({\xLA+0.30},  1.00) {$\scriptstyle \m{A}$};
    \node[style=none] (LLabB_top) at ({\xLB+0.30},  1.00) {$\scriptstyle \m{B}$};
    \node[style=none] (LLabA_bot) at ({\xLA+0.30}, -1.00) {$\scriptstyle \m{A}$};
    \node[style=none] (LLabB_bot) at ({\xLB+0.30}, -1.00) {$\scriptstyle \m{B}$};

    \def\boxPad{0.35}
    \node[style=none] (LR1) at ({\xLA-\boxPad}, \yMidTop) {};
    \node[style=none] (LR2) at ({\xLB+\boxPad}, \yMidTop) {};
    \node[style=none] (LR3) at ({\xLB+\boxPad}, \yMidBot) {};
    \node[style=none] (LR4) at ({\xLA-\boxPad}, \yMidBot) {};
    \node[style=none] (LBoxLab) at ({(\xLA+\xLB)/2}, 0) {$\Pi_{TL}$};

    \node[style=none] (Eq)  at (-1.6,0) {$=$};
    \node[style=none] (Sum) at (-0.3,0) {$\sum_{i,j}$};

    \def\xRA{ 1.2}
    \def\xRB{ 3.0}

    \def\wTri{0.55}
    \def\yBaseTop{ 1.05}  
    \def\yBaseBot{-1.05}  
    \def\yWireTop{ 1.55}
    \def\yWireBot{-1.55}

    \node[style=none] (ATipTop) at (\xRA, { \gap/2}) {};
    \node[style=none] (ATipBot) at (\xRA, {- \gap/2}) {};

    \node[style=none] (BTipTop) at (\xRB, { \gap/2}) {};
    \node[style=none] (BTipBot) at (\xRB, {- \gap/2}) {};

    \node[style=none] (AeL) at ({\xRA-\wTri}, \yBaseTop) {};
    \node[style=none] (AeR) at ({\xRA+\wTri}, \yBaseTop) {};
    \node[style=none] (AeLab) at (\xRA, 0.62) {$v_i$};

    \node[style=none] (AtL) at ({\xRA-\wTri}, \yBaseBot) {};
    \node[style=none] (AtR) at ({\xRA+\wTri}, \yBaseBot) {};
    \node[style=none] (AtLab) at (\xRA,-0.62) {$t_i$};

    \node[style=none] (AUpL)  at ({\xRA-\dw}, \yWireTop) {};
    \node[style=none] (AUpR)  at ({\xRA+\dw}, \yWireTop) {};
    \node[style=none] (AUpL2) at ({\xRA-\dw}, \yBaseTop) {};
    \node[style=none] (AUpR2) at ({\xRA+\dw}, \yBaseTop) {};

    \node[style=none] (ADnL)  at ({\xRA-\dw}, \yBaseBot) {};
    \node[style=none] (ADnR)  at ({\xRA+\dw}, \yBaseBot) {};
    \node[style=none] (ADnL2) at ({\xRA-\dw}, \yWireBot) {};
    \node[style=none] (ADnR2) at ({\xRA+\dw}, \yWireBot) {};

    \node[style=none] (RLabA_top) at ({\xRA+0.35},  1.35) {$\scriptstyle \m{A}$};
    \node[style=none] (RLabA_bot) at ({\xRA+0.35}, -1.35) {$\scriptstyle \m{A}$};

    \node[style=none] (BeL) at ({\xRB-\wTri}, \yBaseTop) {};
    \node[style=none] (BeR) at ({\xRB+\wTri}, \yBaseTop) {};
    \node[style=none] (BeLab) at (\xRB, 0.62) {$v_j$};

    \node[style=none] (BtL) at ({\xRB-\wTri}, \yBaseBot) {};
    \node[style=none] (BtR) at ({\xRB+\wTri}, \yBaseBot) {};
    \node[style=none] (BtLab) at (\xRB,-0.62) {$t_j$};

    \node[style=none] (BUpL)  at ({\xRB-\dw}, \yWireTop) {};
    \node[style=none] (BUpR)  at ({\xRB+\dw}, \yWireTop) {};
    \node[style=none] (BUpL2) at ({\xRB-\dw}, \yBaseTop) {};
    \node[style=none] (BUpR2) at ({\xRB+\dw}, \yBaseTop) {};

    \node[style=none] (BDnL)  at ({\xRB-\dw}, \yBaseBot) {};
    \node[style=none] (BDnR)  at ({\xRB+\dw}, \yBaseBot) {};
    \node[style=none] (BDnL2) at ({\xRB-\dw}, \yWireBot) {};
    \node[style=none] (BDnR2) at ({\xRB+\dw}, \yWireBot) {};

    \node[style=none] (RLabB_top) at ({\xRB+0.35},  1.35) {$\scriptstyle \m{B}$};
    \node[style=none] (RLabB_bot) at ({\xRB+0.35}, -1.35) {$\scriptstyle \m{B}$};

  \end{pgfonlayer}

  \begin{pgfonlayer}{edgelayer}

    \draw (LTA1.center) -- (LTA2.center);
    \draw (LTA3.center) -- (LTA4.center);
    \draw (LBA1.center) -- (LBA2.center);
    \draw (LBA3.center) -- (LBA4.center);

    \draw (LTB1.center) -- (LTB2.center);
    \draw (LTB3.center) -- (LTB4.center);
    \draw (LBB1.center) -- (LBB2.center);
    \draw (LBB3.center) -- (LBB4.center);

    \draw (LR1.center) -- (LR2.center) -- (LR3.center) -- (LR4.center) -- cycle;

    \draw (AeL.center) -- (AeR.center) -- (ATipTop.center) -- cycle;
    \draw (AtL.center) -- (AtR.center) -- (ATipBot.center) -- cycle;

    \draw (BeL.center) -- (BeR.center) -- (BTipTop.center) -- cycle;
    \draw (BtL.center) -- (BtR.center) -- (BTipBot.center) -- cycle;

    \draw (AUpL2.center) -- (AUpL.center);
    \draw (AUpR2.center) -- (AUpR.center);
    \draw (ADnL.center) -- (ADnL2.center);
    \draw (ADnR.center) -- (ADnR2.center);

    \draw (BUpL2.center) -- (BUpL.center);
    \draw (BUpR2.center) -- (BUpR.center);
    \draw (BDnL.center) -- (BDnL2.center);
    \draw (BDnR.center) -- (BDnR2.center);

  \end{pgfonlayer}
\end{tikzpicture}%
}}
\end{equation}

    \item Linear combination of effect-state channels: Consider two sets $\{\omega^{\m{A}}_i\}_{i=1}^{d_A}\subset \m{S_A}$ and $\{\omega^B_k\}_{k=1}^{d_B}\subset \m{S_B}$ of states forming bases for the local vector spaces $A$ and $B$ and sets of effects $\{e^{\m{A}}_j\}_{j=1}^{d_A}$ and $\{e^{\m{B}}_l\}_{l=1}^{d_B}$ forming bases for the local vector spaces $A^*$ and $B^*$. Then, $\Pi_{\rm TL}$ can be written as

    \begin{align}
    \label{eq:PiTL_EffectState}
        &\Pi_{\rm TL}(\bullet) = \sum_{i,j,k,l}c^{\m{A}}_{ij}c^{\m{B}}_{kl}\left(\omega^{\m{A}}_i\boxtimes \omega^B_k \right)\circ \left[e^{\m{A}}_j\boxtimes e^{\m{B}}_l\right](\bullet)
    \end{align}
  where  $c_{ij}^{\m{A}}$ are the coefficients of the inverse of matrix $\mathds{M}^{\m{A}}$, where $(\mathds{M}^{\m{A}})_{ij}=e^{\m{A}}_j(\omega^{\m{A}}_i)$, and similarly for $c_{kl}^{\m{B}}$. Diagrammatically,
  \begin{equation}
\label{eq:Pi_TL_coeff_decomposition}
\vcenter{\hbox{%
\begin{tikzpicture}
  \begin{pgfonlayer}{nodelayer}

    \def\dw{0.08}    
    \def\gap{0.16}   

    \def\xLA{-4.6}
    \def\xLB{-3.2}

    \def\yTop{ 1.5}
    \def\yMidTop{ 0.6}
    \def\yMidBot{-0.6}
    \def\yBot{-1.5}

    \node[style=none] (LTA1) at ({\xLA-\dw}, \yTop) {};
    \node[style=none] (LTA2) at ({\xLA-\dw}, \yMidTop) {};
    \node[style=none] (LBA1) at ({\xLA-\dw}, \yMidBot) {};
    \node[style=none] (LBA2) at ({\xLA-\dw}, \yBot) {};

    \node[style=none] (LTA3) at ({\xLA+\dw}, \yTop) {};
    \node[style=none] (LTA4) at ({\xLA+\dw}, \yMidTop) {};
    \node[style=none] (LBA3) at ({\xLA+\dw}, \yMidBot) {};
    \node[style=none] (LBA4) at ({\xLA+\dw}, \yBot) {};

    \node[style=none] (LTB1) at ({\xLB-\dw}, \yTop) {};
    \node[style=none] (LTB2) at ({\xLB-\dw}, \yMidTop) {};
    \node[style=none] (LBB1) at ({\xLB-\dw}, \yMidBot) {};
    \node[style=none] (LBB2) at ({\xLB-\dw}, \yBot) {};

    \node[style=none] (LTB3) at ({\xLB+\dw}, \yTop) {};
    \node[style=none] (LTB4) at ({\xLB+\dw}, \yMidTop) {};
    \node[style=none] (LBB3) at ({\xLB+\dw}, \yMidBot) {};
    \node[style=none] (LBB4) at ({\xLB+\dw}, \yBot) {};

    \node[style=none] (LLabA_top) at ({\xLA+0.30},  1.00) {$\scriptstyle \m{A}$};
    \node[style=none] (LLabB_top) at ({\xLB+0.30},  1.00) {$\scriptstyle \m{B}$};
    \node[style=none] (LLabA_bot) at ({\xLA+0.30}, -1.00) {$\scriptstyle \m{A}$};
    \node[style=none] (LLabB_bot) at ({\xLB+0.30}, -1.00) {$\scriptstyle \m{B}$};

    \def\boxPad{0.35}
    \node[style=none] (LR1) at ({\xLA-\boxPad}, \yMidTop) {};
    \node[style=none] (LR2) at ({\xLB+\boxPad}, \yMidTop) {};
    \node[style=none] (LR3) at ({\xLB+\boxPad}, \yMidBot) {};
    \node[style=none] (LR4) at ({\xLA-\boxPad}, \yMidBot) {};
    \node[style=none] (LBoxLab) at ({(\xLA+\xLB)/2}, 0) {$\Pi_{TL}$};

    \node[style=none] (Eq)   at (-1.6,0) {$=$};
    \node[style=none] (Sum)  at (-0.55,0) {$\sum$};
    \node[style=none] (Coef) at (0.70,0.0) {$c^{\m{A}}_{ij}\,c^{\m{B}}_{k\ell}$};

    \def\xRA{ 2.3}   
    \def\xRB{ 4.1}   

    \def\wTri{0.65}
    \def\yBaseTop{ 1.05}
    \def\yBaseBot{-1.05}
    \def\yWireTop{ 1.55}
    \def\yWireBot{-1.55}

    \node[style=none] (ATipTop) at (\xRA, { \gap/2}) {};
    \node[style=none] (ATipBot) at (\xRA, {- \gap/2}) {};

    \node[style=none] (BTipTop) at (\xRB, { \gap/2}) {};
    \node[style=none] (BTipBot) at (\xRB, {- \gap/2}) {};

    \node[style=none] (AomL) at ({\xRA-\wTri}, \yBaseTop) {};
    \node[style=none] (AomR) at ({\xRA+\wTri}, \yBaseTop) {};
    \node[style=none] (AomLab) at (\xRA, 0.65) {$\omega_i$};

    \node[style=none] (AeL) at ({\xRA-\wTri}, \yBaseBot) {};
    \node[style=none] (AeR) at ({\xRA+\wTri}, \yBaseBot) {};
    \node[style=none] (AeLab) at (\xRA,-0.65) {$e_j$};

    \node[style=none] (AUpL)  at ({\xRA-\dw}, \yWireTop) {};
    \node[style=none] (AUpR)  at ({\xRA+\dw}, \yWireTop) {};
    \node[style=none] (AUpL2) at ({\xRA-\dw}, \yBaseTop) {};
    \node[style=none] (AUpR2) at ({\xRA+\dw}, \yBaseTop) {};

    \node[style=none] (ADnL)  at ({\xRA-\dw}, \yBaseBot) {};
    \node[style=none] (ADnR)  at ({\xRA+\dw}, \yBaseBot) {};
    \node[style=none] (ADnL2) at ({\xRA-\dw}, \yWireBot) {};
    \node[style=none] (ADnR2) at ({\xRA+\dw}, \yWireBot) {};

    \node[style=none] (RLabA_top) at ({\xRA+0.35},  1.25) {$\scriptstyle \m{A}$};
    \node[style=none] (RLabA_bot) at ({\xRA+0.35}, -1.25) {$\scriptstyle \m{A}$};

    \node[style=none] (BomL) at ({\xRB-\wTri}, \yBaseTop) {};
    \node[style=none] (BomR) at ({\xRB+\wTri}, \yBaseTop) {};
    \node[style=none] (BomLab) at (\xRB, 0.65) {$\omega_k$};

    \node[style=none] (BeL) at ({\xRB-\wTri}, \yBaseBot) {};
    \node[style=none] (BeR) at ({\xRB+\wTri}, \yBaseBot) {};
    \node[style=none] (BeLab) at (\xRB,-0.65) {$e_\ell$};

    \node[style=none] (BUpL)  at ({\xRB-\dw}, \yWireTop) {};
    \node[style=none] (BUpR)  at ({\xRB+\dw}, \yWireTop) {};
    \node[style=none] (BUpL2) at ({\xRB-\dw}, \yBaseTop) {};
    \node[style=none] (BUpR2) at ({\xRB+\dw}, \yBaseTop) {};

    \node[style=none] (BDnL)  at ({\xRB-\dw}, \yBaseBot) {};
    \node[style=none] (BDnR)  at ({\xRB+\dw}, \yBaseBot) {};
    \node[style=none] (BDnL2) at ({\xRB-\dw}, \yWireBot) {};
    \node[style=none] (BDnR2) at ({\xRB+\dw}, \yWireBot) {};

    \node[style=none] (RLabB_top) at ({\xRB+0.35},  1.25) {$\scriptstyle \m{B}$};
    \node[style=none] (RLabB_bot) at ({\xRB+0.35}, -1.25) {$\scriptstyle \m{B}$};

  \end{pgfonlayer}

  \begin{pgfonlayer}{edgelayer}

    \draw (LTA1.center) -- (LTA2.center);
    \draw (LTA3.center) -- (LTA4.center);
    \draw (LBA1.center) -- (LBA2.center);
    \draw (LBA3.center) -- (LBA4.center);

    \draw (LTB1.center) -- (LTB2.center);
    \draw (LTB3.center) -- (LTB4.center);
    \draw (LBB1.center) -- (LBB2.center);
    \draw (LBB3.center) -- (LBB4.center);

    \draw (LR1.center) -- (LR2.center) -- (LR3.center) -- (LR4.center) -- cycle;

    \draw (AomL.center) -- (AomR.center) -- (ATipTop.center) -- cycle;
    \draw (AeL.center)  -- (AeR.center)  -- (ATipBot.center) -- cycle;

    \draw (BomL.center) -- (BomR.center) -- (BTipTop.center) -- cycle;
    \draw (BeL.center)  -- (BeR.center)  -- (BTipBot.center) -- cycle;

    \draw (AUpL2.center) -- (AUpL.center);
    \draw (AUpR2.center) -- (AUpR.center);
    \draw (ADnL.center) -- (ADnL2.center);
    \draw (ADnR.center) -- (ADnR2.center);

    \draw (BUpL2.center) -- (BUpL.center);
    \draw (BUpR2.center) -- (BUpR.center);
    \draw (BDnL.center) -- (BDnL2.center);
    \draw (BDnR.center) -- (BDnR2.center);

  \end{pgfonlayer}
\end{tikzpicture}%
}}
\end{equation}

  \end{enumerate} 

\end{proposition}

The proof of Proposition~\ref{propPiTL} is given in Appendix~\ref{projprop}. Here we provide a brief sketch and indicate how these decompositions lead to an intuitive understanding of $\Pi_{\rm TL}$. The decomposition in terms of bases of $A$ and $B$ and their dual bases follows from the defining property of dual bases, namely that $t^{A^*}_{i'}(v^A_i)=\delta_{ii'}$. This ensures that the operator in Eq.~\eqref{eq:Pitl_DualBases} acts trivially on any component $v_{\rm TL}\in AB_{\otimes}$ of a vector $v\in AB$. Moreover, it is a linear combination of elements belonging to the annihilator of $H_S$ (and an analogous argument applies to its dual action on $AB^*$), and therefore maps every component in $H_S$ to $0^{AB}$. The characterization in terms of effect--state channels follows from the previous one by specializing to the case in which $\{v_i^A\}$ and $\{v_j^B\}$ are sets of states and expanding the effects in terms of their corresponding dual bases.

We now turn to the intuition behind the action of $\Pi_{\rm TL}$ and its connection to tomographic locality, focusing in particular on Eq.~\eqref{eq:Pi_TL_coeff_decomposition}. As shown in Refs.~\cite{Schmid2024structuretheorem,hardy2011reformulatingreconstructingquantumtheory}, this specific linear combination of effect--state pairs coincides with the identity transformation in tomographically-local theories. That is, in a TL theory, one has

\begin{equation}
\label{eq:local_tomography_completeness}
\vcenter{\hbox{%
\begin{tikzpicture}
  \begin{pgfonlayer}{nodelayer}

    \def\dw{0.08}    
    \def\gap{0.16}   

    \node[style=none] (Sum)  at (-3.2,0.0) {$\sum_{i,j}$};
    \node[style=none] (Coef) at (-1.9,0.0) {$c^{\m{A}}_{ij}$};

    \def\xA{-0.6}

    \def\wTri{0.65}
    \def\yBaseTop{ 0.85}   
    \def\yBaseBot{-0.85}   
    \def\yWireTop{ 1.35}
    \def\yWireBot{-1.35}

    \node[style=none] (TipTop) at (\xA, { \gap/2}) {};
    \node[style=none] (TipBot) at (\xA, {- \gap/2}) {};

    \node[style=none] (TopL) at ({\xA-\wTri}, \yBaseTop) {};
    \node[style=none] (TopR) at ({\xA+\wTri}, \yBaseTop) {};
    \node[style=none] (TopLab) at (\xA, 0.50) {$\omega_i$};

    \node[style=none] (BotL) at ({\xA-\wTri}, \yBaseBot) {};
    \node[style=none] (BotR) at ({\xA+\wTri}, \yBaseBot) {};
    \node[style=none] (BotLab) at (\xA,-0.52) {$e_j$};

    \node[style=none] (UpL)  at ({\xA-\dw}, \yWireTop) {};
    \node[style=none] (UpR)  at ({\xA+\dw}, \yWireTop) {};
    \node[style=none] (UpL2) at ({\xA-\dw}, \yBaseTop) {};
    \node[style=none] (UpR2) at ({\xA+\dw}, \yBaseTop) {};

    \node[style=none] (DnL)  at ({\xA-\dw}, \yBaseBot) {};
    \node[style=none] (DnR)  at ({\xA+\dw}, \yBaseBot) {};
    \node[style=none] (DnL2) at ({\xA-\dw}, \yWireBot) {};
    \node[style=none] (DnR2) at ({\xA+\dw}, \yWireBot) {};

    \node[style=none] (LabTop) at ({\xA+0.35},  1.05) {$\scriptstyle \m{A}$};
    \node[style=none] (LabBot) at ({\xA+0.35}, -1.05) {$\scriptstyle \m{A}$};

    \node[style=none] (Eq)   at (1.0,0.0) {$=$};
    \node[style=none] (EqTL) at (1.0,0.45) {$\scriptstyle \mathrm{TL}$};

    \def\xId{2.6}

    \node[style=none] (IdUpL)  at ({\xId-\dw},  1.35) {};
    \node[style=none] (IdUpR)  at ({\xId+\dw},  1.35) {};
    \node[style=none] (IdDnL)  at ({\xId-\dw}, -1.35) {};
    \node[style=none] (IdDnR)  at ({\xId+\dw}, -1.35) {};

    \node[style=none] (IdLab)  at ({\xId+0.35}, 0.0) {$\scriptstyle \m{A}$};

  \end{pgfonlayer}

  \begin{pgfonlayer}{edgelayer}

    \draw (TopL.center) -- (TopR.center) -- (TipTop.center) -- cycle;
    \draw (BotL.center) -- (BotR.center) -- (TipBot.center) -- cycle;

    \draw (UpL2.center) -- (UpL.center);
    \draw (UpR2.center) -- (UpR.center);
    \draw (DnL.center)  -- (DnL2.center);
    \draw (DnR.center)  -- (DnR2.center);

    \draw (IdUpL.center) -- (IdDnL.center);
    \draw (IdUpR.center) -- (IdDnR.center);

  \end{pgfonlayer}
\end{tikzpicture}%
}}
\end{equation}
while in a general theory one has
\begin{equation}
\label{eq:EffectStateNOTId}
\vcenter{\hbox{%
\begin{tikzpicture}
  \begin{pgfonlayer}{nodelayer}

    \def\dw{0.08}    
    \def\gap{0.16}   

    \node[style=none] (Sum)  at (-3.2,0.0) {$\sum_{i,j}$};
    \node[style=none] (Coef) at (-1.9,0.0) {$c^{\m{A}}_{ij}$};

    \def\xA{-0.6}

    \def\wTri{0.65}
    \def\yBaseTop{ 0.85}   
    \def\yBaseBot{-0.85}   
    \def\yWireTop{ 1.35}
    \def\yWireBot{-1.35}

    \node[style=none] (TipTop) at (\xA, { \gap/2}) {};
    \node[style=none] (TipBot) at (\xA, {- \gap/2}) {};

    \node[style=none] (TopL) at ({\xA-\wTri}, \yBaseTop) {};
    \node[style=none] (TopR) at ({\xA+\wTri}, \yBaseTop) {};
    \node[style=none] (TopLab) at (\xA, 0.50) {$\omega_i$};

    \node[style=none] (BotL) at ({\xA-\wTri}, \yBaseBot) {};
    \node[style=none] (BotR) at ({\xA+\wTri}, \yBaseBot) {};
    \node[style=none] (BotLab) at (\xA,-0.52) {$e_j$};

    \node[style=none] (UpL)  at ({\xA-\dw}, \yWireTop) {};
    \node[style=none] (UpR)  at ({\xA+\dw}, \yWireTop) {};
    \node[style=none] (UpL2) at ({\xA-\dw}, \yBaseTop) {};
    \node[style=none] (UpR2) at ({\xA+\dw}, \yBaseTop) {};

    \node[style=none] (DnL)  at ({\xA-\dw}, \yBaseBot) {};
    \node[style=none] (DnR)  at ({\xA+\dw}, \yBaseBot) {};
    \node[style=none] (DnL2) at ({\xA-\dw}, \yWireBot) {};
    \node[style=none] (DnR2) at ({\xA+\dw}, \yWireBot) {};

    \node[style=none] (LabTop) at ({\xA+0.35},  1.05) {$\scriptstyle \m{A}$};
    \node[style=none] (LabBot) at ({\xA+0.35}, -1.05) {$\scriptstyle \m{A}$};

    \node[style=none] (Eq)   at (1.0,0.0) {$\neq$};

    \def\xId{2.6}

    \node[style=none] (IdUpL)  at ({\xId-\dw},  1.35) {};
    \node[style=none] (IdUpR)  at ({\xId+\dw},  1.35) {};
    \node[style=none] (IdDnL)  at ({\xId-\dw}, -1.35) {};
    \node[style=none] (IdDnR)  at ({\xId+\dw}, -1.35) {};

    \node[style=none] (IdLab)  at ({\xId+0.35}, 0.0) {$\scriptstyle \m{A}$};

  \end{pgfonlayer}

  \begin{pgfonlayer}{edgelayer}

    \draw (TopL.center) -- (TopR.center) -- (TipTop.center) -- cycle;
    \draw (BotL.center) -- (BotR.center) -- (TipBot.center) -- cycle;

    \draw (UpL2.center) -- (UpL.center);
    \draw (UpR2.center) -- (UpR.center);
    \draw (DnL.center)  -- (DnL2.center);
    \draw (DnR.center)  -- (DnR2.center);

    \draw (IdUpL.center) -- (IdDnL.center);
    \draw (IdUpR.center) -- (IdDnR.center);

  \end{pgfonlayer}
\end{tikzpicture}%
}}
\end{equation}

In fact, such a channel might not even be valid, ie, might not be a physical operation. However, in every theory, this mathematical object acts \emph{locally} as the identity, that is, its action {\em on system $\mathcal A$ alone} is the same as the action of the identity. We represent this as
\begin{equation}
\label{eq:EffectStateLocallyEquivId}
\vcenter{\hbox{%
\begin{tikzpicture}
  \begin{pgfonlayer}{nodelayer}

    \def\dw{0.08}    
    \def\gap{0.16}   

    \node[style=none] (Sum)  at (-3.2,0.0) {$\sum_{i,j}$};
    \node[style=none] (Coef) at (-1.9,0.0) {$c^{\m{A}}_{ij}$};

    \def\xA{-0.6}

    \def\wTri{0.65}
    \def\yBaseTop{ 0.85}   
    \def\yBaseBot{-0.85}   
    \def\yWireTop{ 1.35}
    \def\yWireBot{-1.35}

    \node[style=none] (TipTop) at (\xA, { \gap/2}) {};
    \node[style=none] (TipBot) at (\xA, {- \gap/2}) {};

    \node[style=none] (TopL) at ({\xA-\wTri}, \yBaseTop) {};
    \node[style=none] (TopR) at ({\xA+\wTri}, \yBaseTop) {};
    \node[style=none] (TopLab) at (\xA, 0.50) {$\omega_i$};

    \node[style=none] (BotL) at ({\xA-\wTri}, \yBaseBot) {};
    \node[style=none] (BotR) at ({\xA+\wTri}, \yBaseBot) {};
    \node[style=none] (BotLab) at (\xA,-0.52) {$e_j$};

    \node[style=none] (UpL)  at ({\xA-\dw}, \yWireTop) {};
    \node[style=none] (UpR)  at ({\xA+\dw}, \yWireTop) {};
    \node[style=none] (UpL2) at ({\xA-\dw}, \yBaseTop) {};
    \node[style=none] (UpR2) at ({\xA+\dw}, \yBaseTop) {};

    \node[style=none] (DnL)  at ({\xA-\dw}, \yBaseBot) {};
    \node[style=none] (DnR)  at ({\xA+\dw}, \yBaseBot) {};
    \node[style=none] (DnL2) at ({\xA-\dw}, \yWireBot) {};
    \node[style=none] (DnR2) at ({\xA+\dw}, \yWireBot) {};

    \node[style=none] (LabTop) at ({\xA+0.35},  1.05) {$\scriptstyle \m{A}$};
    \node[style=none] (LabBot) at ({\xA+0.35}, -1.05) {$\scriptstyle \m{A}$};

    \node[style=none] (Eq)   at (1.0,0.0) {$\simeq$};
    \node[style=none] (EqTL) at (1.35,-0.40) {$\scriptstyle \mathrm{L}$};

    \def\xId{2.6}

    \node[style=none] (IdUpL)  at ({\xId-\dw},  1.35) {};
    \node[style=none] (IdUpR)  at ({\xId+\dw},  1.35) {};
    \node[style=none] (IdDnL)  at ({\xId-\dw}, -1.35) {};
    \node[style=none] (IdDnR)  at ({\xId+\dw}, -1.35) {};

    \node[style=none] (IdLab)  at ({\xId+0.35}, 0.0) {$\scriptstyle \m{A}$};
    \node[style=none] (Punctuation)  at ({\xId+0.71}, -0.39) {,};
  \end{pgfonlayer}

  \begin{pgfonlayer}{edgelayer}

    \draw (TopL.center) -- (TopR.center) -- (TipTop.center) -- cycle;
    \draw (BotL.center) -- (BotR.center) -- (TipBot.center) -- cycle;

    \draw (UpL2.center) -- (UpL.center);
    \draw (UpR2.center) -- (UpR.center);
    \draw (DnL.center)  -- (DnL2.center);
    \draw (DnR.center)  -- (DnR2.center);

    \draw (IdUpL.center) -- (IdDnL.center);
    \draw (IdUpR.center) -- (IdDnR.center);

  \end{pgfonlayer}
\end{tikzpicture}%
}}
\end{equation}
where we introduce the symbol $\simeq_L$ to mean that the process is locally equivalent to the identity. 
Similarly, in a general GPT, we have for a pair of systems that 
\begin{equation}
\label{eq:PiTL_NOTequal_identity}
\vcenter{\hbox{%
\begin{tikzpicture}
  \begin{pgfonlayer}{nodelayer}

    \def\dw{0.08} 

    \def\xLA{-3.8}
    \def\xLB{-2.4}

    \def\yTop{ 1.5}
    \def\yBoxTop{ 0.6}
    \def\yBoxBot{-0.6}
    \def\yBot{-1.5}

    \node[style=none] (LTA1) at ({\xLA-\dw}, \yTop) {};
    \node[style=none] (LTA2) at ({\xLA-\dw}, \yBoxTop) {};
    \node[style=none] (LTA3) at ({\xLA+\dw}, \yTop) {};
    \node[style=none] (LTA4) at ({\xLA+\dw}, \yBoxTop) {};

    \node[style=none] (LBA1) at ({\xLA-\dw}, \yBoxBot) {};
    \node[style=none] (LBA2) at ({\xLA-\dw}, \yBot) {};
    \node[style=none] (LBA3) at ({\xLA+\dw}, \yBoxBot) {};
    \node[style=none] (LBA4) at ({\xLA+\dw}, \yBot) {};

    \node[style=none] (LTB1) at ({\xLB-\dw}, \yTop) {};
    \node[style=none] (LTB2) at ({\xLB-\dw}, \yBoxTop) {};
    \node[style=none] (LTB3) at ({\xLB+\dw}, \yTop) {};
    \node[style=none] (LTB4) at ({\xLB+\dw}, \yBoxTop) {};

    \node[style=none] (LBB1) at ({\xLB-\dw}, \yBoxBot) {};
    \node[style=none] (LBB2) at ({\xLB-\dw}, \yBot) {};
    \node[style=none] (LBB3) at ({\xLB+\dw}, \yBoxBot) {};
    \node[style=none] (LBB4) at ({\xLB+\dw}, \yBot) {};

    \node[style=none] (LLabA_top) at ({\xLA+0.30},  1.00) {$\scriptstyle \m{A}$};
    \node[style=none] (LLabB_top) at ({\xLB+0.30},  1.00) {$\scriptstyle \m{B}$};
    \node[style=none] (LLabA_bot) at ({\xLA+0.30}, -1.00) {$\scriptstyle \m{A}$};
    \node[style=none] (LLabB_bot) at ({\xLB+0.30}, -1.00) {$\scriptstyle \m{B}$};

    \def\boxPad{0.45}
    \node[style=none] (BR1) at ({\xLA-\boxPad}, \yBoxTop) {};
    \node[style=none] (BR2) at ({\xLB+\boxPad}, \yBoxTop) {};
    \node[style=none] (BR3) at ({\xLB+\boxPad}, \yBoxBot) {};
    \node[style=none] (BR4) at ({\xLA-\boxPad}, \yBoxBot) {};
    \node[style=none] (BoxLab) at ({(\xLA+\xLB)/2}, 0) {$\Pi_{TL}$};

    \node[style=none] (Eq)   at (-0.9,0.0) {$\neq$};

    \def\xRA{1.0}
    \def\xRB{2.2}

    \node[style=none] (RTA1) at ({\xRA-\dw}, \yTop) {};
    \node[style=none] (RTA2) at ({\xRA-\dw}, \yBot) {};
    \node[style=none] (RTA3) at ({\xRA+\dw}, \yTop) {};
    \node[style=none] (RTA4) at ({\xRA+\dw}, \yBot) {};

    \node[style=none] (RTB1) at ({\xRB-\dw}, \yTop) {};
    \node[style=none] (RTB2) at ({\xRB-\dw}, \yBot) {};
    \node[style=none] (RTB3) at ({\xRB+\dw}, \yTop) {};
    \node[style=none] (RTB4) at ({\xRB+\dw}, \yBot) {};

    \node[style=none] (RLabA) at ({\xRA+0.30}, 0.0) {$\scriptstyle \m{A}$};
    \node[style=none] (RLabB) at ({\xRB+0.30}, 0.0) {$\scriptstyle \m{B}$};

  \end{pgfonlayer}

  \begin{pgfonlayer}{edgelayer}

    \draw (LTA1.center) -- (LTA2.center);
    \draw (LTA3.center) -- (LTA4.center);
    \draw (LBA1.center) -- (LBA2.center);
    \draw (LBA3.center) -- (LBA4.center);

    \draw (LTB1.center) -- (LTB2.center);
    \draw (LTB3.center) -- (LTB4.center);
    \draw (LBB1.center) -- (LBB2.center);
    \draw (LBB3.center) -- (LBB4.center);

    \draw (BR1.center) -- (BR2.center) -- (BR3.center) -- (BR4.center) -- cycle;

    \draw (RTA1.center) -- (RTA2.center);
    \draw (RTA3.center) -- (RTA4.center);

    \draw (RTB1.center) -- (RTB2.center);
    \draw (RTB3.center) -- (RTB4.center);

  \end{pgfonlayer}
\end{tikzpicture}%
}}
\end{equation}

even though
\begin{equation}
\label{eq:PiTL_equals_identity_TL}
\vcenter{\hbox{%
\begin{tikzpicture}
  \begin{pgfonlayer}{nodelayer}

    \def\dw{0.08} 

    \def\xLA{-3.8}
    \def\xLB{-2.4}

    \def\yTop{ 1.5}
    \def\yBoxTop{ 0.6}
    \def\yBoxBot{-0.6}
    \def\yBot{-1.5}

    \node[style=none] (LTA1) at ({\xLA-\dw}, \yTop) {};
    \node[style=none] (LTA2) at ({\xLA-\dw}, \yBoxTop) {};
    \node[style=none] (LTA3) at ({\xLA+\dw}, \yTop) {};
    \node[style=none] (LTA4) at ({\xLA+\dw}, \yBoxTop) {};

    \node[style=none] (LBA1) at ({\xLA-\dw}, \yBoxBot) {};
    \node[style=none] (LBA2) at ({\xLA-\dw}, \yBot) {};
    \node[style=none] (LBA3) at ({\xLA+\dw}, \yBoxBot) {};
    \node[style=none] (LBA4) at ({\xLA+\dw}, \yBot) {};

    \node[style=none] (LTB1) at ({\xLB-\dw}, \yTop) {};
    \node[style=none] (LTB2) at ({\xLB-\dw}, \yBoxTop) {};
    \node[style=none] (LTB3) at ({\xLB+\dw}, \yTop) {};
    \node[style=none] (LTB4) at ({\xLB+\dw}, \yBoxTop) {};

    \node[style=none] (LBB1) at ({\xLB-\dw}, \yBoxBot) {};
    \node[style=none] (LBB2) at ({\xLB-\dw}, \yBot) {};
    \node[style=none] (LBB3) at ({\xLB+\dw}, \yBoxBot) {};
    \node[style=none] (LBB4) at ({\xLB+\dw}, \yBot) {};

    \node[style=none] (LLabA_top) at ({\xLA+0.30},  1.00) {$\scriptstyle \m{A}$};
    \node[style=none] (LLabB_top) at ({\xLB+0.30},  1.00) {$\scriptstyle \m{B}$};
    \node[style=none] (LLabA_bot) at ({\xLA+0.30}, -1.00) {$\scriptstyle \m{A}$};
    \node[style=none] (LLabB_bot) at ({\xLB+0.30}, -1.00) {$\scriptstyle \m{B}$};

    \def\boxPad{0.45}
    \node[style=none] (BR1) at ({\xLA-\boxPad}, \yBoxTop) {};
    \node[style=none] (BR2) at ({\xLB+\boxPad}, \yBoxTop) {};
    \node[style=none] (BR3) at ({\xLB+\boxPad}, \yBoxBot) {};
    \node[style=none] (BR4) at ({\xLA-\boxPad}, \yBoxBot) {};
    \node[style=none] (BoxLab) at ({(\xLA+\xLB)/2}, 0) {$\Pi_{TL}$};

    \node[style=none] (Eq)   at (-0.9,0.0) {$=$};
    \node[style=none] (EqTL) at (-0.9,0.45) {$\scriptstyle \mathrm{TL}$};

    \def\xRA{1.0}
    \def\xRB{2.2}

    \node[style=none] (RTA1) at ({\xRA-\dw}, \yTop) {};
    \node[style=none] (RTA2) at ({\xRA-\dw}, \yBot) {};
    \node[style=none] (RTA3) at ({\xRA+\dw}, \yTop) {};
    \node[style=none] (RTA4) at ({\xRA+\dw}, \yBot) {};

    \node[style=none] (RTB1) at ({\xRB-\dw}, \yTop) {};
    \node[style=none] (RTB2) at ({\xRB-\dw}, \yBot) {};
    \node[style=none] (RTB3) at ({\xRB+\dw}, \yTop) {};
    \node[style=none] (RTB4) at ({\xRB+\dw}, \yBot) {};

    \node[style=none] (RLabA) at ({\xRA+0.30}, 0.0) {$\scriptstyle \m{A}$};
    \node[style=none] (RLabB) at ({\xRB+0.30}, 0.0) {$\scriptstyle \m{B}$};
    \node[style=none] (Punctuation)  at ({\xRB+0.71}, -0.39) {.};
  \end{pgfonlayer}

  \begin{pgfonlayer}{edgelayer}

    \draw (LTA1.center) -- (LTA2.center);
    \draw (LTA3.center) -- (LTA4.center);
    \draw (LBA1.center) -- (LBA2.center);
    \draw (LBA3.center) -- (LBA4.center);

    \draw (LTB1.center) -- (LTB2.center);
    \draw (LTB3.center) -- (LTB4.center);
    \draw (LBB1.center) -- (LBB2.center);
    \draw (LBB3.center) -- (LBB4.center);

    \draw (BR1.center) -- (BR2.center) -- (BR3.center) -- (BR4.center) -- cycle;

    \draw (RTA1.center) -- (RTA2.center);
    \draw (RTA3.center) -- (RTA4.center);

    \draw (RTB1.center) -- (RTB2.center);
    \draw (RTB3.center) -- (RTB4.center);

  \end{pgfonlayer}
\end{tikzpicture}%
}}
\end{equation}
Thus, $\Pi_{\rm TL}$ coincides with the identity when constructed within a tomographically local theory, but differs from the identity whenever the parallel composition fails to be tomographically local. In such cases, it is the parallel composition of channels that are locally indistinguishable from the identity, yet can differ from it at the global level.

We now define a complementary operator $\Pi_{\rm TNL}$. 
\begin{definition}[Projection onto the tomographically non-local subspace]
\label{def: PinTL}
    Consider a composite system $\mathcal{AB}$ and the associated projection $\Pi_{TL}$ onto its tomographically local subspace. We define the \emph{projection $\Pi_{\rm TNL}$ onto the tomographically-nonlocal subspace} as 
    \begin{align}
        \Pi_{\rm TNL} := \mathds{1}^{AB} - \Pi_{TL},
    \end{align}
    where $\mathds{1}^{AB}$ is the identity operator on $AB$.
    Given that $\Pi_{\rm TL}:v\to v_{\rm TL}$ for all $v\in AB$, 
\begin{align}
    \Pi_{\rm TNL} (v_{\rm TL}+h)= h \quad \text{ for all }v_{\rm TL}\in A\otimes B,h\in H_S.
\end{align}
\end{definition}

\begin{equation}
\label{eq:PiNTL}
\vcenter{\hbox{
We will represent $\Pi_{\rm TNL}$  diagrammatically as
\begin{tikzpicture}
	\begin{pgfonlayer}{nodelayer}

		\node [style=none] (LTA) at (-3.6,1.7) {};
		\node [style=none] (LTB) at (-2.4,1.7) {};
		\node [style=none] (LTA2) at (-3.6,0.75) {};
		\node [style=none] (LTB2) at (-2.4,0.75) {};

		\node [style=none] (LBA) at (-3.6,-0.75) {};
		\node [style=none] (LBB) at (-2.4,-0.75) {};
		\node [style=none] (LBA2) at (-3.6,-1.7) {};
		\node [style=none] (LBB2) at (-2.4,-1.7) {};

		\node [style=none] (LA1) at (-3.9,1.3) {$\scriptstyle \mathcal A$};
		\node [style=none] (LB1) at (-2.1,1.3) {$\scriptstyle \mathcal B$};
		\node [style=none] (LA2) at (-3.9,-1.3) {$\scriptstyle \mathcal A$};
		\node [style=none] (LB2) at (-2.1,-1.3) {$\scriptstyle \mathcal B$};

		\node [style=none] (R1) at (-4.15,0.75) {};
		\node [style=none] (R2) at (-1.85,0.75) {};
		\node [style=none] (R3) at (-1.85,-0.75) {};
		\node [style=none] (R4) at (-4.15,-0.75) {};
		\node [style=none] (RectLabel) at (-3.0,0) {$\Pi_{TNL}$};

		\end{pgfonlayer}

	\begin{pgfonlayer}{edgelayer}

		\draw[double] (LTA.center) to (LTA2.center);
		\draw[double] (LTB.center) to (LTB2.center);
		\draw[double] (LBA.center) to (LBA2.center);
		\draw[double] (LBB.center) to (LBB2.center);

		\draw (R1.center) to (R2.center);
		\draw (R2.center) to (R3.center);
		\draw (R3.center) to (R4.center);
		\draw (R4.center) to (R1.center);
	\end{pgfonlayer}
\end{tikzpicture}
}}
\end{equation}
$\Pi_{\rm TNL}$ is a projection ($\Pi_{\rm TNL}\circ\Pi_{\rm TNL}=\Pi_{\rm TNL}$) that destroys the tomographically local component of any vector in $AB$. Just as for Prop.~\ref{prop:DualPiTL}, a similar feature holds for the action of $\Pi_{\rm TNL}$ on effects, namely $\Pi^*_{\rm TNL}(w):=w\circ \Pi_{\rm TNL}$ for $w\in (AB)^*$:
\begin{proposition} 
[Dual of $\Pi_{\rm TNL}$]
\label{prop:DualPinTL}
The operator $\Pi_{\rm TNL}^*: (AB)^*\to (AB)^*$ defined by 
\begin{align}
    \Pi_{\rm TNL}^*(w)=w\circ \Pi_{\rm TNL}.
\end{align}
is the projector onto the tomographically non-local subspace $H_E$. 
\end{proposition}

The proof of Prop.~\ref{prop:DualPinTL} is totally analogous to that of Prop.~\ref{prop:DualPiTL}.

This operator also has some important properties, proved in Appendix~\ref{projprop} (see Prop.~\ref{prop17}). 
\begin{enumerate}
    \item $\Pi_{\rm TNL}=\mathbf{0}$ if and only if the GPT system is tomographically local;
    \item $\Pi_{\rm TNL}$ does not preserve normalization. Vectors in its image have normalization 0, since $u^{\m{AB}}\circ\Pi_{\rm TNL}=0^{(AB)^*}$;
    \item $\mathsf{ Im}(\Pi_{\rm TNL})\bigcap S_{\m{AB}}=\{0^{AB}\}$; 
    \item For a state $\omega$, $\Pi_{\rm TNL}(\omega)=0\iff \omega\in AB_{\otimes}$;  similarly for effects, $e\circ \Pi_{\rm TNL}=0\iff e\in AB^*_{\otimes}$. As a consequence of property $3$, for any normalized state $\omega$, $\Pi_{\rm TNL}(\omega)\not\in S_{\m{AB}}\iff \omega\not\in AB_{\otimes}$.
\end{enumerate}

In analogy to $\Pi_{\rm TL}$, the projection $\Pi_{\rm TNL}$ can also be decomposed as a sum of vectors in $A,B,AB$ and their duals  or as a linear combination of states/effects. 

\begin{proposition}
    The projection onto the holistic subspaces can be written equivalently in the following two ways:
    \begin{enumerate}
        \item Combination of basis vectors: Consider a basis $\{v_{ij}\}_{ij}\bigcup\{h_k\}_k$ of $AB$ such that $\{v_{ij}\}\subset AB_{\otimes}$ and 
        $\{h_k\} \subset H_S$. Now, take the dual basis of $AB^*$, which will have the form $\{t_{ij}\}_{ij}\bigcup\{k'\}_{k'}$ with $\{t_{ij}\}\subset AB^*_{\otimes}$ and 
        $\{h'_{k'}\}_{k'}\subset H_E$. Then, $\Pi_{\rm TNL}$ can be written as
        \begin{align}
            \Pi_{\rm TNL}(\bullet)= \sum_{k}h_k\circ (h_k'[\bullet]).
        \end{align}
        Diagrammatically,
        \begin{equation}
\label{eq:Pi_nTL_hourglass}
\vcenter{\hbox{
        \begin{tikzpicture}
	\begin{pgfonlayer}{nodelayer}

		\node [style=none] (LTA) at (-3.6,1.4) {};
		\node [style=none] (LTB) at (-2.4,1.4) {};
		\node [style=none] (LTA2) at (-3.6,0.6) {};
		\node [style=none] (LTB2) at (-2.4,0.6) {};

		\node [style=none] (LBA) at (-3.6,-0.6) {};
		\node [style=none] (LBB) at (-2.4,-0.6) {};
		\node [style=none] (LBA2) at (-3.6,-1.4) {};
		\node [style=none] (LBB2) at (-2.4,-1.4) {};

		\node [style=none] (LA1) at (-3.9,1.0) {$\scriptstyle \mathcal A$};
		\node [style=none] (LB1) at (-2.1,1.0) {$\scriptstyle \mathcal B$};
		\node [style=none] (LA2) at (-3.9,-1.0) {$\scriptstyle \mathcal A$};
		\node [style=none] (LB2) at (-2.1,-1.0) {$\scriptstyle \mathcal B$};

		\node [style=none] (R1) at (-3.9,0.6) {};
		\node [style=none] (R2) at (-2.1,0.6) {};
		\node [style=none] (R3) at (-2.1,-0.6) {};
		\node [style=none] (R4) at (-3.9,-0.6) {};
		\node [style=none] (RectLabel) at (-3.0,0) {$\Pi_{TNL}$};

		\node [style=none] (EQ) at (-1.2,0) {$=$};
		\node [style=none] (Sumk) at (0.2,0) {$\sum_{k}$};


		\node [style=none] (Tip) at (2.2,0) {};

		\node [style=none] (RTA)  at (1.55,2.05) {};
		\node [style=none] (RTB)  at (2.85,2.05) {};
		\node [style=none] (RTA2) at (1.55,0.95) {};
		\node [style=none] (RTB2) at (2.85,0.95) {};

		\node [style=none] (RA1) at (1.25,1.45) {$\scriptstyle \mathcal A$};
		\node [style=none] (RB1) at (3.15,1.45) {$\scriptstyle \mathcal B$};

		\node [style=none] (RBA)  at (1.55,-0.95) {};
		\node [style=none] (RBB)  at (2.85,-0.95) {};
		\node [style=none] (RBA2) at (1.55,-2.05) {};
		\node [style=none] (RBB2) at (2.85,-2.05) {};

		\node [style=none] (RA2) at (1.25,-1.45) {$\scriptstyle \mathcal A$};
		\node [style=none] (RB2) at (3.15,-1.45) {$\scriptstyle \mathcal B$};

		\node [style=none] (TTopL) at (0.90,0.95) {};
		\node [style=none] (TTopR) at (3.50,0.95) {};
		\node [style=none] (TopLabel) at (2.2,0.55) {$h_k$};

		\node [style=none] (TBotL) at (0.90,-0.95) {};
		\node [style=none] (TBotR) at (3.50,-0.95) {};
		\node [style=none] (BotLabel) at (2.2,-0.55) {$h'_k$};

	\end{pgfonlayer}

	\begin{pgfonlayer}{edgelayer}

		\draw[double] (LTA.center) to (LTA2.center);
		\draw[double] (LTB.center) to (LTB2.center);
		\draw[double] (LBA.center) to (LBA2.center);
		\draw[double] (LBB.center) to (LBB2.center);

		\draw (R1.center) to (R2.center);
		\draw (R2.center) to (R3.center);
		\draw (R3.center) to (R4.center);
		\draw (R4.center) to (R1.center);

		\draw[double] (RTA.center) to (RTA2.center);
		\draw[double] (RTB.center) to (RTB2.center);
		\draw[double] (RBA.center) to (RBA2.center);
		\draw[double] (RBB.center) to (RBB2.center);

		\draw (TTopL.center) to (TTopR.center);
		\draw (TTopR.center) to (Tip.center);
		\draw (Tip.center) to (TTopL.center);

		\draw (TBotL.center) to (TBotR.center);
		\draw (TBotR.center) to (Tip.center);
		\draw (Tip.center) to (TBotL.center);

	\end{pgfonlayer}
\end{tikzpicture}
}}
\end{equation}
        \item Linear combination of effect state channels and $\Pi_{\rm TL}$: Consider a basis of $AB$ $\{s_i\}$ such that $\{s_i\}\subset \m{S_{AB}}$ and a basis of $AB^*$ constituted of effects, $\{e_i\}\subset E_{AB}$. Then, $\Pi_{\rm TNL}$ can be written as
        \begin{align}
            \Pi_{\rm TNL}(\bullet) = \sum_{ij} c^{\m{AB}}_{ij} s_i\circ e_j[\bullet] - \Pi_{\rm TL},
        \end{align}
        where $c^{\m{AB}}_{ij}$ are again the elements of the inverse matrix $\mathbb{M}^{\m{AB}}$, with $(\mathbb{M}^{\m{AB}})_{ij}=e_j(s_i)$. 
        Diagrammatically,

        \begin{equation}
\label{eq:Pi_nTL_GlobalMinusTL}
        \vcenter{\hbox{
\begin{tikzpicture}[baseline={(eq.base)}]
  \begin{pgfonlayer}{nodelayer}

    \node[style=none] (LTA)  at (-4.60, 2.00) {};
    \node[style=none] (LTB)  at (-3.40, 2.00) {};
    \node[style=none] (LTA2) at (-4.60, 0.70) {};
    \node[style=none] (LTB2) at (-3.40, 0.70) {};

    \node[style=none] (LBA)  at (-4.60,-0.70) {};
    \node[style=none] (LBB)  at (-3.40,-0.70) {};
    \node[style=none] (LBA2) at (-4.60,-2.00) {};
    \node[style=none] (LBB2) at (-3.40,-2.00) {};

    \node[style=none] (LA1) at (-4.90, 1.35) {$\scriptstyle \m{A}$};
    \node[style=none] (LB1) at (-3.10, 1.35) {$\scriptstyle \m{B}$};
    \node[style=none] (LA2) at (-4.90,-1.35) {$\scriptstyle \m{A}$};
    \node[style=none] (LB2) at (-3.10,-1.35) {$\scriptstyle \m{B}$};

    \node[style=none] (LR1) at (-5.10, 0.70) {};
    \node[style=none] (LR2) at (-2.90, 0.70) {};
    \node[style=none] (LR3) at (-2.90,-0.70) {};
    \node[style=none] (LR4) at (-5.10,-0.70) {};
    \node[style=none] (LRectLab) at (-4.00, 0.00) {$\Pi_{TNL}$};

    \node[style=none] (eq)  at (-1.90, 0.00) {$=$};
    \node[style=none] (sum) at (-0.10, 0.00) {$\sum_{i,j} c_{ij}$};

    \node[style=none] (MTA)  at (1.05, 2.00) {};
    \node[style=none] (MTB)  at (2.35, 2.00) {};
    \node[style=none] (MTA2) at (1.05, 0.90) {};
    \node[style=none] (MTB2) at (2.35, 0.90) {};

    \node[style=none] (MBA)  at (1.05,-0.90) {};
    \node[style=none] (MBB)  at (2.35,-0.90) {};
    \node[style=none] (MBA2) at (1.05,-2.00) {};
    \node[style=none] (MBB2) at (2.35,-2.00) {};

    \node[style=none] (MA1) at (0.75, 1.35) {$\scriptstyle \m{A}$};
    \node[style=none] (MB1) at (2.65, 1.35) {$\scriptstyle \m{B}$};
    \node[style=none] (MA2) at (0.75,-1.35) {$\scriptstyle \m{A}$};
    \node[style=none] (MB2) at (2.65,-1.35) {$\scriptstyle \m{B}$};

    \node[style=none] (MTip)  at (1.70, 0.00) {};
    \node[style=none] (MTopL) at (0.40, 0.90) {};
    \node[style=none] (MTopR) at (3.00, 0.90) {};
    \node[style=none] (MBotL) at (0.40,-0.90) {};
    \node[style=none] (MBotR) at (3.00,-0.90) {};

    \node[style=none] (MTopLab) at (1.70, 0.52) {$\omega_i$};
    \node[style=none] (MBotLab) at (1.70,-0.52) {$e_j$};

    \node[style=none] (minus) at (3.8, 0.00) {$-$};

    \node[style=none] (RTA)  at (5.60, 2.00) {};
    \node[style=none] (RTB)  at (6.80, 2.00) {};
    \node[style=none] (RTA2) at (5.60, 0.70) {};
    \node[style=none] (RTB2) at (6.80, 0.70) {};

    \node[style=none] (RBA)  at (5.60,-0.70) {};
    \node[style=none] (RBB)  at (6.80,-0.70) {};
    \node[style=none] (RBA2) at (5.60,-2.00) {};
    \node[style=none] (RBB2) at (6.80,-2.00) {};

    \node[style=none] (RA1) at (5.30, 1.35) {$\scriptstyle \m{A}$};
    \node[style=none] (RB1) at (7.10, 1.35) {$\scriptstyle \m{B}$};
    \node[style=none] (RA2) at (5.30,-1.35) {$\scriptstyle \m{A}$};
    \node[style=none] (RB2) at (7.10,-1.35) {$\scriptstyle \m{B}$};

    \node[style=none] (RR1) at (5.10, 0.70) {};
    \node[style=none] (RR2) at (7.30, 0.70) {};
    \node[style=none] (RR3) at (7.30,-0.70) {};
    \node[style=none] (RR4) at (5.10,-0.70) {};
    \node[style=none] (RRectLab) at (6.20, 0.00) {$\Pi_{TL}$};

    \node[style=none] (eqbase) at (-1.90,0) {};
    \node[style=none] (eq) at (-1.90,0) {}; 
  \end{pgfonlayer}

  \begin{pgfonlayer}{edgelayer}

    \draw[double] (LTA.center) to (LTA2.center);
    \draw[double] (LTB.center) to (LTB2.center);
    \draw[double] (LBA.center) to (LBA2.center);
    \draw[double] (LBB.center) to (LBB2.center);

    \draw (LR1.center) to (LR2.center);
    \draw (LR2.center) to (LR3.center);
    \draw (LR3.center) to (LR4.center);
    \draw (LR4.center) to (LR1.center);

    \draw[double] (MTA.center) to (MTA2.center);
    \draw[double] (MTB.center) to (MTB2.center);
    \draw[double] (MBA.center) to (MBA2.center);
    \draw[double] (MBB.center) to (MBB2.center);

    \draw (MTopL.center) to (MTopR.center);
    \draw (MTopR.center) to (MTip.center);
    \draw (MTip.center)  to (MTopL.center);

    \draw (MBotL.center) to (MBotR.center);
    \draw (MBotR.center) to (MTip.center);
    \draw (MTip.center)  to (MBotL.center);

    \draw[double] (RTA.center) to (RTA2.center);
    \draw[double] (RTB.center) to (RTB2.center);
    \draw[double] (RBA.center) to (RBA2.center);
    \draw[double] (RBB.center) to (RBB2.center);

    \draw (RR1.center) to (RR2.center);
    \draw (RR2.center) to (RR3.center);
    \draw (RR3.center) to (RR4.center);
    \draw (RR4.center) to (RR1.center);

  \end{pgfonlayer}
\end{tikzpicture}
}}
\end{equation}

    \end{enumerate}
    
\end{proposition}
The above proposition follows from Def.~\ref{def: PinTL} and Prop.~\ref{propPiTL}.

\section{Tomographically-nonlocal entanglement}

We now apply these tools to study entanglement in general GPTs, and in particular in tomographically-nonlocal \blk ones (though still confined to prepare-and-measure scenarios). 
As per Definition~\ref{def: EntangledStates}, every state that is not in the convex hull of the product states is entangled. There are two qualitatively distinct ways in which this might happen. The first occurs when a state has a component {\em in } the span but outside the convex hull of the product states. This kind of entanglement, which we call {\em tomographically-local entanglement}  is the kind present in  unrestricted quantum theory, since {\em all} such  entangled quantum states are in the span of the unrestricted product states. \blk In GPTs that are tomographically nonlocal, however, a new possibility arises: states can have a component outside even the {\em span} of the product states, and such states exhibit a novel form of entanglement that we call {\em tomographically-nonlocal entanglement}. In general theories, both kinds of entanglement may be present at once in a single state. 

All of these considerations hold equally well for effects. 
Thus, we can formalise \blk two definitions:

\begin{definition}[Tomographically-local entanglement] 
\label{def:TLEntenglement}
A normalized state $\omega^{\m{AB}}\in \Omega_{\m{AB}}$ is said to have tomographically-local entanglement (or TL-entanglement) if and only if the component of the state in the tomographically local subspace $AB_\otimes$ is not separable: \blk
\begin{equation}
\label{eq:TLEntanglementDef}
\vcenter{\hbox{
\begin{tikzpicture}
	\begin{pgfonlayer}{nodelayer}

		\node [style=none] (TA) at (-0.6,1.4) {};
		\node [style=none] (TB) at (0.6,1.4) {};

		\node [style=none] (TA2) at (-0.6,0.7) {};
		\node [style=none] (TB2) at (0.6,0.7) {};

		\node [style=none] (Alabel1) at (-0.9,1.1) {$\scriptstyle \mathcal A$};
		\node [style=none] (Blabel1) at (0.9,1.1) {$\scriptstyle \mathcal B$};

		\node [style=none] (R1) at (-0.9,0.7) {};
		\node [style=none] (R2) at (0.9,0.7) {};
		\node [style=none] (R3) at (0.9,-0.1) {};
		\node [style=none] (R4) at (-0.9,-0.1) {};

		\node [style=none] (RectLabel) at (0,0.3) {$\Pi_{TL}$};

		\node [style=none] (BA) at (-0.6,-0.1) {};
		\node [style=none] (BB) at (0.6,-0.1) {};

		\node [style=none] (BA2) at (-0.6,-0.8) {};
		\node [style=none] (BB2) at (0.6,-0.8) {};

		\node [style=none] (Alabel2) at (-0.9,-0.45) {$\scriptstyle \mathcal A$};
		\node [style=none] (Blabel2) at (0.9,-0.45) {$\scriptstyle \mathcal B$};

		\node [style=none] (TL) at (-0.9,-0.8) {};
		\node [style=none] (TR) at (0.9,-0.8) {};
		\node [style=none] (TC) at (0,-1.8) {};
		\node [style=none] (W)  at (0,-1.25) {$\omega$};

		\node [style=none] (NEQ) at (2.2,-0.5) {};

	\end{pgfonlayer}

	\begin{pgfonlayer}{edgelayer}

		\draw[double] (TA.center) to (TA2.center);
		\draw[double] (TB.center) to (TB2.center);

		\draw (R1.center) to (R2.center);
		\draw (R2.center) to (R3.center);
		\draw (R3.center) to (R4.center);
		\draw (R4.center) to (R1.center);

		\draw[double] (BA.center) to (BA2.center);
		\draw[double] (BB.center) to (BB2.center);

		\draw (TL.center) to (TR.center);
		\draw (TR.center) to (TC.center);
		\draw (TC.center) to (TL.center);

	\end{pgfonlayer}
\end{tikzpicture}}}\not\in \mathsf{Sep}(\Omega^{\mathcal{AB}})\,.
\end{equation}
This is equivalent to stating that there is no probability distribution $\{p_i\}$, states $\{\omega_i^A\}\subset S_{A}$ and $\{\nu_i^B\}\subset S_{B}$, and vector $\tilde{h}$ in the holistic subspace $H_S$ such that 
\begin{equation}
\label{eq:StateswithoutTLEntanglement}
\vcenter{\hbox{
\begin{tikzpicture}
	\begin{pgfonlayer}{nodelayer}

		\node [style=none] (L1) at (-4.8,0) {};
		\node [style=none] (R1) at (-3.0,0) {};
		\node [style=none] (C1) at (-3.9,-1.2) {};
		\node [style=none] (W1) at (-3.9,-0.4) {$\omega$};

		\node [style=none] (T1a) at (-4.4,0) {};
		\node [style=none] (T1b) at (-3.4,0) {};

		\node [style=none] (T1aL) at (-4.44,0) {};
		\node [style=none] (T1aR) at (-4.36,0) {};
		\node [style=none] (T1bL) at (-3.44,0) {};
		\node [style=none] (T1bR) at (-3.36,0) {};

		\node [style=none] (A1) at (-4.15,0.5) {$\scriptstyle \m{A}$};
		\node [style=none] (B1) at (-3.15,0.5) {$\scriptstyle \m{B}$};

		\node [style=none] (eq)  at (-1.7,0) {$=$};
		\node [style=none] (Sij) at (-0.3,0) {$\sum_{i} p_{i}$};

		\node [style=none] (L2) at (0.7,0) {};
		\node [style=none] (R2) at (2.2,0) {};
		\node [style=none] (C2) at (1.45,-1.0) {};
		\node [style=none] (W2) at (1.45,-0.33) {$\omega_i$};

		\node [style=none] (T2) at (1.45,0) {};
		\node [style=none] (T2L) at (1.41,0) {};
		\node [style=none] (T2R) at (1.49,0) {};
		\node [style=none] (A2) at (1.70,0.42) {$\scriptstyle \m{A}$};

		\node [style=none] (L3) at (2.4,0) {};
		\node [style=none] (R3) at (3.9,0) {};
		\node [style=none] (C3) at (3.15,-1.0) {};
		\node [style=none] (W3) at (3.15,-0.33) {$\nu_i$};

		\node [style=none] (T3) at (3.15,0) {};
		\node [style=none] (T3L) at (3.11,0) {};
		\node [style=none] (T3R) at (3.19,0) {};
		\node [style=none] (B3) at (3.40,0.42) {$\scriptstyle \m{B}$};

		\node [style=none] (plus) at (4.8,0) {$+$};

		\node [style=none] (L4) at (5.8,0) {};
		\node [style=none] (R4) at (7.9,0) {};
		\node [style=none] (C4) at (6.85,-1.0) {};
		\node [style=none] (W4) at (6.85,-0.38) {$\tilde{h}$};

		\node [style=none] (T4a) at (6.35,0) {};
		\node [style=none] (T4b) at (7.35,0) {};

		\node [style=none] (T4aL) at (6.31,0) {};
		\node [style=none] (T4aR) at (6.39,0) {};
		\node [style=none] (T4bL) at (7.31,0) {};
		\node [style=none] (T4bR) at (7.39,0) {};

		\node [style=none] (A4) at (6.60,0.5) {$\scriptstyle \m{A}$};
		\node [style=none] (B4) at (7.60,0.5) {$\scriptstyle \m{B}$};

	\end{pgfonlayer}

	\begin{pgfonlayer}{edgelayer}

		\draw (L1.center) to (R1.center);
		\draw (R1.center) to (C1.center);
		\draw (C1.center) to (L1.center);

		\draw (L2.center) to (R2.center);
		\draw (R2.center) to (C2.center);
		\draw (C2.center) to (L2.center);

		\draw (L3.center) to (R3.center);
		\draw (R3.center) to (C3.center);
		\draw (C3.center) to (L3.center);

		\draw (L4.center) to (R4.center);
		\draw (R4.center) to (C4.center);
		\draw (C4.center) to (L4.center);

		\draw[line width=0.6pt] (T1aL.center) to +(0,0.9);
		\draw[line width=0.6pt] (T1aR.center) to +(0,0.9);

		\draw[line width=0.6pt] (T1bL.center) to +(0,0.9);
		\draw[line width=0.6pt] (T1bR.center) to +(0,0.9);

		\draw[line width=0.6pt] (T2L.center) to +(0,0.9);
		\draw[line width=0.6pt] (T2R.center) to +(0,0.9);

		\draw[line width=0.6pt] (T3L.center) to +(0,0.9);
		\draw[line width=0.6pt] (T3R.center) to +(0,0.9);

		\draw[line width=0.6pt] (T4aL.center) to +(0,0.9);
		\draw[line width=0.6pt] (T4aR.center) to +(0,0.9);

		\draw[line width=0.6pt] (T4bL.center) to +(0,0.9);
		\draw[line width=0.6pt] (T4bR.center) to +(0,0.9);

	\end{pgfonlayer}
\end{tikzpicture}
}}.
\end{equation}
An analogous definition can be given for effects.
\end{definition}

Notice \blk that whether or not an object has tomographically-local entanglement is independent of what component of that object (if any) is in the holistic subspace. Whether or not an object has a component in the holistic subspace {\em is}, however, precisely what determines whether or not the object has the second form of entanglement:

\begin{definition}[Tomographically-nonlocal entanglement] 
\label{def:nonTLEntenglement}
A state is said to have tomographically-nonlocal entanglement (or TNL entanglement) if and only if it has a component in the  holistic subspace:
\begin{equation}
\label{eq:TLEntangledStateDEF}
\vcenter{\hbox{\begin{tikzpicture}
	\begin{pgfonlayer}{nodelayer}

		\node [style=none] (TA) at (-0.6,1.4) {};
		\node [style=none] (TB) at (0.6,1.4) {};

		\node [style=none] (TA2) at (-0.6,0.7) {};
		\node [style=none] (TB2) at (0.6,0.7) {};

		\node [style=none] (Alabel1) at (-0.9,1.1) {$\scriptstyle \mathcal A$};
		\node [style=none] (Blabel1) at (0.9,1.1) {$\scriptstyle \mathcal B$};

		\node [style=none] (R1) at (-0.9,0.7) {};
		\node [style=none] (R2) at (0.9,0.7) {};
		\node [style=none] (R3) at (0.9,-0.1) {};
		\node [style=none] (R4) at (-0.9,-0.1) {};

		\node [style=none] (RectLabel) at (0,0.3) {$\Pi_{TNL}$};

		\node [style=none] (BA) at (-0.6,-0.1) {};
		\node [style=none] (BB) at (0.6,-0.1) {};

		\node [style=none] (BA2) at (-0.6,-0.8) {};
		\node [style=none] (BB2) at (0.6,-0.8) {};

		\node [style=none] (Alabel2) at (-0.9,-0.45) {$\scriptstyle \mathcal A$};
		\node [style=none] (Blabel2) at (0.9,-0.45) {$\scriptstyle \mathcal B$};

		\node [style=none] (TL) at (-0.9,-0.8) {};
		\node [style=none] (TR) at (0.9,-0.8) {};
		\node [style=none] (TC) at (0,-1.8) {};
		\node [style=none] (W)  at (0,-1.25) {$\omega$};

		\node [style=none] (NEQ) at (2.2,-0.5) {$\neq 0_{AB}$};

	\end{pgfonlayer}

	\begin{pgfonlayer}{edgelayer}

		\draw[double] (TA.center) to (TA2.center);
		\draw[double] (TB.center) to (TB2.center);

		\draw (R1.center) to (R2.center);
		\draw (R2.center) to (R3.center);
		\draw (R3.center) to (R4.center);
		\draw (R4.center) to (R1.center);

		\draw[double] (BA.center) to (BA2.center);
		\draw[double] (BB.center) to (BB2.center);

		\draw (TL.center) to (TR.center);
		\draw (TR.center) to (TC.center);
		\draw (TC.center) to (TL.center);

	\end{pgfonlayer}
\end{tikzpicture}
}},
\end{equation}
where $0_{AB}$ is the zero state  (the zero vector in $AB$). Equivalently, a state is said to be TNL entangled if and only if it has a component outside the span of the product states; i.e., if and only if there are no real coefficients $\{r_i\}_i$ and states $\{\omega_i^A\}\subset \Omega_{\m{A}}$ and $\{\nu_i^B\}\subset \Omega_{\m{B}}$ such that
\begin{equation}
\label{eq:StateWIthoutTNLEntanglement}
\vcenter{\hbox{
\begin{tikzpicture}
	\begin{pgfonlayer}{nodelayer}

		\node [style=none] (L1) at (-4.8,0) {};
		\node [style=none] (R1) at (-3.0,0) {};
		\node [style=none] (C1) at (-3.9,-1.2) {};
		\node [style=none] (W1) at (-3.9,-0.4) {$\omega$};

		\node [style=none] (T1a) at (-4.4,0) {};
		\node [style=none] (T1b) at (-3.4,0) {};

		\node [style=none] (T1aL) at (-4.44,0) {};
		\node [style=none] (T1aR) at (-4.36,0) {};
		\node [style=none] (T1bL) at (-3.44,0) {};
		\node [style=none] (T1bR) at (-3.36,0) {};

		\node [style=none] (A1) at (-4.15,0.5) {$\scriptstyle \m{A}$};
		\node [style=none] (B1) at (-3.15,0.5) {$\scriptstyle \m{B}$};

		\node [style=none] (eq)  at (-2.1,0) {$=$};
		\node [style=none] (Sij) at (-0.35,0) {$\sum_{ij} r_{ij}$};

		\node [style=none] (L2) at (0.7,0) {};
		\node [style=none] (R2) at (2.2,0) {};
		\node [style=none] (C2) at (1.45,-1.0) {};
		\node [style=none] (W2) at (1.45,-0.33) {$\omega_i$};

		\node [style=none] (T2) at (1.45,0) {};
		\node [style=none] (T2L) at (1.41,0) {};
		\node [style=none] (T2R) at (1.49,0) {};
		\node [style=none] (A2) at (1.70,0.42) {$\scriptstyle \m{A}$};

		\node [style=none] (L3) at (2.4,0) {};
		\node [style=none] (R3) at (3.9,0) {};
		\node [style=none] (C3) at (3.15,-1.0) {};
		\node [style=none] (W3) at (3.15,-0.33) {$\nu_j$};

		\node [style=none] (T3) at (3.15,0) {};
		\node [style=none] (T3L) at (3.11,0) {};
		\node [style=none] (T3R) at (3.19,0) {};
		\node [style=none] (B3) at (3.40,0.42) {$\scriptstyle \m{B}$};

	\end{pgfonlayer}

	\begin{pgfonlayer}{edgelayer}

		\draw (L1.center) to (R1.center);
		\draw (R1.center) to (C1.center);
		\draw (C1.center) to (L1.center);

		\draw (L2.center) to (R2.center);
		\draw (R2.center) to (C2.center);
		\draw (C2.center) to (L2.center);

		\draw (L3.center) to (R3.center);
		\draw (R3.center) to (C3.center);
		\draw (C3.center) to (L3.center);

		\draw[line width=0.6pt] (T1aL.center) to +(0,0.9);
		\draw[line width=0.6pt] (T1aR.center) to +(0,0.9);

		\draw[line width=0.6pt] (T1bL.center) to +(0,0.9);
		\draw[line width=0.6pt] (T1bR.center) to +(0,0.9);

		\draw[line width=0.6pt] (T2L.center) to +(0,0.9);
		\draw[line width=0.6pt] (T2R.center) to +(0,0.9);

		\draw[line width=0.6pt] (T3L.center) to +(0,0.9);
		\draw[line width=0.6pt] (T3R.center) to +(0,0.9);

	\end{pgfonlayer}
\end{tikzpicture}
}}.
\end{equation}
An analogous definition can be given for effects.
\end{definition}

Recall that a state that lacks TNL entanglement is mapped to the $0_{\rm AB}$ vector by the projector $\Pi_{\rm TNL}$, which is a separable (unnormalized) state. Additionally, $\Pi_{\rm TNL}$ takes every state that is TNL entangled to a nonzero vector in $H_S$, which is necessarily not a state  (as we argued after Def.~\ref{def:HolisticSubspaces}) \blk -- and thus lives outside of $\mathsf{Sep}[\Omega]$. Therefore,  one could write Definition~\ref{def: PinTL} (for states) in an equivalent way which is similar to that of  Definition~\ref{def:PiTL}: \blk TNL entangled states are those $\omega\in S_{\m{AB}}$ such that  $\Pi_{\rm TNL}(\omega)\not\in\mathsf{Sep}[\Omega]$. But we remark that this form of the Definition is less direct and may hide the fact that what is important for TNL entanglement is the component of the state in the holistic subspace. 
\blk

States with TNL entanglement can be viewed as those that cannot be generated with Local Operations and Shared Quasiprobability Distrubutions, or LOSQP, where shared quasiprobability distributions should be contrasted with the shared randomness in LOSR~\cite{Wolfe2020quantifyingbell,Schmid2020typeindependent,Schmid2023understanding} and the classical communication in LOCC~\cite{nielsen2010quantum}, the two most natural classes of free operations that characterize TL entanglement. (Physically, it is not clear what it would mean to share a quasiprobability distribution, but this mathematical characterization of TNL entanglement may still be useful.)

Clearly, every separable state (effect) lacks both kinds of entanglement. 
In any TNL theory, there are some states that have TNL entanglement {\em and} some effects that have TNL entanglement. (One could not construct a theory that had TNL entangled states but not effects, or vice versa, since this would imply a failure of tomography. ) In some theories, like real quantum theory, \blk a state can have either TL entanglement or TNL entanglement, or both. In others, like bilocal classical theory, the only kind of entanglement a state can have is TNL entanglement. And in TL theories, like unrestricted  quantum theory and classical theory, the only kind of entanglement a state can have is TL entanglement.

The following table summarizes the different kinds of entanglement in various common theories. 

\begin{center}
\begin{tabular}{|l|c|c|}
\hline
 & \textbf{has TL entanglement} & \textbf{has TNL entanglement} \\
\hline
\textbf{Classical Theory}  & $\times$   & $\times$ \\
\hline
\textbf{Bilocal Classical Theory} & $\times$   & \checkmark \\
\hline
\textbf{Unrestricted Quantum Theory}  & \checkmark & $\times$ \\
\hline
\textbf{Fermionic Quantum Theory} & \checkmark   & \checkmark \\
\hline
\textbf{Real Quantum Theory} & \checkmark & \checkmark \\
\hline
\end{tabular}
\end{center}

 We conclude this section by showing a sufficient condition for a state to feature both types of entanglement.  

\begin{proposition}\label{propnboth}
    If $\Pi_{TL}(\omega^{\m{AB}})$ is not in the state space  $S_{\m{AB}}$ (despite being a vector in $(AB)_{\otimes}$), \blk then it necessarily has both TL entanglement and TNL entanglement.
\end{proposition}

\begin{proof}
First, note that $\Pi_{TL}(\omega^{\m{AB}})\not\in S_{\m{AB}}\implies \Pi_{TL}(\omega^{\m{AB}})\not\in \mathsf{Sep}(\omega^{\m{AB}})$, which means that $\omega^{\m{AB}}$ has TL entanglement.  
Now let us prove that $\omega^{\m{AB}}$ also has to carry TNL entanglement. Recall that we can express $\omega^{\m{AB}}$ as $\omega^{\m{AB}}=v_{\rm TL}+h \in AB$, where $v_{\rm TL}$ and $h$ are vectors that are not necessarily states. If $\Pi_{TL}(\omega^{\m{AB}})\not\in S_{\m{AB}}$ then $v_{\rm TL}$ is not a valid state, and hence $\Pi_{TL}(\omega^{\m{AB}})\neq \omega^{\m{AB}}$. Hence, for $\omega^{\m{AB}}$ to be a state it needs to have a non-zero component outside of the TL subspace $AB_\otimes$, which implies that it has TNL entanglement. 
\end{proof}

\section{Example: Real Quantum Theory}
\label{sec:RQT}

To get more intuition on these concepts, we flesh them out in the specific context of real quantum theory (RQT){\cite{Wootters_2010,Chiribella_QuantumFromPrinciples2016}}, especially the case of two-levels systems, called rebits.

\subsection{Systems in Real Quantum Theory}
\label{sec:SystemsRQT}

Systems in real quantum theory (RQT) can be defined via the usual Hilbert--space formulation of quantum mechanics, but restricted to \emph{real} vector spaces rather than complex ones. Concretely, every $d$-level system is associated with a real Hilbert space $\mathcal H^{\mathbb R}_d$ instead of a complex one. The GPT vector space associated with such a system is then isomorphic to the real vector space of $d\times d$ symmetric matrices acting on $\mathcal H^{\mathbb R}_d$, $\mathrm{Sym}_d(\mathbb{R})$. Normalized states (including mixed ones) correspond to real symmetric positive semidefinite matrices of unit trace, while measurements are associated with real-valued POVMs.

Let us exemplify RQT systems with the simplest case of a $2$-level system, called a rebit, which we denote by $\m{A}_{\rm rebit}$. In the traditional construction, it is associated with a two-dimensional real Hilbert space $\mathcal H^{\mathbb R}_2 \cong \mathbb R^2$, and the GPT vector space associated with it, denoted by $A_{\rm rebit}$, is isomorphic to the three-dimensional real vector space of $2\times2$ real symmetric matrices.

The set of normalized states, $\Omega^{\m{A}_{\rm rebit}}$, consists of all density operators acting on $\mathcal H^{\mathbb R}_2$. These density operators can be expanded using the basis $\{\mathds{1},\sigma_x,\sigma_z\}$ for $A_{\rm rebit}$ as
\begin{equation}
\label{eq:RQTstates}
  \rho = \frac12\left(\mathds{1} + a_x\,\sigma_x + a_z\,\sigma_z\right),
\end{equation}
with real coefficients $(a_x,a_z)$ satisfying $a_x^2 + a_z^2 \le 1$,  where $\sigma_x$ and $\sigma_z$ are the X and Z Pauli operators respectively. \blk Comparing Eq.~\eqref{eq:RQTstates} with the analogous expression for a qubit, one immediately notes that the Pauli matrix $\sigma_y$ (which has purely imaginary entries) does not appear in this expansion. Consequently, the Bloch vector $(a_x,a_z)$ lives in a two-dimensional disk rather than in a three-dimensional ball,
and the normalized state space $\Omega^{\m{A}_{\rm rebit}}$ has dimension $2$ instead of $3$.
The full rebit state space $\mathcal S_{\m{A}_{\rm rebit}} = \mathsf{ConvHull}[0,\Omega^{\m{A}_{\rm rebit}}]$,
which also contains subnormalized states, is therefore a three-dimensional convex set spanning the GPT vector space $A_{\rm rebit}$.

Any rebit effect can be written as
\begin{equation}
\label{eq:RQTeffects}
  M = \frac12\left(b_0\,\mathds{1} + b_x\,\sigma_x + b_z\,\sigma_z\right),
\end{equation}
with $b_0 \ge 0$ and $(b_x,b_z)\in\mathbb R^2$ satisfying
$b_x^2 + b_z^2 \le \min\{b_0^2,(1-b_0)^2\}$. The Born rule then takes the form
\begin{align}
\Tr(\rho M) = \tfrac12\left(b_0 + a_x b_x + a_z b_z\right).
\end{align}
All such real-valued POVM elements constitute the effect space $E_{\m{A}_{\rm rebit}}$.

Note that both $\Omega^{\m{A}_{\rm rebit}}$ and $E_{\m{A}_{\rm rebit}}$ can be viewed as subsets of the corresponding state and effect spaces of a qubit (in usual complex quantum theory), namely those operators with no imaginary components in a fixed basis. The same idea extends to arbitrary systems in real quantum theory: states and effects can always be represented using real symmetric operators, analogously to Eqs.~\eqref{eq:RQTstates} and~\eqref{eq:RQTeffects} which can be thought of as a subset of the states in the analogous complex quantum system.

\subsection{Composite systems in real quantum theory}
\label{sec:CompositesRQT}

For composite systems, the tensor-product rule still applies for the Hilbert spaces: a bipartite system $\m{AB}$ has Hilbert space $\mathcal{ H^{\mathbb{R}}_A}\otimes\mathcal {H^{\mathbb{R}}_B}$. 
Here one finds a striking difference from unrestricted (complex) quantum theory: the composite GPT vector space $AB$ has higher dimensionality than $A\otimes B$. This is a general fact in composite systems in real quantum theory, wich we exemplify again with the two-rebit systems case. 

For two rebits\footnote{In this subsection, we are using the notation $\m{A}$ instead of $\m{A}_{\rm rebit}$ in order to have a lighter notation -- and similarly using $AB$ for the GPT vector space spanned by the states of $\m{AB}$ instead of $AB_{\rm rebit}$.}
$\m{A}$ and $\m{B}$, the joint Hilbert space is isomorphic to $\mathbb{R}^2\otimes \mathbb{R}^2\cong \mathbb{R}^4$ and the GPT vector space $AB$ is isomorphic to the space of $4\times4$ symmetric matrices.  Notice that a real $4\times4$ symmetric matrix has 10 independent parameters (16 real parameters minus symmetry constraints yields 10), while $A\otimes B$ has dimension $9$, since both $A$ and $B$ are three-dimensional.  Hence, there is one extra holistic degree of freedom.
This extra parameter is most easily seen in the Pauli basis for $AB$, since a general real-valued bipartite state can be written as
\begin{align}
  \label{eq:RQTCompositeStates}
  \rho_{AB}=\tfrac14\Bigl(\mathds{1}\otimes\mathds{1} &+ \sum_{i\in\{x,z\}}a_i\,\sigma_i\otimes\mathds{1} + \sum_{j\in\{x,z\}}b_j\,\mathds{1}\otimes\sigma_j 
  + \sum_{i,j\in\{x,z\}}T_{ij}\,\sigma_i\otimes\sigma_j \Bigr) 
  + t\,(\sigma_y\otimes\sigma_y),
\end{align}
where $t=\Tr(\rho_{AB}\,\sigma_y\otimes\sigma_y)$. We see that this component is not the product of local components, $\mathds{1},\sigma_x$ and $\sigma_z$ for each local system. The extra degree of freedom arises because although the state of a single rebit cannot have any component of $\sigma_y$ (which is not real-valued), the state of a pair of rebits can have a component of $\sigma_y\otimes\sigma_y$, which is real-valued:
\begin{align}
\sigma_y\otimes\sigma_y =
\begin{pmatrix}
0&0&0&-1\\
0&0&1&0\\
0&1&0&0\\
-1&0&0&0
\end{pmatrix}.
\end{align}
In other words, while $t\sigma_y\otimes\sigma_y\in AB$ for all $t$, $t\sigma_y\otimes\sigma_y\not\in AB_{\otimes}$ if $t\neq 0$. 

Recall that effects in real quantum theory are represented by positive semidefinite operators bounded above by the identity, with the additional restriction that they be real-valued matrices. In analogy to states, effects are also represented by $4\times 4$ real-valued positive semidefinite matrices, and thus admit a decomposition similar to Eq.~\eqref{eq:RQTCompositeStates}. Again, local effects can have components $\mathds{1},\sigma_x,\sigma_z$ but never $\sigma_y$ components. Thus, product effects $M^{\m{A}}\boxtimes M^{\m{B}}$ (or effects in their span) can have no $\sigma_y\otimes\sigma_y$ component, which implies that every effect of two rebits that lives in the span of product effects also lack a $\sigma_y\otimes\sigma_y$ component. In other words, $s\sigma_y\otimes\sigma_y\in AB$, but $s\sigma_y\otimes\sigma_y\not\in AB^*_{\otimes}$ for all $s\neq 0$.

Therefore, using the Born rule for real quantum systems, $\Tr(M\rho)$,  we see that the component $t\,\sigma_y\otimes\sigma_y$ of a state $\rho$ is never probed by any effect $M\in AB^*_{\otimes}$ — and it is the only component that is oblivious to all such effects. Hence, the holistic state space is given by
\begin{align}
\label{eq:H_STwoRebits}
    H_S=\mathsf{Span}\{\sigma_y\otimes\sigma_y\}.
\end{align}

In other words, one cannot learn this parameter by any local measurements; rather, it can only be determined by a global measurement such as the one associated with the Hermitian operator $\sigma_y\otimes\sigma_y$. (The remaining 9 parameters of a state, of course, simply correspond to products of the three elements of the single-system basis.) \blk

One can make a similar argument for effects, which shows that the $s\sigma_y\otimes\sigma_y\in AB^*$ of an effect is the only one that cannot be probed by states in $AB^{\otimes}$, which implies
\begin{align}
    H_E=\mathsf{Span}\{\sigma_y\otimes\sigma_y\}.
\end{align}

Thus, for any two-rebit state written as in Eq.~\eqref{eq:RQTCompositeStates}, the action of $\Pi_{\rm TL}$ can be written as
\begin{align}
\label{eq:PiTLTwoRebits}
\Pi_{\rm TL}(\rho) = \tfrac14\Bigl(\mathds{1}\otimes\mathds{1}
+ \sum_{i\in\{x,z\}}a_i\,\sigma_i\otimes\mathds{1}
+ \sum_{j\in\{x,z\}}b_j\,\mathds{1}\otimes\sigma_j
+ \sum_{i,j\in\{x,z\}}T_{ij}\,\sigma_i\otimes\sigma_j \Bigr),
\end{align}
while
\begin{align}
\label{eq:PiTnLTwoRebits}
    \Pi_{\rm TNL}(\rho)=t\,\sigma_y\otimes\sigma_y.
\end{align}

Let us now take a closer look at $\Pi_{\rm TNL}(\rho)$ and $\Pi_{\rm TNL}(\rho)$. On the one hand, as discussed after Def.~\ref{def:HolisticSubspaces}, $\Pi_{\rm TNL}(\rho)$ is not a valid state of the joint system, and has trace 0.  On the other hand, in the case of $\Pi_{TL}(\rho)$ both alternatives can happen: \blk $\Pi_{TL}(\rho)$ might be a valid state but it might also fail to be in $\m{S}_{AB}$ if the initial state had non-null component in the holistic subspace. For instance, the Bell state $\ket{\Phi^+}\bra{\Phi^+}$ is mapped to the matrix $\frac14 (\mathds{1}\otimes\mathds{1}+\sigma_x\otimes\sigma_x+\sigma_z\otimes\sigma_z)$, which belongs to $AB_{\otimes}$ but not to $\m{S_{AB}}$, as this is not a positive semidefinite matrix and thus not a valid state.

While we have illustrated this phenomenon using two rebits, recall that the appearance of such holistic degrees of freedom is generic for composite systems in real quantum theory: whenever local subsystems forbid certain complex quantum observables due to imaginary components, these observables can reappear in pairs at the composite level, leading to a systematic failure of tomographic locality. This can also be seen comparing the dimension of the GPT vector space $AB$ with that of $A\otimes B$: denote by $d_A$ and $d_B$ the number of levels of real quantum systems $\m{A}$ and $\m{B}$, respectively and note that
positive symmetric real-valued $(d_Ad_B)\times (d_Ad_B)$ matrices require $\frac{d_Ad_B(d_Ad_B+1)}{2}$ parameters to be determined. This is a higher number than $(\frac{d_A(d_A+1)}{2})(\frac{d_B(d_B+1)}{2})$ for all nontrivial values of $d_A,d_B\not\in\{0,1\}$. Thus, $H_S$ will never be trivial for such composite real quantum systems.

Finally, note that any state with non-null component in the holistic subspace -- in the two-rebits example, any state $\rho$ with $t\sigma_y\otimes\sigma_y$ and $t\neq 0$ --  cannot belong to $AB_\otimes$, and thus cannot be separable, which means that it is an entangled state. This happens with any GPT system with holistic degrees of freedom. According to the traditional definition of entanglement in GPT systems, this is essentially all that one would be able to say: holistic components implies entanglement. With our proposed definitions, we can analyze the entanglement for two-rebits in more depth, as we show next.

\subsection{Entanglement in real quantum theory}
\label{sec:EntanglementRQT}

Entanglement in real quantum theory (and in any GPT) is defined just as in standard quantum mechanics: a bipartite state is \emph{separable} if it can be written as a convex sum of product states $\rho_{AB}=\sum_k p_k\,\rho_A^{(k)}\otimes\rho_B^{(k)}$ within the theory. Otherwise it is \emph{entangled}. 

Although every real quantum state is also a valid complex quantum state and the definition of entanglement is formally the same in both theories, whether a given density matrix represents an entangled state can depend on whether it is regarded as a state of real quantum theory or of complex quantum theory.  In fact,  because the scope of product states is limited in a real quantum system relative to the analogous system in standard quantum theory, some states in RQT are entangled even though states of the same form would be separable in standard (complex) quantum theory. That is, a given density matrix might not admit a separable decomposition in terms of density matrices on real Hilbert spaces, but do admit a separable decomposition when complex amplitudes are allowed in the underlying Hilbert space.
For instance, consider the following states for a pair of rebits:
\begin{align}
\label{eq:StateswRQT}
    \omega_{\pm}^{AB} =\frac{\mathds{1}\pm\sigma_y\otimes\sigma_y}{4}.
\end{align}
Since these states have non-null holistic components, they are not separable in RQT, i.e., they are entangled. Each of these states, however, is \emph{separable} in standard quantum theory, as these can equivalently be written as 
\begin{equation}
\label{eq:OmegapmSeparableCQT}
\omega_{\pm}^{AB} = \frac{1}{2}(\ketbra{+y}\otimes \ketbra{\pm y} + \ketbra{-y} \otimes \ketbra{\mp
 y})
\end{equation}
(which can be easily seen using the fact that $\ketbra{\pm y} = \frac12 (\mathds{1}\pm \sigma_y$)), and in bare complex quantum theory $\ket{\pm y}\bra{\pm y}$ are valid states, while these are not valid states in real quantum theory.

The states $\omega^{AB}_{\pm}$ in Eq.~\eqref{eq:StateswRQT} showcase some puzzling features of entanglement in real quantum theory: there are maximally entangled states (as quantified by concurrence) that can be shared by an arbitrary number of parties, drastically breaking monogamy of entanglement~\cite{Wootters_2010}. Moreover, the decomposition in Eq.~\eqref{eq:OmegapmSeparableCQT} can be used to show that these ``maximally entangled'' states cannot violate any Bell inequality; additionally, it can be shown that $\omega_{\pm}^{\m{AB}}$ can be locally broadcast~\cite{Weilenmann_2025}, which is never the case for entangled states in unrestricted (complex) quantum systems. 

Notice that some familiar entangled states within quantum theory are also entangled states in RQT. For example, the Bell state $\ket{\Phi^+} = \frac{1}{\sqrt{2}}(\ket{00}+\ket{11})$ corresponds to the two-rebit density matrix $\tfrac14(\mathds{1}\otimes\mathds{1}+\sigma_x\otimes\sigma_x+\sigma_z\otimes\sigma_z-\sigma_y\otimes\sigma_y)$, and is entangled in either theory. This example shows that not all entangled states in RQT exhibit these exotic features: in contrast to the states $\omega^{AB}_{\pm}$, the Bell state $\ket{\Phi^+}$ remains monogamous, cannot be locally broadcast, and can violate Bell inequalities in RQT as well.

Let us now analyze the entanglement of $\omega_{\pm}$ and $\ket{\Phi}$ in the light of the new definitions of tomographically-local and tomographically-nonlocal entanglement. Given the action of $\Pi_{\rm TL}$ and $\Pi_{\rm TNL}$ in pairs of rebits (Eqs.~\eqref{eq:PiTLTwoRebits} and \eqref{eq:PiTnLTwoRebits}), one can check that
\begin{align}
    \Pi_{\rm TL}(\omega_{\pm})&=\frac{\mathds{1}}{4};\\
    \Pi_{\rm TNL}(\omega_{\pm})&=  \pm \frac{1}{4} \sigma_y\otimes\sigma_y. 
\end{align}
These Equations imply that $\omega_{\pm}$ \emph{carry only TNL entanglement} (as per Definitions~\ref{def:TLEntenglement} and \ref{def:nonTLEntenglement}), since $\frac{\mathds{1}\otimes\mathds{1}}{4}$ is separable and $\pm\sigma_y\otimes\sigma_y\neq 0^{AB}$. The state $\ket{\Phi}$, however, \emph{carry both forms of entanglement}, a fact that one can see using Proposition~\ref{propnboth} together with the fact that $\Pi_{TL}( \ketbra{\Phi^+}{\Phi^+}) = \tfrac14(\mathds{1}\otimes\mathds{1}+\sigma_x\otimes\sigma_x+\sigma_z\otimes\sigma_z)$ is not a valid two-rebit state. One could wonder whether a pair of rebits could have TL entanglement alone, which the following result shows is impossible:

\begin{proposition}
\label{prop:TwoRebitsNonSepIffTNLEntangled}
    For a pair of rebits, a state is nonseparable if and only if it has TNL entanglement. 
\end{proposition}
\begin{proof}
This follows from the fact that all non-separable states must have a $\sigma_y\otimes\sigma_y$ component, as proven in  Ref.~\cite{caves2000entanglementformationarbitrarystate}. Therefore, a two-rebit state $s^{\m{AB}}$ is non-separable if and only if $\Pi_{\rm TNL}(s^{\m{AB}})\neq 0$.
\end{proof}

The following result, proven in App.~\ref{sec:ProofTwoRebitsLackingTLEProposition}, provides two conditions for two-rebit states that carry only TNL entanglement.
\begin{proposition}[Two-rebit states carrying only TNL entanglement]
\label{prop:TwoRebitLackingTLEntanglement}
Let $\rho^{\m{AB}}\in\Omega_{\m{AB}}$ be a two-rebit state. Then:
\begin{enumerate}
    \item $\rho^{\m{AB}}$ lacks TL-entanglement if and only if
    $\Pi_{\rm TL}(\rho^{\m{AB}})$ is a valid two-rebit state, i.e.,
    $\Pi_{\rm TL}(\rho^{\m{AB}})\in\Omega_{\m{AB}}$.
    \item Let $\iota$ denote the natural inclusion of two-rebit states into
    two-qubit states in complex quantum theory. If $\iota(\rho^{\m{AB}})$
    is separable in (unrestricted) complex quantum theory, then $\rho^{\m{AB}}$ lacks
    TL-entanglement.
\end{enumerate}
\end{proposition}

In other words, the first condition says that a two-rebit state $\rho^{\m{AB}}$ lacks TL-entanglement if and only if the matrix resulting from forgetting its $\sigma_y\otimes\sigma_y$ component also represents a valid state of two-rebits. The second tell us that, if $\rho^{\m{AB}}$ represented a state in bare complex quantum theory (instead of RQT) and it was separable (in complex quantum theory), then the state $\rho^{\m{AB}}$ cannot posses TL entanglement \emph{in real quantum theory}. As we saw before, the two features apply to the $\omega_{\pm}$ states, while they do not hold for the Bell states, like $\ket{\Phi^+}$.

We suspect that the striking operational differences between states $\omega_{\pm}$ and $\ket{\Phi^+}$, even though those have maximal concurrence, are precisely due to the fact that $\omega_{\pm}$ lacks TL entanglement, whereas $\ket{\Phi^+}$ has it. Indeed, in the next section we show that states without TL entanglement cannot be used to violate Bell inequalities, nor to demonstrate steering or enable teleportation, in any tomographically-nonlocal theory. From this perspective, the fact that some two-rebit states are maximally entangled yet completely useless for these tasks is not mysterious at all: it is simply a consequence of the absence of TL entanglement. Before turning to these results, we show how the possibility of locally broadcasting entangled states of pairs of rebits can also be understood through the distinction between TL and TNL entanglement.

\begin{proposition}
\label{prop: LocallyBroadRebits}
If a two-rebit state $\rho^{\m{AB}}$ is locally broadcastable (in real quantum theory), then it lacks TL-entanglement.  
\end{proposition}
\begin{proof}
    If $\rho^{AB}$ is locally broadcastable in real quantum theory, then the two-qubit state represented by the same density matrix is also locally broadcastable in (unrestricted) complex quantum theory. That is, considering the natural inclusion $\iota$ of rebits into qubits, $\iota(\rho^{\m{AB}})$, $\rho^{\m{AB}}$ being locally broadcastable implies $\iota(\rho^{\m{AB}})$ is also locally broadcastable. This follows from the fact that valid operations in RQT are a subset of the valid operations of complex quantum theory, so the local operations that broadcast $\rho^{AB}$ can also be included in complex quantum theory to locally broadcast $\iota(\rho^{\m{AB}})$.  

    In unrestricted complex quantum theory, a state  is locally broadcastable if and only if it is a classical-classical state~\cite{PianiNoLocalBroadcasting_2008,piani2016localbroadcastingquantumcorrelations}, which implies it is separable. Thus, $\iota(\rho^{\m{AB}})$ being locally broadcastable implies $\iota(\rho^{\m{AB}})$ is separable in complex quantum theory. By condition 2. of Prop.~\ref{prop:TwoRebitLackingTLEntanglement}, this implies that $\rho^{\m{AB}}$ has no TL-entanglement.
\end{proof}

This Proposition shows that, at least in the two-rebits case, TL-entanglement cannot be locally broadcast—exactly as in complex quantum theory. This clarifies the otherwise puzzling fact that some (maximally) entangled states in RQT are locally broadcastable~\cite{Weilenmann_2025}: this can occur only when the entanglement present is entirely tomographically nonlocal. This provides a unified and intuitive view of local broadcasting in both unrestricted complex quantum theory and in the two-rebits case: local broadcastability signals the absence of TL-entanglement.

An open question is whether this is the case in general GPTs, i.e., if in all TNL GPTs,  states with TL entanglement are not locally broadcastable.  We defer this study to future work.

\section{TNL entanglement is useless for nonlocality, steering, and teleportation}

Tomographically-nonlocal entanglement often behaves in ways that are qualitatively distinct from tomographically-local entanglement. This is illustrated by the following three theorems, which show that TNL entanglement is useless for generating nonclassical correlations in Bell scenarios, steering scenarios, and teleportation scenarios, respectively.

\begin{proposition}
    Consider a bipartite GPT system $\m{AB}$. Any state that lacks tomographically-local entanglement cannot violate any Bell inequality (even if it has tomographically-nonlocal entanglement). 
\end{proposition} 
\begin{proof}
Recall that any state of $AB$ can be decomposed as (see Eq.~\eqref{eq:DecompositionStateGeneral})
\begin{equation}
\label{eq:DecompositionStatesAB}
\vcenter{\hbox{
\begin{tikzpicture}
	\begin{pgfonlayer}{nodelayer}

		\node [style=none] (L1) at (-4.8,0) {};
		\node [style=none] (R1) at (-3.0,0) {};
		\node [style=none] (C1) at (-3.9,-1.2) {};
		\node [style=none] (W1) at (-3.9,-0.4) {$\omega$};

		\node [style=none] (T1a) at (-4.4,0) {};
		\node [style=none] (T1b) at (-3.4,0) {};

		\node [style=none] (T1aL) at (-4.44,0) {};
		\node [style=none] (T1aR) at (-4.36,0) {};
		\node [style=none] (T1bL) at (-3.44,0) {};
		\node [style=none] (T1bR) at (-3.36,0) {};

		\node [style=none] (A1) at (-4.15,0.5) {$\scriptstyle \m{A}$};
		\node [style=none] (B1) at (-3.15,0.5) {$\scriptstyle \m{B}$};

		\node [style=none] (eq)  at (-2.1,0) {$=$};
		\node [style=none] (Sij) at (-0.35,0) {$\sum_{ij} r_{ij}$};

		\node [style=none] (L2) at (0.7,0) {};
		\node [style=none] (R2) at (2.2,0) {};
		\node [style=none] (C2) at (1.45,-1.0) {};
		\node [style=none] (W2) at (1.45,-0.33) {$\omega_i$};

		\node [style=none] (T2) at (1.45,0) {};
		\node [style=none] (T2L) at (1.41,0) {};
		\node [style=none] (T2R) at (1.49,0) {};
		\node [style=none] (A2) at (1.70,0.42) {$\scriptstyle \m{A}$};

		\node [style=none] (L3) at (2.4,0) {};
		\node [style=none] (R3) at (3.9,0) {};
		\node [style=none] (C3) at (3.15,-1.0) {};
		\node [style=none] (W3) at (3.15,-0.33) {$\omega_j$};

		\node [style=none] (T3) at (3.15,0) {};
		\node [style=none] (T3L) at (3.11,0) {};
		\node [style=none] (T3R) at (3.19,0) {};
		\node [style=none] (B3) at (3.40,0.42) {$\scriptstyle \m{B}$};

		\node [style=none] (plus) at (4.5,0) {$+$};
		\node [style=none] (St) at (5.7,0) {$\sum_{t} r_{t}$};

		\node [style=none] (L4) at (6.8,0) {};
		\node [style=none] (R4) at (8.9,0) {};
		\node [style=none] (C4) at (7.85,-1.0) {};
		\node [style=none] (W4) at (7.85,-0.35) {$\tilde{t}$};

		\node [style=none] (T4a) at (7.35,0) {};
		\node [style=none] (T4b) at (8.35,0) {};

		\node [style=none] (T4aL) at (7.31,0) {};
		\node [style=none] (T4aR) at (7.39,0) {};
		\node [style=none] (T4bL) at (8.31,0) {};
		\node [style=none] (T4bR) at (8.39,0) {};

		\node [style=none] (A4) at (7.60,0.5) {$\scriptstyle \m{A}$};
		\node [style=none] (B4) at (8.60,0.5) {$\scriptstyle \m{B}$};

	\end{pgfonlayer}

	\begin{pgfonlayer}{edgelayer}

		\draw (L1.center) to (R1.center);
		\draw (R1.center) to (C1.center);
		\draw (C1.center) to (L1.center);

		\draw (L2.center) to (R2.center);
		\draw (R2.center) to (C2.center);
		\draw (C2.center) to (L2.center);

		\draw (L3.center) to (R3.center);
		\draw (R3.center) to (C3.center);
		\draw (C3.center) to (L3.center);

		\draw (L4.center) to (R4.center);
		\draw (R4.center) to (C4.center);
		\draw (C4.center) to (L4.center);

		\draw[line width=0.6pt] (T1aL.center) to +(0,0.9);
		\draw[line width=0.6pt] (T1aR.center) to +(0,0.9);

		\draw[line width=0.6pt] (T1bL.center) to +(0,0.9);
		\draw[line width=0.6pt] (T1bR.center) to +(0,0.9);

		\draw[line width=0.6pt] (T2L.center) to +(0,0.9);
		\draw[line width=0.6pt] (T2R.center) to +(0,0.9);

		\draw[line width=0.6pt] (T3L.center) to +(0,0.9);
		\draw[line width=0.6pt] (T3R.center) to +(0,0.9);

		\draw[line width=0.6pt] (T4aL.center) to +(0,0.9);
		\draw[line width=0.6pt] (T4aR.center) to +(0,0.9);

		\draw[line width=0.6pt] (T4bL.center) to +(0,0.9);
		\draw[line width=0.6pt] (T4bR.center) to +(0,0.9);

	\end{pgfonlayer}
\end{tikzpicture}\,.
}}
\end{equation}

Since $\omega^{\m{AB}}$ lacks TL-entanglement, the component of the state in $AB_{\otimes}$ is in fact a convex combination (rather than a more general linear combination) of product states, as in Eq.~\eqref{eq:StateswithoutTLEntanglement}, repeated here:
\begin{equation}
\label{eq:DecompositionStatesWithoutTL}
\vcenter{\hbox{
\begin{tikzpicture}
	\begin{pgfonlayer}{nodelayer}

		\node [style=none] (L1) at (-4.8,0) {};
		\node [style=none] (R1) at (-3.0,0) {};
		\node [style=none] (C1) at (-3.9,-1.2) {};
		\node [style=none] (W1) at (-3.9,-0.4) {$\omega$};

		\node [style=none] (T1a) at (-4.4,0) {};
		\node [style=none] (T1b) at (-3.4,0) {};

		\node [style=none] (T1aL) at (-4.44,0) {};
		\node [style=none] (T1aR) at (-4.36,0) {};
		\node [style=none] (T1bL) at (-3.44,0) {};
		\node [style=none] (T1bR) at (-3.36,0) {};

		\node [style=none] (A1) at (-4.15,0.5) {$\scriptstyle \m{A}$};
		\node [style=none] (B1) at (-3.15,0.5) {$\scriptstyle \m{B}$};

		\node [style=none] (eq)  at (-2.1,0) {$=$};
		\node [style=none] (Sij) at (-0.35,0) {$\sum_{i} p_{i}$};

		\node [style=none] (L2) at (0.7,0) {};
		\node [style=none] (R2) at (2.2,0) {};
		\node [style=none] (C2) at (1.45,-1.0) {};
		\node [style=none] (W2) at (1.45,-0.33) {$\omega_i$};

		\node [style=none] (T2) at (1.45,0) {};
		\node [style=none] (T2L) at (1.41,0) {};
		\node [style=none] (T2R) at (1.49,0) {};
		\node [style=none] (A2) at (1.70,0.42) {$\scriptstyle \m{A}$};

		\node [style=none] (L3) at (2.4,0) {};
		\node [style=none] (R3) at (3.9,0) {};
		\node [style=none] (C3) at (3.15,-1.0) {};
		\node [style=none] (W3) at (3.15,-0.33) {$\nu_i$};

		\node [style=none] (T3) at (3.15,0) {};
		\node [style=none] (T3L) at (3.11,0) {};
		\node [style=none] (T3R) at (3.19,0) {};
		\node [style=none] (B3) at (3.40,0.42) {$\scriptstyle \m{B}$};

		\node [style=none] (plus) at (5.2,0) {$+$};

		\node [style=none] (L4) at (6.8,0) {};
		\node [style=none] (R4) at (8.9,0) {};
		\node [style=none] (C4) at (7.85,-1.0) {};
		\node [style=none] (W4) at (7.85,-0.38) {$\tilde{h}$};

		\node [style=none] (T4a) at (7.35,0) {};
		\node [style=none] (T4b) at (8.35,0) {};

		\node [style=none] (T4aL) at (7.31,0) {};
		\node [style=none] (T4aR) at (7.39,0) {};
		\node [style=none] (T4bL) at (8.31,0) {};
		\node [style=none] (T4bR) at (8.39,0) {};

		\node [style=none] (A4) at (7.60,0.5) {$\scriptstyle \m{A}$};
		\node [style=none] (B4) at (8.60,0.5) {$\scriptstyle \m{B}$};

	\end{pgfonlayer}

	\begin{pgfonlayer}{edgelayer}

		\draw (L1.center) to (R1.center);
		\draw (R1.center) to (C1.center);
		\draw (C1.center) to (L1.center);

		\draw (L2.center) to (R2.center);
		\draw (R2.center) to (C2.center);
		\draw (C2.center) to (L2.center);

		\draw (L3.center) to (R3.center);
		\draw (R3.center) to (C3.center);
		\draw (C3.center) to (L3.center);

		\draw (L4.center) to (R4.center);
		\draw (R4.center) to (C4.center);
		\draw (C4.center) to (L4.center);

		\draw[line width=0.6pt] (T1aL.center) to +(0,0.9);
		\draw[line width=0.6pt] (T1aR.center) to +(0,0.9);

		\draw[line width=0.6pt] (T1bL.center) to +(0,0.9);
		\draw[line width=0.6pt] (T1bR.center) to +(0,0.9);

		\draw[line width=0.6pt] (T2L.center) to +(0,0.9);
		\draw[line width=0.6pt] (T2R.center) to +(0,0.9);

		\draw[line width=0.6pt] (T3L.center) to +(0,0.9);
		\draw[line width=0.6pt] (T3R.center) to +(0,0.9);

		\draw[line width=0.6pt] (T4aL.center) to +(0,0.9);
		\draw[line width=0.6pt] (T4aR.center) to +(0,0.9);

		\draw[line width=0.6pt] (T4bL.center) to +(0,0.9);
		\draw[line width=0.6pt] (T4bR.center) to +(0,0.9);

	\end{pgfonlayer}
\end{tikzpicture}\,,
}}
\end{equation}
 with $\tilde{h}=\sum_tr_t\tilde{t}$ being the holistic component. 
 
The idea of the proof is as follows. 
 The component $\tilde{h}$  of the state lying in the holistic subspace is annihilated by all product effects, and the only effects used in Bell scenarios are product effects. Therefore, this component of the state does not contribute in any way to the observed probabilities in the scenario, and any correlations observed come solely from the component of the state living in $AB_{\otimes }$, which by assumption is not entangled, and so cannot violate any Bell inequalities.  In what follows, we spell out this proof idea (which will be useful for subsequent proofs in this manuscript).

Consider a bipartite Bell scenario with a finite number of outcomes and of measurements. 
 Consider the set of all conditional probabilities of obtaining outcomes $a,b$  when Alice implements local measurement $x$ and Bob implements  local measurement $y$ on their subsystems: $\{\{p(a,b|x,y)\}_{a,b}\}_{x,y}$. These correlations violate some Bell inequality for the scenario if and only if each joint probability $p(a,b|x,y)$ does not admit a decomposition of the form  $p(a,b|x,y)=\sum_{\lambda}p(a|x,\lambda)p(b|y,\lambda)p(\lambda)$. Diagrammatically, $p(a,b|x,y)$ can be written as
\begin{equation}
p(a,b|x,y)=\vcenter{\hbox{%
\begin{tikzpicture}
  \begin{pgfonlayer}{nodelayer}

    \def\dw{0.07}          
    \def\xA{-1.3}         
    \def\xB{ 1.3}         

    \def\yEffBase{1.05}
    \def\yEffTip{ 2.17}
    \def\triHalfW{0.91}    

    \def\yOmBase{0.00}
    \def\yOmTip{-1.10}
    \def\omHalfW{1.90}

    \node[style=none] (eBL) at ({\xA-\triHalfW}, \yEffBase) {};
    \node[style=none] (eBR) at ({\xA+\triHalfW}, \yEffBase) {};
    \node[style=none] (eT)  at (\xA, \yEffTip) {};
    \node[style=none] (eLab) at (\xA+0.1, 1.40) {$e_{a|x}$};

    \node[style=none] (fBL) at ({\xB-\triHalfW}, \yEffBase) {};
    \node[style=none] (fBR) at ({\xB+\triHalfW}, \yEffBase) {};
    \node[style=none] (fT)  at (\xB, \yEffTip) {};
    \node[style=none] (fLab) at (\xB, 1.40) {$f_{b|y}$};

    \node[style=none] (OmL) at (-\omHalfW, \yOmBase) {};
    \node[style=none] (OmR) at ( \omHalfW, \yOmBase) {};
    \node[style=none] (OmT) at (0, \yOmTip) {};
    \node[style=none] (OmLab) at (0,-0.45) {$\omega$};

    \node[style=none] (Adn1) at ({\xA-\dw}, \yOmBase) {};
    \node[style=none] (Adn2) at ({\xA+\dw}, \yOmBase) {};
    \node[style=none] (Aup1) at ({\xA-\dw}, \yEffBase) {};
    \node[style=none] (Aup2) at ({\xA+\dw}, \yEffBase) {};
    \node[style=none] (Alab) at ({\xA+0.28}, 0.55) {$\scriptstyle \m{A}$};

    \node[style=none] (Bdn1) at ({\xB-\dw}, \yOmBase) {};
    \node[style=none] (Bdn2) at ({\xB+\dw}, \yOmBase) {};
    \node[style=none] (Bup1) at ({\xB-\dw}, \yEffBase) {};
    \node[style=none] (Bup2) at ({\xB+\dw}, \yEffBase) {};
    \node[style=none] (Blab) at ({\xB+0.28}, 0.55) {$\scriptstyle \m{B}$};

  \end{pgfonlayer}

  \begin{pgfonlayer}{edgelayer}

    \draw (eBL.center) -- (eBR.center) -- (eT.center) -- cycle;
    \draw (fBL.center) -- (fBR.center) -- (fT.center) -- cycle;

    \draw (OmL.center) -- (OmR.center) -- (OmT.center) -- cycle;

    \draw[line width=0.6pt] (Adn1.center) -- (Aup1.center);
    \draw[line width=0.6pt] (Adn2.center) -- (Aup2.center);

    \draw[line width=0.6pt] (Bdn1.center) -- (Bup1.center);
    \draw[line width=0.6pt] (Bdn2.center) -- (Bup2.center);

  \end{pgfonlayer}
\end{tikzpicture}%
}}\,.
\end{equation}

Now, if $\omega$ lacks TL-entanglement, 

\begin{equation}\label{eq:your_label_here}
\vcenter{\hbox{%
\begin{tikzpicture}[baseline={(base.center)}]
\begin{pgfonlayer}{nodelayer}

\node (base) at (0,0) {}; 

\def\dw{0.07}            

\def\yEffBase{1.05}
\def\yEffTip{ 2.17}
\def\triHalfW{0.91}

\def\yBigBase{0.00}
\def\yBigTip{-1.8}
\def\bigHalfW{2.1}

\def\ySmBase{0.00}
\def\ySmTip{-1.05}
\def\smHalfW{0.75}

\def\yWireLab{0.55}

%
\newcommand{\EffectTri}[3]{%
  \node[style=none] (#2BL) at ({#3-\triHalfW}, \yEffBase) {};
  \node[style=none] (#2BR) at ({#3+\triHalfW}, \yEffBase) {};
  \node[style=none] (#2T)  at (#3, \yEffTip) {};
  \node[style=none] (#2Lab) at (#3, 1.40) {$#1$};
}

\newcommand{\InvTri}[7]{%
  \node[style=none] (#2L) at ({#3-#4}, #5) {};
  \node[style=none] (#2R) at ({#3+#4}, #5) {};
  \node[style=none] (#2T) at (#3, #6) {};
  \node[style=none] (#2Lab) at (#3, #7) {$#1$};
}

\newcommand{\DoubleWire}[4]{%
  \node[style=none] (#1dn1) at ({#2-\dw}, #3) {};
  \node[style=none] (#1dn2) at ({#2+\dw}, #3) {};
  \node[style=none] (#1up1) at ({#2-\dw}, #4) {};
  \node[style=none] (#1up2) at ({#2+\dw}, #4) {};
}

\begin{scope}[xshift=-8.0cm]
  \def\xA{-1.48}
  \def\xB{ 1.48}

  \EffectTri{e_{a|x}}{LHe}{\xA}
  \EffectTri{f_{b|y}}{LHf}{\xB}

  \InvTri{\omega}{LHOm}{0}{\bigHalfW}{\yBigBase}{\yBigTip}{-0.70}

  \DoubleWire{LHA}{\xA}{\yBigBase}{\yEffBase}
  \DoubleWire{LHB}{\xB}{\yBigBase}{\yEffBase}

  \node[style=none] (LHAlab) at ({\xA+0.30}, \yWireLab) {$\scriptstyle \m{A}$};
  \node[style=none] (LHBhlab) at ({\xB+0.30}, \yWireLab) {$\scriptstyle \m{B}$};
\end{scope}

\node[style=none] (EQ) at (-4.2,0.10) {$=$};
\node[style=none] (SUM) at (-1.7,0.10) {$\sum_{i} p_i$};

\begin{scope}[xshift=1.0cm]
  \def\xMA{-0.9}
  \EffectTri{e_{a|x}}{Me}{\xMA}
  \InvTri{\omega_i}{Mwi}{\xMA}{\smHalfW}{\ySmBase}{\ySmTip}{-0.35}
  \DoubleWire{MWA}{\xMA}{\ySmBase}{\yEffBase}
  \node[style=none] (MWAlab) at ({\xMA+0.28}, \yWireLab) {$\scriptstyle \m{A}$};

  \def\xMB{1.35}
  \EffectTri{f_{b|y}}{Mf}{\xMB}
  \InvTri{\nu_i}{Mnu}{\xMB}{\smHalfW}{\ySmBase}{\ySmTip}{-0.35}
  \DoubleWire{MWB}{\xMB}{\ySmBase}{\yEffBase}
  \node[style=none] (MWBlab) at ({\xMB+0.28}, \yWireLab) {$\scriptstyle \m{B}$};
\end{scope}

\node[style=none] (PLUS) at (4.9,0.10) {$+$};

\begin{scope}[xshift=9.0cm]
  \def\xA{-1.48}
  \def\xB{ 1.48}

  \EffectTri{e_{a|x}}{RHe}{\xA}
  \EffectTri{f_{b|y}}{RHf}{\xB}

  \InvTri{\tilde h}{RHh}{0}{\bigHalfW}{\yBigBase}{\yBigTip}{-0.70}

  \DoubleWire{RHA}{\xA}{\yBigBase}{\yEffBase}
  \DoubleWire{RHB}{\xB}{\yBigBase}{\yEffBase}

  \node[style=none] (RHAlab) at ({\xA+0.30}, \yWireLab) {$\scriptstyle \m{A}$};
  \node[style=none] (RHBhlab) at ({\xB+0.30}, \yWireLab) {$\scriptstyle \m{B}$};
\end{scope}

\end{pgfonlayer}

\begin{pgfonlayer}{edgelayer}

  \draw (LHeBL.center) -- (LHeBR.center) -- (LHeT.center) -- cycle;
  \draw (LHfBL.center) -- (LHfBR.center) -- (LHfT.center) -- cycle;
  \draw (LHOmL.center) -- (LHOmR.center) -- (LHOmT.center) -- cycle;

  \draw[line width=0.6pt] (LHAdn1.center) -- (LHAup1.center);
  \draw[line width=0.6pt] (LHAdn2.center) -- (LHAup2.center);
  \draw[line width=0.6pt] (LHBdn1.center) -- (LHBup1.center);
  \draw[line width=0.6pt] (LHBdn2.center) -- (LHBup2.center);

  \draw (MeBL.center) -- (MeBR.center) -- (MeT.center) -- cycle;
  \draw (MwiL.center) -- (MwiR.center) -- (MwiT.center) -- cycle;
  \draw[line width=0.6pt] (MWAdn1.center) -- (MWAup1.center);
  \draw[line width=0.6pt] (MWAdn2.center) -- (MWAup2.center);

  \draw (MfBL.center) -- (MfBR.center) -- (MfT.center) -- cycle;
  \draw (MnuL.center) -- (MnuR.center) -- (MnuT.center) -- cycle;
  \draw[line width=0.6pt] (MWBdn1.center) -- (MWBup1.center);
  \draw[line width=0.6pt] (MWBdn2.center) -- (MWBup2.center);

  \draw (RHeBL.center) -- (RHeBR.center) -- (RHeT.center) -- cycle;
  \draw (RHfBL.center) -- (RHfBR.center) -- (RHfT.center) -- cycle;
  \draw (RHhL.center) -- (RHhR.center) -- (RHhT.center) -- cycle;

  \draw[line width=0.6pt] (RHAdn1.center) -- (RHAup1.center);
  \draw[line width=0.6pt] (RHAdn2.center) -- (RHAup2.center);
  \draw[line width=0.6pt] (RHBdn1.center) -- (RHBup1.center);
  \draw[line width=0.6pt] (RHBdn2.center) -- (RHBup2.center);

\end{pgfonlayer}
\end{tikzpicture}%
}}\,,
\end{equation}

and since $\tilde{h}\in H_S$, it evaluates to $0$ on all product effects, so the last term vanishes. Therefore, we end up with
\begin{equation}
p(a,b|x,y)=\vcenter{\hbox{%
\begin{tikzpicture}[baseline={(base.center)}]
\begin{pgfonlayer}{nodelayer}
  \node (base) at (0,0) {}; 

  \def\dw{0.07}
  \def\yEffBase{1.05}
  \def\yEffTip{ 2.17}
  \def\triHalfW{0.91}

  \def\ySmBase{0.00}
  \def\ySmTip{-1.05}
  \def\smHalfW{0.75}

  \def\yWireLab{0.55}

  \def\xMA{0.00}   
  \def\xMB{2.25}   

  \node[style=none] (SUM) at (-1.55,0.65) {$\sum_i p_i$};

  \newcommand{\EffectTri}[3]{%
    \node[style=none] (#2BL) at ({#3-\triHalfW-0.2}, \yEffBase) {};
    \node[style=none] (#2BR) at ({#3+\triHalfW}, \yEffBase) {};
    \node[style=none] (#2T)  at (#3, \yEffTip) {};
    \node[style=none] (#2Lab) at (#3, 1.40) {$#1$};
  }
  \newcommand{\InvTri}[7]{%
    \node[style=none] (#2L) at ({#3-#4}, #5) {};
    \node[style=none] (#2R) at ({#3+#4}, #5) {};
    \node[style=none] (#2T) at (#3, #6) {};
    \node[style=none] (#2Lab) at (#3, #7) {$#1$};
  }
  \newcommand{\DoubleWire}[4]{%
    \node[style=none] (#1dn1) at ({#2-\dw}, #3) {};
    \node[style=none] (#1dn2) at ({#2+\dw}, #3) {};
    \node[style=none] (#1up1) at ({#2-\dw}, #4) {};
    \node[style=none] (#1up2) at ({#2+\dw}, #4) {};
  }

  \EffectTri{e_{a|x}}{Ae}{\xMA}
  \InvTri{\omega_i}{Aom}{\xMA}{\smHalfW}{\ySmBase}{\ySmTip}{-0.35}
  \DoubleWire{Aw}{\xMA}{\ySmBase}{\yEffBase}
  \node[style=none] (Alab) at ({\xMA+0.28}, \yWireLab) {$\scriptstyle \m{A}$};

  \EffectTri{f_{b|y}}{Bf}{\xMB}
  \InvTri{\nu_i}{Bnu}{\xMB}{\smHalfW}{\ySmBase}{\ySmTip}{-0.35}
  \DoubleWire{Bw}{\xMB}{\ySmBase}{\yEffBase}
  \node[style=none] (Blab) at ({\xMB+0.28}, \yWireLab) {$\scriptstyle \m{B}$};

\end{pgfonlayer}

\begin{pgfonlayer}{edgelayer}
  \draw (AeBL.center) -- (AeBR.center) -- (AeT.center) -- cycle;
  \draw (BfBL.center) -- (BfBR.center) -- (BfT.center) -- cycle;

  \draw (AomL.center) -- (AomR.center) -- (AomT.center) -- cycle;
  \draw (BnuL.center) -- (BnuR.center) -- (BnuT.center) -- cycle;

  \draw[line width=0.6pt] (Awdn1.center) -- (Awup1.center);
  \draw[line width=0.6pt] (Awdn2.center) -- (Awup2.center);
  \draw[line width=0.6pt] (Bwdn1.center) -- (Bwup1.center);
  \draw[line width=0.6pt] (Bwdn2.center) -- (Bwup2.center);
\end{pgfonlayer}

\end{tikzpicture}%
}}=\sum_i p_i p(a|x,i)p(b|y,i).
\end{equation}

The final expression for the joint probability $p(ab|xy)$ provides a local hidden variable model for it, which shows that $p(ab|xy)$ cannot violate Bell inequality. 
 \end{proof}
The result above makes it clear why it is entirely natural that the rebit states $\omega^{AB}_{\pm}$, although maximally entangled (as quantified by concurrence), cannot violate any Bell inequality: the state is indeed entangled, but all of its entanglement is entirely of the TNL type, which cannot contribute to Bell violations, regardless of which GPT is considered.

\begin{proposition} Consider a bipartite system $\m{AB}$. Any state in  $\Omega^{\m{AB}}$ that lacks tomographically-local entanglement is useless for steering (even if it has tomographically-nonlocal entanglement).
\end{proposition}
\begin{proof}
    Consider the action of an effect $e\in E_{\m{A}}$ on an entangled state $\omega^{\m{AB}}$ that lacks TL entanglement. Given the general decomposition for a state $\omega^{\m{AB}}$ without TL entanglement (Eq.~\eqref{eq:StateswithoutTLEntanglement}), the conditional state (which belongs to $B$ as per the requirements of GPT composite systems, Def.~\ref{def: CompositionRequirements}) is
    \begin{equation}
    B\ni
\vcenter{\hbox{%
\begin{tikzpicture}
\begin{pgfonlayer}{nodelayer}
\node[style=none] (LwTL) at (-5.20,0.00) {};   
\node[style=none] (LwTR) at (-3.60,0.00) {};   
\node[style=none] (LwB)  at (-4.40,-1.30) {};  
\node[style=none] (Lw)   at (-4.40,-0.55) {$\omega$};

\node[style=none] (LAu1) at (-4.85,0.00) {};
\node[style=none] (LAu2) at (-4.75,0.00) {};
\node[style=none] (LBu1) at (-3.95,0.00) {};
\node[style=none] (LBu2) at (-3.85,0.00) {};

\node[style=none] (LAt1) at (-4.85,1.25) {};
\node[style=none] (LAt2) at (-4.75,1.25) {};
\node[style=none] (LBt1) at (-3.95,1.25) {};
\node[style=none] (LBt2) at (-3.85,1.25) {};

\node[style=none] (LAlab) at (-5.25,0.62) {$\scriptstyle \m{A}$};
\node[style=none] (LBlab) at (-3.45,0.62) {$\scriptstyle \m{B}$};

\node[style=none] (LeBL) at (-5.25,1.25) {};
\node[style=none] (LeBR) at (-4.35,1.25) {};
\node[style=none] (LeT)  at (-4.80,2.20) {};
\node[style=none] (Le)   at (-4.80,1.72) {$e$};

\node[style=none] (Eq)   at (-2.78,0.05) {$=$};
\node[style=none] (Sum)  at (-1.46,0.13) {$\sum_i p_i$};

\node[style=none] (Mw1TL) at (-0.59,0.00) {};
\node[style=none] (Mw1TR) at ( 0.59,0.00) {};
\node[style=none] (Mw1B)  at ( 0.00,-1.04) {};
\node[style=none] (Mw1)   at ( 0.00,-0.35) {$\omega_i$};

\node[style=none] (MAu1) at (-0.08,0.00) {};
\node[style=none] (MAu2) at ( 0.08,0.00) {};
\node[style=none] (MAt1) at (-0.08,1.25) {};
\node[style=none] (MAt2) at ( 0.08,1.25) {};
\node[style=none] (MAlab) at (-0.35,0.62) {$\scriptstyle \m{A}$};

\node[style=none] (MeBL) at (-0.45,1.25) {};
\node[style=none] (MeBR) at ( 0.45,1.25) {};
\node[style=none] (MeT)  at ( 0.00,2.20) {};
\node[style=none] (Me)   at ( 0.00,1.72) {$e$};

\node[style=none] (Mw2TL) at (0.81,0.00) {};
\node[style=none] (Mw2TR) at (2.19,0.00) {};
\node[style=none] (Mw2B)  at (1.50,-0.95) {};
\node[style=none] (Mw2)   at (1.50,-0.35) {$\nu_i$};

\node[style=none] (MBu1) at (1.42,0.00) {};
\node[style=none] (MBu2) at (1.58,0.00) {};
\node[style=none] (MBt1) at (1.42,1.25) {};
\node[style=none] (MBt2) at (1.58,1.25) {};
\node[style=none] (MBlab) at (2.25,0.62) {$\scriptstyle \m{B}$};

\node[style=none] (Plus) at (3.05,0.05) {$+$};

\node[style=none] (RwTL) at (4.15,0.00) {};
\node[style=none] (RwTR) at (5.75,0.00) {};
\node[style=none] (RwB)  at (4.95,-1.30) {};
\node[style=none] (Rnu)  at (4.95,-0.55) {$\tilde{h}$};

\node[style=none] (RAu1) at (4.50,0.00) {};
\node[style=none] (RAu2) at (4.60,0.00) {};
\node[style=none] (RBu1) at (5.30,0.00) {};
\node[style=none] (RBu2) at (5.40,0.00) {};
\node[style=none] (RAt1) at (4.50,1.25) {};
\node[style=none] (RAt2) at (4.60,1.25) {};
\node[style=none] (RBt1) at (5.30,1.25) {};
\node[style=none] (RBt2) at (5.40,1.25) {};

\node[style=none] (RAlab) at (4.05,0.62) {$\scriptstyle \m{A}$};
\node[style=none] (RBlab) at (5.95,0.62) {$\scriptstyle \m{B}$};

\node[style=none] (ReBL) at (4.10,1.25) {};
\node[style=none] (ReBR) at (5.00,1.25) {};
\node[style=none] (ReT)  at (4.55,2.20) {};
\node[style=none] (Re)   at (4.55,1.72) {$e$};
\end{pgfonlayer}

\begin{pgfonlayer}{edgelayer}
\draw (LwTL.center) -- (LwTR.center) -- (LwB.center) -- cycle;
\draw (LeBL.center) -- (LeBR.center) -- (LeT.center) -- cycle;

\draw[line width=0.6pt] (LAu1.center) -- (LAt1.center);
\draw[line width=0.6pt] (LAu2.center) -- (LAt2.center);
\draw[line width=0.6pt] (LBu1.center) -- (LBt1.center);
\draw[line width=0.6pt] (LBu2.center) -- (LBt2.center);

\draw (Mw1TL.center) -- (Mw1TR.center) -- (Mw1B.center) -- cycle;
\draw (MeBL.center) -- (MeBR.center) -- (MeT.center) -- cycle;
\draw[line width=0.6pt] (MAu1.center) -- (MAt1.center);
\draw[line width=0.6pt] (MAu2.center) -- (MAt2.center);

\draw (Mw2TL.center) -- (Mw2TR.center) -- (Mw2B.center) -- cycle;
\draw[line width=0.6pt] (MBu1.center) -- (MBt1.center);
\draw[line width=0.6pt] (MBu2.center) -- (MBt2.center);

\draw (RwTL.center) -- (RwTR.center) -- (RwB.center) -- cycle;
\draw (ReBL.center) -- (ReBR.center) -- (ReT.center) -- cycle;
\draw[line width=0.6pt] (RAu1.center) -- (RAt1.center);
\draw[line width=0.6pt] (RAu2.center) -- (RAt2.center);
\draw[line width=0.6pt] (RBu1.center) -- (RBt1.center);
\draw[line width=0.6pt] (RBu2.center) -- (RBt2.center);
\end{pgfonlayer}
\end{tikzpicture}%
}}\vcenter{\hbox{%
\begin{tikzpicture}
\begin{pgfonlayer}{nodelayer}
\node[style=none] (Imp) at (-2.17,0.00) {$\implies$};

\node[style=none] (wTL) at (-0.80,0.00) {};
\node[style=none] (wTR) at ( 0.80,0.00) {};
\node[style=none] (wB)  at ( 0.00,-1.20) {};
\node[style=none] (nu)  at ( 0.00,-0.52) {$\tilde{h}$};

\node[style=none] (Au1) at (-0.25,0.00) {};
\node[style=none] (Au2) at (-0.15,0.00) {};
\node[style=none] (Bu1) at ( 0.35,0.00) {};
\node[style=none] (Bu2) at ( 0.45,0.00) {};
\node[style=none] (At1) at (-0.25,1.10) {};
\node[style=none] (At2) at (-0.15,1.10) {};
\node[style=none] (Bt1) at ( 0.35,1.10) {};
\node[style=none] (Bt2) at ( 0.45,1.10) {};

\node[style=none] (eBL) at (-0.70,1.10) {};
\node[style=none] (eBR) at ( 0.10,1.10) {};
\node[style=none] (eT)  at (-0.30,1.95) {};
\node[style=none] (e)   at (-0.30,1.52) {$e$};

\node[style=none] (in)  at (1.20,0.00) {$\in$};
\node[style=none] (Bset) at (2.15,0.00) { ${B}.$};
\end{pgfonlayer}

\begin{pgfonlayer}{edgelayer}
\draw (wTL.center) -- (wTR.center) -- (wB.center) -- cycle;
\draw (eBL.center) -- (eBR.center) -- (eT.center) -- cycle;

\draw[line width=0.6pt] (Au1.center) -- (At1.center);
\draw[line width=0.6pt] (Au2.center) -- (At2.center);
\draw[line width=0.6pt] (Bu1.center) -- (Bt1.center);
\draw[line width=0.6pt] (Bu2.center) -- (Bt2.center);
\end{pgfonlayer}
\end{tikzpicture}%
}}
\end{equation}
    where the implication follows from the fact that the first term to the rhs of the equal sign is also in $B$, and $B$ is a vector space. We now analyze which vector of $B$ this can be.
    Recall that by the definition of the holistic-state subspace $H_S$, one has

\begin{equation}
    \begin{tikzpicture}
\begin{pgfonlayer}{nodelayer}

\node[style=none] (nuTL)   at (0.105,0.00) {};     
\node[style=none] (nuTR)   at (1.995,0.00) {};     
\node[style=none] (nuB)    at (1.050,-1.44) {};    
\node[style=none] (nuLab)  at (1.050,-0.63) {$\tilde{h}$};


\node[style=none] (AuL) at (0.58,0.00) {};
\node[style=none] (AuR) at (0.68,0.00) {};
\node[style=none] (AtL) at (0.58,1.25) {};
\node[style=none] (AtR) at (0.68,1.25) {};
\node[style=none] (Alab) at (0.88,0.62) {$\scriptstyle \m{A}$};

\node[style=none] (BuL) at (1.42,0.00) {};
\node[style=none] (BuR) at (1.52,0.00) {};
\node[style=none] (BtL) at (1.42,1.25) {};
\node[style=none] (BtR) at (1.52,1.25) {};
\node[style=none] (Blab) at (1.72,0.62) {$\scriptstyle \m{B}$};

\node[style=none] (eBL)  at (0.245,1.25) {};   
\node[style=none] (eBR)  at (1.015,1.25) {};   
\node[style=none] (eT)   at (0.630,2.13) {};
\node[style=none] (eLab) at (0.630,1.67) {$e$};

\node[style=none] (fBL)  at (1.085,1.25) {};   
\node[style=none] (fBR)  at (1.855,1.25) {};   
\node[style=none] (fT)   at (1.470,2.13) {};
\node[style=none] (fLab) at (1.38,1.65) {$f$};

\node[style=none] (Eq)     at (3.05,0.05) {$=$};
\node[style=none] (Zero)   at (4.05,0.05) {$0$};
\node[style=none] (Forall) at (5.10,0.05) {$\forall$};


\node[style=none] (RBbL) at (6.10,0.00) {};
\node[style=none] (RBbR) at (6.20,0.00) {};
\node[style=none] (RBtL) at (6.10,0.95) {};
\node[style=none] (RBtR) at (6.20,0.95) {};
\node[style=none] (RBlabel) at (6.42,0.45) {$\scriptstyle \m{B}$};

\node[style=none] (RfBL)  at (5.82,0.95) {};   
\node[style=none] (RfBR)  at (6.59,0.95) {};   
\node[style=none] (RfT)   at (6.15,1.775) {};  
\node[style=none] (RfLab) at (6.07,1.31) {$f$};

\node[style=none] (In) at (7.25,0.05) {$\in$};
\node[style=none] (EB) at (8.35,0.05) {$E_{\m{B}}$};

\end{pgfonlayer}

\begin{pgfonlayer}{edgelayer}
\draw (nuTL.center) -- (nuTR.center) -- (nuB.center) -- cycle;
\draw (eBL.center)  -- (eBR.center)  -- (eT.center)  -- cycle;
\draw (fBL.center)  -- (fBR.center)  -- (fT.center)  -- cycle;

\draw[line width=0.6pt] (AuL.center) -- (AtL.center);
\draw[line width=0.6pt] (AuR.center) -- (AtR.center);
\draw[line width=0.6pt] (BuL.center) -- (BtL.center);
\draw[line width=0.6pt] (BuR.center) -- (BtR.center);

\draw[line width=0.6pt] (RBbL.center) -- (RBtL.center);
\draw[line width=0.6pt] (RBbR.center) -- (RBtR.center);
\draw (RfBL.center) -- (RfBR.center) -- (RfT.center) -- cycle;
\end{pgfonlayer}
\end{tikzpicture}
\end{equation}

 which implies that
 \begin{equation}
\begin{tikzpicture}
\begin{pgfonlayer}{nodelayer}

\node[style=none] (nuTL)   at (0.10,0.00) {};
\node[style=none] (nuTR)   at (1.99,0.00) {};
\node[style=none] (nuB)    at (1.05,-1.30) {};
\node[style=none] (nuLab)  at (1.05,-0.52) {$\tilde{h}$};

\node[style=none] (AuL) at (0.58,0.00) {};
\node[style=none] (AuR) at (0.68,0.00) {};
\node[style=none] (AtL) at (0.58,1.25) {};
\node[style=none] (AtR) at (0.68,1.25) {};
\node[style=none] (Alab) at (0.88,0.62) {$\scriptstyle \m{A}$};

\node[style=none] (BuL) at (1.42,0.00) {};
\node[style=none] (BuR) at (1.52,0.00) {};
\node[style=none] (BtL) at (1.42,1.25) {};
\node[style=none] (BtR) at (1.52,1.25) {};
\node[style=none] (Blab) at (1.72,0.62) {$\scriptstyle \m{B}$};

\node[style=none] (eBL)  at (0.245,1.25) {};
\node[style=none] (eBR)  at (1.015,1.25) {};
\node[style=none] (eT)   at (0.630,2.13) {};
\node[style=none] (eLab) at (0.630,1.67) {$e$};

\node[style=none] (Eq) at (2.85,0.05) {$=$};


\node[style=none] (RBtL) at (4.10,1.10) {};
\node[style=none] (RBtR) at (4.20,1.10) {};
\node[style=none] (RBbL) at (4.10,0.00) {};
\node[style=none] (RBbR) at (4.20,0.00) {};
\node[style=none] (RBlabel) at (4.42,0.55) {$\scriptstyle \m{B}$};

\node[style=none] (zTL)  at (3.65,0.00) {};
\node[style=none] (zTR)  at (4.65,0.00) {};
\node[style=none] (zB)   at (4.15,-0.91) {};
\node[style=none] (zLab) at (4.15,-0.35) {$\vec{0}$};


\end{pgfonlayer}

\begin{pgfonlayer}{edgelayer}
\draw (nuTL.center) -- (nuTR.center) -- (nuB.center) -- cycle;
\draw (eBL.center)  -- (eBR.center)  -- (eT.center)  -- cycle;

\draw[line width=0.6pt] (AuL.center) -- (AtL.center);
\draw[line width=0.6pt] (AuR.center) -- (AtR.center);
\draw[line width=0.6pt] (BuL.center) -- (BtL.center);
\draw[line width=0.6pt] (BuR.center) -- (BtR.center);

\draw[line width=0.6pt] (RBtL.center) -- (RBbL.center);
\draw[line width=0.6pt] (RBtR.center) -- (RBbR.center);
\draw (zTL.center) -- (zTR.center) -- (zB.center) -- cycle;
\end{pgfonlayer}
\end{tikzpicture}
\end{equation}
    (since the effects $f^{\m{B}}$ span $B^*$), i.e., the holistic component does not contribute to the conditional state of subsystem $\m{B}$.
    Therefore, we can conclude that the conditional state can be written as
    \begin{equation}
    \begin{tikzpicture}
\begin{pgfonlayer}{nodelayer}

\node[style=none] (wTL)   at (0.10,0.00) {};
\node[style=none] (wTR)   at (1.99,0.00) {};
\node[style=none] (wB)    at (1.05,-1.44) {};
\node[style=none] (wLab)  at (1.08,-0.63) {$\omega$};

\node[style=none] (AuL) at (0.58,0.00) {};
\node[style=none] (AuR) at (0.68,0.00) {};
\node[style=none] (AtL) at (0.58,1.25) {};
\node[style=none] (AtR) at (0.68,1.25) {};
\node[style=none] (Alab) at (0.88,0.62) {$\scriptstyle \m{A}$};

\node[style=none] (BuL) at (1.42,0.00) {};
\node[style=none] (BuR) at (1.52,0.00) {};
\node[style=none] (BtL) at (1.42,1.25) {};
\node[style=none] (BtR) at (1.52,1.25) {};
\node[style=none] (Blab) at (1.72,0.62) {$\scriptstyle \m{B}$};

\node[style=none] (eBL)  at (0.245,1.25) {};
\node[style=none] (eBR)  at (1.015,1.25) {};
\node[style=none] (eT)   at (0.630,2.13) {};
\node[style=none] (eLab) at (0.630,1.67) {$e$};

\node[style=none] (Eq)   at (2.80,0.0) {$=$};
\node[style=none] (Sum)  at (3.95,0.0) {$\sum_i p_i$};


\node[style=none] (wi1TL)  at (4.82,0.00) {};
\node[style=none] (wi1TR)  at (5.98,0.00) {};
\node[style=none] (wi1B)   at (5.40,-0.98) {};
\node[style=none] (wi1Lab) at (5.42,-0.38) {$\omega_i$};

\node[style=none] (RAuL) at (5.35,0.00) {};
\node[style=none] (RAuR) at (5.45,0.00) {};
\node[style=none] (RAtL) at (5.35,1.25) {};
\node[style=none] (RAtR) at (5.45,1.25) {};
\node[style=none] (RAlab) at (5.65,0.62) {$\scriptstyle \m{A}$};

\node[style=none] (ReBL)  at (5.015,1.25) {};
\node[style=none] (ReBR)  at (5.785,1.25) {};
\node[style=none] (ReT)   at (5.400,2.13) {};
\node[style=none] (ReLab) at (5.400,1.67) {$e$};

\node[style=none] (wi2TL)  at (6.12,0.00) {};
\node[style=none] (wi2TR)  at (7.28,0.00) {};
\node[style=none] (wi2B)   at (6.70,-0.98) {};
\node[style=none] (wi2Lab) at (6.70,-0.38) {$\nu_i$};

\node[style=none] (RBuL) at (6.65,0.00) {};
\node[style=none] (RBuR) at (6.75,0.00) {};
\node[style=none] (RBtL) at (6.65,1.25) {};
\node[style=none] (RBtR) at (6.75,1.25) {};
\node[style=none] (RBlab) at (6.95,0.62) {$\scriptstyle \m{B}$};

\end{pgfonlayer}

\begin{pgfonlayer}{edgelayer}
\draw (wTL.center) -- (wTR.center) -- (wB.center) -- cycle;
\draw (eBL.center) -- (eBR.center) -- (eT.center) -- cycle;

\draw[line width=0.6pt] (AuL.center) -- (AtL.center);
\draw[line width=0.6pt] (AuR.center) -- (AtR.center);
\draw[line width=0.6pt] (BuL.center) -- (BtL.center);
\draw[line width=0.6pt] (BuR.center) -- (BtR.center);

\draw (wi1TL.center) -- (wi1TR.center) -- (wi1B.center) -- cycle;
\draw (ReBL.center)  -- (ReBR.center)  -- (ReT.center)  -- cycle;
\draw[line width=0.6pt] (RAuL.center) -- (RAtL.center);
\draw[line width=0.6pt] (RAuR.center) -- (RAtR.center);

\draw (wi2TL.center) -- (wi2TR.center) -- (wi2B.center) -- cycle;
\draw[line width=0.6pt] (RBuL.center) -- (RBtL.center);
\draw[line width=0.6pt] (RBuR.center) -- (RBtR.center);
\end{pgfonlayer}
\end{tikzpicture}
\end{equation}
    The last expression, equivalent to $\sum_i p_ip(e|i)\nu_i^{\m{B}}$, implies that any assemblage constructed from $\omega^{\m{AB}}$ by one party performing local measurements will admit of a local hidden state model~\cite{steeringwiseman}. Hence, no assemblage constructed from $\omega^{\m{AB}}$  is useful for steering (i.e., it is  manifestly LOSR-free~\cite{Schmid2020typeindependent,Zjawin2023quantifyingepr}). \blk    
\end{proof}

\begin{proposition}
\label{Prop.: TNLEAloneUselessForSteering}Consider a bipartite system $\m{AB}$. Any state in  $\Omega^{\m{AB}}$ that lacks tomographically-local entanglement is useless for teleportation (even if it has tomographically-nonlocal entanglement).
\end{proposition}
\begin{proof}
A teleportation protocol (acting on a particular state $\psi$) and using a bipartite state $\omega^{\m{AB}}$ as a resource starts with a joint measurement made on $\psi$ and the $\m{A}$ subsystem. Then, Alice might communicate the outcome of the measurement to Bob, which, depending on the outcome, might perform an operation on his subsystem $\m{B}$. Representing a measurement as an operation from the input systems to a classical GPT system $\m{C}$, we can represent a teleportation protocol in the following diagram
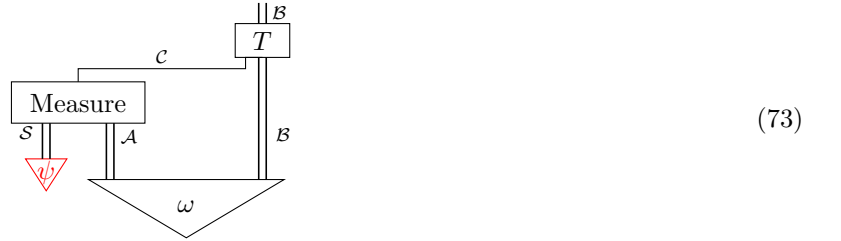
\begin{equation}
\vcenter{\hbox{%
\begin{tikzpicture}[baseline={(base.center)}]
\begin{pgfonlayer}{nodelayer}

\node[style=none] (base) at (0,0) {};

\def\dw{0.10}                 

\def\BellCx{2.90}
\def\BellCy{2.05}
\def\BellW{3.55}
\def\BellH{1.10}

\node[style=none] (BellTL) at ({\BellCx-\BellW/2}, {\BellCy+\BellH/2}) {};
\node[style=none] (BellTR) at ({\BellCx+\BellW/2}, {\BellCy+\BellH/2}) {};
\node[style=none] (BellBR) at ({\BellCx+\BellW/2}, {\BellCy-\BellH/2}) {};
\node[style=none] (BellBL) at ({\BellCx-\BellW/2}, {\BellCy-\BellH/2}) {};
\node[style=none] (BellLab) at (\BellCx,\BellCy) {$\mathrm{Measure}$};

\def\Sx{2.05}
\node[style=none] (SuL1) at ({\Sx-\dw}, {\BellCy-\BellH/2}) {};
\node[style=none] (SuR1) at ({\Sx+\dw}, {\BellCy-\BellH/2}) {};
\node[style=none] (SLab) at (1.52,1.25) {$\scriptstyle \m{S}$};

\def\Ax{3.75}
\node[style=none] (AuL1) at ({\Ax-\dw}, {\BellCy-\BellH/2}) {};
\node[style=none] (AuR1) at ({\Ax+\dw}, {\BellCy-\BellH/2}) {};
\node[style=none] (ALab) at (4.25,1.15) {$\scriptstyle \m{A}$};

\node[style=none] (Cout0) at (\BellCx, {\BellCy+\BellH/2}) {};
\node[style=none] (Ccorner1) at (\BellCx, 2.95) {};        

\def\Psix{\Sx}
\node[style=none] (PsiTL)  at ({\Psix-0.55}, 0.55) {};
\node[style=none] (PsiTR)  at ({\Psix+0.55}, 0.55) {};
\node[style=none] (PsiB)   at (\Psix,-0.30) {};
\node[style=none] (PsiLab) at (\Psix,0.18) {{\color{red}$\psi$}};

\node[style=none] (SuL0) at ({\Sx-\dw}, 0.55) {};
\node[style=none] (SuR0) at ({\Sx+\dw}, 0.55) {};

\def\TCx{7.80}
\def\TCy{3.70}
\def\TW{1.45}
\def\TH{0.90}

\node[style=none] (TTL) at ({\TCx-\TW/2}, {\TCy+\TH/2}) {};
\node[style=none] (TTR) at ({\TCx+\TW/2}, {\TCy+\TH/2}) {};
\node[style=none] (TBR) at ({\TCx+\TW/2}, {\TCy-\TH/2}) {};
\node[style=none] (TBL) at ({\TCx-\TW/2}, {\TCy-\TH/2}) {};
\node[style=none] (TLab) at (\TCx,\TCy) {$T$};

\node[style=none] (Cin) at (\TCx, {\TCy-\TH/2}) {};

\def\Bx{\TCx}

\node[style=none] (AuL0) at ({\Ax-\dw}, 0.00) {};
\node[style=none] (AuR0) at ({\Ax+\dw}, 0.00) {};
\node[style=none] (BuL0) at ({\Bx-\dw}, 0.00) {};
\node[style=none] (BuR0) at ({\Bx+\dw}, 0.00) {};

\node[style=none] (OmTL)  at ({0.5*(\Ax+\Bx)-2.60}, 0.00) {};
\node[style=none] (OmTR)  at ({0.5*(\Ax+\Bx)+2.60}, 0.00) {};
\node[style=none] (OmB)   at ({0.5*(\Ax+\Bx)}, -1.55) {};
\node[style=none] (OmLab) at ({0.5*(\Ax+\Bx)}, -0.70) {$\omega$};

\node[style=none] (BuL_in) at ({\Bx-\dw}, {\TCy-\TH/2}) {};
\node[style=none] (BuR_in) at ({\Bx+\dw}, {\TCy-\TH/2}) {};

\node[style=none] (BuL_out0) at ({\Bx-\dw}, {\TCy+\TH/2}) {};
\node[style=none] (BuR_out0) at ({\Bx+\dw}, {\TCy+\TH/2}) {};
\node[style=none] (BuL_out1) at ({\Bx-\dw}, 4.70) {};
\node[style=none] (BuR_out1) at ({\Bx+\dw}, 4.70) {};

\node[style=none] (BLab_top) at ({\Bx+0.45},4.45) {$\scriptstyle \m{B}$};
\node[style=none] (BLab_mid) at ({\Bx+0.55},1.20) {$\scriptstyle \m{B}$};

\node[style=none] (Ccorner2) at ({\Bx-0.45}, 2.95) {};

\node[style=none] (Ccorner3) at ({\Bx-0.45}, {\TCy-\TH/2}) {};

\node[style=none] (CLab) at ({0.5*(\BellCx+(\Bx-0.45))}, 3.26) {$\scriptstyle \m{C}$};

\end{pgfonlayer}

\begin{pgfonlayer}{edgelayer}

\draw (BellTL.center) -- (BellTR.center) -- (BellBR.center) -- (BellBL.center) -- cycle;

\draw (TTL.center) -- (TTR.center) -- (TBR.center) -- (TBL.center) -- cycle;

\draw (OmTL.center) -- (OmTR.center) -- (OmB.center) -- cycle;

\draw[red] (PsiTL.center) -- (PsiTR.center) -- (PsiB.center) -- cycle;

\draw[line width=0.6pt] (SuL0.center) -- (SuL1.center);
\draw[line width=0.6pt] (SuR0.center) -- (SuR1.center);

\draw[line width=0.6pt] (AuL0.center) -- (AuL1.center);
\draw[line width=0.6pt] (AuR0.center) -- (AuR1.center);

\draw[line width=0.6pt] (BuL0.center) -- (BuL_in.center);
\draw[line width=0.6pt] (BuR0.center) -- (BuR_in.center);
\draw[line width=0.6pt] (BuL_out0.center) -- (BuL_out1.center);
\draw[line width=0.6pt] (BuR_out0.center) -- (BuR_out1.center);

\draw (Cout0.center) -- (Ccorner1.center) -- (Ccorner2.center) -- (Ccorner3.center) -- (Cin.center);

    \end{pgfonlayer}
    \end{tikzpicture}%
    }}
\end{equation}
State $\omega$ allows for perfect teleportation if and only if there exist a measurement and T operations such that the final state of subsystem $\m{B}$ is equivalent to $\psi$.   Following Refs.~\cite{PhysRevLett.119.110501,Schmid2020typeindependent}, we now focus on the part of the protocol prior to the outcome $a$ being communicated to the party on the right-hand-side:
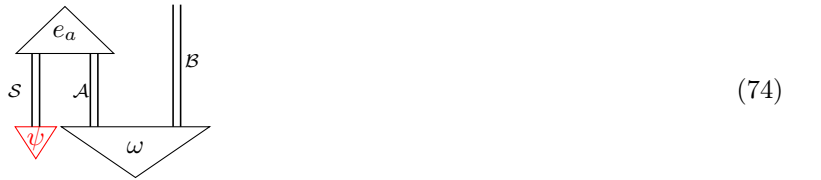
\begin{equation}
\label{eq:DiagramTeleportage}
    \vcenter{\hbox{%
\begin{tikzpicture}[baseline={(base.center)}]
\begin{pgfonlayer}{nodelayer}
\node[style=none] (base) at (0,0) {};

\pgfmathsetmacro{\dw}{0.10}       

\pgfmathsetmacro{\Sx}{0.00}
\pgfmathsetmacro{\Ax}{1.55}
\pgfmathsetmacro{\Bx}{3.75}

\pgfmathsetmacro{\yPsiTop}{0.60}
\pgfmathsetmacro{\yOmTop}{0.60}
\pgfmathsetmacro{\yEaBase}{2.55}
\pgfmathsetmacro{\yBTop}{3.85}

\pgfmathsetmacro{\xEa}{0.5*(\Sx+\Ax)}
\pgfmathsetmacro{\xOm}{0.5*(\Ax+\Bx)}
\pgfmathsetmacro{\xPsi}{\Sx}

\pgfmathsetmacro{\eaHalfW}{1.30}
\pgfmathsetmacro{\eaH}{1.25}

\node[style=none] (EaBL) at (\xEa-\eaHalfW,\yEaBase) {};
\node[style=none] (EaBR) at (\xEa+\eaHalfW,\yEaBase) {};
\node[style=none] (EaT)  at (\xEa,\yEaBase+\eaH) {};
\node[style=none] (EaLab) at (\xEa,\yEaBase+0.55) {$e_a$};

\pgfmathsetmacro{\psiHalfW}{0.55}
\pgfmathsetmacro{\psiH}{0.85}

\node[style=none] (PsiTL) at (\xPsi-\psiHalfW,\yPsiTop) {};
\node[style=none] (PsiTR) at (\xPsi+\psiHalfW,\yPsiTop) {};
\node[style=none] (PsiB)  at (\xPsi,\yPsiTop-\psiH) {};
\node[style=none] (PsiLab) at (\xPsi,\yPsiTop-0.32) {{\color{red}$\psi$}};

\pgfmathsetmacro{\omHalfW}{2.10}
\pgfmathsetmacro{\omH}{1.35}

\node[style=none] (OmTL) at (\xOm-\omHalfW+0.13,\yOmTop) {};
\node[style=none] (OmTR) at (\xOm+\omHalfW-0.13,\yOmTop) {};
\node[style=none] (OmB)  at (\xOm,\yOmTop-\omH) {};
\node[style=none] (OmLab) at (\xOm,\yOmTop-0.55) {$\omega$};

\node[style=none] (SuL0) at (\Sx-\dw,\yPsiTop) {};
\node[style=none] (SuR0) at (\Sx+\dw,\yPsiTop) {};
\node[style=none] (SuL1) at (\Sx-\dw,\yEaBase) {};
\node[style=none] (SuR1) at (\Sx+\dw,\yEaBase) {};
\node[style=none] (SLab) at (\Sx-0.55,1.55) {$\scriptstyle \m{S}$};

\node[style=none] (AuL0) at (\Ax-\dw,\yOmTop) {};
\node[style=none] (AuR0) at (\Ax+\dw,\yOmTop) {};
\node[style=none] (AuL1) at (\Ax-\dw,\yEaBase) {};
\node[style=none] (AuR1) at (\Ax+\dw,\yEaBase) {};
\node[style=none] (ALab) at (\Ax-0.35,1.55) {$\scriptstyle \m{A}$};

\node[style=none] (BuL0) at (\Bx-\dw,\yOmTop) {};
\node[style=none] (BuR0) at (\Bx+\dw,\yOmTop) {};
\node[style=none] (BuL1) at (\Bx-\dw,\yBTop) {};
\node[style=none] (BuR1) at (\Bx+\dw,\yBTop) {};
\node[style=none] (BLab) at (\Bx+0.40,2.35) {$\scriptstyle \m{B}$};

\end{pgfonlayer}

\begin{pgfonlayer}{edgelayer}
\draw (EaBL.center) -- (EaBR.center) -- (EaT.center) -- cycle;
\draw[red] (PsiTL.center) -- (PsiTR.center) -- (PsiB.center) -- cycle;
\draw (OmTL.center) -- (OmTR.center) -- (OmB.center) -- cycle;

\draw[line width=0.6pt] (SuL0.center) -- (SuL1.center);
\draw[line width=0.6pt] (SuR0.center) -- (SuR1.center);

\draw[line width=0.6pt] (AuL0.center) -- (AuL1.center);
\draw[line width=0.6pt] (AuR0.center) -- (AuR1.center);

\draw[line width=0.6pt] (BuL0.center) -- (BuL1.center);
\draw[line width=0.6pt] (BuR0.center) -- (BuR1.center);
\end{pgfonlayer}
\end{tikzpicture}%
}}
\end{equation}

    for some effect $e_{a}$.
Since the state $\omega^{\m{AB}}$ lacks TL-entanglement, Eq.~\eqref{eq:StateswithoutTLEntanglement} gives
    \begin{equation}
\vcenter{\hbox{
\begin{tikzpicture}
	\begin{pgfonlayer}{nodelayer}

		\node [style=none] (L1) at (-4.8,0) {};
		\node [style=none] (R1) at (-3.0,0) {};
		\node [style=none] (C1) at (-3.9,-1.2) {};
		\node [style=none] (W1) at (-3.9,-0.4) {$\omega$};

		\node [style=none] (T1a) at (-4.4,0) {};
		\node [style=none] (T1b) at (-3.4,0) {};

		\node [style=none] (T1aL) at (-4.44,0) {};
		\node [style=none] (T1aR) at (-4.36,0) {};
		\node [style=none] (T1bL) at (-3.44,0) {};
		\node [style=none] (T1bR) at (-3.36,0) {};

		\node [style=none] (A1) at (-4.15,0.5) {$\scriptstyle \m{A}$};
		\node [style=none] (B1) at (-3.15,0.5) {$\scriptstyle \m{B}$};

		\node [style=none] (eq)  at (-1.7,0) {$=$};
		\node [style=none] (Sij) at (-0.3,0) {$\sum_{i} p_{i}$};

		\node [style=none] (L2) at (0.7,0) {};
		\node [style=none] (R2) at (2.2,0) {};
		\node [style=none] (C2) at (1.45,-1.0) {};
		\node [style=none] (W2) at (1.45,-0.33) {$\omega_i$};

		\node [style=none] (T2) at (1.45,0) {};
		\node [style=none] (T2L) at (1.41,0) {};
		\node [style=none] (T2R) at (1.49,0) {};
		\node [style=none] (A2) at (1.70,0.42) {$\scriptstyle \m{A}$};

		\node [style=none] (L3) at (2.4,0) {};
		\node [style=none] (R3) at (3.9,0) {};
		\node [style=none] (C3) at (3.15,-1.0) {};
		\node [style=none] (W3) at (3.15,-0.33) {$\nu_i$};

		\node [style=none] (T3) at (3.15,0) {};
		\node [style=none] (T3L) at (3.11,0) {};
		\node [style=none] (T3R) at (3.19,0) {};
		\node [style=none] (B3) at (3.40,0.42) {$\scriptstyle \m{B}$};

		\node [style=none] (plus) at (4.8,0) {$+$};

		\node [style=none] (L4) at (5.8,0) {};
		\node [style=none] (R4) at (7.9,0) {};
		\node [style=none] (C4) at (6.85,-1.0) {};
		\node [style=none] (W4) at (6.85,-0.38) {$\tilde{h}$};

		\node [style=none] (T4a) at (6.35,0) {};
		\node [style=none] (T4b) at (7.35,0) {};

		\node [style=none] (T4aL) at (6.31,0) {};
		\node [style=none] (T4aR) at (6.39,0) {};
		\node [style=none] (T4bL) at (7.31,0) {};
		\node [style=none] (T4bR) at (7.39,0) {};

		\node [style=none] (A4) at (6.60,0.5) {$\scriptstyle \m{A}$};
		\node [style=none] (B4) at (7.60,0.5) {$\scriptstyle \m{B}$};

	\end{pgfonlayer}

	\begin{pgfonlayer}{edgelayer}

		\draw (L1.center) to (R1.center);
		\draw (R1.center) to (C1.center);
		\draw (C1.center) to (L1.center);

		\draw (L2.center) to (R2.center);
		\draw (R2.center) to (C2.center);
		\draw (C2.center) to (L2.center);

		\draw (L3.center) to (R3.center);
		\draw (R3.center) to (C3.center);
		\draw (C3.center) to (L3.center);

		\draw (L4.center) to (R4.center);
		\draw (R4.center) to (C4.center);
		\draw (C4.center) to (L4.center);

		\draw[line width=0.6pt] (T1aL.center) to +(0,0.9);
		\draw[line width=0.6pt] (T1aR.center) to +(0,0.9);

		\draw[line width=0.6pt] (T1bL.center) to +(0,0.9);
		\draw[line width=0.6pt] (T1bR.center) to +(0,0.9);

		\draw[line width=0.6pt] (T2L.center) to +(0,0.9);
		\draw[line width=0.6pt] (T2R.center) to +(0,0.9);

		\draw[line width=0.6pt] (T3L.center) to +(0,0.9);
		\draw[line width=0.6pt] (T3R.center) to +(0,0.9);

		\draw[line width=0.6pt] (T4aL.center) to +(0,0.9);
		\draw[line width=0.6pt] (T4aR.center) to +(0,0.9);

		\draw[line width=0.6pt] (T4bL.center) to +(0,0.9);
		\draw[line width=0.6pt] (T4bR.center) to +(0,0.9);

	\end{pgfonlayer}
\end{tikzpicture}
}}\,.
\end{equation}
Substituting this into the expression  from Eq.~\eqref{eq:DiagramTeleportage},  we have the sum of two terms, the latter of which is 
\begin{equation}
    \vcenter{\hbox{%
\begin{tikzpicture}[baseline={(base.center)}]
\begin{pgfonlayer}{nodelayer}
\node[style=none] (base) at (0,0) {};

\pgfmathsetmacro{\dw}{0.10}

\pgfmathsetmacro{\Sx}{0.00}
\pgfmathsetmacro{\Ax}{1.55}
\pgfmathsetmacro{\Bx}{3.75}

\pgfmathsetmacro{\yPsiTop}{0.60}
\pgfmathsetmacro{\yOmTop}{0.60}
\pgfmathsetmacro{\yEaBase}{2.55}
\pgfmathsetmacro{\yBTop}{3.85}

\pgfmathsetmacro{\xEa}{0.5*(\Sx+\Ax)}
\pgfmathsetmacro{\xOm}{0.5*(\Ax+\Bx)}
\pgfmathsetmacro{\xPsi}{\Sx}

\pgfmathsetmacro{\eaHalfW}{1.30}
\pgfmathsetmacro{\eaH}{1.25}

\node[style=none] (EaBL)  at (\xEa-\eaHalfW,\yEaBase) {};
\node[style=none] (EaBR)  at (\xEa+\eaHalfW,\yEaBase) {};
\node[style=none] (EaT)   at (\xEa,\yEaBase+\eaH) {};
\node[style=none] (EaLab) at (\xEa,\yEaBase+0.55) {$e_a$};

\pgfmathsetmacro{\psiHalfW}{0.55}
\pgfmathsetmacro{\psiH}{0.85}

\node[style=none] (PsiTL)  at (\xPsi-\psiHalfW,\yPsiTop) {};
\node[style=none] (PsiTR)  at (\xPsi+\psiHalfW,\yPsiTop) {};
\node[style=none] (PsiB)   at (\xPsi,\yPsiTop-\psiH) {};
\node[style=none] (PsiLab) at (\xPsi,\yPsiTop-0.32) {{\color{red}$\psi$}};

\pgfmathsetmacro{\omHalfW}{1.70}   
\pgfmathsetmacro{\omH}{1.10}       

\node[style=none] (OmTL)  at (\xOm-\omHalfW,\yOmTop) {};
\node[style=none] (OmTR)  at (\xOm+\omHalfW,\yOmTop) {};
\node[style=none] (OmB)   at (\xOm,\yOmTop-\omH) {};
\node[style=none] (OmLab) at (\xOm,\yOmTop-0.48) {$\tilde{h}$};

\node[style=none] (SuL0) at (\Sx-\dw,\yPsiTop) {};
\node[style=none] (SuR0) at (\Sx+\dw,\yPsiTop) {};
\node[style=none] (SuL1) at (\Sx-\dw,\yEaBase) {};
\node[style=none] (SuR1) at (\Sx+\dw,\yEaBase) {};
\node[style=none] (SLab) at (\Sx-0.43,1.55) {$\scriptstyle \m{S}$};

\node[style=none] (AuL0) at (\Ax-\dw,\yOmTop) {};
\node[style=none] (AuR0) at (\Ax+\dw,\yOmTop) {};
\node[style=none] (AuL1) at (\Ax-\dw,\yEaBase) {};
\node[style=none] (AuR1) at (\Ax+\dw,\yEaBase) {};
\node[style=none] (ALab) at (\Ax-0.37,1.30) {$\scriptstyle \m{A}$};

\node[style=none] (BuL0) at (\Bx-\dw,\yOmTop) {};
\node[style=none] (BuR0) at (\Bx+\dw,\yOmTop) {};
\node[style=none] (BuL1) at (\Bx-\dw,\yBTop) {};
\node[style=none] (BuR1) at (\Bx+\dw,\yBTop) {};
\node[style=none] (BLab) at (\Bx+0.40,2.35) {$\scriptstyle \m{B}$};

\pgfmathsetmacro{\yAcut}{1.55}     

\pgfmathsetmacro{\boxL}{-0.62}
\pgfmathsetmacro{\boxB}{-0.50}
\pgfmathsetmacro{\boxT}{4.05}

\pgfmathsetmacro{\boxRlow}{0.90}

\pgfmathsetmacro{\boxRhigh}{2.05}

\node[style=none] (BoxTL)   at (\boxL,\boxT) {};
\node[style=none] (BoxTR)   at (\boxRhigh,\boxT) {};
\node[style=none] (BoxStep) at (\boxRhigh,\yAcut) {};
\node[style=none] (BoxStepL)at (\boxRlow,\yAcut) {};
\node[style=none] (BoxBR)   at (\boxRlow,\boxB) {};
\node[style=none] (BoxBL)   at (\boxL,\boxB) {};

\node[style=none] (EeffLab) at ({0.5*(\boxL+\boxRhigh)}, {\boxT+0.25}) {$e_{\mathrm{eff}}$};

\end{pgfonlayer}

\begin{pgfonlayer}{edgelayer}
\draw (EaBL.center) -- (EaBR.center) -- (EaT.center) -- cycle;
\draw[red] (PsiTL.center) -- (PsiTR.center) -- (PsiB.center) -- cycle;
\draw (OmTL.center) -- (OmTR.center) -- (OmB.center) -- cycle;

\draw[line width=0.6pt] (SuL0.center) -- (SuL1.center);
\draw[line width=0.6pt] (SuR0.center) -- (SuR1.center);

\draw[line width=0.6pt] (AuL0.center) -- (AuL1.center);
\draw[line width=0.6pt] (AuR0.center) -- (AuR1.center);

\draw[line width=0.6pt] (BuL0.center) -- (BuL1.center);
\draw[line width=0.6pt] (BuR0.center) -- (BuR1.center);

\draw[draw=gray, line width=0.9pt]
  (BoxTL.center) -- (BoxTR.center) -- (BoxStep.center) -- (BoxStepL.center)
  -- (BoxBR.center) -- (BoxBL.center) -- cycle;

\end{pgfonlayer}
\end{tikzpicture}%
}}= \begin{tikzpicture}
\begin{pgfonlayer}{nodelayer}

\node[style=none] (nuTL)   at (0.10,0.00) {};
\node[style=none] (nuTR)   at (1.99,0.00) {};
\node[style=none] (nuB)    at (1.05,-1.30) {};
\node[style=none] (nuLab)  at (1.05,-0.52) {$\tilde{h}$};

\node[style=none] (AuL) at (0.58,0.00) {};
\node[style=none] (AuR) at (0.68,0.00) {};
\node[style=none] (AtL) at (0.58,1.25) {};
\node[style=none] (AtR) at (0.68,1.25) {};
\node[style=none] (Alab) at (0.88,0.62) {$\scriptstyle \m{A}$};

\node[style=none] (BuL) at (1.42,0.00) {};
\node[style=none] (BuR) at (1.52,0.00) {};
\node[style=none] (BtL) at (1.42,1.25) {};
\node[style=none] (BtR) at (1.52,1.25) {};
\node[style=none] (Blab) at (1.72,0.62) {$\scriptstyle \m{B}$};

\node[style=none] (eBL)  at (0.01,1.25) {};
\node[style=none] (eBR)  at (1.250,1.25) {};
\node[style=none] (eT)   at (0.630,2.13) {};
\node[style=none] (eLab) at (0.630,1.44) {${e_{\rm eff}}$};

\node[style=none] (Eq) at (2.85,0.05) {$=$};


\node[style=none] (RBtL) at (4.10,1.10) {};
\node[style=none] (RBtR) at (4.20,1.10) {};
\node[style=none] (RBbL) at (4.10,0.00) {};
\node[style=none] (RBbR) at (4.20,0.00) {};
\node[style=none] (RBlabel) at (4.42,0.55) {$\scriptstyle \m{B}$};

\node[style=none] (zTL)  at (3.65,0.00) {};
\node[style=none] (zTR)  at (4.65,0.00) {};
\node[style=none] (zB)   at (4.15,-0.91) {};
\node[style=none] (zLab) at (4.15,-0.35) {$\vec{0}$};


\end{pgfonlayer}

\begin{pgfonlayer}{edgelayer}
\draw (nuTL.center) -- (nuTR.center) -- (nuB.center) -- cycle;
\draw (eBL.center)  -- (eBR.center)  -- (eT.center)  -- cycle;

\draw[line width=0.6pt] (AuL.center) -- (AtL.center);
\draw[line width=0.6pt] (AuR.center) -- (AtR.center);
\draw[line width=0.6pt] (BuL.center) -- (BtL.center);
\draw[line width=0.6pt] (BuR.center) -- (BtR.center);

\draw[line width=0.6pt] (RBtL.center) -- (RBbL.center);
\draw[line width=0.6pt] (RBtR.center) -- (RBbR.center);
\draw (zTL.center) -- (zTR.center) -- (zB.center) -- cycle;
\end{pgfonlayer}
\end{tikzpicture}
\end{equation}
    where $e_{\rm eff}$ is some effect on $A$ (defined by composing $\psi$ with the left subsystem of $e_a$). 
So the contribution to the state of $B$ coming from the component $\tilde{h}$ of the state in the holistic subspace is a constant vector (in fact, the zero vector, as we saw in the proof of Prop.~\ref{Prop.: TNLEAloneUselessForSteering}), \emph{independent of $\psi$}. Consequently, it is useless for achieving teleportation, since no correction operation on $B$ could map this constant vector to state $\psi$ for any unknown state $\psi$ (even if one conditions the correction on the value of $a$, as is allowed in the teleportation protocol).

The remaining term in the sum is simply 
\begin{equation}
\sum_{i} p_{i}\;
\vcenter{\hbox{%
\begin{tikzpicture}[baseline={(base.center)}]
\begin{pgfonlayer}{nodelayer}

\node[style=none] (base) at (0,0) {};

\pgfmathsetmacro{\dw}{0.10}  

\pgfmathsetmacro{\Sx}{0.00}
\pgfmathsetmacro{\Ax}{1.55}
\pgfmathsetmacro{\Vx}{3.25}

\pgfmathsetmacro{\yBotTop}{0.35}   
\pgfmathsetmacro{\yBotTip}{-0.85}  
\pgfmathsetmacro{\yEaBase}{1.55}   
\pgfmathsetmacro{\yEaTop}{2.85}    
\pgfmathsetmacro{\yBTop}{2.70}     

\pgfmathsetmacro{\EaInset}{0.55}

\pgfmathsetmacro{\EaBLx}{\Sx-\EaInset}
\pgfmathsetmacro{\EaBRx}{\Ax+\EaInset}
\pgfmathsetmacro{\EaCx}{0.5*(\EaBLx+\EaBRx)}

\node[style=none] (EaBL) at (\EaBLx,\yEaBase) {};
\node[style=none] (EaBR) at (\EaBRx,\yEaBase) {};
\node[style=none] (EaT)  at (\EaCx,\yEaTop)  {};
\node[style=none] (EaLab) at (\EaCx,2.20) {$e_a$};

\node[style=none] (EaS) at (\Sx,\yEaBase) {};
\node[style=none] (EaA) at (\Ax,\yEaBase) {};

\pgfmathsetmacro{\triHalfW}{0.65}

\node[style=none] (PsiTL) at (\Sx-\triHalfW,\yBotTop) {};
\node[style=none] (PsiTR) at (\Sx+\triHalfW,\yBotTop) {};
\node[style=none] (PsiB)  at (\Sx,\yBotTip) {};
\node[style=none] (PsiLab) at (\Sx,-0.15) {{\color{red}$\psi$}};

\node[style=none] (WiTL) at (\Ax-\triHalfW,\yBotTop) {};
\node[style=none] (WiTR) at (\Ax+\triHalfW,\yBotTop) {};
\node[style=none] (WiB)  at (\Ax,\yBotTip) {};
\node[style=none] (WiLab) at (\Ax,-0.15) {$\omega_i$};

\node[style=none] (ViTL) at (\Vx-\triHalfW,\yBotTop) {};
\node[style=none] (ViTR) at (\Vx+\triHalfW,\yBotTop) {};
\node[style=none] (ViB)  at (\Vx,\yBotTip) {};
\node[style=none] (ViLab) at (\Vx,-0.15) {$\nu_i$};

\node[style=none] (SuL0) at (\Sx-\dw,\yBotTop) {};
\node[style=none] (SuR0) at (\Sx+\dw,\yBotTop) {};
\node[style=none] (SuL1) at (\Sx-\dw,\yEaBase) {};
\node[style=none] (SuR1) at (\Sx+\dw,\yEaBase) {};
\node[style=none] (SLab) at (\Sx-0.48,1.00) {$\scriptstyle \m{S}$};

\node[style=none] (AuL0) at (\Ax-\dw,\yBotTop) {};
\node[style=none] (AuR0) at (\Ax+\dw,\yBotTop) {};
\node[style=none] (AuL1) at (\Ax-\dw,\yEaBase) {};
\node[style=none] (AuR1) at (\Ax+\dw,\yEaBase) {};
\node[style=none] (ALab) at (\Ax-0.40,1.00) {$\scriptstyle \m{A}$};

\node[style=none] (BuL0) at (\Vx-\dw,\yBotTop) {};
\node[style=none] (BuR0) at (\Vx+\dw,\yBotTop) {};
\node[style=none] (BuL1) at (\Vx-\dw,\yBTop) {};
\node[style=none] (BuR1) at (\Vx+\dw,\yBTop) {};
\node[style=none] (BLab) at (\Vx+0.32,1.45) {$\scriptstyle \m{B}$};

\end{pgfonlayer}

\begin{pgfonlayer}{edgelayer}

\draw (EaBL.center) -- (EaBR.center) -- (EaT.center) -- cycle;

\draw[red] (PsiTL.center) -- (PsiTR.center) -- (PsiB.center) -- cycle;

\draw (WiTL.center) -- (WiTR.center) -- (WiB.center) -- cycle;

\draw (ViTL.center) -- (ViTR.center) -- (ViB.center) -- cycle;

\draw[line width=0.6pt] (SuL0.center) -- (SuL1.center);
\draw[line width=0.6pt] (SuR0.center) -- (SuR1.center);

\draw[line width=0.6pt] (AuL0.center) -- (AuL1.center);
\draw[line width=0.6pt] (AuR0.center) -- (AuR1.center);

\draw[line width=0.6pt] (BuL0.center) -- (BuL1.center);
\draw[line width=0.6pt] (BuR0.center) -- (BuR1.center);

\end{pgfonlayer}
\end{tikzpicture}%
}}
\end{equation}

which has the form of a so-called   LOSR-free teleportage~\cite{Schmid2020typeindependent}, and so is useless for teleportation.
\end{proof}

A stronger version of each of these results also holds: even for a state with TL entanglement, the amount of nonclassical correlations that one can generate in each of these three scenarios cannot be increased by increasing the component of TNL entanglement in the state; indeed, it does not depend on that component at all. This is clear immediately from the proofs above. 

The common feature that drives all of the above results is the fact that in each scenario, the two halves of an entangled state are accessed locally rather than jointly. Since local effects annihilate the holistic component of a state, they are not sensitive to tomographically-nonlocal entanglement. In other information-processing protocols, however, inaccessibility to local effects might be advantageous or maybe joint measurements may be incorporated, and in such a case there is potential for tomographically-nonlocal entanglement to be relevant. We will show that this is the case for dense coding and data hiding in the next section. 

\section{Tasks for which TNL entanglement is useful}

\subsection{Using TNL entanglement for dense coding}

Dense coding is a paradigmatic information-processing task in complex quantum theory (CQT)~\cite{BennettWiesner1992}.
In its simplest bipartite form, two parties share a system prepared in an entangled state; by applying a local
operation on one subsystem and then transmitting this subsystem, one party can convey more classical
information than would otherwise be possible with an unentangled system of the same local
dimension -- in the perfect case, the classical communication capacity is doubled. It is known that (in standard quantum theory) entanglement is necessary for dense coding~\cite{HorodeckiDenseCoding2001}, and in fact the perfect case requires maximally entangled Bell states~\cite{Nayak_2023}. Since standard quantum systems carry no holistic degrees of freedom, tomographically local entanglement is the sole form of entanglement responsible for such a quantum advantage.\footnote{Note that the traditional example of quantum dense coding with qubits use states, effects and operations that are also valid in real quantum theory -- but in this case both kinds of entanglement are present.}

Could there be theories in which dense coding is driven by tomographically \emph{non}local entanglement instead? 
The above question is answered in the affirmative by bilocal classical theory (BCT)~\cite{d2020classicality}.
In BCT, there is no TL entanglement at all: all entanglement present in the theory is of
the TNL type. Nevertheless, dense coding is possible, as we exemplify below and was first proven in Ref.~\cite{d2020classicality}.
This demonstrates that TL entanglement is not necessary for dense coding in some theories.

\begin{example}
    \label{BCTDenseCoding}
Consider a composite system $AB$ in \emph{bilocal classical theory} (BCT), where the local state spaces
$\m{S}_A$ and $\m{S}_B$ are those of a classical bit, with pure states $\lvert{0}\rparen$ and $\lvert{1}\rparen$ (as well as the local effect spaces are identical to those of classical bits). In standard
classical theory, the composite would be a two-bit system with four extremal states. In BCT, by contrast,
the composite state space $\m{S}_{AB}$ possesses an additional \emph{holistic degree of freedom}
$s\in\{+,-\}$, and therefore has eight extremal pure states
\begin{align}
\Big\{ \lvert{(ij)^s}\rparen_{AB} : i,j\in\{0,1\},\ s\in\{+,-\} \Big\}.
\end{align}

Product states correspond to complete ignorance about the holistic degree of freedom:
\begin{align}
\lvert{i}_A\rparen\lvert {j}_B\rparen
\;=\;
\frac{1}{2}\Big( \lvert{(ij)^+}\rparen_{AB} + \lvert{(ij)^-}\rparen_{AB} \Big),
\end{align}
so every extremal bipartite state $\lvert{(ij)^s}\rparen_{AB}$ has definite $s$ and is therefore entangled. A similar rule follows for products of effects. Thus, the sign $s$ is invisible to all local measurements but can be accessed by global ones.

Crucially, local reversible transformations can act independently on the local bit and the holistic bit, i.e., valid reversible transformations on Alice's side can map $\lvert{(ab)^s}\rparen$ to any state of the form $\lvert{(a'b)^{s'}}\rparen$ with $a'$ taking values in $\{0,1\}$ and $s'$ taking values in $\{+,-\}$.

Now, suppose Alice and Bob share the entangled state
$\lvert{(0b)^-}\rparen_{AB}$, with $b\in\{0,1\}$ (the particular value of Bob's local bit $b$ plays no role in the protocol). To send a two-bit message $(x,y)\in\{0,1\}^2$, Alice locally
encodes $x$ in her bit value first (i.e., she implements $\lvert{(0b)^s}\rparen\mapsto\lvert{(xb)^s}\rparen$) and $y$ in the holistic sign $s$ (which is possible in BCT since local operations can independently tweak the value of $s$), preparing one of the four states
\begin{align}
\lvert{(0b)^-}\rparen_{AB},\quad
\lvert{(0b)^+}\rparen_{AB},\quad
\lvert{(1b)^-}\rparen_{AB},\quad
\lvert{(1b)^+}\rparen_{AB}.
\end{align}
She then sends her single local system to Bob. Once Bob holds both subsystems, he performs a global
measurement that perfectly distinguishes these four states and recovers both bits.

Thus, in BCT, two classical bits can be transmitted by sending a single classical-bit system, with the
dense-coding advantage arising from a controllable holistic degree of freedom rather than from
nonclassical local systems.
\end{example}

The feature enabling dense coding in BCT is twofold.
First, the theory possesses holistic degrees of freedom, that is, degrees of freedom
lying outside the span of product states, which can carry information inaccessible to
purely local measurements. Second, and crucially, the theory must admit local
transformations that can act nontrivially and independently on these holistic degrees of freedom.
Indeed, in BCT, local operations can modify the holistic component
of a shared state, allowing one party to encode information into it prior to transmission.  An even more striking example is bit-flip twirled c-bit world~\cite{centeno2024twirledworldssymmetryinducedfailures}, in which the local systems are trivial, yet a shared holistic degree of freedom, controllable by local operations, allows Alice to transmit one classical bit by sending a trivial system.

Dense coding is \emph{not} possible in all tomographically-nonlocal theories. Counterexamples are provided by some minimal operational probabilistic theories~\cite{Rolino2024,Rolino_MOPTs_2025}, which are theories with very restricted set of operations. In those cases, even though holistic degrees of freedom might exist, the operations needed by Alice to encode information beyond her local classical capacity are not allowed, thus making dense coding impossible. The idea is the following. Consider the minimal version of BCT, called MBCT~\cite{Rolino_MOPTs_2025}: in such a minimal theory, the reversible transformations are restricted to compositions of identities and swaps, while the remaining channels are just discard-and-prepare ones. Notice that, importantly, the number of available reversible transformations is smaller than the number of discard-and-prepare channels whose associated states are perfectly discriminable: this is in stark contrast with Example~\ref{BCTDenseCoding} from BCT. Thus, via a counting argument, it can be checked that dense coding is impossible in MBCT by using the following property. In MBCT, it is impossible to perfectly discriminate any two bipartite states such that: (i) they have the same reduced state on both subsystems; and (ii) one of them is a product state. This means that the reversible transformations in MBCT are not more useful than measure-and-prepare channels, and accordingly the channels that are available in MBCT do not allow Alice to exceed her local classical capacity.

In summary, dense coding does not require tomographically-nonlocal entanglement in general, as shown by complex quantum theory, but it can be powered solely by TNL entanglement in suitably structured theories such as BCT. What ultimately matters for TNL to power dense coding is not just the existence of holistic degrees of freedom but the ability to locally manipulate them.

\subsection{Using TNL entanglement for data hiding}
\label{sec:Dense coding}

In this Section, we show that tomographically-nonlocal entanglement can also be used to accomplish a task which is impossible in complex quantum theory~\cite{DiVincenzo_2003_DataHidingPerfectCase}; namely, perfectly secure data hiding. We consider here the simpler case of hiding just a bit in a bipartite scenario, and we consider the standard case of hiding data against LOCC measurements. Here we show that tomographically nonlocal entanglement is required for success in such a task; in fact,  in some theories, states with tomographically non-local entanglement alone can be the resource for success.

In data hiding protocols~\cite{Terhal_DataHidingFirst}, a bit $x$ is encoded using  two bipartite  states $\{s^{\m{AB}}_x\}_x$, which should be perfectly distinguishable. The subsystems are then distributed to spatially separated parties, Alice (who holds subsystem ${\m{A}}$) and Bob (who holds subsystem $\m{B}$). Alice and Bob can then perform any local operation and communicate classically in order to try to learn the value of $x$.  In the perfectly secure case, Alice and Bob should never be able to recover any information about $x$~\cite{DiVincenzo_2003_DataHidingPerfectCase}. In other words, for such a scheme to be possible in a given theory,  there should be two perfectly distinguishable states $\{s^{\m{AB}}_x\}_x$ that produce the same local statistics for any given product effect.

The above explanation already suggests that achieving success in such a task (i.e., successfully hide)  is impossible in standard quantum theory~\cite{DiVincenzo_2002_DataHidingNoquantumPerfectCase,Matthews_2009_DistinguishabilityQStatesUnderRestrictedMeasurements}, or, in fact, in any tomographically local theory. Indeed,  by implementing LOCC operations in TL theories, one is always able to reconstruct the state of a system using many copies, which implies that any distinct pair of states should lead to distinct statistics for some product effect.\footnote{Note, however, that data hiding is possible in the non-perfect case within bare quantum theory, and can be made arbitrarily close to the perfect case by increasing the local dimension of systems asymptotically~\cite{DiVincenzo_2002_DataHidingNoquantumPerfectCase}. Reference ~\cite{Lami_2018_DataHidingGPTs} also analyzed how TL GPTs can perform data hiding in the non-perfect case, and against different strategies for Alice and Bob, besides LOCC.} That is, only tomographically nonlocal theories can successfully implement such a task. Moreover, as we show next, any encoding scheme that allows for perfectly secure data hiding requires tomographically non-local entanglement.

\begin{proposition}
\label{prop.:PerfectDataHidingRequiresTnLE}
    Consider a composite GPT system $\m{AB}$. Any encoding scheme $\{\omega_x^{\m{AB}}\}$ that allows for perfectly secure data hiding requires tomographically non-local entanglement, i.e., either $\omega_0^{\m{AB}}$ or $\omega_{1}^{AB}$ (or both) have holistic degrees of freedom.
\end{proposition}
\begin{proof}
    An encoding scheme defined by the pair of states $\{\omega^{\m{AB}}_0,\omega^{\m{AB}}_1\}$ achieves perfectly secure data hiding iff the following two conditions hold:
    \begin{enumerate}
        \item[i.] Product effects cannot separate $\omega^{\m{AB}}_0$ and $\omega^{\m{AB}}_1$, ie: 
        \begin{equation}
            \vcenter{\hbox{%
            \begin{tikzpicture}[baseline={(base.center)}]
  \begin{pgfonlayer}{nodelayer}
    \node[style=none] (base) at (0,0) {};

    \def\yTop{0.90}
    \def\yMid{0.10}
    \def\yBot{-0.85}

    \node[style=none] (eTL) at (-1.00,\yTop) {};
    \node[style=none] (eTR) at (-0.20,\yTop) {};
    \node[style=none] (eTT) at (-0.60,1.65) {};
    \node[style=none] (eLab) at (-0.60,1.25) {$e$};

    \node[style=none] (fTL) at ( 0.20,\yTop) {};
    \node[style=none] (fTR) at ( 1.00,\yTop) {};
    \node[style=none] (fTT) at ( 0.60,1.65) {};
    \node[style=none] (fLab) at ( 0.60,1.25) {$f$};

    \node[style=none] (AwL_top) at (-0.68,\yTop) {};
    \node[style=none] (AwR_top) at (-0.52,\yTop) {};
    \node[style=none] (AwL_mid) at (-0.68,\yMid) {};
    \node[style=none] (AwR_mid) at (-0.52,\yMid) {};
    \node[style=none] (AwLab) at (-0.92,0.55) {$\scriptstyle \m{A}$};

    \node[style=none] (BwL_top) at ( 0.52,\yTop) {};
    \node[style=none] (BwR_top) at ( 0.68,\yTop) {};
    \node[style=none] (BwL_mid) at ( 0.52,\yMid) {};
    \node[style=none] (BwR_mid) at ( 0.68,\yMid) {};
    \node[style=none] (BwLab) at ( 0.92,0.55) {$\scriptstyle \m{B}$};

    \node[style=none] (sTL) at (-1.10,\yMid) {};
    \node[style=none] (sTR) at ( 1.10,\yMid) {};
    \node[style=none] (sTC) at ( 0.00,\yBot) {};
    \node[style=none] (sNu) at ( 0.00,-0.30) {$\omega_0$};

  \end{pgfonlayer}

  \begin{pgfonlayer}{edgelayer}
    \draw (eTL.center) -- (eTR.center);
    \draw (eTR.center) -- (eTT.center);
    \draw (eTT.center) -- (eTL.center);

    \draw (fTL.center) -- (fTR.center);
    \draw (fTR.center) -- (fTT.center);
    \draw (fTT.center) -- (fTL.center);

    \draw[line width=0.6pt] (AwL_top.center) -- (AwL_mid.center);
    \draw[line width=0.6pt] (AwR_top.center) -- (AwR_mid.center);

    \draw[line width=0.6pt] (BwL_top.center) -- (BwL_mid.center);
    \draw[line width=0.6pt] (BwR_top.center) -- (BwR_mid.center);

    \draw (sTL.center) -- (sTR.center);
    \draw (sTR.center) -- (sTC.center);
    \draw (sTC.center) -- (sTL.center);

  \end{pgfonlayer}
            \end{tikzpicture}%
            }}= \vcenter{\hbox{%
\begin{tikzpicture}[baseline={(base.center)}]
  \begin{pgfonlayer}{nodelayer}
    \node[style=none] (base) at (0,0) {};

    \def\yTop{0.90}
    \def\yMid{0.10}
    \def\yBot{-0.85}

    \node[style=none] (eTL) at (-1.00,\yTop) {};
    \node[style=none] (eTR) at (-0.20,\yTop) {};
    \node[style=none] (eTT) at (-0.60,1.65) {};
    \node[style=none] (eLab) at (-0.60,1.25) {$e$};

    \node[style=none] (fTL) at ( 0.20,\yTop) {};
    \node[style=none] (fTR) at ( 1.00,\yTop) {};
    \node[style=none] (fTT) at ( 0.60,1.65) {};
    \node[style=none] (fLab) at ( 0.60,1.25) {$f$};

    \node[style=none] (AwL_top) at (-0.68,\yTop) {};
    \node[style=none] (AwR_top) at (-0.52,\yTop) {};
    \node[style=none] (AwL_mid) at (-0.68,\yMid) {};
    \node[style=none] (AwR_mid) at (-0.52,\yMid) {};
    \node[style=none] (AwLab) at (-0.92,0.55) {$\scriptstyle \m{A}$};

    \node[style=none] (BwL_top) at ( 0.52,\yTop) {};
    \node[style=none] (BwR_top) at ( 0.68,\yTop) {};
    \node[style=none] (BwL_mid) at ( 0.52,\yMid) {};
    \node[style=none] (BwR_mid) at ( 0.68,\yMid) {};
    \node[style=none] (BwLab) at ( 0.92,0.55) {$\scriptstyle \m{B}$};

    \node[style=none] (sTL) at (-1.10,\yMid) {};
    \node[style=none] (sTR) at ( 1.10,\yMid) {};
    \node[style=none] (sTC) at ( 0.00,\yBot) {};
    \node[style=none] (sNu) at ( 0.00,-0.30) {$\omega_1$};

  \end{pgfonlayer}

  \begin{pgfonlayer}{edgelayer}
    \draw (eTL.center) -- (eTR.center);
    \draw (eTR.center) -- (eTT.center);
    \draw (eTT.center) -- (eTL.center);

    \draw (fTL.center) -- (fTR.center);
    \draw (fTR.center) -- (fTT.center);
    \draw (fTT.center) -- (fTL.center);

    \draw[line width=0.6pt] (AwL_top.center) -- (AwL_mid.center);
    \draw[line width=0.6pt] (AwR_top.center) -- (AwR_mid.center);

    \draw[line width=0.6pt] (BwL_top.center) -- (BwL_mid.center);
    \draw[line width=0.6pt] (BwR_top.center) -- (BwR_mid.center);

    \draw (sTL.center) -- (sTR.center);
    \draw (sTR.center) -- (sTC.center);
    \draw (sTC.center) -- (sTL.center);

  \end{pgfonlayer}
\end{tikzpicture}%
}}\;\;\forall\;
            \vcenter{\hbox{%
            \begin{tikzpicture}[baseline={(base.center)}]
            \begin{pgfonlayer}{nodelayer}
            \node[style=none] (base) at (0,0) {};

            \def\yTop{0.00}
            \def\yBot{-0.80}

            \node[style=none] (eTL) at (-1.00,\yTop) {};
            \node[style=none] (eTR) at (-0.20,\yTop) {};
            \node[style=none] (eTT) at (-0.60,0.75) {};
            \node[style=none] (eLab) at (-0.60,0.35) {$e$};

            \node[style=none] (AwL_top) at (-0.68,\yTop) {};
            \node[style=none] (AwR_top) at (-0.52,\yTop) {};
            \node[style=none] (AwL_bot) at (-0.68,\yBot) {};
            \node[style=none] (AwR_bot) at (-0.52,\yBot) {};
            \node[style=none] (AwLab) at (-0.92,-0.40) {$\scriptstyle \m{A}$};

        \end{pgfonlayer}

        \begin{pgfonlayer}{edgelayer}
            \draw (eTL.center) -- (eTR.center);
            \draw (eTR.center) -- (eTT.center);
            \draw (eTT.center) -- (eTL.center);

            \draw[line width=0.6pt] (AwL_top.center) -- (AwL_bot.center);
            \draw[line width=0.6pt] (AwR_top.center) -- (AwR_bot.center);

        \end{pgfonlayer}
        \end{tikzpicture}%
        }}
        \in E_{\m{A}}\,,\,\vcenter{\hbox{%
        \begin{tikzpicture}[baseline={(base.center)}]
        \begin{pgfonlayer}{nodelayer}
            \node[style=none] (base) at (0,0) {};

            \def\yTop{0.00}
            \def\yBot{-0.80}

            \node[style=none] (fTL) at (-1.00,\yTop) {};
            \node[style=none] (fTR) at (-0.20,\yTop) {};
            \node[style=none] (fTT) at (-0.60,0.75) {};
            \node[style=none] (fLab) at (-0.60,0.35) {$f$};

            \node[style=none] (BwL_top) at (-0.68,\yTop) {};
            \node[style=none] (BwR_top) at (-0.52,\yTop) {};
            \node[style=none] (BwL_bot) at (-0.68,\yBot) {};
            \node[style=none] (BwR_bot) at (-0.52,\yBot) {};
            \node[style=none] (bwLab) at (-0.92,-0.40) {$\scriptstyle \m{B}$};

    \end{pgfonlayer}

    \begin{pgfonlayer}{edgelayer}
    \draw (fTL.center) -- (fTR.center);
    \draw (fTR.center) -- (fTT.center);
    \draw (fTT.center) -- (fTL.center);

    \draw[line width=0.6pt] (AwL_top.center) -- (BwL_bot.center);
    \draw[line width=0.6pt] (AwR_top.center) -- (BwR_bot.center);

  \end{pgfonlayer}
    \end{tikzpicture}%
    }}
    \in E_{\m{B}};
\end{equation} 
\item[ii.] States $\omega^{\m{AB}}_0$ and $\omega^{\m{AB}}_1$ are perfectly distinguishable. That is, there exists a (global) measurement $\{e^{\m{AB}}_0,e^{\m{AB}}_1\}\subset E_{\m{AB}}$ (thus satisfying $e^{\m{AB}}_0+e^{\m{AB}}_1=u^{\m{AB}}$), such that
\begin{equation}
\vcenter{\hbox{%
\begin{tikzpicture}[baseline={(base.center)}]
  \begin{pgfonlayer}{nodelayer}
    \node[style=none] (base) at (0,0) {};

    \node[style=none] (eTL)  at (-0.95, 1.20) {};
    \node[style=none] (eTR)  at ( 0.95, 1.20) {};
    \node[style=none] (eTT)  at ( 0.00, 2.20) {};
    \node[style=none] (eLab) at ( 0.00, 1.65) {$e_i$};

    \node[style=none] (wTL)  at (-0.95, 0.00) {};
    \node[style=none] (wTR)  at ( 0.95, 0.00) {};
    \node[style=none] (wTC)  at ( 0.00,-1.10) {};
    \node[style=none] (wLab) at ( 0.00,-0.45) {$\omega_x$};

    \node[style=none] (AL_top) at (-0.60, 1.20) {};
    \node[style=none] (AR_top) at (-0.44, 1.20) {};
    \node[style=none] (AL_bot) at (-0.60, 0.00) {};
    \node[style=none] (AR_bot) at (-0.44, 0.00) {};
    \node[style=none] (Alab)   at (-0.88, 0.60) {$\scriptstyle \m{A}$};

    \node[style=none] (BL_top) at ( 0.44, 1.20) {};
    \node[style=none] (BR_top) at ( 0.60, 1.20) {};
    \node[style=none] (BL_bot) at ( 0.44, 0.00) {};
    \node[style=none] (BR_bot) at ( 0.60, 0.00) {};
    \node[style=none] (Blab)   at ( 0.88, 0.60) {$\scriptstyle \m{B}$};

  \end{pgfonlayer}

  \begin{pgfonlayer}{edgelayer}
    \draw (eTL.center) -- (eTR.center);
    \draw (eTR.center) -- (eTT.center);
    \draw (eTT.center) -- (eTL.center);

    \draw (wTL.center) -- (wTR.center);
    \draw (wTR.center) -- (wTC.center);
    \draw (wTC.center) -- (wTL.center);

    \draw[line width=0.6pt] (AL_top.center) -- (AL_bot.center);
    \draw[line width=0.6pt] (AR_top.center) -- (AR_bot.center);

    \draw[line width=0.6pt] (BL_top.center) -- (BL_bot.center);
    \draw[line width=0.6pt] (BR_top.center) -- (BR_bot.center);

  \end{pgfonlayer}
\end{tikzpicture}%
}}=\delta_{ix}.%
\end{equation}

    \end{enumerate}
Now, the first conditions implies that $\omega_0^{\m{AB}}$ and $\omega_1^{\m{AB}}$ cannot differ in the tomographically local subspace, ie
\begin{equation}
    \vcenter{\hbox{
\begin{tikzpicture}
	\begin{pgfonlayer}{nodelayer}

		\node [style=none] (TA) at (-0.6,1.4) {};
		\node [style=none] (TB) at (0.6,1.4) {};

		\node [style=none] (TA2) at (-0.6,0.7) {};
		\node [style=none] (TB2) at (0.6,0.7) {};

		\node [style=none] (Alabel1) at (-0.9,1.1) {$\scriptstyle \mathcal A$};
		\node [style=none] (Blabel1) at (0.9,1.1) {$\scriptstyle \mathcal B$};

		\node [style=none] (R1) at (-0.9,0.7) {};
		\node [style=none] (R2) at (0.9,0.7) {};
		\node [style=none] (R3) at (0.9,-0.1) {};
		\node [style=none] (R4) at (-0.9,-0.1) {};

		\node [style=none] (RectLabel) at (0,0.3) {$\Pi_{TL}$};

		\node [style=none] (BA) at (-0.6,-0.1) {};
		\node [style=none] (BB) at (0.6,-0.1) {};

		\node [style=none] (BA2) at (-0.6,-0.8) {};
		\node [style=none] (BB2) at (0.6,-0.8) {};

		\node [style=none] (Alabel2) at (-0.9,-0.45) {$\scriptstyle \mathcal A$};
		\node [style=none] (Blabel2) at (0.9,-0.45) {$\scriptstyle \mathcal B$};

		\node [style=none] (TL) at (-0.9,-0.8) {};
		\node [style=none] (TR) at (0.9,-0.8) {};
		\node [style=none] (TC) at (0,-1.8) {};
		\node [style=none] (W)  at (0,-1.25) {$\omega_0$};

		\node [style=none] (NEQ) at (2.2,-0.5) {};

	\end{pgfonlayer}

	\begin{pgfonlayer}{edgelayer}

		\draw[double] (TA.center) to (TA2.center);
		\draw[double] (TB.center) to (TB2.center);

		\draw (R1.center) to (R2.center);
		\draw (R2.center) to (R3.center);
		\draw (R3.center) to (R4.center);
		\draw (R4.center) to (R1.center);

		\draw[double] (BA.center) to (BA2.center);
		\draw[double] (BB.center) to (BB2.center);

		\draw (TL.center) to (TR.center);
		\draw (TR.center) to (TC.center);
		\draw (TC.center) to (TL.center);

	\end{pgfonlayer}
\end{tikzpicture}}} = \vcenter{\hbox{
\begin{tikzpicture}
	\begin{pgfonlayer}{nodelayer}

		\node [style=none] (TA) at (-0.6,1.4) {};
		\node [style=none] (TB) at (0.6,1.4) {};

		\node [style=none] (TA2) at (-0.6,0.7) {};
		\node [style=none] (TB2) at (0.6,0.7) {};

		\node [style=none] (Alabel1) at (-0.9,1.1) {$\scriptstyle \mathcal A$};
		\node [style=none] (Blabel1) at (0.9,1.1) {$\scriptstyle \mathcal B$};

		\node [style=none] (R1) at (-0.9,0.7) {};
		\node [style=none] (R2) at (0.9,0.7) {};
		\node [style=none] (R3) at (0.9,-0.1) {};
		\node [style=none] (R4) at (-0.9,-0.1) {};

		\node [style=none] (RectLabel) at (0,0.3) {$\Pi_{TL}$};

		\node [style=none] (BA) at (-0.6,-0.1) {};
		\node [style=none] (BB) at (0.6,-0.1) {};

		\node [style=none] (BA2) at (-0.6,-0.8) {};
		\node [style=none] (BB2) at (0.6,-0.8) {};

		\node [style=none] (Alabel2) at (-0.9,-0.45) {$\scriptstyle \mathcal A$};
		\node [style=none] (Blabel2) at (0.9,-0.45) {$\scriptstyle \mathcal B$};

		\node [style=none] (TL) at (-0.9,-0.8) {};
		\node [style=none] (TR) at (0.9,-0.8) {};
		\node [style=none] (TC) at (0,-1.8) {};
		\node [style=none] (W)  at (0,-1.25) {$\omega_1$};

		\node [style=none] (NEQ) at (2.2,-0.5) {};

	\end{pgfonlayer}

	\begin{pgfonlayer}{edgelayer}

		\draw[double] (TA.center) to (TA2.center);
		\draw[double] (TB.center) to (TB2.center);

		\draw (R1.center) to (R2.center);
		\draw (R2.center) to (R3.center);
		\draw (R3.center) to (R4.center);
		\draw (R4.center) to (R1.center);

		\draw[double] (BA.center) to (BA2.center);
		\draw[double] (BB.center) to (BB2.center);

		\draw (TL.center) to (TR.center);
		\draw (TR.center) to (TC.center);
		\draw (TC.center) to (TL.center);

	\end{pgfonlayer}
\end{tikzpicture}}}\implies 
\vcenter{\hbox{%
\begin{tikzpicture}[baseline={(base.center)}]
  \begin{pgfonlayer}{nodelayer}
    \node[style=none] (base) at (0,0) {};

    \node[style=none] (wTL)  at (-0.90,0.00) {};
    \node[style=none] (wTR)  at ( 0.90,0.00) {};
    \node[style=none] (wTC)  at ( 0.00,-1.18) {};
    \node[style=none] (wLab) at ( 0.00,-0.45) {$\omega_x$};

    \node[style=none] (AL_bot) at (-0.58,0.00) {};
    \node[style=none] (AR_bot) at (-0.42,0.00) {};
    \node[style=none] (AL_top) at (-0.58,1.10) {};
    \node[style=none] (AR_top) at (-0.42,1.10) {};
    \node[style=none] (Alab)   at (-0.88,0.70) {$\scriptstyle \m{A}$};

    \node[style=none] (BL_bot) at ( 0.42,0.00) {};
    \node[style=none] (BR_bot) at ( 0.58,0.00) {};
    \node[style=none] (BL_top) at ( 0.42,1.10) {};
    \node[style=none] (BR_top) at ( 0.58,1.10) {};
    \node[style=none] (Blab)   at ( 0.88,0.70) {$\scriptstyle \m{B}$};

  \end{pgfonlayer}

  \begin{pgfonlayer}{edgelayer}
    \draw (wTL.center) -- (wTR.center);
    \draw (wTR.center) -- (wTC.center);
    \draw (wTC.center) -- (wTL.center);

    \draw[line width=0.6pt] (AL_bot.center) -- (AL_top.center);
    \draw[line width=0.6pt] (AR_bot.center) -- (AR_top.center);

    \draw[line width=0.6pt] (BL_bot.center) -- (BL_top.center);
    \draw[line width=0.6pt] (BR_bot.center) -- (BR_top.center);
  \end{pgfonlayer}
\end{tikzpicture}%
}} = \vcenter{\hbox{%
\begin{tikzpicture}[baseline={(base.center)}]
  \begin{pgfonlayer}{nodelayer}
    \node[style=none] (base) at (0,0) {};

    \node[style=none] (wTL)  at (-0.90,0.00) {};
    \node[style=none] (wTR)  at ( 0.90,0.00) {};
    \node[style=none] (wTC)  at ( 0.00,-1.18) {};
    \node[style=none] (wLab) at ( -0.05,-0.35) {$v_{\rm TL}$};

    \node[style=none] (AL_bot) at (-0.58,0.00) {};
    \node[style=none] (AR_bot) at (-0.42,0.00) {};
    \node[style=none] (AL_top) at (-0.58,1.10) {};
    \node[style=none] (AR_top) at (-0.42,1.10) {};
    \node[style=none] (Alab)   at (-0.88,0.70) {$\scriptstyle \m{A}$};

    \node[style=none] (BL_bot) at ( 0.42,0.00) {};
    \node[style=none] (BR_bot) at ( 0.58,0.00) {};
    \node[style=none] (BL_top) at ( 0.42,1.10) {};
    \node[style=none] (BR_top) at ( 0.58,1.10) {};
    \node[style=none] (Blab)   at ( 0.88,0.70) {$\scriptstyle \m{B}$};

  \end{pgfonlayer}

  \begin{pgfonlayer}{edgelayer}
    \draw (wTL.center) -- (wTR.center);
    \draw (wTR.center) -- (wTC.center);
    \draw (wTC.center) -- (wTL.center);

    \draw[line width=0.6pt] (AL_bot.center) -- (AL_top.center);
    \draw[line width=0.6pt] (AR_bot.center) -- (AR_top.center);

    \draw[line width=0.6pt] (BL_bot.center) -- (BL_top.center);
    \draw[line width=0.6pt] (BR_bot.center) -- (BR_top.center);
  \end{pgfonlayer}
\end{tikzpicture}%
}} + \vcenter{\hbox{%
\begin{tikzpicture}[baseline={(base.center)}]
  \begin{pgfonlayer}{nodelayer}
    \node[style=none] (base) at (0,0) {};

    \node[style=none] (wTL)  at (-0.90,0.00) {};
    \node[style=none] (wTR)  at ( 0.90,0.00) {};
    \node[style=none] (wTC)  at ( 0.00,-1.18) {};
    \node[style=none] (wLab) at ( 0.00,-0.45) {$\tilde{h}_x$};

    \node[style=none] (AL_bot) at (-0.58,0.00) {};
    \node[style=none] (AR_bot) at (-0.42,0.00) {};
    \node[style=none] (AL_top) at (-0.58,1.10) {};
    \node[style=none] (AR_top) at (-0.42,1.10) {};
    \node[style=none] (Alab)   at (-0.88,0.70) {$\scriptstyle \m{A}$};

    \node[style=none] (BL_bot) at ( 0.42,0.00) {};
    \node[style=none] (BR_bot) at ( 0.58,0.00) {};
    \node[style=none] (BL_top) at ( 0.42,1.10) {};
    \node[style=none] (BR_top) at ( 0.58,1.10) {};
    \node[style=none] (Blab)   at ( 0.88,0.70) {$\scriptstyle \m{B}$};

  \end{pgfonlayer}

  \begin{pgfonlayer}{edgelayer}
    \draw (wTL.center) -- (wTR.center);
    \draw (wTR.center) -- (wTC.center);
    \draw (wTC.center) -- (wTL.center);

    \draw[line width=0.6pt] (AL_bot.center) -- (AL_top.center);
    \draw[line width=0.6pt] (AR_bot.center) -- (AR_top.center);

    \draw[line width=0.6pt] (BL_bot.center) -- (BL_top.center);
    \draw[line width=0.6pt] (BR_bot.center) -- (BR_top.center);
  \end{pgfonlayer}
\end{tikzpicture}%
}}%
\end{equation}
Where $v_{\rm TL}\in AB_{\otimes}$ is independent of $x$, and $\tilde{h}_x\in H_S$. Now, the second condition implies that $\omega^{\m{AB}}_0\neq \omega^{\m{AB}}_1$, which implies  $\tilde{h}_0\neq \tilde{h}_1$ and therefore, for at least some value of $x$,  $h_x\neq{0}_{AB}$.
\end{proof}

So far, we have shown that tomographically non-local entanglement is necessary for perfectly secure data hiding schemes. A natural question is whether tomographically non-local entanglement alone can be sufficient for perfectly secure data hiding in some theories. The following result, which is naturally suggested by Proposition~\ref{prop.:PerfectDataHidingRequiresTnLE}, defines a class of GPTs in which this occurs.
\begin{proposition}
\label{prop.:ConditionsPerfectDataHiding}
    Consider a GPT containing a composite GPT system $\m{AB}$. If there exists a perfectly distinguishable pair of states $\omega^{\m{AB}}_{x}$ satisfying
    \begin{enumerate}
        \item[i.] Product effects cannot separate $\omega^{\m{AB}}_0$ and $\omega^{\m{AB}}_1$, ie: 
        \begin{equation}
            \vcenter{\hbox{%
            \begin{tikzpicture}[baseline={(base.center)}]
  \begin{pgfonlayer}{nodelayer}
    \node[style=none] (base) at (0,0) {};

    \def\yTop{0.90}
    \def\yMid{0.10}
    \def\yBot{-0.85}

    \node[style=none] (eTL) at (-1.00,\yTop) {};
    \node[style=none] (eTR) at (-0.20,\yTop) {};
    \node[style=none] (eTT) at (-0.60,1.65) {};
    \node[style=none] (eLab) at (-0.60,1.25) {$e$};

    \node[style=none] (fTL) at ( 0.20,\yTop) {};
    \node[style=none] (fTR) at ( 1.00,\yTop) {};
    \node[style=none] (fTT) at ( 0.60,1.65) {};
    \node[style=none] (fLab) at ( 0.60,1.25) {$f$};

    \node[style=none] (AwL_top) at (-0.68,\yTop) {};
    \node[style=none] (AwR_top) at (-0.52,\yTop) {};
    \node[style=none] (AwL_mid) at (-0.68,\yMid) {};
    \node[style=none] (AwR_mid) at (-0.52,\yMid) {};
    \node[style=none] (AwLab) at (-0.92,0.55) {$\scriptstyle \m{A}$};

    \node[style=none] (BwL_top) at ( 0.52,\yTop) {};
    \node[style=none] (BwR_top) at ( 0.68,\yTop) {};
    \node[style=none] (BwL_mid) at ( 0.52,\yMid) {};
    \node[style=none] (BwR_mid) at ( 0.68,\yMid) {};
    \node[style=none] (BwLab) at ( 0.92,0.55) {$\scriptstyle \m{B}$};

    \node[style=none] (sTL) at (-1.10,\yMid) {};
    \node[style=none] (sTR) at ( 1.10,\yMid) {};
    \node[style=none] (sTC) at ( 0.00,\yBot) {};
    \node[style=none] (sNu) at ( 0.00,-0.30) {$\omega_0$};

  \end{pgfonlayer}

  \begin{pgfonlayer}{edgelayer}
    \draw (eTL.center) -- (eTR.center);
    \draw (eTR.center) -- (eTT.center);
    \draw (eTT.center) -- (eTL.center);

    \draw (fTL.center) -- (fTR.center);
    \draw (fTR.center) -- (fTT.center);
    \draw (fTT.center) -- (fTL.center);

    \draw[line width=0.6pt] (AwL_top.center) -- (AwL_mid.center);
    \draw[line width=0.6pt] (AwR_top.center) -- (AwR_mid.center);

    \draw[line width=0.6pt] (BwL_top.center) -- (BwL_mid.center);
    \draw[line width=0.6pt] (BwR_top.center) -- (BwR_mid.center);

    \draw (sTL.center) -- (sTR.center);
    \draw (sTR.center) -- (sTC.center);
    \draw (sTC.center) -- (sTL.center);

  \end{pgfonlayer}
            \end{tikzpicture}%
            }}= \vcenter{\hbox{%
\begin{tikzpicture}[baseline={(base.center)}]
  \begin{pgfonlayer}{nodelayer}
    \node[style=none] (base) at (0,0) {};

    \def\yTop{0.90}
    \def\yMid{0.10}
    \def\yBot{-0.85}

    \node[style=none] (eTL) at (-1.00,\yTop) {};
    \node[style=none] (eTR) at (-0.20,\yTop) {};
    \node[style=none] (eTT) at (-0.60,1.65) {};
    \node[style=none] (eLab) at (-0.60,1.25) {$e$};

    \node[style=none] (fTL) at ( 0.20,\yTop) {};
    \node[style=none] (fTR) at ( 1.00,\yTop) {};
    \node[style=none] (fTT) at ( 0.60,1.65) {};
    \node[style=none] (fLab) at ( 0.60,1.25) {$f$};

    \node[style=none] (AwL_top) at (-0.68,\yTop) {};
    \node[style=none] (AwR_top) at (-0.52,\yTop) {};
    \node[style=none] (AwL_mid) at (-0.68,\yMid) {};
    \node[style=none] (AwR_mid) at (-0.52,\yMid) {};
    \node[style=none] (AwLab) at (-0.92,0.55) {$\scriptstyle \m{A}$};

    \node[style=none] (BwL_top) at ( 0.52,\yTop) {};
    \node[style=none] (BwR_top) at ( 0.68,\yTop) {};
    \node[style=none] (BwL_mid) at ( 0.52,\yMid) {};
    \node[style=none] (BwR_mid) at ( 0.68,\yMid) {};
    \node[style=none] (BwLab) at ( 0.92,0.55) {$\scriptstyle \m{B}$};

    \node[style=none] (sTL) at (-1.10,\yMid) {};
    \node[style=none] (sTR) at ( 1.10,\yMid) {};
    \node[style=none] (sTC) at ( 0.00,\yBot) {};
    \node[style=none] (sNu) at ( 0.00,-0.30) {$\omega_1$};

  \end{pgfonlayer}

  \begin{pgfonlayer}{edgelayer}
    \draw (eTL.center) -- (eTR.center);
    \draw (eTR.center) -- (eTT.center);
    \draw (eTT.center) -- (eTL.center);

    \draw (fTL.center) -- (fTR.center);
    \draw (fTR.center) -- (fTT.center);
    \draw (fTT.center) -- (fTL.center);

    \draw[line width=0.6pt] (AwL_top.center) -- (AwL_mid.center);
    \draw[line width=0.6pt] (AwR_top.center) -- (AwR_mid.center);

    \draw[line width=0.6pt] (BwL_top.center) -- (BwL_mid.center);
    \draw[line width=0.6pt] (BwR_top.center) -- (BwR_mid.center);

    \draw (sTL.center) -- (sTR.center);
    \draw (sTR.center) -- (sTC.center);
    \draw (sTC.center) -- (sTL.center);

  \end{pgfonlayer}
\end{tikzpicture}%
}}\;\;\forall\;
            \vcenter{\hbox{%
            \begin{tikzpicture}[baseline={(base.center)}]
            \begin{pgfonlayer}{nodelayer}
            \node[style=none] (base) at (0,0) {};

            \def\yTop{0.00}
            \def\yBot{-0.80}

            \node[style=none] (eTL) at (-1.00,\yTop) {};
            \node[style=none] (eTR) at (-0.20,\yTop) {};
            \node[style=none] (eTT) at (-0.60,0.75) {};
            \node[style=none] (eLab) at (-0.60,0.35) {$e$};

            \node[style=none] (AwL_top) at (-0.68,\yTop) {};
            \node[style=none] (AwR_top) at (-0.52,\yTop) {};
            \node[style=none] (AwL_bot) at (-0.68,\yBot) {};
            \node[style=none] (AwR_bot) at (-0.52,\yBot) {};
            \node[style=none] (AwLab) at (-0.92,-0.40) {$\scriptstyle \m{A}$};

        \end{pgfonlayer}

        \begin{pgfonlayer}{edgelayer}
            \draw (eTL.center) -- (eTR.center);
            \draw (eTR.center) -- (eTT.center);
            \draw (eTT.center) -- (eTL.center);

            \draw[line width=0.6pt] (AwL_top.center) -- (AwL_bot.center);
            \draw[line width=0.6pt] (AwR_top.center) -- (AwR_bot.center);

        \end{pgfonlayer}
        \end{tikzpicture}%
        }}
        \in E_{\m{A}}\,,\,\vcenter{\hbox{%
        \begin{tikzpicture}[baseline={(base.center)}]
        \begin{pgfonlayer}{nodelayer}
            \node[style=none] (base) at (0,0) {};

            \def\yTop{0.00}
            \def\yBot{-0.80}

            \node[style=none] (fTL) at (-1.00,\yTop) {};
            \node[style=none] (fTR) at (-0.20,\yTop) {};
            \node[style=none] (fTT) at (-0.60,0.75) {};
            \node[style=none] (fLab) at (-0.60,0.35) {$f$};

            \node[style=none] (BwL_top) at (-0.68,\yTop) {};
            \node[style=none] (BwR_top) at (-0.52,\yTop) {};
            \node[style=none] (BwL_bot) at (-0.68,\yBot) {};
            \node[style=none] (BwR_bot) at (-0.52,\yBot) {};
            \node[style=none] (bwLab) at (-0.92,-0.40) {$\scriptstyle \m{B}$};

    \end{pgfonlayer}

    \begin{pgfonlayer}{edgelayer}
    \draw (fTL.center) -- (fTR.center);
    \draw (fTR.center) -- (fTT.center);
    \draw (fTT.center) -- (fTL.center);

    \draw[line width=0.6pt] (AwL_top.center) -- (BwL_bot.center);
    \draw[line width=0.6pt] (AwR_top.center) -- (BwR_bot.center);

  \end{pgfonlayer}
    \end{tikzpicture}%
    }}
    \in E_{\m{B}};
\end{equation} 
\item[ii.] Both states $\omega^{\m{AB}}_0$ and $\omega^{\m{AB}}_1$ lack TL-entanglement, ie:
\begin{equation}
\label{eq:DataHidingCondition2}
\left\{\vcenter{\hbox{
\begin{tikzpicture}
	\begin{pgfonlayer}{nodelayer}

		\node [style=none] (TA) at (-0.6,1.4) {};
		\node [style=none] (TB) at (0.6,1.4) {};

		\node [style=none] (TA2) at (-0.6,0.7) {};
		\node [style=none] (TB2) at (0.6,0.7) {};

		\node [style=none] (Alabel1) at (-0.9,1.1) {$\scriptstyle \mathcal A$};
		\node [style=none] (Blabel1) at (0.9,1.1) {$\scriptstyle \mathcal B$};

		\node [style=none] (R1) at (-0.9,0.7) {};
		\node [style=none] (R2) at (0.9,0.7) {};
		\node [style=none] (R3) at (0.9,-0.1) {};
		\node [style=none] (R4) at (-0.9,-0.1) {};

		\node [style=none] (RectLabel) at (0,0.3) {$\Pi_{TL}$};

		\node [style=none] (BA) at (-0.6,-0.1) {};
		\node [style=none] (BB) at (0.6,-0.1) {};

		\node [style=none] (BA2) at (-0.6,-0.8) {};
		\node [style=none] (BB2) at (0.6,-0.8) {};

		\node [style=none] (Alabel2) at (-0.9,-0.45) {$\scriptstyle \mathcal A$};
		\node [style=none] (Blabel2) at (0.9,-0.45) {$\scriptstyle \mathcal B$};

		\node [style=none] (TL) at (-0.9,-0.8) {};
		\node [style=none] (TR) at (0.9,-0.8) {};
		\node [style=none] (TC) at (0,-1.8) {};
		\node [style=none] (W)  at (0,-1.25) {$\omega_0$};

		\node [style=none] (NEQ) at (2.2,-0.5) {};

	\end{pgfonlayer}

	\begin{pgfonlayer}{edgelayer}

		\draw[double] (TA.center) to (TA2.center);
		\draw[double] (TB.center) to (TB2.center);

		\draw (R1.center) to (R2.center);
		\draw (R2.center) to (R3.center);
		\draw (R3.center) to (R4.center);
		\draw (R4.center) to (R1.center);

		\draw[double] (BA.center) to (BA2.center);
		\draw[double] (BB.center) to (BB2.center);

		\draw (TL.center) to (TR.center);
		\draw (TR.center) to (TC.center);
		\draw (TC.center) to (TL.center);

	\end{pgfonlayer}
\end{tikzpicture}}}\right\} \subset \mathsf{Sep}[\Omega_{\m{AB}}]\,,
\end{equation}
\end{enumerate}
    then a perfectly secure data hiding protocol can be implemented with tomographically non-local entanglement alone.
\end{proposition}

Many known tomographically non-local GPTs satisfy the conditions of Proposition~\ref{prop.:ConditionsPerfectDataHiding}, such as bilocal classical theories, bit-fliped twirled c-bit world, and RQT. The next example shows in more detail the particular case of perfectly secure data hiding  using rebits~\cite{Wootters_2010}, where success without tomographically-local entanglement is possible.
\blk
\begin{example}[Perfectly secure data hiding using rebits]
\label{example: DataHidingRebits}
    Consider a world governed by real quantum theory. The following encoding scheme of a bit $x$ in  states $\omega^{\m{AB}}_x\in\m{S_{AB}}$, achieves perfectly secure data hiding:
    \begin{align}
    \label{eq:rebitStatesEncodingBit}
        \omega_{x}^{AB} = \frac{\mathds{1}\otimes\mathds{1}+(-1)^{x}\sigma_y\otimes\sigma_y}{4}.
    \end{align}

These two states are orthogonal in RQT, i.e., one can perfectly discriminate them by implementing  the   measurement $\{e_0:= \Pi_{\omega_0},e_1:= \Pi_{\omega_1}\}\subset E_{\m{AB}}$, \blk where $\Pi_{\omega_0}=\frac{1}{2}(\mathds{1}+\sigma_y\otimes\sigma_y)$ and $\Pi_{\omega_1}=\frac{1}{2}(\mathds{1}-\sigma_y\otimes\sigma_y)$. To see that this is a valid measurement in RQT, note that every effect is real and this is also a valid measurement in complex quantum theory, as the effects can be written as   $e_0=\ketbra{\Psi_+}{\Psi_+}+\ketbra{\Phi_-}{\Phi_-}$ and $e_1=\ketbra{\Psi_-}{\Psi_-}+\ketbra{\Phi_+}{\Phi_+}\}$, where $\ket{\Psi_\pm} = \frac{1}{\sqrt{2}}(\ket{01}\pm \ket{10})$ and $\ket{\Phi_\pm}= \frac{1}{\sqrt{2}}(\ket{00}\pm \ket{11})$.  

Now, notice that the only difference between $\omega^{\m{AB}}_0$ and $\omega^{\m{AB}}_1$ is the sign of the component $\sigma_y\otimes\sigma_y$, which spans the holistic-state subspace $H_S$ (recall Eq.~\eqref{eq:H_STwoRebits}). Since no local measurement can access the holistic degrees of freedom, no local measurements can give any information about this bit. One can only find out the bit by making a joint measurement, such as the measurement described above.

Finally, notice that $\Pi_{\rm TL}(\omega^{\m{AB}})=\frac{\mathds{1}\otimes\mathds{1}}{4}$, which is separable.
\end{example}

Ref.~\cite{Wootters_2010} proposed this protocol for rebits for the first time and further generalized it to an arbitrary number of parties. The case of the other mentioned GPTs are analogous and show that some operational tasks can be accomplished with tomographically non-local entanglement alone that are impossible without it (and thus impossible in all TL theories). 
\blk

\subsection{Local encoding and LOCC-decoding of hidden data allowed by extra TNL-entanglement} 

In this Section we show that, in certain theories where tomographic locality fails, one of the parties (e.g., Alice) can locally encode data into a bipartite system in a way that remains hidden from all local measurements. Traditionally, data-hiding schemes assume that the hidden data is encoded by a third party at the preparation stage. Here, we adapt the data-hiding paradigm to a setting in which Alice and Bob distribute a hiding carrier in advance, and Alice may only later acquire the data she wishes to store. The information is then encoded locally by Alice, yet remains inaccessible to any adversary who gains access to only one subsystem. We further show that if Alice and Bob share an additional pair of systems exhibiting tomographically non-local entanglement, they can use LOCC to decode the hidden data (which is otherwise inaccessible by LOCC alone) such that any adversary with access only to their classical communication is oblivious to Alice's bit. Equivalently, the extra TNL entangled resource activates the implementation of the joint measurement required for decoding. 

For the local encoding part to work,  local operations must be able to affect holistic degrees of freedom, as we saw is needed for dense coding in Section~\ref{sec:Dense coding}.  We start exemplifying how this can be done in RQT.

\begin{example}
\label{example: LocalEncodingDataHiding}Suppose Alice and Bob each hold a share of a two-rebit system prepared in state $\omega^{AB}_0=\frac{\mathds{1}\otimes\mathds{1}+\sigma_y\otimes\sigma_y}{4}$. Then, Alice  can locally encode a bit $x$ into the pair of states $\{\omega^{AB}_0,\omega^{AB}_1\}$ defined in Eq.~\eqref{eq:rebitStatesEncodingBit}, by applying a local operation. 
To see this, note that one can toggle between the two states above by  the unitary superoperator $(\mathcal{Z}^A)^x$ on $A$. 
   This is clear since $\mathcal{Z}(\sigma_y):=\sigma_z\sigma_y\sigma_z=-\sigma_y$, and so
   \begin{align}
        (\mathcal{Z}^A)^x\otimes \mathds{1}^B(\omega^{\m{AB}}_+) = \frac{1}{4}\left[\mathds{1}^{AB}+ \mathcal{Z}^x(\sigma_y^A)\otimes\sigma_y^B\right] = \frac{1}{4}\left[\mathds{1}^{AB}+ (-1)^x\sigma_y^A\otimes\sigma_y^B\right].
    \end{align}
\end{example}

Notice that an analogous mechanism works in BCT and other theories in which local operations can affect the holistic degrees of freedom. For instance, a simple adaptation of  Example~\ref{BCTDenseCoding} allows to achieve such a local encoding in BCT. Again, in theories where local operations never impact the holistic degrees of freedom (such as MBCT~\cite{Rolino_MOPTs_2025}), this is impossible.

Now we show that if Alice and Bob share an extra two-rebit state $\omega^{A'B'}$, that can carry TNL entanglement only, they can use LOCC on their systems so that Bob retrieves Alice's bit. Moreover, Alice can publicly announce her outcomes and only Bob will have access to $x$, if he keeps his systems and outcomes private.

\begin{proposition}[Decoding secret bit with TNL entangled state plus LOCC in RQT]
    Suppose Alice and Bob share the following state encoding a bit $x\in\{0,1\}$:
    \begin{align}
        \omega^{\m{AB}}_x = \frac{1}{4}\left[\mathds{1}+(-1)^{x}\sigma_y\otimes\sigma_y\right].
    \end{align}
    Suppose that they want Bob to know the bit, but they do not have a trusted channel to send rebit particles (so they can't put both systems together to access the value of $x$ through a local joint measurement). Suppose they have access to LOCC  and share an extra pair of rebit systems, $\omega^{A'B'}$. State $\omega^{A'B'}$ can carry TNL entanglement alone and still allow for retrieving bit $x$.
\end{proposition}

\begin{proof}
We need to find a composite system $\m{A'B'}$ and a state $\omega^{\m{A'B'}}$ such that, when Alice performs a local measurement on $\m{AA'}$,  the  state of systems $\m{BB'}$ on Bob's side conditioned on Alice's outcome are such that Bob can locally perform a joint measurement on $\m{BB'}$ to read out the value of $x$. We will see that we can take (i) $\m{A'}$ as isomorphic to $\m{A}$, and $\m{B'}$ as isomorphic to $\m{B}$, (ii) $\omega^{\m{A'B'}} = \omega^{\m{AB}}_0$, and (iii) the local measurement of Alice on $\m{AA'}$ to be $\{e_0,e_1\}$, i.e., the measurement that perfectly distinguishes among $\omega_0^{\m{AB}}$ and $\omega_1^{\m{AB}}$.

In diagrammatic jargon, what we need to prove is that the following two equalities hold:

\begin{equation}
\label{eq:LOCC-DecodingwExtraTNLE1}
\vcenter{\hbox{%
\begin{tikzpicture}[baseline={(base.center)}, WireThick/.style={line width=0.6pt}]
\begin{pgfonlayer}{nodelayer}
\node[style=none] (base) at (0,0) {};

\node[style=none] (Om0TL)  at (-0.78,0.35) {};
\node[style=none] (Om0TR)  at ( 3.82,0.35) {};
\node[style=none] (Om0B)   at ( 1.47,-1.30) {};
\node[style=none] (Om0Lab) at ( 1.47,-0.52) {$\omega_0$};

\coordinate (cApr) at (0.00,0);
\coordinate (cA)   at (1.20,0);
\coordinate (cB)   at (2.15,0);
\coordinate (cBpr) at (3.10,0);

\node[style=none] (OmxTL)  at (0.95,1.80) {};
\node[style=none] (OmxTR)  at (2.40,1.80) {};
\node[style=none] (OmxB)   at (1.675,0.70) {};
\node[style=none] (OmxLab) at (1.675,1.34) {$\omega_x$};

\node[style=none] (EBL)  at (-0.60,3.25) {};
\node[style=none] (EBR)  at ( 1.80,3.25) {};
\node[style=none] (ET)   at ( 0.60,4.35) {};
\node[style=none] (ELab) at ( 0.60,3.80) {$e_0$};

\def\DW{0.10}

\def\yOm0Top{0.35}
\def\yOmxTop{1.80}
\def\yEBase{3.25}
\def\yTopWires{4.10}

\node[style=none] (AprL0) at (-\DW,\yOm0Top) {};
\node[style=none] (AprR0) at ( \DW,\yOm0Top) {};
\node[style=none] (AprL1) at (-\DW,\yEBase) {};
\node[style=none] (AprR1) at ( \DW,\yEBase) {};
\node[style=none] (AprLab) at (-0.55,1.65) {$\scriptstyle \m{A}'$};

\node[style=none] (AL0) at (1.20-\DW,\yOmxTop) {};
\node[style=none] (AR0) at (1.20+\DW,\yOmxTop) {};
\node[style=none] (AL1) at (1.20-\DW,\yEBase) {};
\node[style=none] (AR1) at (1.20+\DW,\yEBase) {};
\node[style=none] (ALab) at (0.91,2.47) {$\scriptstyle \m{A}$};

\node[style=none] (BL0) at (2.15-\DW,\yOmxTop) {};
\node[style=none] (BR0) at (2.15+\DW,\yOmxTop) {};
\node[style=none] (BL1) at (2.15-\DW,\yTopWires) {};
\node[style=none] (BR1) at (2.15+\DW,\yTopWires) {};
\node[style=none] (BLab) at (2.42,2.47) {$\scriptstyle \m{B}$};

\node[style=none] (BprL0) at (3.10-\DW,\yOm0Top) {};
\node[style=none] (BprR0) at (3.10+\DW,\yOm0Top) {};
\node[style=none] (BprL1) at (3.10-\DW,\yTopWires) {};
\node[style=none] (BprR1) at (3.10+\DW,\yTopWires) {};
\node[style=none] (BprLab) at (3.10+0.45,2.55) {$\scriptstyle \m{B}'$};

\node[style=none] (Eq) at (4.90,2.10) {$=$};
\node[style=none] (p0) at (6.00,2.10) {$p(0)$};

\node[style=none] (ROmTL)  at (7.10,1.80) {};
\node[style=none] (ROmTR)  at (9.10,1.80) {};
\node[style=none] (ROmB)   at (8.10,0.70) {};
\node[style=none] (ROmLab) at (8.10,1.43) {$\omega_x$};

\def\xRB{7.55}
\def\xRBp{8.65}

\node[style=none] (RB_L0) at (\xRB-\DW,1.80) {};
\node[style=none] (RB_R0) at (\xRB+\DW,1.80) {};
\node[style=none] (RB_L1) at (\xRB-\DW,4.10) {};
\node[style=none] (RB_R1) at (\xRB+\DW,4.10) {};
\node[style=none] (RB_Lab) at (\xRB-0.35,2.55) {$\scriptstyle \m{B}$};

\node[style=none] (RBp_L0) at (\xRBp-\DW,1.80) {};
\node[style=none] (RBp_R0) at (\xRBp+\DW,1.80) {};
\node[style=none] (RBp_L1) at (\xRBp-\DW,4.10) {};
\node[style=none] (RBp_R1) at (\xRBp+\DW,4.10) {};
\node[style=none] (RBp_Lab) at (\xRBp+0.35,2.55) {$\scriptstyle \m{B}'$};

\end{pgfonlayer}

\begin{pgfonlayer}{edgelayer}

\draw (Om0TL.center) -- (Om0TR.center) -- (Om0B.center) -- cycle;
\draw (OmxTL.center) -- (OmxTR.center) -- (OmxB.center) -- cycle;
\draw (EBL.center) -- (EBR.center) -- (ET.center) -- cycle;

\draw[WireThick] (AprL0.center) -- (AprL1.center);
\draw[WireThick] (AprR0.center) -- (AprR1.center);

\draw[WireThick] (AL0.center) -- (AL1.center);
\draw[WireThick] (AR0.center) -- (AR1.center);

\draw[WireThick] (BL0.center) -- (BL1.center);
\draw[WireThick] (BR0.center) -- (BR1.center);

\draw[WireThick] (BprL0.center) -- (BprL1.center);
\draw[WireThick] (BprR0.center) -- (BprR1.center);

\draw (ROmTL.center) -- (ROmTR.center) -- (ROmB.center) -- cycle;

\draw[WireThick] (RB_L0.center) -- (RB_L1.center);
\draw[WireThick] (RB_R0.center) -- (RB_R1.center);

\draw[WireThick] (RBp_L0.center) -- (RBp_L1.center);
\draw[WireThick] (RBp_R0.center) -- (RBp_R1.center);

\end{pgfonlayer}
\end{tikzpicture}%
}}
\end{equation}
and
\begin{equation}
\label{eq:LOCC-DecodingwExtraTNLE2}
\vcenter{\hbox{%
\begin{tikzpicture}[baseline={(base.center)}, WireThick/.style={line width=0.6pt}]
\begin{pgfonlayer}{nodelayer}
\node[style=none] (base) at (0,0) {};

\node[style=none] (Om0TL)  at (-0.78,0.35) {};
\node[style=none] (Om0TR)  at ( 3.82,0.35) {};
\node[style=none] (Om0B)   at ( 1.47,-1.30) {};
\node[style=none] (Om0Lab) at ( 1.47,-0.52) {$\omega_0$};

\coordinate (cApr) at (0.00,0);
\coordinate (cA)   at (1.20,0);
\coordinate (cB)   at (2.15,0);
\coordinate (cBpr) at (3.10,0);

\node[style=none] (OmxTL)  at (0.95,1.80) {};
\node[style=none] (OmxTR)  at (2.40,1.80) {};
\node[style=none] (OmxB)   at (1.675,0.70) {};
\node[style=none] (OmxLab) at (1.675,1.34) {$\omega_x$};

\node[style=none] (EBL)  at (-0.60,3.25) {};
\node[style=none] (EBR)  at ( 1.80,3.25) {};
\node[style=none] (ET)   at ( 0.60,4.35) {};
\node[style=none] (ELab) at ( 0.60,3.80) {$e_1$};

\def\DW{0.10}

\def\yOm0Top{0.35}
\def\yOmxTop{1.80}
\def\yEBase{3.25}
\def\yTopWires{4.10}

\node[style=none] (AprL0) at (-\DW,\yOm0Top) {};
\node[style=none] (AprR0) at ( \DW,\yOm0Top) {};
\node[style=none] (AprL1) at (-\DW,\yEBase) {};
\node[style=none] (AprR1) at ( \DW,\yEBase) {};
\node[style=none] (AprLab) at (-0.55,1.65) {$\scriptstyle \m{A}'$};

\node[style=none] (AL0) at (1.20-\DW,\yOmxTop) {};
\node[style=none] (AR0) at (1.20+\DW,\yOmxTop) {};
\node[style=none] (AL1) at (1.20-\DW,\yEBase) {};
\node[style=none] (AR1) at (1.20+\DW,\yEBase) {};
\node[style=none] (ALab) at (0.91,2.47) {$\scriptstyle \m{A}$};

\node[style=none] (BL0) at (2.15-\DW,\yOmxTop) {};
\node[style=none] (BR0) at (2.15+\DW,\yOmxTop) {};
\node[style=none] (BL1) at (2.15-\DW,\yTopWires) {};
\node[style=none] (BR1) at (2.15+\DW,\yTopWires) {};
\node[style=none] (BLab) at (2.42,2.47) {$\scriptstyle \m{B}$};

\node[style=none] (BprL0) at (3.10-\DW,\yOm0Top) {};
\node[style=none] (BprR0) at (3.10+\DW,\yOm0Top) {};
\node[style=none] (BprL1) at (3.10-\DW,\yTopWires) {};
\node[style=none] (BprR1) at (3.10+\DW,\yTopWires) {};
\node[style=none] (BprLab) at (3.10+0.45,2.55) {$\scriptstyle \m{B}'$};

\node[style=none] (Eq) at (4.90,2.10) {$=$};
\node[style=none] (p0) at (6.00,2.10) {$p(1)$};

\node[style=none] (ROmTL)  at (7.10,1.80) {};
\node[style=none] (ROmTR)  at (9.10,1.80) {};
\node[style=none] (ROmB)   at (8.10,0.70) {};
\node[style=none] (ROmLab) at (8.10,1.43) {$\omega_{\neg x}$};

\def\xRB{7.55}
\def\xRBp{8.65}

\node[style=none] (RB_L0) at (\xRB-\DW,1.80) {};
\node[style=none] (RB_R0) at (\xRB+\DW,1.80) {};
\node[style=none] (RB_L1) at (\xRB-\DW,4.10) {};
\node[style=none] (RB_R1) at (\xRB+\DW,4.10) {};
\node[style=none] (RB_Lab) at (\xRB-0.35,2.55) {$\scriptstyle \m{B}$};

\node[style=none] (RBp_L0) at (\xRBp-\DW,1.80) {};
\node[style=none] (RBp_R0) at (\xRBp+\DW,1.80) {};
\node[style=none] (RBp_L1) at (\xRBp-\DW,4.10) {};
\node[style=none] (RBp_R1) at (\xRBp+\DW,4.10) {};
\node[style=none] (RBp_Lab) at (\xRBp+0.35,2.55) {$\scriptstyle \m{B}'$};

\end{pgfonlayer}

\begin{pgfonlayer}{edgelayer}

\draw (Om0TL.center) -- (Om0TR.center) -- (Om0B.center) -- cycle;
\draw (OmxTL.center) -- (OmxTR.center) -- (OmxB.center) -- cycle;
\draw (EBL.center) -- (EBR.center) -- (ET.center) -- cycle;

\draw[WireThick] (AprL0.center) -- (AprL1.center);
\draw[WireThick] (AprR0.center) -- (AprR1.center);

\draw[WireThick] (AL0.center) -- (AL1.center);
\draw[WireThick] (AR0.center) -- (AR1.center);

\draw[WireThick] (BL0.center) -- (BL1.center);
\draw[WireThick] (BR0.center) -- (BR1.center);

\draw[WireThick] (BprL0.center) -- (BprL1.center);
\draw[WireThick] (BprR0.center) -- (BprR1.center);

\draw (ROmTL.center) -- (ROmTR.center) -- (ROmB.center) -- cycle;

\draw[WireThick] (RB_L0.center) -- (RB_L1.center);
\draw[WireThick] (RB_R0.center) -- (RB_R1.center);

\draw[WireThick] (RBp_L0.center) -- (RBp_L1.center);
\draw[WireThick] (RBp_R0.center) -- (RBp_R1.center);

\end{pgfonlayer}
\end{tikzpicture}%
}}\,,
\end{equation}
where $\omega_{\neg x}$ is $\omega_1$ if $x=0$ or $\omega_{0}$ if $x=1$. \blk
Eqs.~\eqref{eq:LOCC-DecodingwExtraTNLE1} and \eqref{eq:LOCC-DecodingwExtraTNLE2} imply that Bob can find out the value of $x$ by implementing the measurement that distinguished $\omega_0$ and $\omega_1$, provided that he knows Alice's outcome.  The formal proof goes as follows.

Suppose the extra rebit pair is in the following state
    \begin{align}
    \omega^{\m{A'B'}} = \frac{1}{4}\left[\mathds{1}+\sigma_y\otimes\sigma_y\right].
    \end{align}
    Then, Alice and Bob can do as follows:
    Alice performs  the measurement  
   $\{e_0,e_1\}$ on her systems $\m{AA'}$ and communicates her outcome $a$ to Bob. Bob then makes an analogous measurement on his systems, $\m{BB'}$, obtaining outcome $b$.   What we will see now is that $x=a\oplus b$, and hence Bob can compute the secret bit $x$ locally from $a$ and $b$.

 The systems that Alice and Bob share are described by the state
    \begin{align}
        \omega^{\m{AA'BB'}} &= \omega^{\m{AB}}_x\otimes\omega^{\m{A'B'}} = \frac{1}{16}\left[\mathds{1}^{AB}+(-1)^{x}\sigma^A_y\otimes\sigma^B_y\right]\left[\mathds{1}^{A'B'}+\sigma^{A'}_y\otimes\sigma^{B'}_y\right]\\ &= \frac{1}{16}\left[\mathds{1}^{ABA'B'} + \mathds{1}^{AB}\otimes\sigma_y^{A'}\otimes\sigma_y^{B'}+(-1)^{x} \sigma_y^{A}\otimes\sigma_y^{B}\otimes\mathds{1}^{A'B'} +(-1)^{x} \sigma_y^{A}\otimes\sigma_y^{B}\otimes\sigma_y^{A'}\otimes\sigma_y^{B'}\right].
    \end{align}

    Now, if Alice implements the measurement 
   $\{e_0^{\m{AA'}},e_1^{\m{AA'}}\}$, she gets outcome $0$ with probability $\frac{1}{2}$ and outcome $1$ with probability $\frac{1}{2}$. The conditional (subnormalized) state on Bob's side associated to $a=0$ is
    \begin{align}
    \label{eq:Assemblage0NTL}
        \sigma^{\m{BB'}}_{a=0}:={\rm Tr}_{\m{AA'}}[e^{\m{AA'}}_0\omega^{\m{ABA'B'}}]= \frac{1}{2}\frac{1}{4}\left[\mathds{1}+(-1)^{x}\sigma_y^B\otimes\sigma_y^{B'}\right],
    \end{align}
    while the conditional (subnormalized) state associated to $a=1$ is
        \begin{align}
            \label{eq:Assemblage1NTL}
        \sigma^{\m{BB'}}_{a=1}:={\rm Tr}_{\m{AA'}}[e^{\m{AA'}}_1\omega^{\m{ABA'B'}}]= \frac{1}{2}\frac{1}{4}\left[\mathds{1}-(-1)^{x}\sigma_y^B\otimes\sigma_y^{B'}\right].
    \end{align}
    Thus, if Bob knows that Alice got outcome $a=0$, Bob knows he has state $\omega^{\m{BB'}}_x$ on his side, while if Alice got outcome $a=1$, he has $\omega^{\m{BB'}}_{\neg x}$. Hence, when Bob performs locally the measurement $\{e_0^{\m{BB'}},e_1^{\m{BB'}}\}$, he obtains $x$ when $a=0$ and $\neg x$ when $a=1$. It follows that $x=a\oplus b$, as we needed to prove. 
    
    An alternative way to see that $x=a\oplus b$, is by calculating the joint probability of Alice and Bob obtaining outcomes $0$ or $1$ when performing the measurements above, namely
    \begin{align}
        p(00)=p(11) = \frac{1}{4}(1+(-1)^{x})\\
        p(01)=p(10) = \frac{1}{4}(1-(-1)^{x}).
    \end{align}
    Hence, if $x=0$, Alice and Bob have perfectly correlated outcomes, and if $x=1$ they have perfectly anti-correlated outcomes. Notice that, anyone who has access to $a$ but not to $b$ cannot know whether the outcomes are correlated or anti-correlated, and thus cannot know $x$.

In summary, if the two parties share an extra pair of rebits in the state $\omega^{\m{A'B'}}$ (which has only TNL entanglement) and if they can communicate a classical bit, then they can also retrieve the secret bit without the need to bring their rebit systems together physically or via trusted rebit channels. Moreover, someone with access to Alices's bit alone cannot recover $x$.
\end{proof}

This effective implementation of the joint measurement using TNL entanglement and LOCC generalizes to a class of GPT systems, as we show next.
\begin{proposition}
    [Generalization to other TNL GPT systems]
Consider a composite GPT system $\m{AB}$ in which:
\begin{enumerate}
    \item[i.] There exists a pair of states $\{\omega^{\m{AB}}_x\}_{x\in\{0,1\}}\subset\Omega_{\m{AB}}$ that carry TNL entanglement alone; 
    \item[ii.] $e^{\m{A}}\boxtimes e^{\m{B}}(\omega^{\m{AB}}_x)$ is independent of $x$ for all $e^{\m{A}}\in E_{\m{A}}$ and $e^{\m{B}}\in E_{\m{B}}$;
    \item[iii.] There exists a joint measurement $\{e^{\m{AB}}_0,e^{\m{AB}}_1\}\subset E_{\m{AB}}$ that perfectly distinguishes $\omega^{\m{AB}}_0$ and $\omega^{\m{AB}}_{1}$;
    \item[iv.] There exist systems $\m{A'}$ isomorphic to $\m{A}$ and $\m{B'}$ isomorphic to $\m{B}$, and state $\omega^{\m{A'B'}} = \omega_0$, such that \blk \\
    $e^{\m{AA'}}_0\boxtimes \mathds{1}^{\m{BB'}}(\omega^{\m{AB}}_x\boxtimes \omega^{\m{A'B'}})=p(0)\omega^{\m{BB'}}_x$ and $e^{\m{AA'}}_1\boxtimes \mathds{1}^{\m{BB'}}(\omega^{\m{AB}}_x\boxtimes \omega^{\m{A'B'}})=p(1)\omega^{\m{BB'}}_{\neg x}$, i.e.,

    \begin{equation}
\label{eq:LOCC-DecodingwExtraTNLE--GPTs1}
\vcenter{\hbox{%
\begin{tikzpicture}[baseline={(base.center)}, WireThick/.style={line width=0.6pt}]
\begin{pgfonlayer}{nodelayer}
\node[style=none] (base) at (0,0) {};

\node[style=none] (Om0TL)  at (-0.78,0.35) {};
\node[style=none] (Om0TR)  at ( 3.82,0.35) {};
\node[style=none] (Om0B)   at ( 1.47,-1.30) {};
\node[style=none] (Om0Lab) at ( 1.47,-0.52) {$\omega_0$};

\coordinate (cApr) at (0.00,0);
\coordinate (cA)   at (1.20,0);
\coordinate (cB)   at (2.15,0);
\coordinate (cBpr) at (3.10,0);

\node[style=none] (OmxTL)  at (0.95,1.80) {};
\node[style=none] (OmxTR)  at (2.40,1.80) {};
\node[style=none] (OmxB)   at (1.675,0.70) {};
\node[style=none] (OmxLab) at (1.675,1.34) {$\omega_x$};

\node[style=none] (EBL)  at (-0.60,3.25) {};
\node[style=none] (EBR)  at ( 1.80,3.25) {};
\node[style=none] (ET)   at ( 0.60,4.35) {};
\node[style=none] (ELab) at ( 0.60,3.80) {$e_0$};

\def\DW{0.10}

\def\yOm0Top{0.35}
\def\yOmxTop{1.80}
\def\yEBase{3.25}
\def\yTopWires{4.10}

\node[style=none] (AprL0) at (-\DW,\yOm0Top) {};
\node[style=none] (AprR0) at ( \DW,\yOm0Top) {};
\node[style=none] (AprL1) at (-\DW,\yEBase) {};
\node[style=none] (AprR1) at ( \DW,\yEBase) {};
\node[style=none] (AprLab) at (-0.55,1.65) {$\scriptstyle \m{A}'$};

\node[style=none] (AL0) at (1.20-\DW,\yOmxTop) {};
\node[style=none] (AR0) at (1.20+\DW,\yOmxTop) {};
\node[style=none] (AL1) at (1.20-\DW,\yEBase) {};
\node[style=none] (AR1) at (1.20+\DW,\yEBase) {};
\node[style=none] (ALab) at (0.91,2.47) {$\scriptstyle \m{A}$};

\node[style=none] (BL0) at (2.15-\DW,\yOmxTop) {};
\node[style=none] (BR0) at (2.15+\DW,\yOmxTop) {};
\node[style=none] (BL1) at (2.15-\DW,\yTopWires) {};
\node[style=none] (BR1) at (2.15+\DW,\yTopWires) {};
\node[style=none] (BLab) at (2.42,2.47) {$\scriptstyle \m{B}$};

\node[style=none] (BprL0) at (3.10-\DW,\yOm0Top) {};
\node[style=none] (BprR0) at (3.10+\DW,\yOm0Top) {};
\node[style=none] (BprL1) at (3.10-\DW,\yTopWires) {};
\node[style=none] (BprR1) at (3.10+\DW,\yTopWires) {};
\node[style=none] (BprLab) at (3.10+0.45,2.55) {$\scriptstyle \m{B}'$};

\node[style=none] (Eq) at (4.90,2.10) {$=$};
\node[style=none] (p0) at (6.00,2.10) {$p(0)$};

\node[style=none] (ROmTL)  at (7.10,1.80) {};
\node[style=none] (ROmTR)  at (9.10,1.80) {};
\node[style=none] (ROmB)   at (8.10,0.70) {};
\node[style=none] (ROmLab) at (8.10,1.43) {$\omega_x$};

\def\xRB{7.55}
\def\xRBp{8.65}

\node[style=none] (RB_L0) at (\xRB-\DW,1.80) {};
\node[style=none] (RB_R0) at (\xRB+\DW,1.80) {};
\node[style=none] (RB_L1) at (\xRB-\DW,4.10) {};
\node[style=none] (RB_R1) at (\xRB+\DW,4.10) {};
\node[style=none] (RB_Lab) at (\xRB-0.35,2.55) {$\scriptstyle \m{B}$};

\node[style=none] (RBp_L0) at (\xRBp-\DW,1.80) {};
\node[style=none] (RBp_R0) at (\xRBp+\DW,1.80) {};
\node[style=none] (RBp_L1) at (\xRBp-\DW,4.10) {};
\node[style=none] (RBp_R1) at (\xRBp+\DW,4.10) {};
\node[style=none] (RBp_Lab) at (\xRBp+0.35,2.55) {$\scriptstyle \m{B}'$};

\end{pgfonlayer}

\begin{pgfonlayer}{edgelayer}

\draw (Om0TL.center) -- (Om0TR.center) -- (Om0B.center) -- cycle;
\draw (OmxTL.center) -- (OmxTR.center) -- (OmxB.center) -- cycle;
\draw (EBL.center) -- (EBR.center) -- (ET.center) -- cycle;

\draw[WireThick] (AprL0.center) -- (AprL1.center);
\draw[WireThick] (AprR0.center) -- (AprR1.center);

\draw[WireThick] (AL0.center) -- (AL1.center);
\draw[WireThick] (AR0.center) -- (AR1.center);

\draw[WireThick] (BL0.center) -- (BL1.center);
\draw[WireThick] (BR0.center) -- (BR1.center);

\draw[WireThick] (BprL0.center) -- (BprL1.center);
\draw[WireThick] (BprR0.center) -- (BprR1.center);

\draw (ROmTL.center) -- (ROmTR.center) -- (ROmB.center) -- cycle;

\draw[WireThick] (RB_L0.center) -- (RB_L1.center);
\draw[WireThick] (RB_R0.center) -- (RB_R1.center);

\draw[WireThick] (RBp_L0.center) -- (RBp_L1.center);
\draw[WireThick] (RBp_R0.center) -- (RBp_R1.center);

\end{pgfonlayer}
\end{tikzpicture}%
}}
\end{equation}
and
\begin{equation}
\label{eq:LOCC-DecodingwExtraTNLE--GPTs2}
\vcenter{\hbox{%
\begin{tikzpicture}[baseline={(base.center)}, WireThick/.style={line width=0.6pt}]
\begin{pgfonlayer}{nodelayer}
\node[style=none] (base) at (0,0) {};

\node[style=none] (Om0TL)  at (-0.78,0.35) {};
\node[style=none] (Om0TR)  at ( 3.82,0.35) {};
\node[style=none] (Om0B)   at ( 1.47,-1.30) {};
\node[style=none] (Om0Lab) at ( 1.47,-0.52) {$\omega_0$};

\coordinate (cApr) at (0.00,0);
\coordinate (cA)   at (1.20,0);
\coordinate (cB)   at (2.15,0);
\coordinate (cBpr) at (3.10,0);

\node[style=none] (OmxTL)  at (0.95,1.80) {};
\node[style=none] (OmxTR)  at (2.40,1.80) {};
\node[style=none] (OmxB)   at (1.675,0.70) {};
\node[style=none] (OmxLab) at (1.675,1.34) {$\omega_x$};

\node[style=none] (EBL)  at (-0.60,3.25) {};
\node[style=none] (EBR)  at ( 1.80,3.25) {};
\node[style=none] (ET)   at ( 0.60,4.35) {};
\node[style=none] (ELab) at ( 0.60,3.80) {$e_1$};

\def\DW{0.10}

\def\yOm0Top{0.35}
\def\yOmxTop{1.80}
\def\yEBase{3.25}
\def\yTopWires{4.10}

\node[style=none] (AprL0) at (-\DW,\yOm0Top) {};
\node[style=none] (AprR0) at ( \DW,\yOm0Top) {};
\node[style=none] (AprL1) at (-\DW,\yEBase) {};
\node[style=none] (AprR1) at ( \DW,\yEBase) {};
\node[style=none] (AprLab) at (-0.55,1.65) {$\scriptstyle \m{A}'$};

\node[style=none] (AL0) at (1.20-\DW,\yOmxTop) {};
\node[style=none] (AR0) at (1.20+\DW,\yOmxTop) {};
\node[style=none] (AL1) at (1.20-\DW,\yEBase) {};
\node[style=none] (AR1) at (1.20+\DW,\yEBase) {};
\node[style=none] (ALab) at (0.91,2.47) {$\scriptstyle \m{A}$};

\node[style=none] (BL0) at (2.15-\DW,\yOmxTop) {};
\node[style=none] (BR0) at (2.15+\DW,\yOmxTop) {};
\node[style=none] (BL1) at (2.15-\DW,\yTopWires) {};
\node[style=none] (BR1) at (2.15+\DW,\yTopWires) {};
\node[style=none] (BLab) at (2.42,2.47) {$\scriptstyle \m{B}$};

\node[style=none] (BprL0) at (3.10-\DW,\yOm0Top) {};
\node[style=none] (BprR0) at (3.10+\DW,\yOm0Top) {};
\node[style=none] (BprL1) at (3.10-\DW,\yTopWires) {};
\node[style=none] (BprR1) at (3.10+\DW,\yTopWires) {};
\node[style=none] (BprLab) at (3.10+0.45,2.55) {$\scriptstyle \m{B}'$};

\node[style=none] (Eq) at (4.90,2.10) {$=$};
\node[style=none] (p0) at (6.00,2.10) {$p(1)$};

\node[style=none] (ROmTL)  at (7.10,1.80) {};
\node[style=none] (ROmTR)  at (9.10,1.80) {};
\node[style=none] (ROmB)   at (8.10,0.70) {};
\node[style=none] (ROmLab) at (8.10,1.43) {$\omega_{\neg x}$};

\def\xRB{7.55}
\def\xRBp{8.65}

\node[style=none] (RB_L0) at (\xRB-\DW,1.80) {};
\node[style=none] (RB_R0) at (\xRB+\DW,1.80) {};
\node[style=none] (RB_L1) at (\xRB-\DW,4.10) {};
\node[style=none] (RB_R1) at (\xRB+\DW,4.10) {};
\node[style=none] (RB_Lab) at (\xRB-0.35,2.55) {$\scriptstyle \m{B}$};

\node[style=none] (RBp_L0) at (\xRBp-\DW,1.80) {};
\node[style=none] (RBp_R0) at (\xRBp+\DW,1.80) {};
\node[style=none] (RBp_L1) at (\xRBp-\DW,4.10) {};
\node[style=none] (RBp_R1) at (\xRBp+\DW,4.10) {};
\node[style=none] (RBp_Lab) at (\xRBp+0.35,2.55) {$\scriptstyle \m{B}'$};

\end{pgfonlayer}

\begin{pgfonlayer}{edgelayer}

\draw (Om0TL.center) -- (Om0TR.center) -- (Om0B.center) -- cycle;
\draw (OmxTL.center) -- (OmxTR.center) -- (OmxB.center) -- cycle;
\draw (EBL.center) -- (EBR.center) -- (ET.center) -- cycle;

\draw[WireThick] (AprL0.center) -- (AprL1.center);
\draw[WireThick] (AprR0.center) -- (AprR1.center);

\draw[WireThick] (AL0.center) -- (AL1.center);
\draw[WireThick] (AR0.center) -- (AR1.center);

\draw[WireThick] (BL0.center) -- (BL1.center);
\draw[WireThick] (BR0.center) -- (BR1.center);

\draw[WireThick] (BprL0.center) -- (BprL1.center);
\draw[WireThick] (BprR0.center) -- (BprR1.center);

\draw (ROmTL.center) -- (ROmTR.center) -- (ROmB.center) -- cycle;

\draw[WireThick] (RB_L0.center) -- (RB_L1.center);
\draw[WireThick] (RB_R0.center) -- (RB_R1.center);

\draw[WireThick] (RBp_L0.center) -- (RBp_L1.center);
\draw[WireThick] (RBp_R0.center) -- (RBp_R1.center);

\end{pgfonlayer}
\end{tikzpicture}%
}}\,.
\end{equation}
\end{enumerate}
Then, the GPT allows for secure secret sharing. 

Notice that the first three conditions ensure that data hiding is possible with TNL alone (recall Prop.~\ref{prop.:ConditionsPerfectDataHiding}). In turn, condition (iv) ensures that one can remotely implement the required joint measurement for securely reading the hidden data.\footnote{Roughly, this fourth condition is akin to demanding that a particular kind of teleportation is possible within the theory, or that a particular kind of entanglement swapping is possible within the theory.}
\end{proposition}
\begin{proof}
Conditions $i.-iv.$ imply:
\begin{equation}
\vcenter{\hbox{%
\begin{tikzpicture}[baseline={(base.center)}]
\begin{pgfonlayer}{nodelayer}
\node[style=none] (base) at (0,0) {};


\def\DW{0.10}

\def\yOm0Top{0.35}
\def\yEBase{3.25}
\def\yETop{4.35}

\def\yOmxTop{1.80}

\node[style=none] (Om0TL)  at (-0.78,\yOm0Top) {};
\node[style=none] (Om0TR)  at ( 4.85,\yOm0Top) {};
\node[style=none] (Om0B)   at ( 2.035,-1.60) {};
\node[style=none] (Om0Lab) at ( 2.035,-0.70) {$\omega_0$};

\node[style=none] (OmxTL)  at (0.78,\yOmxTop) {};
\node[style=none] (OmxTR)  at (3.43,\yOmxTop) {};
\node[style=none] (OmxB)   at (2.10,0.70) {};
\node[style=none] (OmxLab) at (2.10,1.30) {$\omega_x$};


\node[style=none] (E1BL)  at (-0.60,\yEBase) {};
\node[style=none] (E1BR)  at ( 1.80,\yEBase) {};
\node[style=none] (E1T)   at ( 0.60,\yETop)  {};
\node[style=none] (E1Lab) at ( 0.60,3.80) {$e_0$};

\node[style=none] (ExBL)  at ( 2.26,\yEBase) {};
\node[style=none] (ExBR)  at ( 4.66,\yEBase) {};
\node[style=none] (ExT)   at ( 3.46,\yETop)  {};
\node[style=none] (ExLab) at ( 3.46,3.80) {$e_x$};

\def\xApr{0.00}
\def\xA{1.20}
\def\xB{3.00}
\def\xBpr{4.20}


\node[style=none] (AprL0) at (\xApr-\DW,\yOm0Top) {};
\node[style=none] (AprR0) at (\xApr+\DW,\yOm0Top) {};
\node[style=none] (AprL1) at (\xApr-\DW,\yEBase)  {};
\node[style=none] (AprR1) at (\xApr+\DW,\yEBase)  {};
\node[style=none] (AprLab) at (\xApr-0.55,1.65) {$\scriptstyle \m{A}'$};

\node[style=none] (AL0) at (\xA-\DW,\yOmxTop) {};
\node[style=none] (AR0) at (\xA+\DW,\yOmxTop) {};
\node[style=none] (AL1) at (\xA-\DW,\yEBase)  {};
\node[style=none] (AR1) at (\xA+\DW,\yEBase)  {};
\node[style=none] (ALab) at (\xA-0.35,2.45) {$\scriptstyle \m{A}$};

\node[style=none] (BL0) at (\xB-\DW,\yOmxTop) {};
\node[style=none] (BR0) at (\xB+\DW,\yOmxTop) {};
\node[style=none] (BL1) at (\xB-\DW,\yEBase)  {};
\node[style=none] (BR1) at (\xB+\DW,\yEBase)  {};
\node[style=none] (BLab) at (\xB+0.35,2.45) {$\scriptstyle \m{B}$};

\node[style=none] (BprL0) at (\xBpr-\DW,\yOm0Top) {};
\node[style=none] (BprR0) at (\xBpr+\DW,\yOm0Top) {};
\node[style=none] (BprL1) at (\xBpr-\DW,\yEBase)  {};
\node[style=none] (BprR1) at (\xBpr+\DW,\yEBase)  {};
\node[style=none] (BprLab) at (\xBpr+0.45,1.65) {$\scriptstyle \m{B}'$};

\end{pgfonlayer}

\begin{pgfonlayer}{edgelayer}

\draw (Om0TL.center) -- (Om0TR.center) -- (Om0B.center) -- cycle;
\draw (OmxTL.center) -- (OmxTR.center) -- (OmxB.center) -- cycle;

\draw (E1BL.center) -- (E1BR.center) -- (E1T.center) -- cycle;
\draw (ExBL.center) -- (ExBR.center) -- (ExT.center) -- cycle;

\draw[line width=0.6pt] (AprL0.center) -- (AprL1.center);
\draw[line width=0.6pt] (AprR0.center) -- (AprR1.center);

\draw[line width=0.6pt] (AL0.center) -- (AL1.center);
\draw[line width=0.6pt] (AR0.center) -- (AR1.center);

\draw[line width=0.6pt] (BL0.center) -- (BL1.center);
\draw[line width=0.6pt] (BR0.center) -- (BR1.center);

\draw[line width=0.6pt] (BprL0.center) -- (BprL1.center);
\draw[line width=0.6pt] (BprR0.center) -- (BprR1.center);

\end{pgfonlayer}
\end{tikzpicture}%
}} = p(0)
\end{equation}
and
\begin{equation}
\vcenter{\hbox{%
\begin{tikzpicture}[baseline={(base.center)}]
\begin{pgfonlayer}{nodelayer}
\node[style=none] (base) at (0,0) {};


\def\DW{0.10}

\def\yOm0Top{0.35}
\def\yEBase{3.25}
\def\yETop{4.35}

\def\yOmxTop{1.80}

\node[style=none] (Om0TL)  at (-0.78,\yOm0Top) {};
\node[style=none] (Om0TR)  at ( 4.85,\yOm0Top) {};
\node[style=none] (Om0B)   at ( 2.035,-1.60) {};
\node[style=none] (Om0Lab) at ( 2.035,-0.70) {$\omega_0$};

\node[style=none] (OmxTL)  at (0.78,\yOmxTop) {};
\node[style=none] (OmxTR)  at (3.43,\yOmxTop) {};
\node[style=none] (OmxB)   at (2.10,0.70) {};
\node[style=none] (OmxLab) at (2.10,1.30) {$\omega_x$};


\node[style=none] (E1BL)  at (-0.60,\yEBase) {};
\node[style=none] (E1BR)  at ( 1.80,\yEBase) {};
\node[style=none] (E1T)   at ( 0.60,\yETop)  {};
\node[style=none] (E1Lab) at ( 0.60,3.80) {$e_1$};

\node[style=none] (ExBL)  at ( 2.26,\yEBase) {};
\node[style=none] (ExBR)  at ( 4.66,\yEBase) {};
\node[style=none] (ExT)   at ( 3.46,\yETop)  {};
\node[style=none] (ExLab) at ( 3.46,3.80) {$e_x$};

\def\xApr{0.00}
\def\xA{1.20}
\def\xB{3.00}
\def\xBpr{4.20}


\node[style=none] (AprL0) at (\xApr-\DW,\yOm0Top) {};
\node[style=none] (AprR0) at (\xApr+\DW,\yOm0Top) {};
\node[style=none] (AprL1) at (\xApr-\DW,\yEBase)  {};
\node[style=none] (AprR1) at (\xApr+\DW,\yEBase)  {};
\node[style=none] (AprLab) at (\xApr-0.55,1.65) {$\scriptstyle \m{A}'$};

\node[style=none] (AL0) at (\xA-\DW,\yOmxTop) {};
\node[style=none] (AR0) at (\xA+\DW,\yOmxTop) {};
\node[style=none] (AL1) at (\xA-\DW,\yEBase)  {};
\node[style=none] (AR1) at (\xA+\DW,\yEBase)  {};
\node[style=none] (ALab) at (\xA-0.35,2.45) {$\scriptstyle \m{A}$};

\node[style=none] (BL0) at (\xB-\DW,\yOmxTop) {};
\node[style=none] (BR0) at (\xB+\DW,\yOmxTop) {};
\node[style=none] (BL1) at (\xB-\DW,\yEBase)  {};
\node[style=none] (BR1) at (\xB+\DW,\yEBase)  {};
\node[style=none] (BLab) at (\xB+0.35,2.45) {$\scriptstyle \m{B}$};

\node[style=none] (BprL0) at (\xBpr-\DW,\yOm0Top) {};
\node[style=none] (BprR0) at (\xBpr+\DW,\yOm0Top) {};
\node[style=none] (BprL1) at (\xBpr-\DW,\yEBase)  {};
\node[style=none] (BprR1) at (\xBpr+\DW,\yEBase)  {};
\node[style=none] (BprLab) at (\xBpr+0.45,1.65) {$\scriptstyle \m{B}'$};

\end{pgfonlayer}

\begin{pgfonlayer}{edgelayer}

\draw (Om0TL.center) -- (Om0TR.center) -- (Om0B.center) -- cycle;
\draw (OmxTL.center) -- (OmxTR.center) -- (OmxB.center) -- cycle;

\draw (E1BL.center) -- (E1BR.center) -- (E1T.center) -- cycle;
\draw (ExBL.center) -- (ExBR.center) -- (ExT.center) -- cycle;

\draw[line width=0.6pt] (AprL0.center) -- (AprL1.center);
\draw[line width=0.6pt] (AprR0.center) -- (AprR1.center);

\draw[line width=0.6pt] (AL0.center) -- (AL1.center);
\draw[line width=0.6pt] (AR0.center) -- (AR1.center);

\draw[line width=0.6pt] (BL0.center) -- (BL1.center);
\draw[line width=0.6pt] (BR0.center) -- (BR1.center);

\draw[line width=0.6pt] (BprL0.center) -- (BprL1.center);
\draw[line width=0.6pt] (BprR0.center) -- (BprR1.center);

\end{pgfonlayer}
\end{tikzpicture}%
}} = 0.
\end{equation}
Furthermore, they imply:
\begin{equation}
\vcenter{\hbox{%
\begin{tikzpicture}[baseline={(base.center)}]
\begin{pgfonlayer}{nodelayer}
\node[style=none] (base) at (0,0) {};


\def\DW{0.10}

\def\yOm0Top{0.35}
\def\yEBase{3.25}
\def\yETop{4.35}

\def\yOmxTop{1.80}

\node[style=none] (Om0TL)  at (-0.78,\yOm0Top) {};
\node[style=none] (Om0TR)  at ( 4.85,\yOm0Top) {};
\node[style=none] (Om0B)   at ( 2.035,-1.60) {};
\node[style=none] (Om0Lab) at ( 2.035,-0.70) {$\omega_0$};

\node[style=none] (OmxTL)  at (0.78,\yOmxTop) {};
\node[style=none] (OmxTR)  at (3.43,\yOmxTop) {};
\node[style=none] (OmxB)   at (2.10,0.70) {};
\node[style=none] (OmxLab) at (2.10,1.30) {$\omega_x$};


\node[style=none] (E1BL)  at (-0.60,\yEBase) {};
\node[style=none] (E1BR)  at ( 1.80,\yEBase) {};
\node[style=none] (E1T)   at ( 0.60,\yETop)  {};
\node[style=none] (E1Lab) at ( 0.60,3.80) {$e_0$};

\node[style=none] (ExBL)  at ( 2.26,\yEBase) {};
\node[style=none] (ExBR)  at ( 4.66,\yEBase) {};
\node[style=none] (ExT)   at ( 3.46,\yETop)  {};
\node[style=none] (ExLab) at ( 3.54,3.80) {$e_{\neg x}$};

\def\xApr{0.00}
\def\xA{1.20}
\def\xB{3.00}
\def\xBpr{4.20}


\node[style=none] (AprL0) at (\xApr-\DW,\yOm0Top) {};
\node[style=none] (AprR0) at (\xApr+\DW,\yOm0Top) {};
\node[style=none] (AprL1) at (\xApr-\DW,\yEBase)  {};
\node[style=none] (AprR1) at (\xApr+\DW,\yEBase)  {};
\node[style=none] (AprLab) at (\xApr-0.55,1.65) {$\scriptstyle \m{A}'$};

\node[style=none] (AL0) at (\xA-\DW,\yOmxTop) {};
\node[style=none] (AR0) at (\xA+\DW,\yOmxTop) {};
\node[style=none] (AL1) at (\xA-\DW,\yEBase)  {};
\node[style=none] (AR1) at (\xA+\DW,\yEBase)  {};
\node[style=none] (ALab) at (\xA-0.35,2.45) {$\scriptstyle \m{A}$};

\node[style=none] (BL0) at (\xB-\DW,\yOmxTop) {};
\node[style=none] (BR0) at (\xB+\DW,\yOmxTop) {};
\node[style=none] (BL1) at (\xB-\DW,\yEBase)  {};
\node[style=none] (BR1) at (\xB+\DW,\yEBase)  {};
\node[style=none] (BLab) at (\xB+0.35,2.45) {$\scriptstyle \m{B}$};

\node[style=none] (BprL0) at (\xBpr-\DW,\yOm0Top) {};
\node[style=none] (BprR0) at (\xBpr+\DW,\yOm0Top) {};
\node[style=none] (BprL1) at (\xBpr-\DW,\yEBase)  {};
\node[style=none] (BprR1) at (\xBpr+\DW,\yEBase)  {};
\node[style=none] (BprLab) at (\xBpr+0.45,1.65) {$\scriptstyle \m{B}'$};

\end{pgfonlayer}

\begin{pgfonlayer}{edgelayer}

\draw (Om0TL.center) -- (Om0TR.center) -- (Om0B.center) -- cycle;
\draw (OmxTL.center) -- (OmxTR.center) -- (OmxB.center) -- cycle;

\draw (E1BL.center) -- (E1BR.center) -- (E1T.center) -- cycle;
\draw (ExBL.center) -- (ExBR.center) -- (ExT.center) -- cycle;

\draw[line width=0.6pt] (AprL0.center) -- (AprL1.center);
\draw[line width=0.6pt] (AprR0.center) -- (AprR1.center);

\draw[line width=0.6pt] (AL0.center) -- (AL1.center);
\draw[line width=0.6pt] (AR0.center) -- (AR1.center);

\draw[line width=0.6pt] (BL0.center) -- (BL1.center);
\draw[line width=0.6pt] (BR0.center) -- (BR1.center);

\draw[line width=0.6pt] (BprL0.center) -- (BprL1.center);
\draw[line width=0.6pt] (BprR0.center) -- (BprR1.center);

\end{pgfonlayer}
\end{tikzpicture}%
}} = 0
\end{equation}
and
\begin{equation}
\vcenter{\hbox{%
\begin{tikzpicture}[baseline={(base.center)}]
\begin{pgfonlayer}{nodelayer}
\node[style=none] (base) at (0,0) {};


\def\DW{0.10}

\def\yOm0Top{0.35}
\def\yEBase{3.25}
\def\yETop{4.35}

\def\yOmxTop{1.80}

\node[style=none] (Om0TL)  at (-0.78,\yOm0Top) {};
\node[style=none] (Om0TR)  at ( 4.85,\yOm0Top) {};
\node[style=none] (Om0B)   at ( 2.035,-1.60) {};
\node[style=none] (Om0Lab) at ( 2.035,-0.70) {$\omega_0$};

\node[style=none] (OmxTL)  at (0.78,\yOmxTop) {};
\node[style=none] (OmxTR)  at (3.43,\yOmxTop) {};
\node[style=none] (OmxB)   at (2.10,0.70) {};
\node[style=none] (OmxLab) at (2.10,1.30) {$\omega_x$};


\node[style=none] (E1BL)  at (-0.60,\yEBase) {};
\node[style=none] (E1BR)  at ( 1.80,\yEBase) {};
\node[style=none] (E1T)   at ( 0.60,\yETop)  {};
\node[style=none] (E1Lab) at ( 0.60,3.80) {$e_1$};

\node[style=none] (ExBL)  at ( 2.26,\yEBase) {};
\node[style=none] (ExBR)  at ( 4.66,\yEBase) {};
\node[style=none] (ExT)   at ( 3.46,\yETop)  {};
\node[style=none] (ExLab) at ( 3.54,3.80) {$e_{\neg x}$};

\def\xApr{0.00}
\def\xA{1.20}
\def\xB{3.00}
\def\xBpr{4.20}


\node[style=none] (AprL0) at (\xApr-\DW,\yOm0Top) {};
\node[style=none] (AprR0) at (\xApr+\DW,\yOm0Top) {};
\node[style=none] (AprL1) at (\xApr-\DW,\yEBase)  {};
\node[style=none] (AprR1) at (\xApr+\DW,\yEBase)  {};
\node[style=none] (AprLab) at (\xApr-0.55,1.65) {$\scriptstyle \m{A}'$};

\node[style=none] (AL0) at (\xA-\DW,\yOmxTop) {};
\node[style=none] (AR0) at (\xA+\DW,\yOmxTop) {};
\node[style=none] (AL1) at (\xA-\DW,\yEBase)  {};
\node[style=none] (AR1) at (\xA+\DW,\yEBase)  {};
\node[style=none] (ALab) at (\xA-0.35,2.45) {$\scriptstyle \m{A}$};

\node[style=none] (BL0) at (\xB-\DW,\yOmxTop) {};
\node[style=none] (BR0) at (\xB+\DW,\yOmxTop) {};
\node[style=none] (BL1) at (\xB-\DW,\yEBase)  {};
\node[style=none] (BR1) at (\xB+\DW,\yEBase)  {};
\node[style=none] (BLab) at (\xB+0.35,2.45) {$\scriptstyle \m{B}$};

\node[style=none] (BprL0) at (\xBpr-\DW,\yOm0Top) {};
\node[style=none] (BprR0) at (\xBpr+\DW,\yOm0Top) {};
\node[style=none] (BprL1) at (\xBpr-\DW,\yEBase)  {};
\node[style=none] (BprR1) at (\xBpr+\DW,\yEBase)  {};
\node[style=none] (BprLab) at (\xBpr+0.45,1.65) {$\scriptstyle \m{B}'$};

\end{pgfonlayer}

\begin{pgfonlayer}{edgelayer}

\draw (Om0TL.center) -- (Om0TR.center) -- (Om0B.center) -- cycle;
\draw (OmxTL.center) -- (OmxTR.center) -- (OmxB.center) -- cycle;

\draw (E1BL.center) -- (E1BR.center) -- (E1T.center) -- cycle;
\draw (ExBL.center) -- (ExBR.center) -- (ExT.center) -- cycle;

\draw[line width=0.6pt] (AprL0.center) -- (AprL1.center);
\draw[line width=0.6pt] (AprR0.center) -- (AprR1.center);

\draw[line width=0.6pt] (AL0.center) -- (AL1.center);
\draw[line width=0.6pt] (AR0.center) -- (AR1.center);

\draw[line width=0.6pt] (BL0.center) -- (BL1.center);
\draw[line width=0.6pt] (BR0.center) -- (BR1.center);

\draw[line width=0.6pt] (BprL0.center) -- (BprL1.center);
\draw[line width=0.6pt] (BprR0.center) -- (BprR1.center);

\end{pgfonlayer}
\end{tikzpicture}%
}} = p(1).
\end{equation}
Therefore, if $x=0$, then $p(00)+p(11)=1$ while if $x=1$, $p(01)+p(10)=1$. This means that, if Alice gets output $a$ for her measurement and Bob has output $b$ for his measurement, $x=a\oplus b$. Thus if either party knows both $a$ and $b$, they can decode $x$.
\end{proof}

A natural question is when a GPT system carries features $i.-iv.$, mainly condition $iv.$. 

\begin{conjecture}
\label{conjectureEmbeddingTLConditioni} TNL-systems of the same type that can be embedded into TL systems and obey conditions $(i)$ to $(iii)$ above, also obey condition $(iv)$. 
\end{conjecture}

The intuition behind this conjecture is that a well-behaved embedding requires that state spaces of systems $\m{C}$ and $\m{D}$ (from the TL theory) have extra dimensions when compared to  $\m{A}$ and $\m{B}$,  so that their product can accommodate the holistic degrees of freedom of $\m{AB}$ -- similarly to what happens in the two-rebits case, in which the $\sigma_y$ plays the role of the additional component. We then expect that $H_S$ will be made of products of such extra dimensions (again, as it happens to $\sigma_y\otimes\sigma_y$ in the two rebits case). If this holds true, the embedding of $\omega^{\m{AB}}_x\boxtimes \omega^{\m{A'B'}}_0$ into $CDC'D' $ will have products of all four extra components, which can be assessed by local measurements in the same way that $\sigma_y^A\otimes\sigma_y^{A'}\otimes\sigma_y^B\otimes\sigma_y^{B'}$ can be assessed by Alice with a measurement on her $\m{AA'}$ two-rebits subsystems. If this holds true, one can prove that many TNL-theories allow for the protocol described in this section.

\section{Future directions}

Our work introduced the notion of tomographically-nonlocal  entanglement, and begun the study of its structural and operational significance. This opens several avenues for future work.
A recurring theme is that several no-go results traditionally phrased as ``task $A$ is impossible with entanglement'' rely---often implicitly---on the assumption of tomographic locality. As we argued above, tomographic locality may fail even in quantum theory (when considering fermions or fundamental superselection rules). Our framework therefore shows that such results need to be formulated more carefully; for instance, by translating them into statements of the form ``task $A$ is impossible \emph{with TL-entanglement}.'' A critical question, then is to  explore whether these tasks become possible when TNL-entanglement is present.  As just one example of this, perfectly secure data hiding does become possible, unlike in unrestricted quantum theory. What other familiar claims about the power and limitations of quantum information protocols need to be reassessed in light of these considerations? 

These questions strengthen the case for answering the critical question: {\em does} tomographic locality fail in our quantum world, a possibility already raised in Ref.~\cite{centeno2024twirledworldssymmetryinducedfailures}? Or can every superselection rule be lifted -- even, for example, the one which forbids superposition of states having different parity in the number of particles of fermionic systems?

The distinction between TL and TNL entanglement also has consequences within entanglement theory itself. As an illustrative example, the existence of \emph{pure} bound-entangled states in phase-shift--twirled bosonic world---whose entanglement may be entirely tomographically nonlocal---suggests that phenomena ruled out in unrestricted quantum theory may reappear once attention is restricted to the absence of \emph{tomographically-local} entanglement rather than to entanglement \emph{simpliciter}. It would be valuable to understand how general this pattern is, and to identify further qualitative effects that arise specifically from tomographic nonlocality.

A third direction is to analyze embeddings of TNL systems into TL systems~\cite{SchmidSimplex,MullerGarner,schmid2024shadowssubsystemsgeneralizedprobabilistic}. At a minimum, this is useful as a representational strategy, as can be seen in Ref.~\cite{centeno2024twirledworldssymmetryinducedfailures}. (When such embeddings exist, one may be able to work within the higher-dimensional TL model and interpret TNL entanglement as a distinguished component of the embedded structure.)

Several open questions arise in the multipartite setting: how TL versus TNL entanglement depends on the chosen partition, whether there exist genuinely multipartite forms of TL or TNL entanglement, and how these notions relate to known holism constraints (e.g., limited holism in real quantum theory). Technically, the simple relation $\Pi_{\rm TNL}=\mathds{1}-\Pi_{\rm TL}$ may be insufficient to isolate the TNL content relevant to a given partition, suggesting a richer family of TNL projections may be required.

It is also natural to further investigate the relationship between TNL entanglement and monogamy. For pairs of rebits, non-monogamous states have {\em only} TNL entanglement, but it remains to determine whether this extends to larger numbers of systems and to arbitrary GPTs. Moreover, it remains to be shown whether TNL entanglement is always non-monogamous. Progress here may require an explicitly operational formulation of monogamy.

On the resource-theoretic side, it could be interesting to develop resource theories for TL and/or TNL entanglement, potentially based on LOSR or LOCC for the former, and on a suitable analogue such as LOQR (local operations and shared quasiprobabilities) for the latter. This has potential implications not just for the methodology of foil theories, but also for real-world applications in the presence of constraints from symmetry or superselection rules.

Further work is needed to explore cryptographic and operational applications of TNL entanglement. In forthcoming work, we study unconditionally secure bit-commitment using states with TNL entanglement, and we expect to find additional connections to reference frames -- e.g., identifying when TNL entangled states function as useful reference-frame resources in general GPTs. We also anticipate further insights from studying TNL entanglement in twirled and swirled worlds, including from resolving Conjecture~\ref{conjectureEmbeddingTLConditioni} and analyzing the information-processing possibilities within these theories.

Just as we defined TL and TNL entanglement, it is natural to wonder if one could define TL contextuality and TNL contextuality. We expect that doing so would require combining the methodology of this paper with that of Refs.~\cite{zhang2025reassessingboundaryclassicalnonclassical,zhang2025quantifiers}, which introduces definitions of nonclassicality (in the sense of the failure of noncontextuality) for states, transformations, and measurements (as opposed to nonclassicality for full experiments or theories).

Finally, it would be useful to understand if and when the tomographically local projection of a theory defines a valid subtheory~\cite{schmid2024shadowssubsystemsgeneralizedprobabilistic} of the given GPT. We conjecture that for theories carrying only TNL entanglement (such as BCT) this is always the case. However, it does not generally happen, since we saw in the main text that the TL projection can take valid states to vectors outside the state space.  
Relatedly, one can ask when the tomographically local projection of a theory defines a valid GPT (although not necessarily a subtheory of the original theory). Note that a simple sufficient condition for both of these to be the case is that the tomographically local projector is a valid (physical) process within the starting GPT.  There are likely also connections between this question and Ref.~\cite{Barnum_2023}, which defines a notion of a tomographically local shadow of a theory.

\section*{Acknowledgments}
We thank Rob Spekkens, Y\`{i}l\`{e} Ying and Lucien Hardy for useful discussions.
R.D.B acknowledges support by the Digital Horizon Europe project FoQaCiA, Foundations of quantum computational advantage, GA No.~101070558, funded by the European Union, NSERC (Canada), and UKRI (UK).
This work is partially carried out under IRA Programme, project no.~FENG.02.01-IP.05-0006/23, financed by the FENG program 2021-2027, Priority FENG.02, Measure FENG.02.01., with the support of the FNP. M.E.~was supported by the National Science Centre, Poland (Opus project,
Categorical Foundations of the Non-Classicality of Nature, project no.~2021/41/B/ST2/03149), and conducted part of this research under the support of the National Science Centre, Poland, Grant Sonata 16 no.~2020/39/D/ST2/01234.
J.H.S.~was funded by the European Commission by the QuantERA project ResourceQ under the grant agreement UMO2023/05/Y/ST2/00143.
A.B.S.~and J.H.S.~conducted part of this research while visiting the Okinawa Institute of Science and Technology (OIST) through the Theoretical Sciences Visiting Program (TSVP). D.S. and R.D.B were supported by the Perimeter Institute for Theoretical Physics. Research at Perimeter Institute is supported
in part by the Government of Canada through the
Department of Innovation, Science and Economic
Development and by the Province of Ontario
through the Ministry of Colleges and Universities.
The diagrams were prepared using TikZit.

\bibliography{bib}

\begin{thebibliography}{68}%
\makeatletter
\providecommand \@ifxundefined [1]{%
 \@ifx{#1\undefined}
}%
\providecommand \@ifnum [1]{%
 \ifnum #1\expandafter \@firstoftwo
 \else \expandafter \@secondoftwo
 \fi
}%
\providecommand \@ifx [1]{%
 \ifx #1\expandafter \@firstoftwo
 \else \expandafter \@secondoftwo
 \fi
}%
\providecommand \natexlab [1]{#1}%
\providecommand \enquote  [1]{``#1''}%
\providecommand \bibnamefont  [1]{#1}%
\providecommand \bibfnamefont [1]{#1}%
\providecommand \citenamefont [1]{#1}%
\providecommand \href@noop [0]{\@secondoftwo}%
\providecommand \href [0]{\begingroup \@sanitize@url \@href}%
\providecommand \@href[1]{\@@startlink{#1}\@@href}%
\providecommand \@@href[1]{\endgroup#1\@@endlink}%
\providecommand \@sanitize@url [0]{\catcode `\\12\catcode `\$12\catcode `\&12\catcode `\#12\catcode `\^12\catcode `\_12\catcode `\%12\relax}%
\providecommand \@@startlink[1]{}%
\providecommand \@@endlink[0]{}%
\providecommand \url  [0]{\begingroup\@sanitize@url \@url }%
\providecommand \@url [1]{\endgroup\@href {#1}{\urlprefix }}%
\providecommand \urlprefix  [0]{URL }%
\providecommand \Eprint [0]{\href }%
\providecommand \doibase [0]{https://doi.org/}%
\providecommand \selectlanguage [0]{\@gobble}%
\providecommand \bibinfo  [0]{\@secondoftwo}%
\providecommand \bibfield  [0]{\@secondoftwo}%
\providecommand \translation [1]{[#1]}%
\providecommand \BibitemOpen [0]{}%
\providecommand \bibitemStop [0]{}%
\providecommand \bibitemNoStop [0]{.\EOS\space}%
\providecommand \EOS [0]{\spacefactor3000\relax}%
\providecommand \BibitemShut  [1]{\csname bibitem#1\endcsname}%
\let\auto@bib@innerbib\@empty
\bibitem [{\citenamefont {Barrett}(2006)}]{barrett_informationGPTs_2006}%
  \BibitemOpen
  \bibfield  {author} {\bibinfo {author} {\bibfnamefont {J.}~\bibnamefont {Barrett}},\ }\href@noop {} {\bibinfo {title} {Information processing in generalized probabilistic theories}} (\bibinfo {year} {2006}),\ \Eprint {https://arxiv.org/abs/quant-ph/0508211} {arXiv:quant-ph/0508211 [quant-ph]} \BibitemShut {NoStop}%
\bibitem [{\citenamefont {Müller}(2021)}]{mullerGPTnotes}%
  \BibitemOpen
  \bibfield  {author} {\bibinfo {author} {\bibfnamefont {M.~P.}\ \bibnamefont {Müller}},\ }\bibfield  {title} {\bibinfo {title} {{Probabilistic theories and reconstructions of quantum theory}},\ }\href {https://doi.org/10.21468/SciPostPhysLectNotes.28} {\bibfield  {journal} {\bibinfo  {journal} {SciPost Phys. Lect. Notes}\ ,\ \bibinfo {pages} {28}} (\bibinfo {year} {2021})}\BibitemShut {NoStop}%
\bibitem [{\citenamefont {Plávala}(2023)}]{Plavala_2023_GPTsIntro}%
  \BibitemOpen
  \bibfield  {author} {\bibinfo {author} {\bibfnamefont {M.}~\bibnamefont {Plávala}},\ }\bibfield  {title} {\bibinfo {title} {General probabilistic theories: An introduction},\ }\href {https://doi.org/10.1016/j.physrep.2023.09.001} {\bibfield  {journal} {\bibinfo  {journal} {Physics Reports}\ }\textbf {\bibinfo {volume} {1033}},\ \bibinfo {pages} {1–64} (\bibinfo {year} {2023})}\BibitemShut {NoStop}%
\bibitem [{\citenamefont {Hardy}(2011)}]{hardy2011reformulatingreconstructingquantumtheory}%
  \BibitemOpen
  \bibfield  {author} {\bibinfo {author} {\bibfnamefont {L.}~\bibnamefont {Hardy}},\ }\href {https://arxiv.org/abs/1104.2066} {\bibinfo {title} {Reformulating and reconstructing quantum theory}} (\bibinfo {year} {2011}),\ \Eprint {https://arxiv.org/abs/1104.2066} {arXiv:1104.2066 [quant-ph]} \BibitemShut {NoStop}%
\bibitem [{\citenamefont {Hardy}(2001)}]{hardy2001quantum}%
  \BibitemOpen
  \bibfield  {author} {\bibinfo {author} {\bibfnamefont {L.}~\bibnamefont {Hardy}},\ }\href@noop {} {\bibinfo {title} {Quantum theory from five reasonable axioms}} (\bibinfo {year} {2001}),\ \Eprint {https://arxiv.org/abs/quant-ph/0101012} {arXiv:quant-ph/0101012 [quant-ph]} \BibitemShut {NoStop}%
\bibitem [{\citenamefont {Chiribella}\ \emph {et~al.}(2016)\citenamefont {Chiribella}, \citenamefont {D'Ariano},\ and\ \citenamefont {Perinotti}}]{Chiribella_QuantumFromPrinciples2016}%
  \BibitemOpen
  \bibfield  {author} {\bibinfo {author} {\bibfnamefont {G.}~\bibnamefont {Chiribella}}, \bibinfo {author} {\bibfnamefont {G.~M.}\ \bibnamefont {D'Ariano}},\ and\ \bibinfo {author} {\bibfnamefont {P.}~\bibnamefont {Perinotti}},\ }\href {https://doi.org/10.1007/978-94-017-7303-4} {\emph {\bibinfo {title} {Quantum Theory: Informational Foundations and Foils}}}\ (\bibinfo  {publisher} {Springer Netherlands},\ \bibinfo {year} {2016})\BibitemShut {NoStop}%
\bibitem [{\citenamefont {Mazurek}\ \emph {et~al.}(2021)\citenamefont {Mazurek}, \citenamefont {Pusey}, \citenamefont {Resch},\ and\ \citenamefont {Spekkens}}]{Mazurek_2021}%
  \BibitemOpen
  \bibfield  {author} {\bibinfo {author} {\bibfnamefont {M.~D.}\ \bibnamefont {Mazurek}}, \bibinfo {author} {\bibfnamefont {M.~F.}\ \bibnamefont {Pusey}}, \bibinfo {author} {\bibfnamefont {K.~J.}\ \bibnamefont {Resch}},\ and\ \bibinfo {author} {\bibfnamefont {R.~W.}\ \bibnamefont {Spekkens}},\ }\bibfield  {title} {\bibinfo {title} {Experimentally bounding deviations from quantum theory in the landscape of generalized probabilistic theories},\ }\bibfield  {journal} {\bibinfo  {journal} {PRX Quantum}\ }\textbf {\bibinfo {volume} {2}},\ \href {https://doi.org/10.1103/prxquantum.2.020302} {10.1103/prxquantum.2.020302} (\bibinfo {year} {2021})\BibitemShut {NoStop}%
\bibitem [{\citenamefont {Schmid}\ \emph {et~al.}(2024{\natexlab{a}})\citenamefont {Schmid}, \citenamefont {Selby}, \citenamefont {Pusey},\ and\ \citenamefont {Spekkens}}]{Schmid2024structuretheorem}%
  \BibitemOpen
  \bibfield  {author} {\bibinfo {author} {\bibfnamefont {D.}~\bibnamefont {Schmid}}, \bibinfo {author} {\bibfnamefont {J.~H.}\ \bibnamefont {Selby}}, \bibinfo {author} {\bibfnamefont {M.~F.}\ \bibnamefont {Pusey}},\ and\ \bibinfo {author} {\bibfnamefont {R.~W.}\ \bibnamefont {Spekkens}},\ }\bibfield  {title} {\bibinfo {title} {A structure theorem for generalized-noncontextual ontological models},\ }\href {https://doi.org/10.22331/q-2024-03-14-1283} {\bibfield  {journal} {\bibinfo  {journal} {{Quantum}}\ }\textbf {\bibinfo {volume} {8}},\ \bibinfo {pages} {1283} (\bibinfo {year} {2024}{\natexlab{a}})}\BibitemShut {NoStop}%
\bibitem [{\citenamefont {Müller}\ and\ \citenamefont {Ududec}(2012)}]{Mueller_2012}%
  \BibitemOpen
  \bibfield  {author} {\bibinfo {author} {\bibfnamefont {M.~P.}\ \bibnamefont {Müller}}\ and\ \bibinfo {author} {\bibfnamefont {C.}~\bibnamefont {Ududec}},\ }\bibfield  {title} {\bibinfo {title} {Structure of reversible computation determines the self-duality of quantum theory},\ }\bibfield  {journal} {\bibinfo  {journal} {Physical Review Letters}\ }\textbf {\bibinfo {volume} {108}},\ \href {https://doi.org/10.1103/physrevlett.108.130401} {10.1103/physrevlett.108.130401} (\bibinfo {year} {2012})\BibitemShut {NoStop}%
\bibitem [{\citenamefont {Chiribella}\ \emph {et~al.}(2011)\citenamefont {Chiribella}, \citenamefont {D’Ariano},\ and\ \citenamefont {Perinotti}}]{Chiribella_2011InfoDerivationQT}%
  \BibitemOpen
  \bibfield  {author} {\bibinfo {author} {\bibfnamefont {G.}~\bibnamefont {Chiribella}}, \bibinfo {author} {\bibfnamefont {G.~M.}\ \bibnamefont {D’Ariano}},\ and\ \bibinfo {author} {\bibfnamefont {P.}~\bibnamefont {Perinotti}},\ }\bibfield  {title} {\bibinfo {title} {Informational derivation of quantum theory},\ }\bibfield  {journal} {\bibinfo  {journal} {Physical Review A}\ }\textbf {\bibinfo {volume} {84}},\ \href {https://doi.org/10.1103/physreva.84.012311} {10.1103/physreva.84.012311} (\bibinfo {year} {2011})\BibitemShut {NoStop}%
\bibitem [{\citenamefont {Schmid}(2024)}]{Schmid2024reviewreformulation}%
  \BibitemOpen
  \bibfield  {author} {\bibinfo {author} {\bibfnamefont {D.}~\bibnamefont {Schmid}},\ }\bibfield  {title} {\bibinfo {title} {A review and reformulation of macroscopic realism: resolving its deficiencies using the framework of generalized probabilistic theories},\ }\href {https://doi.org/10.22331/q-2024-01-03-1217} {\bibfield  {journal} {\bibinfo  {journal} {{Quantum}}\ }\textbf {\bibinfo {volume} {8}},\ \bibinfo {pages} {1217} (\bibinfo {year} {2024})}\BibitemShut {NoStop}%
\bibitem [{\citenamefont {Spekkens}(2007)}]{toytheory}%
  \BibitemOpen
  \bibfield  {author} {\bibinfo {author} {\bibfnamefont {R.~W.}\ \bibnamefont {Spekkens}},\ }\bibfield  {title} {\bibinfo {title} {Evidence for the epistemic view of quantum states: A toy theory},\ }\href {https://doi.org/10.1103/PhysRevA.75.032110} {\bibfield  {journal} {\bibinfo  {journal} {Phys. Rev. A}\ }\textbf {\bibinfo {volume} {75}},\ \bibinfo {pages} {032110} (\bibinfo {year} {2007})}\BibitemShut {NoStop}%
\bibitem [{\citenamefont {Gottesman}(1998)}]{gottesman1998heisenbergrepresentationquantumcomputers}%
  \BibitemOpen
  \bibfield  {author} {\bibinfo {author} {\bibfnamefont {D.}~\bibnamefont {Gottesman}},\ }\href {https://arxiv.org/abs/quant-ph/9807006} {\bibinfo {title} {The heisenberg representation of quantum computers}} (\bibinfo {year} {1998}),\ \Eprint {https://arxiv.org/abs/quant-ph/9807006} {arXiv:quant-ph/9807006 [quant-ph]} \BibitemShut {NoStop}%
\bibitem [{\citenamefont {D'Ariano}\ \emph {et~al.}(2014{\natexlab{a}})\citenamefont {D'Ariano}, \citenamefont {Manessi}, \citenamefont {Perinotti},\ and\ \citenamefont {Tosini}}]{Dariano2014feynmanproblem}%
  \BibitemOpen
  \bibfield  {author} {\bibinfo {author} {\bibfnamefont {G.~M.}\ \bibnamefont {D'Ariano}}, \bibinfo {author} {\bibfnamefont {F.}~\bibnamefont {Manessi}}, \bibinfo {author} {\bibfnamefont {P.}~\bibnamefont {Perinotti}},\ and\ \bibinfo {author} {\bibfnamefont {A.}~\bibnamefont {Tosini}},\ }\bibfield  {title} {\bibinfo {title} {{The Feynman problem and fermionic entanglement: Fermionic theory versus qubit theory}},\ }\href {https://doi.org/10.1142/S0217751X14300257} {\bibfield  {journal} {\bibinfo  {journal} {International Journal of Modern Physics A}\ }\textbf {\bibinfo {volume} {29}},\ \bibinfo {pages} {1430025} (\bibinfo {year} {2014}{\natexlab{a}})}\BibitemShut {NoStop}%
\bibitem [{\citenamefont {Wootters}(1990)}]{Wooter_1990}%
  \BibitemOpen
  \bibfield  {author} {\bibinfo {author} {\bibfnamefont {W.~K.}\ \bibnamefont {Wootters}},\ }\bibinfo {title} {{Local Accessibility of Quantum States}},\ in\ \href@noop {} {\emph {\bibinfo {booktitle} {Complexity, Entropy And The Physics Of Information}}},\ \bibinfo {editor} {edited by\ \bibinfo {editor} {\bibfnamefont {W.~H.}\ \bibnamefont {Zurek}}}\ (\bibinfo  {publisher} {CRC Press},\ \bibinfo {address} {Boca Raton},\ \bibinfo {year} {1990})\ pp.\ \bibinfo {pages} {39--46}\BibitemShut {NoStop}%
\bibitem [{\citenamefont {D'Ariano}\ \emph {et~al.}(2014{\natexlab{b}})\citenamefont {D'Ariano}, \citenamefont {Manessi}, \citenamefont {Perinotti},\ and\ \citenamefont {Tosini}}]{darianoFermionic2014}%
  \BibitemOpen
  \bibfield  {author} {\bibinfo {author} {\bibfnamefont {G.~M.}\ \bibnamefont {D'Ariano}}, \bibinfo {author} {\bibfnamefont {F.}~\bibnamefont {Manessi}}, \bibinfo {author} {\bibfnamefont {P.}~\bibnamefont {Perinotti}},\ and\ \bibinfo {author} {\bibfnamefont {A.}~\bibnamefont {Tosini}},\ }\bibfield  {title} {\bibinfo {title} {Fermionic computation is non-local tomographic and violates monogamy of entanglement},\ }\href {https://doi.org/10.1209/0295-5075/107/20009} {\bibfield  {journal} {\bibinfo  {journal} {Europhysics Letters}\ }\textbf {\bibinfo {volume} {107}},\ \bibinfo {pages} {20009} (\bibinfo {year} {2014}{\natexlab{b}})}\BibitemShut {NoStop}%
\bibitem [{\citenamefont {Centeno}\ \emph {et~al.}(2025)\citenamefont {Centeno}, \citenamefont {Erba}, \citenamefont {Galley}, \citenamefont {Schmid}, \citenamefont {Selby}, \citenamefont {Spekkens}, \citenamefont {Soltani}, \citenamefont {Surace}, \citenamefont {Wilce},\ and\ \citenamefont {Y\ifmmode~\bar{\imath}\else \={\i}\fi{}ng}}]{centeno2024twirledworldssymmetryinducedfailures}%
  \BibitemOpen
  \bibfield  {author} {\bibinfo {author} {\bibfnamefont {D.}~\bibnamefont {Centeno}}, \bibinfo {author} {\bibfnamefont {M.}~\bibnamefont {Erba}}, \bibinfo {author} {\bibfnamefont {T.~D.}\ \bibnamefont {Galley}}, \bibinfo {author} {\bibfnamefont {D.}~\bibnamefont {Schmid}}, \bibinfo {author} {\bibfnamefont {J.~H.}\ \bibnamefont {Selby}}, \bibinfo {author} {\bibfnamefont {R.~W.}\ \bibnamefont {Spekkens}}, \bibinfo {author} {\bibfnamefont {S.}~\bibnamefont {Soltani}}, \bibinfo {author} {\bibfnamefont {J.}~\bibnamefont {Surace}}, \bibinfo {author} {\bibfnamefont {A.}~\bibnamefont {Wilce}},\ and\ \bibinfo {author} {\bibfnamefont {Y.}~\bibnamefont {Y\ifmmode~\bar{\imath}\else \={\i}\fi{}ng}},\ }\bibfield  {title} {\bibinfo {title} {Symmetry-induced failures of tomographic locality: Constructing foil theories by twirling},\ }\href {https://doi.org/10.1103/hpmv-15sf} {\bibfield  {journal} {\bibinfo  {journal} {Phys. Rev. A}\ }\textbf {\bibinfo {volume} {112}},\ \bibinfo {pages} {L030202} (\bibinfo {year}
  {2025})}\BibitemShut {NoStop}%
\bibitem [{\citenamefont {Yīng}\ \emph {et~al.}(2025)\citenamefont {Yīng}, \citenamefont {Alañón}, \citenamefont {Centeno}, \citenamefont {Surace}, \citenamefont {Ansanelli}, \citenamefont {Liu}, \citenamefont {Schmid},\ and\ \citenamefont {Spekkens}}]{ying2025quantumtheoryneedscomplex}%
  \BibitemOpen
  \bibfield  {author} {\bibinfo {author} {\bibfnamefont {Y.}~\bibnamefont {Yīng}}, \bibinfo {author} {\bibfnamefont {M.~C.}\ \bibnamefont {Alañón}}, \bibinfo {author} {\bibfnamefont {D.}~\bibnamefont {Centeno}}, \bibinfo {author} {\bibfnamefont {J.}~\bibnamefont {Surace}}, \bibinfo {author} {\bibfnamefont {M.~M.}\ \bibnamefont {Ansanelli}}, \bibinfo {author} {\bibfnamefont {R.}~\bibnamefont {Liu}}, \bibinfo {author} {\bibfnamefont {D.}~\bibnamefont {Schmid}},\ and\ \bibinfo {author} {\bibfnamefont {R.~W.}\ \bibnamefont {Spekkens}},\ }\href {https://arxiv.org/abs/2506.08091} {\bibinfo {title} {On whether quantum theory needs complex numbers: the foil theories perspective}} (\bibinfo {year} {2025}),\ \Eprint {https://arxiv.org/abs/2506.08091} {arXiv:2506.08091 [quant-ph]} \BibitemShut {NoStop}%
\bibitem [{\citenamefont {D'Ariano}\ \emph {et~al.}(2020{\natexlab{a}})\citenamefont {D'Ariano}, \citenamefont {Erba},\ and\ \citenamefont {Perinotti}}]{PhysRevA.101.042118}%
  \BibitemOpen
  \bibfield  {author} {\bibinfo {author} {\bibfnamefont {G.~M.}\ \bibnamefont {D'Ariano}}, \bibinfo {author} {\bibfnamefont {M.}~\bibnamefont {Erba}},\ and\ \bibinfo {author} {\bibfnamefont {P.}~\bibnamefont {Perinotti}},\ }\bibfield  {title} {\bibinfo {title} {Classical theories with entanglement},\ }\href {https://doi.org/10.1103/PhysRevA.101.042118} {\bibfield  {journal} {\bibinfo  {journal} {Phys. Rev. A}\ }\textbf {\bibinfo {volume} {101}},\ \bibinfo {pages} {042118} (\bibinfo {year} {2020}{\natexlab{a}})}\BibitemShut {NoStop}%
\bibitem [{\citenamefont {D'Ariano}\ \emph {et~al.}(2020{\natexlab{b}})\citenamefont {D'Ariano}, \citenamefont {Erba},\ and\ \citenamefont {Perinotti}}]{d2020classicality}%
  \BibitemOpen
  \bibfield  {author} {\bibinfo {author} {\bibfnamefont {G.~M.}\ \bibnamefont {D'Ariano}}, \bibinfo {author} {\bibfnamefont {M.}~\bibnamefont {Erba}},\ and\ \bibinfo {author} {\bibfnamefont {P.}~\bibnamefont {Perinotti}},\ }\bibfield  {title} {\bibinfo {title} {Classicality without local discriminability: Decoupling entanglement and complementarity},\ }\href {https://doi.org/10.1103/physreva.102.052216} {\bibfield  {journal} {\bibinfo  {journal} {Physical Review A}\ }\textbf {\bibinfo {volume} {102}},\ \bibinfo {pages} {052216} (\bibinfo {year} {2020}{\natexlab{b}})}\BibitemShut {NoStop}%
\bibitem [{\citenamefont {Scandolo}(2019)}]{scandolo2019information}%
  \BibitemOpen
  \bibfield  {author} {\bibinfo {author} {\bibfnamefont {C.~M.}\ \bibnamefont {Scandolo}},\ }\href {https://arxiv.org/abs/1901.08054} {\bibinfo {title} {Information-theoretic foundations of thermodynamics in general probabilistic theories}} (\bibinfo {year} {2019}),\ \Eprint {https://arxiv.org/abs/1901.08054} {arXiv:1901.08054 [quant-ph]} \BibitemShut {NoStop}%
\bibitem [{\citenamefont {Chiribella}\ \emph {et~al.}(2024)\citenamefont {Chiribella}, \citenamefont {Giannelli},\ and\ \citenamefont {Scandolo}}]{Chiribella2024}%
  \BibitemOpen
  \bibfield  {author} {\bibinfo {author} {\bibfnamefont {G.}~\bibnamefont {Chiribella}}, \bibinfo {author} {\bibfnamefont {L.}~\bibnamefont {Giannelli}},\ and\ \bibinfo {author} {\bibfnamefont {C.~M.}\ \bibnamefont {Scandolo}},\ }\bibfield  {title} {\bibinfo {title} {Bell nonlocality in classical systems coexisting with other system types},\ }\href {https://doi.org/10.1103/PhysRevLett.132.190201} {\bibfield  {journal} {\bibinfo  {journal} {Phys. Rev. Lett.}\ }\textbf {\bibinfo {volume} {132}},\ \bibinfo {pages} {190201} (\bibinfo {year} {2024})}\BibitemShut {NoStop}%
\bibitem [{\citenamefont {Rolino}\ \emph {et~al.}(2025)\citenamefont {Rolino}, \citenamefont {Erba}, \citenamefont {Tosini},\ and\ \citenamefont {Perinotti}}]{Rolino_MOPTs_2025}%
  \BibitemOpen
  \bibfield  {author} {\bibinfo {author} {\bibfnamefont {D.}~\bibnamefont {Rolino}}, \bibinfo {author} {\bibfnamefont {M.}~\bibnamefont {Erba}}, \bibinfo {author} {\bibfnamefont {A.}~\bibnamefont {Tosini}},\ and\ \bibinfo {author} {\bibfnamefont {P.}~\bibnamefont {Perinotti}},\ }\bibfield  {title} {\bibinfo {title} {Minimal operational theories: classical theories with quantum features},\ }\href {https://doi.org/10.1088/1367-2630/ada850} {\bibfield  {journal} {\bibinfo  {journal} {New Journal of Physics}\ }\textbf {\bibinfo {volume} {27}},\ \bibinfo {pages} {023004} (\bibinfo {year} {2025})}\BibitemShut {NoStop}%
\bibitem [{\citenamefont {Soltani}\ \emph {et~al.}(2025)\citenamefont {Soltani}, \citenamefont {Erba}, \citenamefont {Schmid},\ and\ \citenamefont {Selby}}]{soltani2025decouplinglocalclassicalityclassical}%
  \BibitemOpen
  \bibfield  {author} {\bibinfo {author} {\bibfnamefont {S.}~\bibnamefont {Soltani}}, \bibinfo {author} {\bibfnamefont {M.}~\bibnamefont {Erba}}, \bibinfo {author} {\bibfnamefont {D.}~\bibnamefont {Schmid}},\ and\ \bibinfo {author} {\bibfnamefont {J.~H.}\ \bibnamefont {Selby}},\ }\href {https://arxiv.org/abs/2511.19266} {\bibinfo {title} {Decoupling local classicality from classical explainability: A noncontextual model for bilocal classical theory and a locally-classical but contextual theory}} (\bibinfo {year} {2025}),\ \Eprint {https://arxiv.org/abs/2511.19266} {arXiv:2511.19266 [quant-ph]} \BibitemShut {NoStop}%
\bibitem [{\citenamefont {Erba}\ and\ \citenamefont {Perinotti}(2025)}]{erba2025compositionrulequantumsystems}%
  \BibitemOpen
  \bibfield  {author} {\bibinfo {author} {\bibfnamefont {M.}~\bibnamefont {Erba}}\ and\ \bibinfo {author} {\bibfnamefont {P.}~\bibnamefont {Perinotti}},\ }\href {https://arxiv.org/abs/2411.15964} {\bibinfo {title} {The composition rule for quantum systems is not the only possible one}} (\bibinfo {year} {2025}),\ \Eprint {https://arxiv.org/abs/2411.15964} {arXiv:2411.15964 [quant-ph]} \BibitemShut {NoStop}%
\bibitem [{\citenamefont {Barnum}\ \emph {et~al.}(2020{\natexlab{a}})\citenamefont {Barnum}, \citenamefont {Graydon},\ and\ \citenamefont {Wilce}}]{barnum2020composites}%
  \BibitemOpen
  \bibfield  {author} {\bibinfo {author} {\bibfnamefont {H.}~\bibnamefont {Barnum}}, \bibinfo {author} {\bibfnamefont {M.~A.}\ \bibnamefont {Graydon}},\ and\ \bibinfo {author} {\bibfnamefont {A.}~\bibnamefont {Wilce}},\ }\bibfield  {title} {\bibinfo {title} {Composites and categories of {Euclidean Jordan} algebras},\ }\href {https://doi.org/10.22331/q-2020-11-08-359} {\bibfield  {journal} {\bibinfo  {journal} {Quantum}\ }\textbf {\bibinfo {volume} {4}},\ \bibinfo {pages} {359} (\bibinfo {year} {2020}{\natexlab{a}})}\BibitemShut {NoStop}%
\bibitem [{\citenamefont {Wick}\ \emph {et~al.}(1952)\citenamefont {Wick}, \citenamefont {Wightman},\ and\ \citenamefont {Wigner}}]{Wick_SymmetrisFundamental}%
  \BibitemOpen
  \bibfield  {author} {\bibinfo {author} {\bibfnamefont {G.~C.}\ \bibnamefont {Wick}}, \bibinfo {author} {\bibfnamefont {A.~S.}\ \bibnamefont {Wightman}},\ and\ \bibinfo {author} {\bibfnamefont {E.~P.}\ \bibnamefont {Wigner}},\ }\bibfield  {title} {\bibinfo {title} {The intrinsic parity of elementary particles},\ }\href {https://doi.org/10.1103/PhysRev.88.101} {\bibfield  {journal} {\bibinfo  {journal} {Phys. Rev.}\ }\textbf {\bibinfo {volume} {88}},\ \bibinfo {pages} {101} (\bibinfo {year} {1952})}\BibitemShut {NoStop}%
\bibitem [{\citenamefont {Giulini}(2009)}]{giulini2009superselectionrulesFundamental}%
  \BibitemOpen
  \bibfield  {author} {\bibinfo {author} {\bibfnamefont {D.}~\bibnamefont {Giulini}},\ }\href {https://arxiv.org/abs/0710.1516} {\bibinfo {title} {Superselection rules}} (\bibinfo {year} {2009}),\ \Eprint {https://arxiv.org/abs/0710.1516} {arXiv:0710.1516 [quant-ph]} \BibitemShut {NoStop}%
\bibitem [{\citenamefont {Piani}\ \emph {et~al.}(2008)\citenamefont {Piani}, \citenamefont {Horodecki},\ and\ \citenamefont {Horodecki}}]{PianiNoLocalBroadcasting_2008}%
  \BibitemOpen
  \bibfield  {author} {\bibinfo {author} {\bibfnamefont {M.}~\bibnamefont {Piani}}, \bibinfo {author} {\bibfnamefont {P.}~\bibnamefont {Horodecki}},\ and\ \bibinfo {author} {\bibfnamefont {R.}~\bibnamefont {Horodecki}},\ }\bibfield  {title} {\bibinfo {title} {No-local-broadcasting theorem for multipartite quantum correlations},\ }\href {https://doi.org/10.1103/PhysRevLett.100.090502} {\bibfield  {journal} {\bibinfo  {journal} {Phys. Rev. Lett.}\ }\textbf {\bibinfo {volume} {100}},\ \bibinfo {pages} {090502} (\bibinfo {year} {2008})}\BibitemShut {NoStop}%
\bibitem [{\citenamefont {Weilenmann}\ \emph {et~al.}(2025)\citenamefont {Weilenmann}, \citenamefont {Gisin},\ and\ \citenamefont {Sekatski}}]{Weilenmann_2025}%
  \BibitemOpen
  \bibfield  {author} {\bibinfo {author} {\bibfnamefont {M.}~\bibnamefont {Weilenmann}}, \bibinfo {author} {\bibfnamefont {N.}~\bibnamefont {Gisin}},\ and\ \bibinfo {author} {\bibfnamefont {P.}~\bibnamefont {Sekatski}},\ }\bibfield  {title} {\bibinfo {title} {Partial independence suffices to rule out real quantum theory experimentally},\ }\bibfield  {journal} {\bibinfo  {journal} {Physical Review Letters}\ }\textbf {\bibinfo {volume} {135}},\ \href {https://doi.org/10.1103/3fv7-p8cs} {10.1103/3fv7-p8cs} (\bibinfo {year} {2025})\BibitemShut {NoStop}%
\bibitem [{\citenamefont {Bell}(1964)}]{Bell1964}%
  \BibitemOpen
  \bibfield  {author} {\bibinfo {author} {\bibfnamefont {J.~S.}\ \bibnamefont {Bell}},\ }\bibfield  {title} {\bibinfo {title} {On the einstein podolsky rosen paradox},\ }\href {https://doi.org/10.1103/PhysicsPhysiqueFizika.1.195} {\bibfield  {journal} {\bibinfo  {journal} {Physics Physique Fizika}\ }\textbf {\bibinfo {volume} {1}},\ \bibinfo {pages} {195} (\bibinfo {year} {1964})}\BibitemShut {NoStop}%
\bibitem [{\citenamefont {Brunner}\ \emph {et~al.}(2014)\citenamefont {Brunner}, \citenamefont {Cavalcanti}, \citenamefont {Pironio}, \citenamefont {Scarani},\ and\ \citenamefont {Wehner}}]{Bellreview2014}%
  \BibitemOpen
  \bibfield  {author} {\bibinfo {author} {\bibfnamefont {N.}~\bibnamefont {Brunner}}, \bibinfo {author} {\bibfnamefont {D.}~\bibnamefont {Cavalcanti}}, \bibinfo {author} {\bibfnamefont {S.}~\bibnamefont {Pironio}}, \bibinfo {author} {\bibfnamefont {V.}~\bibnamefont {Scarani}},\ and\ \bibinfo {author} {\bibfnamefont {S.}~\bibnamefont {Wehner}},\ }\bibfield  {title} {\bibinfo {title} {Bell nonlocality},\ }\href {https://doi.org/10.1103/RevModPhys.86.419} {\bibfield  {journal} {\bibinfo  {journal} {Rev. Mod. Phys.}\ }\textbf {\bibinfo {volume} {86}},\ \bibinfo {pages} {419} (\bibinfo {year} {2014})}\BibitemShut {NoStop}%
\bibitem [{\citenamefont {Wiseman}\ \emph {et~al.}(2007)\citenamefont {Wiseman}, \citenamefont {Jones},\ and\ \citenamefont {Doherty}}]{steeringwiseman}%
  \BibitemOpen
  \bibfield  {author} {\bibinfo {author} {\bibfnamefont {H.~M.}\ \bibnamefont {Wiseman}}, \bibinfo {author} {\bibfnamefont {S.~J.}\ \bibnamefont {Jones}},\ and\ \bibinfo {author} {\bibfnamefont {A.~C.}\ \bibnamefont {Doherty}},\ }\bibfield  {title} {\bibinfo {title} {Steering, entanglement, nonlocality, and the einstein-podolsky-rosen paradox},\ }\href {https://doi.org/10.1103/PhysRevLett.98.140402} {\bibfield  {journal} {\bibinfo  {journal} {Phys. Rev. Lett.}\ }\textbf {\bibinfo {volume} {98}},\ \bibinfo {pages} {140402} (\bibinfo {year} {2007})}\BibitemShut {NoStop}%
\bibitem [{\citenamefont {Cavalcanti}\ \emph {et~al.}(2017)\citenamefont {Cavalcanti}, \citenamefont {Skrzypczyk},\ and\ \citenamefont {\ifmmode \check{S}\else \v{S}\fi{}upi\ifmmode~\acute{c}\else \'{c}\fi{}}}]{PhysRevLett.119.110501}%
  \BibitemOpen
  \bibfield  {author} {\bibinfo {author} {\bibfnamefont {D.}~\bibnamefont {Cavalcanti}}, \bibinfo {author} {\bibfnamefont {P.}~\bibnamefont {Skrzypczyk}},\ and\ \bibinfo {author} {\bibfnamefont {I.}~\bibnamefont {\ifmmode \check{S}\else \v{S}\fi{}upi\ifmmode~\acute{c}\else \'{c}\fi{}}},\ }\bibfield  {title} {\bibinfo {title} {All entangled states can demonstrate nonclassical teleportation},\ }\href {https://doi.org/10.1103/PhysRevLett.119.110501} {\bibfield  {journal} {\bibinfo  {journal} {Phys. Rev. Lett.}\ }\textbf {\bibinfo {volume} {119}},\ \bibinfo {pages} {110501} (\bibinfo {year} {2017})}\BibitemShut {NoStop}%
\bibitem [{\citenamefont {Wootters}(2010)}]{Wootters_2010}%
  \BibitemOpen
  \bibfield  {author} {\bibinfo {author} {\bibfnamefont {W.~K.}\ \bibnamefont {Wootters}},\ }\bibfield  {title} {\bibinfo {title} {Entanglement sharing in real-vector-space quantum theory},\ }\href {https://doi.org/10.1007/s10701-010-9488-1} {\bibfield  {journal} {\bibinfo  {journal} {Foundations of Physics}\ }\textbf {\bibinfo {volume} {42}},\ \bibinfo {pages} {19–28} (\bibinfo {year} {2010})}\BibitemShut {NoStop}%
\bibitem [{\citenamefont {Selby}\ \emph {et~al.}(2023)\citenamefont {Selby}, \citenamefont {Schmid}, \citenamefont {Wolfe}, \citenamefont {Sainz}, \citenamefont {Kunjwal},\ and\ \citenamefont {Spekkens}}]{PhysRevA.107.062203}%
  \BibitemOpen
  \bibfield  {author} {\bibinfo {author} {\bibfnamefont {J.~H.}\ \bibnamefont {Selby}}, \bibinfo {author} {\bibfnamefont {D.}~\bibnamefont {Schmid}}, \bibinfo {author} {\bibfnamefont {E.}~\bibnamefont {Wolfe}}, \bibinfo {author} {\bibfnamefont {A.~B.}\ \bibnamefont {Sainz}}, \bibinfo {author} {\bibfnamefont {R.}~\bibnamefont {Kunjwal}},\ and\ \bibinfo {author} {\bibfnamefont {R.~W.}\ \bibnamefont {Spekkens}},\ }\bibfield  {title} {\bibinfo {title} {Accessible fragments of generalized probabilistic theories, cone equivalence, and applications to witnessing nonclassicality},\ }\href {https://doi.org/10.1103/PhysRevA.107.062203} {\bibfield  {journal} {\bibinfo  {journal} {Phys. Rev. A}\ }\textbf {\bibinfo {volume} {107}},\ \bibinfo {pages} {062203} (\bibinfo {year} {2023})}\BibitemShut {NoStop}%
\bibitem [{\citenamefont {Lami}(2018)}]{lami2018nonclassicalcorrelationsquantummechanics}%
  \BibitemOpen
  \bibfield  {author} {\bibinfo {author} {\bibfnamefont {L.}~\bibnamefont {Lami}},\ }\href {https://arxiv.org/abs/1803.02902} {\bibinfo {title} {Non-classical correlations in quantum mechanics and beyond}} (\bibinfo {year} {2018}),\ \Eprint {https://arxiv.org/abs/1803.02902} {arXiv:1803.02902 [quant-ph]} \BibitemShut {NoStop}%
\bibitem [{\citenamefont {Erba}\ \emph {et~al.}(2026)\citenamefont {Erba}, \citenamefont {Perinotti},\ and\ \citenamefont {D'Ariano}}]{erba2026categorical}%
  \BibitemOpen
  \bibfield  {author} {\bibinfo {author} {\bibfnamefont {M.}~\bibnamefont {Erba}}, \bibinfo {author} {\bibfnamefont {P.}~\bibnamefont {Perinotti}},\ and\ \bibinfo {author} {\bibfnamefont {G.~M.}\ \bibnamefont {D'Ariano}},\ }\bibfield  {title} {\bibinfo {title} {{Categorical Physics I – Review of the Operational Probabilistic Theories framework}},\ }\href@noop {} {\bibfield  {journal} {\bibinfo  {journal} {\emph{forthcoming}}\ } (\bibinfo {year} {2026})}\BibitemShut {NoStop}%
\bibitem [{\citenamefont {Chiribella}\ \emph {et~al.}(2010)\citenamefont {Chiribella}, \citenamefont {D’Ariano},\ and\ \citenamefont {Perinotti}}]{Chiribella_2010}%
  \BibitemOpen
  \bibfield  {author} {\bibinfo {author} {\bibfnamefont {G.}~\bibnamefont {Chiribella}}, \bibinfo {author} {\bibfnamefont {G.~M.}\ \bibnamefont {D’Ariano}},\ and\ \bibinfo {author} {\bibfnamefont {P.}~\bibnamefont {Perinotti}},\ }\bibfield  {title} {\bibinfo {title} {Probabilistic theories with purification},\ }\bibfield  {journal} {\bibinfo  {journal} {Physical Review A}\ }\textbf {\bibinfo {volume} {81}},\ \href {https://doi.org/10.1103/physreva.81.062348} {10.1103/physreva.81.062348} (\bibinfo {year} {2010})\BibitemShut {NoStop}%
\bibitem [{\citenamefont {Schmid}\ \emph {et~al.}(2024{\natexlab{b}})\citenamefont {Schmid}, \citenamefont {Selby}, \citenamefont {Rossi}, \citenamefont {Baldijão},\ and\ \citenamefont {Sainz}}]{schmid2024shadowssubsystemsgeneralizedprobabilistic}%
  \BibitemOpen
  \bibfield  {author} {\bibinfo {author} {\bibfnamefont {D.}~\bibnamefont {Schmid}}, \bibinfo {author} {\bibfnamefont {J.~H.}\ \bibnamefont {Selby}}, \bibinfo {author} {\bibfnamefont {V.~P.}\ \bibnamefont {Rossi}}, \bibinfo {author} {\bibfnamefont {R.~D.}\ \bibnamefont {Baldijão}},\ and\ \bibinfo {author} {\bibfnamefont {A.~B.}\ \bibnamefont {Sainz}},\ }\href {https://arxiv.org/abs/2409.13024} {\bibinfo {title} {Shadows and subsystems of generalized probabilistic theories: when tomographic incompleteness is not a loophole for contextuality proofs}} (\bibinfo {year} {2024}{\natexlab{b}}),\ \Eprint {https://arxiv.org/abs/2409.13024} {arXiv:2409.13024 [quant-ph]} \BibitemShut {NoStop}%
\bibitem [{\citenamefont {Hardy}\ and\ \citenamefont {Wootters}(2011)}]{Hardy_2011}%
  \BibitemOpen
  \bibfield  {author} {\bibinfo {author} {\bibfnamefont {L.}~\bibnamefont {Hardy}}\ and\ \bibinfo {author} {\bibfnamefont {W.~K.}\ \bibnamefont {Wootters}},\ }\bibfield  {title} {\bibinfo {title} {Limited holism and real-vector-space quantum theory},\ }\href {https://doi.org/10.1007/s10701-011-9616-6} {\bibfield  {journal} {\bibinfo  {journal} {Foundations of Physics}\ }\textbf {\bibinfo {volume} {42}},\ \bibinfo {pages} {454–473} (\bibinfo {year} {2011})}\BibitemShut {NoStop}%
\bibitem [{\citenamefont {Baldijao}\ \emph {et~al.}(2022)\citenamefont {Baldijao}, \citenamefont {Krumm}, \citenamefont {Garner},\ and\ \citenamefont {Mueller}}]{Baldijao_QDarwinisminGPTs_2022}%
  \BibitemOpen
  \bibfield  {author} {\bibinfo {author} {\bibfnamefont {R.~D.}\ \bibnamefont {Baldijao}}, \bibinfo {author} {\bibfnamefont {M.}~\bibnamefont {Krumm}}, \bibinfo {author} {\bibfnamefont {A.~J.~P.}\ \bibnamefont {Garner}},\ and\ \bibinfo {author} {\bibfnamefont {M.~P.}\ \bibnamefont {Mueller}},\ }\bibfield  {title} {\bibinfo {title} {Quantum darwinism and the spreading of classical information in non-classical theories},\ }\href {https://doi.org/10.22331/q-2022-01-31-636} {\bibfield  {journal} {\bibinfo  {journal} {Quantum}\ }\textbf {\bibinfo {volume} {6}},\ \bibinfo {pages} {636} (\bibinfo {year} {2022})}\BibitemShut {NoStop}%
\bibitem [{\citenamefont {Barnum}\ \emph {et~al.}(2020{\natexlab{b}})\citenamefont {Barnum}, \citenamefont {Graydon},\ and\ \citenamefont {Wilce}}]{Barnum_CathegoriesJordanAlgebras_2020}%
  \BibitemOpen
  \bibfield  {author} {\bibinfo {author} {\bibfnamefont {H.}~\bibnamefont {Barnum}}, \bibinfo {author} {\bibfnamefont {M.~A.}\ \bibnamefont {Graydon}},\ and\ \bibinfo {author} {\bibfnamefont {A.}~\bibnamefont {Wilce}},\ }\bibfield  {title} {\bibinfo {title} {Composites and categories of euclidean jordan algebras},\ }\href {https://doi.org/10.22331/q-2020-11-08-359} {\bibfield  {journal} {\bibinfo  {journal} {Quantum}\ }\textbf {\bibinfo {volume} {4}},\ \bibinfo {pages} {359} (\bibinfo {year} {2020}{\natexlab{b}})}\BibitemShut {NoStop}%
\bibitem [{\citenamefont {Sainz}\ \emph {et~al.}(2018)\citenamefont {Sainz}, \citenamefont {Guryanova}, \citenamefont {Acín},\ and\ \citenamefont {Navascués}}]{Sainz_2018}%
  \BibitemOpen
  \bibfield  {author} {\bibinfo {author} {\bibfnamefont {A.~B.}\ \bibnamefont {Sainz}}, \bibinfo {author} {\bibfnamefont {Y.}~\bibnamefont {Guryanova}}, \bibinfo {author} {\bibfnamefont {A.}~\bibnamefont {Acín}},\ and\ \bibinfo {author} {\bibfnamefont {M.}~\bibnamefont {Navascués}},\ }\bibfield  {title} {\bibinfo {title} {{Almost-Quantum Correlations Violate the No-Restriction Hypothesis}},\ }\bibfield  {journal} {\bibinfo  {journal} {Physical Review Letters}\ }\textbf {\bibinfo {volume} {120}},\ \href {https://doi.org/10.1103/physrevlett.120.200402} {10.1103/physrevlett.120.200402} (\bibinfo {year} {2018})\BibitemShut {NoStop}%
\bibitem [{\citenamefont {Müller}\ \emph {et~al.}(2012)\citenamefont {Müller}, \citenamefont {Dahlsten},\ and\ \citenamefont {Vedral}}]{Muller_2012}%
  \BibitemOpen
  \bibfield  {author} {\bibinfo {author} {\bibfnamefont {M.~P.}\ \bibnamefont {Müller}}, \bibinfo {author} {\bibfnamefont {O.~C.~O.}\ \bibnamefont {Dahlsten}},\ and\ \bibinfo {author} {\bibfnamefont {V.}~\bibnamefont {Vedral}},\ }\bibfield  {title} {\bibinfo {title} {Unifying typical entanglement and coin tossing: on randomization in probabilistic theories},\ }\href {https://doi.org/10.1007/s00220-012-1605-x} {\bibfield  {journal} {\bibinfo  {journal} {Communications in Mathematical Physics}\ }\textbf {\bibinfo {volume} {316}},\ \bibinfo {pages} {441–487} (\bibinfo {year} {2012})}\BibitemShut {NoStop}%
\bibitem [{\citenamefont {Hardy}(2009)}]{hardy2009foliableoperationalstructuresgeneral}%
  \BibitemOpen
  \bibfield  {author} {\bibinfo {author} {\bibfnamefont {L.}~\bibnamefont {Hardy}},\ }\href {https://arxiv.org/abs/0912.4740} {\bibinfo {title} {Foliable operational structures for general probabilistic theories}} (\bibinfo {year} {2009}),\ \Eprint {https://arxiv.org/abs/0912.4740} {arXiv:0912.4740 [quant-ph]} \BibitemShut {NoStop}%
\bibitem [{\citenamefont {Wolfe}\ \emph {et~al.}(2020)\citenamefont {Wolfe}, \citenamefont {Schmid}, \citenamefont {Sainz}, \citenamefont {Kunjwal},\ and\ \citenamefont {Spekkens}}]{Wolfe2020quantifyingbell}%
  \BibitemOpen
  \bibfield  {author} {\bibinfo {author} {\bibfnamefont {E.}~\bibnamefont {Wolfe}}, \bibinfo {author} {\bibfnamefont {D.}~\bibnamefont {Schmid}}, \bibinfo {author} {\bibfnamefont {A.~B.}\ \bibnamefont {Sainz}}, \bibinfo {author} {\bibfnamefont {R.}~\bibnamefont {Kunjwal}},\ and\ \bibinfo {author} {\bibfnamefont {R.~W.}\ \bibnamefont {Spekkens}},\ }\bibfield  {title} {\bibinfo {title} {Quantifying {B}ell: the {R}esource {T}heory of {N}onclassicality of {C}ommon-{C}ause {B}oxes},\ }\href {https://doi.org/10.22331/q-2020-06-08-280} {\bibfield  {journal} {\bibinfo  {journal} {{Quantum}}\ }\textbf {\bibinfo {volume} {4}},\ \bibinfo {pages} {280} (\bibinfo {year} {2020})}\BibitemShut {NoStop}%
\bibitem [{\citenamefont {Schmid}\ \emph {et~al.}(2020)\citenamefont {Schmid}, \citenamefont {Rosset},\ and\ \citenamefont {Buscemi}}]{Schmid2020typeindependent}%
  \BibitemOpen
  \bibfield  {author} {\bibinfo {author} {\bibfnamefont {D.}~\bibnamefont {Schmid}}, \bibinfo {author} {\bibfnamefont {D.}~\bibnamefont {Rosset}},\ and\ \bibinfo {author} {\bibfnamefont {F.}~\bibnamefont {Buscemi}},\ }\bibfield  {title} {\bibinfo {title} {The type-independent resource theory of local operations and shared randomness},\ }\href {https://doi.org/10.22331/q-2020-04-30-262} {\bibfield  {journal} {\bibinfo  {journal} {{Quantum}}\ }\textbf {\bibinfo {volume} {4}},\ \bibinfo {pages} {262} (\bibinfo {year} {2020})}\BibitemShut {NoStop}%
\bibitem [{\citenamefont {Schmid}\ \emph {et~al.}(2023)\citenamefont {Schmid}, \citenamefont {Fraser}, \citenamefont {Kunjwal}, \citenamefont {Sainz}, \citenamefont {Wolfe},\ and\ \citenamefont {Spekkens}}]{Schmid2023understanding}%
  \BibitemOpen
  \bibfield  {author} {\bibinfo {author} {\bibfnamefont {D.}~\bibnamefont {Schmid}}, \bibinfo {author} {\bibfnamefont {T.~C.}\ \bibnamefont {Fraser}}, \bibinfo {author} {\bibfnamefont {R.}~\bibnamefont {Kunjwal}}, \bibinfo {author} {\bibfnamefont {A.~B.}\ \bibnamefont {Sainz}}, \bibinfo {author} {\bibfnamefont {E.}~\bibnamefont {Wolfe}},\ and\ \bibinfo {author} {\bibfnamefont {R.~W.}\ \bibnamefont {Spekkens}},\ }\bibfield  {title} {\bibinfo {title} {Understanding the interplay of entanglement and nonlocality: motivating and developing a new branch of entanglement theory},\ }\href {https://doi.org/10.22331/q-2023-12-04-1194} {\bibfield  {journal} {\bibinfo  {journal} {{Quantum}}\ }\textbf {\bibinfo {volume} {7}},\ \bibinfo {pages} {1194} (\bibinfo {year} {2023})}\BibitemShut {NoStop}%
\bibitem [{\citenamefont {Nielsen}\ and\ \citenamefont {Chuang}(2010)}]{nielsen2010quantum}%
  \BibitemOpen
  \bibfield  {author} {\bibinfo {author} {\bibfnamefont {M.~A.}\ \bibnamefont {Nielsen}}\ and\ \bibinfo {author} {\bibfnamefont {I.~L.}\ \bibnamefont {Chuang}},\ }\href@noop {} {\emph {\bibinfo {title} {Quantum computation and quantum information}}}\ (\bibinfo  {publisher} {Cambridge university press},\ \bibinfo {year} {2010})\BibitemShut {NoStop}%
\bibitem [{\citenamefont {Caves}\ \emph {et~al.}(2000)\citenamefont {Caves}, \citenamefont {Fuchs},\ and\ \citenamefont {Rungta}}]{caves2000entanglementformationarbitrarystate}%
  \BibitemOpen
  \bibfield  {author} {\bibinfo {author} {\bibfnamefont {C.~M.}\ \bibnamefont {Caves}}, \bibinfo {author} {\bibfnamefont {C.~A.}\ \bibnamefont {Fuchs}},\ and\ \bibinfo {author} {\bibfnamefont {P.}~\bibnamefont {Rungta}},\ }\href {https://arxiv.org/abs/quant-ph/0009063} {\bibinfo {title} {Entanglement of formation of an arbitrary state of two rebits}} (\bibinfo {year} {2000}),\ \Eprint {https://arxiv.org/abs/quant-ph/0009063} {arXiv:quant-ph/0009063 [quant-ph]} \BibitemShut {NoStop}%
\bibitem [{\citenamefont {Piani}(2016)}]{piani2016localbroadcastingquantumcorrelations}%
  \BibitemOpen
  \bibfield  {author} {\bibinfo {author} {\bibfnamefont {M.}~\bibnamefont {Piani}},\ }\href {https://arxiv.org/abs/1608.02650} {\bibinfo {title} {Local broadcasting of quantum correlations}} (\bibinfo {year} {2016}),\ \Eprint {https://arxiv.org/abs/1608.02650} {arXiv:1608.02650 [quant-ph]} \BibitemShut {NoStop}%
\bibitem [{\citenamefont {Zjawin}\ \emph {et~al.}(2023)\citenamefont {Zjawin}, \citenamefont {Schmid}, \citenamefont {Hoban},\ and\ \citenamefont {Sainz}}]{Zjawin2023quantifyingepr}%
  \BibitemOpen
  \bibfield  {author} {\bibinfo {author} {\bibfnamefont {B.}~\bibnamefont {Zjawin}}, \bibinfo {author} {\bibfnamefont {D.}~\bibnamefont {Schmid}}, \bibinfo {author} {\bibfnamefont {M.~J.}\ \bibnamefont {Hoban}},\ and\ \bibinfo {author} {\bibfnamefont {A.~B.}\ \bibnamefont {Sainz}},\ }\bibfield  {title} {\bibinfo {title} {Quantifying {EPR}: the resource theory of nonclassicality of common-cause assemblages},\ }\href {https://doi.org/10.22331/q-2023-02-16-926} {\bibfield  {journal} {\bibinfo  {journal} {{Quantum}}\ }\textbf {\bibinfo {volume} {7}},\ \bibinfo {pages} {926} (\bibinfo {year} {2023})}\BibitemShut {NoStop}%
\bibitem [{\citenamefont {Bennett}\ and\ \citenamefont {Wiesner}(1992)}]{BennettWiesner1992}%
  \BibitemOpen
  \bibfield  {author} {\bibinfo {author} {\bibfnamefont {C.~H.}\ \bibnamefont {Bennett}}\ and\ \bibinfo {author} {\bibfnamefont {S.~J.}\ \bibnamefont {Wiesner}},\ }\bibfield  {title} {\bibinfo {title} {Communication via one- and two-particle operators on einstein--podolsky--rosen states},\ }\href@noop {} {\bibfield  {journal} {\bibinfo  {journal} {Phys. Rev. Lett.}\ }\textbf {\bibinfo {volume} {69}},\ \bibinfo {pages} {2881} (\bibinfo {year} {1992})}\BibitemShut {NoStop}%
\bibitem [{\citenamefont {Horodecki}\ \emph {et~al.}(2001)\citenamefont {Horodecki}, \citenamefont {Horodecki},\ and\ \citenamefont {Horodecki}}]{HorodeckiDenseCoding2001}%
  \BibitemOpen
  \bibfield  {author} {\bibinfo {author} {\bibfnamefont {M.}~\bibnamefont {Horodecki}}, \bibinfo {author} {\bibfnamefont {P.}~\bibnamefont {Horodecki}},\ and\ \bibinfo {author} {\bibfnamefont {R.}~\bibnamefont {Horodecki}},\ }\bibfield  {title} {\bibinfo {title} {Dense coding, quantum teleportation and entanglement},\ }\href@noop {} {\bibfield  {journal} {\bibinfo  {journal} {Phys. Rev. A}\ }\textbf {\bibinfo {volume} {63}},\ \bibinfo {pages} {022310} (\bibinfo {year} {2001})}\BibitemShut {NoStop}%
\bibitem [{\citenamefont {Nayak}\ and\ \citenamefont {Yuen}(2023)}]{Nayak_2023}%
  \BibitemOpen
  \bibfield  {author} {\bibinfo {author} {\bibfnamefont {A.}~\bibnamefont {Nayak}}\ and\ \bibinfo {author} {\bibfnamefont {H.}~\bibnamefont {Yuen}},\ }\bibfield  {title} {\bibinfo {title} {Rigidity of superdense coding},\ }\href {https://doi.org/10.1145/3593593} {\bibfield  {journal} {\bibinfo  {journal} {ACM Transactions on Quantum Computing}\ }\textbf {\bibinfo {volume} {4}},\ \bibinfo {pages} {1–39} (\bibinfo {year} {2023})}\BibitemShut {NoStop}%
\bibitem [{\citenamefont {Erba}\ \emph {et~al.}(2024)\citenamefont {Erba}, \citenamefont {Perinotti}, \citenamefont {Rolino},\ and\ \citenamefont {Tosini}}]{Rolino2024}%
  \BibitemOpen
  \bibfield  {author} {\bibinfo {author} {\bibfnamefont {M.}~\bibnamefont {Erba}}, \bibinfo {author} {\bibfnamefont {P.}~\bibnamefont {Perinotti}}, \bibinfo {author} {\bibfnamefont {D.}~\bibnamefont {Rolino}},\ and\ \bibinfo {author} {\bibfnamefont {A.}~\bibnamefont {Tosini}},\ }\bibfield  {title} {\bibinfo {title} {Measurement incompatibility is strictly stronger than disturbance},\ }\href {https://doi.org/10.1103/PhysRevA.109.022239} {\bibfield  {journal} {\bibinfo  {journal} {Phys. Rev. A}\ }\textbf {\bibinfo {volume} {109}},\ \bibinfo {pages} {022239} (\bibinfo {year} {2024})}\BibitemShut {NoStop}%
\bibitem [{\citenamefont {DiVincenzo}\ \emph {et~al.}(2003)\citenamefont {DiVincenzo}, \citenamefont {Hayden},\ and\ \citenamefont {Terhal}}]{DiVincenzo_2003_DataHidingPerfectCase}%
  \BibitemOpen
  \bibfield  {author} {\bibinfo {author} {\bibfnamefont {D.~P.}\ \bibnamefont {DiVincenzo}}, \bibinfo {author} {\bibfnamefont {P.}~\bibnamefont {Hayden}},\ and\ \bibinfo {author} {\bibfnamefont {B.~M.}\ \bibnamefont {Terhal}},\ }\bibfield  {title} {\bibinfo {title} {Hiding quantum data},\ }\href {https://doi.org/10.1023/a:1026013201376} {\bibfield  {journal} {\bibinfo  {journal} {Foundations of Physics}\ }\textbf {\bibinfo {volume} {33}},\ \bibinfo {pages} {1629–1647} (\bibinfo {year} {2003})}\BibitemShut {NoStop}%
\bibitem [{\citenamefont {Terhal}\ \emph {et~al.}(2001)\citenamefont {Terhal}, \citenamefont {DiVincenzo},\ and\ \citenamefont {Leung}}]{Terhal_DataHidingFirst}%
  \BibitemOpen
  \bibfield  {author} {\bibinfo {author} {\bibfnamefont {B.~M.}\ \bibnamefont {Terhal}}, \bibinfo {author} {\bibfnamefont {D.~P.}\ \bibnamefont {DiVincenzo}},\ and\ \bibinfo {author} {\bibfnamefont {D.~W.}\ \bibnamefont {Leung}},\ }\bibfield  {title} {\bibinfo {title} {Hiding bits in bell states},\ }\href {https://doi.org/10.1103/PhysRevLett.86.5807} {\bibfield  {journal} {\bibinfo  {journal} {Phys. Rev. Lett.}\ }\textbf {\bibinfo {volume} {86}},\ \bibinfo {pages} {5807} (\bibinfo {year} {2001})}\BibitemShut {NoStop}%
\bibitem [{\citenamefont {DiVincenzo}\ \emph {et~al.}(2002)\citenamefont {DiVincenzo}, \citenamefont {Leung},\ and\ \citenamefont {Terhal}}]{DiVincenzo_2002_DataHidingNoquantumPerfectCase}%
  \BibitemOpen
  \bibfield  {author} {\bibinfo {author} {\bibfnamefont {D.}~\bibnamefont {DiVincenzo}}, \bibinfo {author} {\bibfnamefont {D.}~\bibnamefont {Leung}},\ and\ \bibinfo {author} {\bibfnamefont {B.}~\bibnamefont {Terhal}},\ }\bibfield  {title} {\bibinfo {title} {Quantum data hiding},\ }\href {https://doi.org/10.1109/18.985948} {\bibfield  {journal} {\bibinfo  {journal} {IEEE Transactions on Information Theory}\ }\textbf {\bibinfo {volume} {48}},\ \bibinfo {pages} {580–598} (\bibinfo {year} {2002})}\BibitemShut {NoStop}%
\bibitem [{\citenamefont {Matthews}\ \emph {et~al.}(2009)\citenamefont {Matthews}, \citenamefont {Wehner},\ and\ \citenamefont {Winter}}]{Matthews_2009_DistinguishabilityQStatesUnderRestrictedMeasurements}%
  \BibitemOpen
  \bibfield  {author} {\bibinfo {author} {\bibfnamefont {W.}~\bibnamefont {Matthews}}, \bibinfo {author} {\bibfnamefont {S.}~\bibnamefont {Wehner}},\ and\ \bibinfo {author} {\bibfnamefont {A.}~\bibnamefont {Winter}},\ }\bibfield  {title} {\bibinfo {title} {Distinguishability of quantum states under restricted families of measurements with an application to quantum data hiding},\ }\href {https://doi.org/10.1007/s00220-009-0890-5} {\bibfield  {journal} {\bibinfo  {journal} {Communications in Mathematical Physics}\ }\textbf {\bibinfo {volume} {291}},\ \bibinfo {pages} {813–843} (\bibinfo {year} {2009})}\BibitemShut {NoStop}%
\bibitem [{\citenamefont {Lami}\ \emph {et~al.}(2018)\citenamefont {Lami}, \citenamefont {Palazuelos},\ and\ \citenamefont {Winter}}]{Lami_2018_DataHidingGPTs}%
  \BibitemOpen
  \bibfield  {author} {\bibinfo {author} {\bibfnamefont {L.}~\bibnamefont {Lami}}, \bibinfo {author} {\bibfnamefont {C.}~\bibnamefont {Palazuelos}},\ and\ \bibinfo {author} {\bibfnamefont {A.}~\bibnamefont {Winter}},\ }\bibfield  {title} {\bibinfo {title} {Ultimate data hiding in quantum mechanics and beyond},\ }\href {https://doi.org/10.1007/s00220-018-3154-4} {\bibfield  {journal} {\bibinfo  {journal} {Communications in Mathematical Physics}\ }\textbf {\bibinfo {volume} {361}},\ \bibinfo {pages} {661–708} (\bibinfo {year} {2018})}\BibitemShut {NoStop}%
\bibitem [{\citenamefont {Schmid}\ \emph {et~al.}(2021)\citenamefont {Schmid}, \citenamefont {Selby}, \citenamefont {Wolfe}, \citenamefont {Kunjwal},\ and\ \citenamefont {Spekkens}}]{SchmidSimplex}%
  \BibitemOpen
  \bibfield  {author} {\bibinfo {author} {\bibfnamefont {D.}~\bibnamefont {Schmid}}, \bibinfo {author} {\bibfnamefont {J.~H.}\ \bibnamefont {Selby}}, \bibinfo {author} {\bibfnamefont {E.}~\bibnamefont {Wolfe}}, \bibinfo {author} {\bibfnamefont {R.}~\bibnamefont {Kunjwal}},\ and\ \bibinfo {author} {\bibfnamefont {R.~W.}\ \bibnamefont {Spekkens}},\ }\bibfield  {title} {\bibinfo {title} {Characterization of noncontextuality in the framework of generalized probabilistic theories},\ }\href {https://doi.org/10.1103/PRXQuantum.2.010331} {\bibfield  {journal} {\bibinfo  {journal} {PRX Quantum}\ }\textbf {\bibinfo {volume} {2}},\ \bibinfo {pages} {010331} (\bibinfo {year} {2021})}\BibitemShut {NoStop}%
\bibitem [{\citenamefont {M\"uller}\ and\ \citenamefont {Garner}(2023)}]{MullerGarner}%
  \BibitemOpen
  \bibfield  {author} {\bibinfo {author} {\bibfnamefont {M.~P.}\ \bibnamefont {M\"uller}}\ and\ \bibinfo {author} {\bibfnamefont {A.~J.~P.}\ \bibnamefont {Garner}},\ }\bibfield  {title} {\bibinfo {title} {Testing quantum theory by generalizing noncontextuality},\ }\href {https://doi.org/10.1103/PhysRevX.13.041001} {\bibfield  {journal} {\bibinfo  {journal} {Phys. Rev. X}\ }\textbf {\bibinfo {volume} {13}},\ \bibinfo {pages} {041001} (\bibinfo {year} {2023})}\BibitemShut {NoStop}%
\bibitem [{\citenamefont {Zhang}\ \emph {et~al.}(2025{\natexlab{a}})\citenamefont {Zhang}, \citenamefont {Schmid}, \citenamefont {Yīng},\ and\ \citenamefont {Spekkens}}]{zhang2025reassessingboundaryclassicalnonclassical}%
  \BibitemOpen
  \bibfield  {author} {\bibinfo {author} {\bibfnamefont {Y.}~\bibnamefont {Zhang}}, \bibinfo {author} {\bibfnamefont {D.}~\bibnamefont {Schmid}}, \bibinfo {author} {\bibfnamefont {Y.}~\bibnamefont {Yīng}},\ and\ \bibinfo {author} {\bibfnamefont {R.~W.}\ \bibnamefont {Spekkens}},\ }\href {https://arxiv.org/abs/2503.05884} {\bibinfo {title} {Reassessing the boundary between classical and nonclassical for individual quantum processes}} (\bibinfo {year} {2025}{\natexlab{a}}),\ \Eprint {https://arxiv.org/abs/2503.05884} {arXiv:2503.05884 [quant-ph]} \BibitemShut {NoStop}%
\bibitem [{\citenamefont {Zhang}\ \emph {et~al.}(2025{\natexlab{b}})\citenamefont {Zhang}, \citenamefont {Y{\=\i}ng},\ and\ \citenamefont {Schmid}}]{zhang2025quantifiers}%
  \BibitemOpen
  \bibfield  {author} {\bibinfo {author} {\bibfnamefont {Y.}~\bibnamefont {Zhang}}, \bibinfo {author} {\bibfnamefont {Y.}~\bibnamefont {Y{\=\i}ng}},\ and\ \bibinfo {author} {\bibfnamefont {D.}~\bibnamefont {Schmid}},\ }\bibfield  {title} {\bibinfo {title} {Quantifiers and witnesses for the nonclassicality of measurements and of states},\ }\href@noop {} {\bibfield  {journal} {\bibinfo  {journal} {arXiv preprint arXiv:2504.02944}\ } (\bibinfo {year} {2025}{\natexlab{b}})}\BibitemShut {NoStop}%
\bibitem [{\citenamefont {Barnum}\ \emph {et~al.}(2023)\citenamefont {Barnum}, \citenamefont {Graydon},\ and\ \citenamefont {Wilce}}]{Barnum_2023}%
  \BibitemOpen
  \bibfield  {author} {\bibinfo {author} {\bibfnamefont {H.}~\bibnamefont {Barnum}}, \bibinfo {author} {\bibfnamefont {M.~A.}\ \bibnamefont {Graydon}},\ and\ \bibinfo {author} {\bibfnamefont {A.}~\bibnamefont {Wilce}},\ }\bibfield  {title} {\bibinfo {title} {Locally tomographic shadows (extended abstract)},\ }\href {https://doi.org/10.4204/eptcs.384.3} {\bibfield  {journal} {\bibinfo  {journal} {Electronic Proceedings in Theoretical Computer Science}\ }\textbf {\bibinfo {volume} {384}},\ \bibinfo {pages} {47–57} (\bibinfo {year} {2023})}\BibitemShut {NoStop}%
\bibitem [{\citenamefont {Roman}(2008)}]{Roman2008}%
  \BibitemOpen
  \bibfield  {author} {\bibinfo {author} {\bibfnamefont {S.}~\bibnamefont {Roman}},\ }\href {https://doi.org/10.1007/978-0-387-72831-5} {\emph {\bibinfo {title} {Advanced Linear Algebra}}},\ \bibinfo {edition} {3rd}\ ed.,\ \bibinfo {series} {Graduate Texts in Mathematics}, Vol.\ \bibinfo {volume} {135}\ (\bibinfo  {publisher} {Springer},\ \bibinfo {address} {New York},\ \bibinfo {year} {2008})\BibitemShut {NoStop}%
\end{thebibliography}%

\appendix

\section{Consequences of composition requirements}
\label{compositionproofs}

\subsection{ Proof of Lemma~\texorpdfstring{\ref{lemma:AtBplusC}}{1}}
\label{app: proofOfLemmaAtBplusC}

Here we present a proof of Lemma~\ref{lemma:AtBplusC}, which let us recall it states the following: \\

Lemma~\ref{lemma:AtBplusC}:  in any composition of GPT systems $\mathcal{A}$ and $\mathcal{B}$ that satisfies Def.~\ref{def: CompositionRequirements},  the following isomorphisms hold:
\begin{equation}
\label{eq:TPStates_app}
        \mathrm{Span}{\lbrace\omega^{\m{A}}\boxtimes\nu^{\m{B}}\mid\omega^{\m{A}}\in S_{\mathcal{A}},\nu^{\m{B}}\in S_{\mathcal{B}}\rbrace}
        \cong
        A\otimes B
    \end{equation}
    and 
    \begin{equation}
    \label{eq:TPEffects_app}
        \mathrm{Span}{\lbrace e^{\m{A}}\boxtimes f^{\m{B}}\mid e^{\m{A}}\in E_{\mathcal{A}},f^{\m{B}}\in E_{\mathcal{B}}\rbrace}
        \cong
        A^*\otimes B^*
        ,
    \end{equation}
and the finite real vector space $AB$ has the form
\begin{align}
    AB \cong A\otimes B\oplus C.
    \label{eq:AssumpComposition_App}
\end{align}
A similar decomposition holds for $AB^*$ and the composition is tomographically local iff $C=\{0\}$.

\begin{proof}
    (Adapted from Ref.~\cite{Barnum_CathegoriesJordanAlgebras_2020}.) Some basic facts about (finite dimensional) linear algebra are useful here\footnote{We believe that a similar claim as Lemma~\ref{lemma:AtBplusC} is also valid for GPTs in which $A$ or $B$ are infinite dimensional. Since we, however, focus on the finite dimensional case, we provide the proof for the case of interest here.}. 

\begin{enumerate} 
\item[(I)] The universal property of the tensor product tells us that given any bilinear map $B:U\times V\to W$ we can factor this (uniquely up to unique isomorphism) through the tensor product, i.e., there exists  a linear map $\widetilde{B}:U\otimes V\to W$ such that $B(u,v)=\tilde{B}(u\otimes v)$ for all $u\in U$ and $v\in V$.
\item[(II)] That any subspace $W\subseteq V$ has a complement $\overline{W}$ such that $W\oplus \overline{W}=V$. Hence, for any injective linear map $L:U\to V$ we have that $U\cong\mathsf{Im}(L)\subseteq V$ and, hence, that $V= \mathsf{Im}(L)\oplus \overline{\mathsf{Im}(L)}\cong U\oplus C$ for some $C\cong\overline{\mathsf{Im}(L)}$.
\end{enumerate}

Now, let us apply these to the situation at hand. We note that $\boxtimes:A\times B\to AB$ is a bilinear map, and hence, following (I), we can consider the map $\widetilde{\boxtimes}:A\otimes B\to AB$.
If we assume that this is injective, then, following (II) we have that $AB\cong (A\otimes B)\oplus C$ which is our desired result. Moreover, since  $a\boxtimes b=\widetilde{\boxtimes}(a\otimes b)$ for all $a\in A$ and $b\in B$, we get $\mathsf{Span}\{\omega^{\m{A}}\boxtimes \omega^{\m{B}}\mid \omega^{\m{A}}\in \m{S_A},\omega^{\m{B}}\in \m{S_B}\}=\mathsf{Span}\{\widetilde{\boxtimes}(\omega^{\m{A}}\otimes \omega^{\m{B}})\mid \omega^{\m{A}}\in \m{S_A},\omega^{\m{B}}\in \m{S_B}\}=\mathsf{Im}[\widetilde{\boxtimes}]$. Again, provided $\widetilde{\boxtimes}$ is injective, $\mathsf{Im}[\widetilde{\boxtimes}]\cong A\otimes B$, which implies Expression~\eqref{eq:TPStates_app}.   

All that remains, is to show that $\widetilde{\boxtimes}$ is indeed injective as we assumed above. To begin recall from item 1 of Def.~\ref{def: CompositionRequirements}  that $e^\mathcal{A}\boxtimes e^\mathcal{B}(\omega^\mathcal{A}\boxtimes\omega^\mathcal{B}) = e^\mathcal{A}(\omega^\mathcal{A})e^\mathcal{B}(\omega^\mathcal{B})$. Using the definition of $\widetilde{\boxtimes}$ we can write this as $e^\mathcal{A}\boxtimes e^\mathcal{B}(\widetilde{\boxtimes}(\omega^\mathcal{A}\otimes\omega^\mathcal{B})) = e^\mathcal{A}(\omega^\mathcal{A})e^\mathcal{B}(\omega^\mathcal{B}) $. Similarly, for the standard tensor product we have that $e^\mathcal{A}\otimes e^\mathcal{B}(\omega^\mathcal{A}\otimes\omega^\mathcal{B}) = e^\mathcal{A}(\omega^\mathcal{A})e^\mathcal{B}(\omega^\mathcal{B})$. Putting these together we therefore have that $e^\mathcal{A}\otimes e^\mathcal{B}(\omega^\mathcal{A}\otimes\omega^\mathcal{B})=e^\mathcal{A}(\omega^\mathcal{A})e^\mathcal{B}(\omega^\mathcal{B}) = e^\mathcal{A}\boxtimes^*e^\mathcal{B}(\widetilde{\boxtimes}(\omega^\mathcal{A}\otimes\omega^\mathcal{B}))$.
As this is true for all $\omega^\mathcal{A}$ and $\omega^\mathcal{B}$ and these span $A\otimes B$, it is therefore the case that \begin{equation}\label{eq:Lem1}
e^\mathcal{A}\boxtimes e^\mathcal{B}(\widetilde{\boxtimes}(\_))  =e^\mathcal{A}\otimes e^\mathcal{B}(\_).
\end{equation}

Next, consider some $X$ such that $\widetilde{\boxtimes}(X)=0$. Substituting this $X$ into Eq.~\eqref{eq:Lem1} gives us that
\begin{eqnarray}
    e^\mathcal{A}\otimes e^\mathcal{B}(X)&=&e^\mathcal{A}\boxtimes e^\mathcal{B}(\widetilde{\boxtimes}(X)) \nonumber\\
    &=&0
\end{eqnarray}
for all $e^\mathcal{A}$ and $e^\mathcal{B}$ and as the $e^\mathcal{A}\otimes e^\mathcal{B}$ span $(A\otimes B)^*$ this can only be the case if $X=0$. Hence, we have shown that $\widetilde{\boxtimes}(X)=0 \iff X=0$ which is to say that $\widetilde{\boxtimes}$ is injective as we assumed. A completely analogous argument leads to the analogous decomposition for $AB^*$ and also implies Expression~\eqref{eq:TPEffects_app}.
\end{proof}
\blk

\subsection{Proof of Lemma \texorpdfstring{\ref{lemma:AtensorB+Holistic}}{2}} \label{lemma2proof}

 Here we present a proof of Lemma~\ref{lemma:AtensorB+Holistic}, which let is recall it states the following: \\

Lemma~\ref{lemma:AtensorB+Holistic}: 
The real vector space associated with any composite $AB$ of two GPT systems $\mathcal{A}$ and $\mathcal{B}$ has the form
    \begin{align}
        AB = (AB_{\otimes})\oplus H_S,
    \end{align}
    and
    \begin{align}
        (AB)^* = (AB^*_{\otimes})\oplus H_E.
    \end{align}
\begin{proof}
We begin by recalling some basic facts from linear algebra and fixing notation. Let $V$ be a real vector space and let $U,W\subseteq V$ be subspaces. The sum $U+W$ is defined as
\begin{align}
U+W := \{u+w \mid u\in U,\ w\in W\}.
\end{align}
Thus, $V=U+W$ means that every vector $v\in V$ can be written as $v=u+w$ for some $u\in U$ and $w\in W$.
In general, such a decomposition needs not be unique. By contrast, $V=U\oplus W$ means that
\begin{align}
\label{eq:DirectsumCondtion}
V=U+W
\quad\text{and}\quad
U\cap W=\{0_V\},
\end{align}
where $0_V$ denotes the zero vector of $V$. In the direct sum case, the decomposition $v=u+w$ is unique.

For a subspace $W\subseteq V$, its annihilator, denoted with a $0$ superscript, $W^0\subseteq V^*$ is defined by
\begin{align}
W^0 := \{\varphi\in V^* \mid \varphi(w)=0 \ \text{for all } w\in W\}.
\end{align}
We will use the following standard identities for annihilators~\cite{Roman2008}:
\begin{enumerate}
    \item[(i)] $\{0_V\}^0 = V^*$  (holds in  finite and infinite dimension);
    \item[(ii)] $(W\cap U)^0 = W^0 + U^0$ (holds in finite and infinite dimension, \cite[Thm.~3.14]{Roman2008}):\footnote{
 Note that throughout this work, $V^*$ is the \emph{algebraic dual} of a vector space $V$, i.e., the vector space of all linear functionals on $V$. If one were instead to consider a restricted dual (for instance, the continuous dual associated with a chosen topology), the relation $(W\cap U)^0 = W^0 + U^0$ need not hold in general. If $V$ is finite-dimensional (as for all GPT vector spaces we consider), we can disregard this distinction.}
    \item[(iii)] $(V^0)^0 = V$ (requires $\dim V < \infty$).
\end{enumerate}

We now prove the statement of the lemma. To establish $AB = AB_{\otimes}\oplus H_S,$ it suffices to show that $AB = AB_{\otimes} + H_S$ and $AB_{\otimes}\cap H_S=\{0_{AB}\}$ (recall Eq.~\eqref{eq:DirectsumCondtion}).
We first prove the sum condition, assuming for the moment that the intersection is trivial.

By definition of the holistic subspaces, $H_E = (AB_{\otimes})^0$, and $H_S = \big((AB)^*_{\otimes}\big)^0$. Applying identity (ii) with $W=AB_{\otimes}$ and $U=H_S$, we obtain
\begin{align}
(AB_{\otimes})^0 + H_S^0 = (AB_{\otimes}\cap H_S)^0.
\end{align}
Assuming $AB_{\otimes}\cap H_S=\{0_{AB}\}$, identity (i) yields
\begin{align}
(AB_{\otimes})^0 + H_S^0 = (AB)^*.
\end{align}
Using the identification $(AB_{\otimes})^0=H_E$ and, by finite dimensionality and identity (iii),
\begin{align}
H_S^0 = \big((AB)^*_{\otimes}\big)^{00} = (AB)^*_{\otimes},
\end{align}
we conclude that
\begin{align}
(AB)^* = (AB)^*_{\otimes} + H_E.
\end{align}
An analogous argument, exchanging the roles of $AB$ and $(AB)^*$, yields $AB = AB_{\otimes} + H_S$.

It remains to prove that $AB_{\otimes}\cap H_S=\{0_{AB}\}$. Let $v\in AB_{\otimes}\cap H_S$.
Since $v\in H_S$, by definition it is annihilated by all product effects:
\begin{align}
(e\boxtimes f)(v)=0
\qquad
\forall\, e\in E_A,\ \forall\, f\in E_B.
\end{align}
By Lemma~\ref{lemma:AtBplusC} and the composition requirements (Def.~\ref{def: CompositionRequirements}), we have $AB_{\otimes}\cong A\otimes B$
and $AB_{\otimes}^*=\mathrm{span}\{e\boxtimes f\}\cong A^*\otimes B^*$.
In particular, in finite dimension\footnote{The proof that product effects separate vectors in $AB_{\otimes}$ we provide here requires finite dimensions, which is the case we are interested in this work. The separation property, however,  holds in infinite dimensional systems as well. Nevertheless,   Lemma~\ref{lemma:AtensorB+Holistic} will not in general be true for systems with infinite-dimensional vector spaces, due to the requirement $(V^0)^0=V$.} product effects separate points of $AB_{\otimes}$: for any two distinct vectors
$v,v'\in AB_{\otimes}$ there exist $e\in E_A$ and $f\in E_B$ such that
$(e\boxtimes f)(v)\neq(e\boxtimes f)(v')$.
Applying this separation property with $v'=0_{AB}$, we conclude that if
$(e\boxtimes f)(v)=0$ for all product effects, then necessarily $v=0_{AB}$.
Hence,$AB_{\otimes}\cap H_S=\{0_{AB}\}$.

The proof that $(AB)^*_{\otimes}\cap H_E=\{0_{(AB)^*}\}$ is entirely analogous. Therefore,
\begin{align}
AB = AB_{\otimes}\oplus H_S,
\qquad
(AB)^* = (AB)^*_{\otimes}\oplus H_E,
\end{align}
as claimed.
\end{proof}

\section{Facts about holistic subspaces and the projections \texorpdfstring{$\Pi_{\rm TL}$}{PiTL} and \texorpdfstring{$\Pi_{\rm TNL}$}{PiTNL}}

\subsection{Relationship between $H_S$ and $H_E$}\label{Hrels}

Any inner product $\langle\cdot\,,\cdot\rangle$ on $AB$ which makes  the subspaces $AB_{\otimes} := {\rm Span}[\{\omega^{\m{A}}\otimes\omega^B\mid \omega^{\m{A}}\in \m{S_A},\omega^{\m{B}}\in\m{S_B}\}]$ and $H_S$ orthogonal
induces a  basis-independent isomorphism between $H_E$ and $H_S^*$, as we now show.

Suppose that $AB$ carries an inner product $\langle\,\cdot\,,\,\cdot\rangle$ satisfying
\begin{align}
\label{eq:innerproductOrthogonalizesHSAotimesB}
    \langle v_{TL},h\rangle=0\,\quad \forall v_{TL}\in AB_{\otimes}\,\, \&\,\,\forall h\in H_S.
\end{align}
For any $v\in AB$, the functional $\langle v\,,\,\cdot\,\rangle$ belongs to $(AB)^*$. We now show that 
\begin{align}
\label{eq:HSandHEDual}
    \langle u\,,\,\cdot\rangle\in H_E \iff u\in H_S.
\end{align}
First note that, due to \eqref{eq:innerproductOrthogonalizesHSAotimesB},  for any $h\in H_S$ the functional $\langle h\,,\,\cdot\rangle$ belongs to the annihilator of $A\otimes B$, and thus to $H_E$. This implies that $\{\langle\, h\,,\,\cdot\,\rangle\}_{h\in H_S}\subseteq H_E$. Now, take a vector $h'\in H_E\subset (AB)^*$. The Riesz representation theorem tells us that there is a unique vector $u_{h'}\in AB$ such that
$\langle u_{h'},v\rangle=h'(v)$ for all $v\in AB$. Since $h'\in H_E$, $\langle u_{h'},v\rangle=0$ for all $v\in AB_{\otimes}$, so $u_{h'}$ must have component zero in the subspace $AB_{\otimes}$. Since $AB=AB_{\otimes}\oplus H_S$, $v_{h'}\in H_S$. This implies that $H_E\subseteq \{\langle\,h\,,\,\cdot\,\rangle\}_{h\in H_S}$ which in turn gives  $H_E=\{\langle\,h\,,\,\cdot\,\rangle\}_{h\in H_S}$. 

The above shows that the subspace $H_E\subset (AB)^*$ gets mapped to $H_S\subset AB$ via the Riesz map $u\mapsto v_{u}$, which is a basis-independent map between an inner product space and its dual. We can thus identify $H_E$, as a subspace of $(AB)^*$, with $H_S^*$ (defined as the dual to $H_S$ seen as a vector space on its own, instead of a subspace on $AB$).

This means that, by introducing this extra structure to the GPT system description (some inner product on $AB$ that treats $AB_{\otimes}$ and $H_S$ as orthogonal), one can essentially identify $H_E$ as the dual to $H_S$ and vice-versa. 

Note that it is always possible to define such an inner product, since we are assuming finite-dimensional vector spaces.\footnote{For instance: choose a basis for $A\otimes B$, $\{w_i\}_{i=1}^m$, and a basis for $H_S$,$\{h_i\}_{i=m+1}^n$. The union of both forms a basis for $AB$, $\{w_i\}\bigcup\{h_i\}$. Now, for any two vectors $v=\sum a_iw_i +\sum b_i h_i $ and $u=\sum a'_iw_i +\sum b'_i h_i $ define $\langle v,u\rangle:= \sum a_ia'_i+\sum b_ib'_i$.}
Thus, one is always free to introduce such an inner product and consequently to think of $H_E$ as the dual of $H_S$.

For different choices of inner product on $AB$, the identification above might not hold.  However,  inner products that treat $AB_{\otimes}$ as orthogonal to $H_S$ capture the splitting of the state space between the tomographically local part and the holistic one, since any holistic degree of freedom has zero component on the non-holistic subspace. Thus, the specific choice of inner product just discussed seems to be the natural one in this context.

\subsection{Representations and Properties of \texorpdfstring{$\Pi_{\rm TL}$}{PiTL} and \texorpdfstring{$\Pi_{\rm TNL}$}{PinTL}}\label{projprop}
Let us begin by proving Prop.~\ref{propPiTL}, which says that $\Pi_{\rm TL}$ can be written as
\begin{equation}
\label{eq:Pi_TL_decomposition_2}
\vcenter{\hbox{%
\begin{tikzpicture}
  \begin{pgfonlayer}{nodelayer}

    \def\dw{0.08}   
    \def\gap{0.16}  

    \def\xLA{-4.6}
    \def\xLB{-3.2}

    \def\yTop{ 1.5}
    \def\yMidTop{ 0.6}
    \def\yMidBot{-0.6}
    \def\yBot{-1.5}

    \node[style=none] (LTA1) at ({\xLA-\dw}, \yTop) {};
    \node[style=none] (LTA2) at ({\xLA-\dw}, \yMidTop) {};
    \node[style=none] (LBA1) at ({\xLA-\dw}, \yMidBot) {};
    \node[style=none] (LBA2) at ({\xLA-\dw}, \yBot) {};

    \node[style=none] (LTA3) at ({\xLA+\dw}, \yTop) {};
    \node[style=none] (LTA4) at ({\xLA+\dw}, \yMidTop) {};
    \node[style=none] (LBA3) at ({\xLA+\dw}, \yMidBot) {};
    \node[style=none] (LBA4) at ({\xLA+\dw}, \yBot) {};

    \node[style=none] (LTB1) at ({\xLB-\dw}, \yTop) {};
    \node[style=none] (LTB2) at ({\xLB-\dw}, \yMidTop) {};
    \node[style=none] (LBB1) at ({\xLB-\dw}, \yMidBot) {};
    \node[style=none] (LBB2) at ({\xLB-\dw}, \yBot) {};

    \node[style=none] (LTB3) at ({\xLB+\dw}, \yTop) {};
    \node[style=none] (LTB4) at ({\xLB+\dw}, \yMidTop) {};
    \node[style=none] (LBB3) at ({\xLB+\dw}, \yMidBot) {};
    \node[style=none] (LBB4) at ({\xLB+\dw}, \yBot) {};

    \node[style=none] (LLabA_top) at ({\xLA+0.30},  1.00) {$\scriptstyle \m{A}$};
    \node[style=none] (LLabB_top) at ({\xLB+0.30},  1.00) {$\scriptstyle \m{B}$};
    \node[style=none] (LLabA_bot) at ({\xLA+0.30}, -1.00) {$\scriptstyle \m{A}$};
    \node[style=none] (LLabB_bot) at ({\xLB+0.30}, -1.00) {$\scriptstyle \m{B}$};

    \def\boxPad{0.35}
    \node[style=none] (LR1) at ({\xLA-\boxPad}, \yMidTop) {};
    \node[style=none] (LR2) at ({\xLB+\boxPad}, \yMidTop) {};
    \node[style=none] (LR3) at ({\xLB+\boxPad}, \yMidBot) {};
    \node[style=none] (LR4) at ({\xLA-\boxPad}, \yMidBot) {};
    \node[style=none] (LBoxLab) at ({(\xLA+\xLB)/2}, 0) {$\Pi_{TL}$};

    \node[style=none] (Eq)  at (-1.6,0) {$=$};
    \node[style=none] (Sum) at (-0.3,0) {$\sum_{i,j}$};

    \def\xRA{ 1.2}
    \def\xRB{ 3.0}

    \def\wTri{0.55}
    \def\yBaseTop{ 1.05}  
    \def\yBaseBot{-1.05}  
    \def\yWireTop{ 1.55}
    \def\yWireBot{-1.55}

    \node[style=none] (ATipTop) at (\xRA, { \gap/2}) {};
    \node[style=none] (ATipBot) at (\xRA, {- \gap/2}) {};

    \node[style=none] (BTipTop) at (\xRB, { \gap/2}) {};
    \node[style=none] (BTipBot) at (\xRB, {- \gap/2}) {};

    \node[style=none] (AeL) at ({\xRA-\wTri}, \yBaseTop) {};
    \node[style=none] (AeR) at ({\xRA+\wTri}, \yBaseTop) {};
    \node[style=none] (AeLab) at (\xRA, 0.62) {$v_i$};

    \node[style=none] (AtL) at ({\xRA-\wTri}, \yBaseBot) {};
    \node[style=none] (AtR) at ({\xRA+\wTri}, \yBaseBot) {};
    \node[style=none] (AtLab) at (\xRA,-0.62) {$t_i$};

    \node[style=none] (AUpL)  at ({\xRA-\dw}, \yWireTop) {};
    \node[style=none] (AUpR)  at ({\xRA+\dw}, \yWireTop) {};
    \node[style=none] (AUpL2) at ({\xRA-\dw}, \yBaseTop) {};
    \node[style=none] (AUpR2) at ({\xRA+\dw}, \yBaseTop) {};

    \node[style=none] (ADnL)  at ({\xRA-\dw}, \yBaseBot) {};
    \node[style=none] (ADnR)  at ({\xRA+\dw}, \yBaseBot) {};
    \node[style=none] (ADnL2) at ({\xRA-\dw}, \yWireBot) {};
    \node[style=none] (ADnR2) at ({\xRA+\dw}, \yWireBot) {};

    \node[style=none] (RLabA_top) at ({\xRA+0.35},  1.35) {$\scriptstyle \m{A}$};
    \node[style=none] (RLabA_bot) at ({\xRA+0.35}, -1.35) {$\scriptstyle \m{A}$};

    \node[style=none] (BeL) at ({\xRB-\wTri}, \yBaseTop) {};
    \node[style=none] (BeR) at ({\xRB+\wTri}, \yBaseTop) {};
    \node[style=none] (BeLab) at (\xRB, 0.62) {$v_j$};

    \node[style=none] (BtL) at ({\xRB-\wTri}, \yBaseBot) {};
    \node[style=none] (BtR) at ({\xRB+\wTri}, \yBaseBot) {};
    \node[style=none] (BtLab) at (\xRB,-0.62) {$t_j$};

    \node[style=none] (BUpL)  at ({\xRB-\dw}, \yWireTop) {};
    \node[style=none] (BUpR)  at ({\xRB+\dw}, \yWireTop) {};
    \node[style=none] (BUpL2) at ({\xRB-\dw}, \yBaseTop) {};
    \node[style=none] (BUpR2) at ({\xRB+\dw}, \yBaseTop) {};

    \node[style=none] (BDnL)  at ({\xRB-\dw}, \yBaseBot) {};
    \node[style=none] (BDnR)  at ({\xRB+\dw}, \yBaseBot) {};
    \node[style=none] (BDnL2) at ({\xRB-\dw}, \yWireBot) {};
    \node[style=none] (BDnR2) at ({\xRB+\dw}, \yWireBot) {};

    \node[style=none] (RLabB_top) at ({\xRB+0.35},  1.35) {$\scriptstyle \m{B}$};
    \node[style=none] (RLabB_bot) at ({\xRB+0.35}, -1.35) {$\scriptstyle \m{B}$};

  \end{pgfonlayer}

  \begin{pgfonlayer}{edgelayer}

    \draw (LTA1.center) -- (LTA2.center);
    \draw (LTA3.center) -- (LTA4.center);
    \draw (LBA1.center) -- (LBA2.center);
    \draw (LBA3.center) -- (LBA4.center);

    \draw (LTB1.center) -- (LTB2.center);
    \draw (LTB3.center) -- (LTB4.center);
    \draw (LBB1.center) -- (LBB2.center);
    \draw (LBB3.center) -- (LBB4.center);

    \draw (LR1.center) -- (LR2.center) -- (LR3.center) -- (LR4.center) -- cycle;

    \draw (AeL.center) -- (AeR.center) -- (ATipTop.center) -- cycle;
    \draw (AtL.center) -- (AtR.center) -- (ATipBot.center) -- cycle;

    \draw (BeL.center) -- (BeR.center) -- (BTipTop.center) -- cycle;
    \draw (BtL.center) -- (BtR.center) -- (BTipBot.center) -- cycle;

    \draw (AUpL2.center) -- (AUpL.center);
    \draw (AUpR2.center) -- (AUpR.center);
    \draw (ADnL.center) -- (ADnL2.center);
    \draw (ADnR.center) -- (ADnR2.center);

    \draw (BUpL2.center) -- (BUpL.center);
    \draw (BUpR2.center) -- (BUpR.center);
    \draw (BDnL.center) -- (BDnL2.center);
    \draw (BDnR.center) -- (BDnR2.center);

  \end{pgfonlayer}
\end{tikzpicture}%
}}
\end{equation}
where $\{v_{i}^A\}\subset{A}$ is a basis for the local vector space $A$ and $\{v^B_j\}\subset B$ is a basis for the local vector space $B$, and $\{t^A_i\}\subset A^*$ and $\{t^B_j\}\subset B^*$ are their dual bases, respectively. Moreover, $\Pi_{\rm TL}$
can equivalently be written as 

  \begin{equation}
\label{eq:Pi_TL_coeff_decomposition_2}
\vcenter{\hbox{%
\begin{tikzpicture}
  \begin{pgfonlayer}{nodelayer}

    \def\dw{0.08}    
    \def\gap{0.16}   

    \def\xLA{-4.6}
    \def\xLB{-3.2}

    \def\yTop{ 1.5}
    \def\yMidTop{ 0.6}
    \def\yMidBot{-0.6}
    \def\yBot{-1.5}

    \node[style=none] (LTA1) at ({\xLA-\dw}, \yTop) {};
    \node[style=none] (LTA2) at ({\xLA-\dw}, \yMidTop) {};
    \node[style=none] (LBA1) at ({\xLA-\dw}, \yMidBot) {};
    \node[style=none] (LBA2) at ({\xLA-\dw}, \yBot) {};

    \node[style=none] (LTA3) at ({\xLA+\dw}, \yTop) {};
    \node[style=none] (LTA4) at ({\xLA+\dw}, \yMidTop) {};
    \node[style=none] (LBA3) at ({\xLA+\dw}, \yMidBot) {};
    \node[style=none] (LBA4) at ({\xLA+\dw}, \yBot) {};

    \node[style=none] (LTB1) at ({\xLB-\dw}, \yTop) {};
    \node[style=none] (LTB2) at ({\xLB-\dw}, \yMidTop) {};
    \node[style=none] (LBB1) at ({\xLB-\dw}, \yMidBot) {};
    \node[style=none] (LBB2) at ({\xLB-\dw}, \yBot) {};

    \node[style=none] (LTB3) at ({\xLB+\dw}, \yTop) {};
    \node[style=none] (LTB4) at ({\xLB+\dw}, \yMidTop) {};
    \node[style=none] (LBB3) at ({\xLB+\dw}, \yMidBot) {};
    \node[style=none] (LBB4) at ({\xLB+\dw}, \yBot) {};

    \node[style=none] (LLabA_top) at ({\xLA+0.30},  1.00) {$\scriptstyle \m{A}$};
    \node[style=none] (LLabB_top) at ({\xLB+0.30},  1.00) {$\scriptstyle \m{B}$};
    \node[style=none] (LLabA_bot) at ({\xLA+0.30}, -1.00) {$\scriptstyle \m{A}$};
    \node[style=none] (LLabB_bot) at ({\xLB+0.30}, -1.00) {$\scriptstyle \m{B}$};

    \def\boxPad{0.35}
    \node[style=none] (LR1) at ({\xLA-\boxPad}, \yMidTop) {};
    \node[style=none] (LR2) at ({\xLB+\boxPad}, \yMidTop) {};
    \node[style=none] (LR3) at ({\xLB+\boxPad}, \yMidBot) {};
    \node[style=none] (LR4) at ({\xLA-\boxPad}, \yMidBot) {};
    \node[style=none] (LBoxLab) at ({(\xLA+\xLB)/2}, 0) {$\Pi_{TL}$};

    \node[style=none] (Eq)   at (-1.6,0) {$=$};
    \node[style=none] (Sum)  at (-0.55,0) {$\sum$};
    \node[style=none] (Coef) at (0.70,0.0) {$c^{\m{A}}_{ij}\,c^{\m{B}}_{k\ell}$};

    \def\xRA{ 2.3}   
    \def\xRB{ 4.1}   

    \def\wTri{0.65}
    \def\yBaseTop{ 1.05}
    \def\yBaseBot{-1.05}
    \def\yWireTop{ 1.55}
    \def\yWireBot{-1.55}

    \node[style=none] (ATipTop) at (\xRA, { \gap/2}) {};
    \node[style=none] (ATipBot) at (\xRA, {- \gap/2}) {};

    \node[style=none] (BTipTop) at (\xRB, { \gap/2}) {};
    \node[style=none] (BTipBot) at (\xRB, {- \gap/2}) {};

    \node[style=none] (AomL) at ({\xRA-\wTri}, \yBaseTop) {};
    \node[style=none] (AomR) at ({\xRA+\wTri}, \yBaseTop) {};
    \node[style=none] (AomLab) at (\xRA, 0.65) {$\omega_i$};

    \node[style=none] (AeL) at ({\xRA-\wTri}, \yBaseBot) {};
    \node[style=none] (AeR) at ({\xRA+\wTri}, \yBaseBot) {};
    \node[style=none] (AeLab) at (\xRA,-0.65) {$e_j$};

    \node[style=none] (AUpL)  at ({\xRA-\dw}, \yWireTop) {};
    \node[style=none] (AUpR)  at ({\xRA+\dw}, \yWireTop) {};
    \node[style=none] (AUpL2) at ({\xRA-\dw}, \yBaseTop) {};
    \node[style=none] (AUpR2) at ({\xRA+\dw}, \yBaseTop) {};

    \node[style=none] (ADnL)  at ({\xRA-\dw}, \yBaseBot) {};
    \node[style=none] (ADnR)  at ({\xRA+\dw}, \yBaseBot) {};
    \node[style=none] (ADnL2) at ({\xRA-\dw}, \yWireBot) {};
    \node[style=none] (ADnR2) at ({\xRA+\dw}, \yWireBot) {};

    \node[style=none] (RLabA_top) at ({\xRA+0.35},  1.25) {$\scriptstyle \m{A}$};
    \node[style=none] (RLabA_bot) at ({\xRA+0.35}, -1.25) {$\scriptstyle \m{A}$};

    \node[style=none] (BomL) at ({\xRB-\wTri}, \yBaseTop) {};
    \node[style=none] (BomR) at ({\xRB+\wTri}, \yBaseTop) {};
    \node[style=none] (BomLab) at (\xRB, 0.65) {$\omega_k$};

    \node[style=none] (BeL) at ({\xRB-\wTri}, \yBaseBot) {};
    \node[style=none] (BeR) at ({\xRB+\wTri}, \yBaseBot) {};
    \node[style=none] (BeLab) at (\xRB,-0.65) {$e_\ell$};

    \node[style=none] (BUpL)  at ({\xRB-\dw}, \yWireTop) {};
    \node[style=none] (BUpR)  at ({\xRB+\dw}, \yWireTop) {};
    \node[style=none] (BUpL2) at ({\xRB-\dw}, \yBaseTop) {};
    \node[style=none] (BUpR2) at ({\xRB+\dw}, \yBaseTop) {};

    \node[style=none] (BDnL)  at ({\xRB-\dw}, \yBaseBot) {};
    \node[style=none] (BDnR)  at ({\xRB+\dw}, \yBaseBot) {};
    \node[style=none] (BDnL2) at ({\xRB-\dw}, \yWireBot) {};
    \node[style=none] (BDnR2) at ({\xRB+\dw}, \yWireBot) {};

    \node[style=none] (RLabB_top) at ({\xRB+0.35},  1.25) {$\scriptstyle \m{B}$};
    \node[style=none] (RLabB_bot) at ({\xRB+0.35}, -1.25) {$\scriptstyle \m{B}$};

  \end{pgfonlayer}

  \begin{pgfonlayer}{edgelayer}

    \draw (LTA1.center) -- (LTA2.center);
    \draw (LTA3.center) -- (LTA4.center);
    \draw (LBA1.center) -- (LBA2.center);
    \draw (LBA3.center) -- (LBA4.center);

    \draw (LTB1.center) -- (LTB2.center);
    \draw (LTB3.center) -- (LTB4.center);
    \draw (LBB1.center) -- (LBB2.center);
    \draw (LBB3.center) -- (LBB4.center);

    \draw (LR1.center) -- (LR2.center) -- (LR3.center) -- (LR4.center) -- cycle;

    \draw (AomL.center) -- (AomR.center) -- (ATipTop.center) -- cycle;
    \draw (AeL.center)  -- (AeR.center)  -- (ATipBot.center) -- cycle;

    \draw (BomL.center) -- (BomR.center) -- (BTipTop.center) -- cycle;
    \draw (BeL.center)  -- (BeR.center)  -- (BTipBot.center) -- cycle;

    \draw (AUpL2.center) -- (AUpL.center);
    \draw (AUpR2.center) -- (AUpR.center);
    \draw (ADnL.center) -- (ADnL2.center);
    \draw (ADnR.center) -- (ADnR2.center);

    \draw (BUpL2.center) -- (BUpL.center);
    \draw (BUpR2.center) -- (BUpR.center);
    \draw (BDnL.center) -- (BDnL2.center);
    \draw (BDnR.center) -- (BDnR2.center);

  \end{pgfonlayer}
\end{tikzpicture}%
,}}
\end{equation}
where $\{\omega_i^{\m{A}}\}$ and $\{\omega_k^{\m{B}}\}$, besides being bases for $A$ and $B$, belong to the local state spaces $\m{S_A}$ and $\m{S_B}$; similarly, $\{e_j^{\m{A}}\}_j\subset E_{\m{A}}$ and $\{e_l^{\m{B}}\}\subset{E_{\m{B}}}$ are valid effects forming  bases for $A^*$ and $B^*$, respectively.

\begin{proof}
We prove Eq.~\eqref{eq:Pi_TL_decomposition_2} first.  
Fix bases $\{v_i^A\}\subset A$, $\{v_j^B\}\subset B$ and their dual bases
$\{t_i^A\}\subset A^\ast$, $\{t_j^B\}\subset B^\ast$, i.e.
$t_i^A(v_{i'}^A)=\delta_{ii'}$ and $t_j^B(v_{j'}^B)=\delta_{jj'}$.
Consider the linear map
\begin{align}
\label{eq:MapP}
P(\,\cdot\,)
:=\sum_{i,j} (v_i^A\boxtimes v_j^B)\circ
[\,t_i^A\boxtimes t_j^B\,](\,\cdot\,).
\end{align}

Since $AB = AB_{\otimes} \oplus H_S$, a basis of $AB$ can be chosen as the union of
a basis of $AB_{\otimes}$ (given by the products $v_i^A\boxtimes v_j^B$)
and a basis $\{h_r\}_r$ of the holistic subspace $H_S$.
For any product basis element $v_{i'}^A\boxtimes v_{j'}^B$ we have
\begin{align}
P(v_{i'}^A\boxtimes v_{j'}^B)
&=\sum_{i,j} (v_i^A\boxtimes v_j^B)\,
t_i^A(v_{i'}^A)\, t_j^B(v_{j'}^B) \\
&=\sum_{i,j} (v_i^A\boxtimes v_j^B)\,
\delta_{ii'}\delta_{jj'}
= v_{i'}^A\boxtimes v_{j'}^B,
\end{align}
where we used the defining property of dual bases.
Thus $P$ fixes all product components.

Moreover, since $\mathrm{Span}\{t_i^A\boxtimes t_j^B\}_{i,j}
= (AB_{\otimes})^\ast$, every vector $h\in H_S$ is annihilated by
$t_i^A\boxtimes t_j^B$, i.e.
\begin{align}
[\,t_i^A\boxtimes t_j^B\,](h)=0
\qquad \forall\, i,j.
\end{align}
Hence $P(h)=0$ for all $h\in H_S$.
Therefore $P$ acts as the identity on $AB_{\otimes}$ and vanishes on $H_S$,
so $P=\Pi_{\rm TL}$ as linear maps on $AB$.
This proves Eq.~\eqref{eq:Pi_TL_decomposition_2}.

\vspace{0.3cm}

For Eq.~\eqref{eq:Pi_TL_coeff_decomposition_2}, we apply
Eq.~\eqref{eq:Pi_TL_decomposition_2} with a particular choice of bases:
take $\{v_i^A\}=\{\omega_i^A\}$ and $\{v_j^B\}=\{\omega_j^B\}$,
where the $\omega$'s are bases of $A$ and $B$ consisting of (valid)
local states, and let $\{t_i^A\}$, $\{t_j^B\}$ be the corresponding dual bases.
Let $\{e_\alpha^A\}\subset A^\ast$ and $\{e_\beta^B\}\subset B^\ast$
be bases of (valid) local effects.
Since the $e$'s span the dual spaces, each dual basis element expands as
\begin{align}
t_i^A=\sum_{\alpha} c^{A}_{i\alpha}\, e_\alpha^A,
\qquad
t_j^B=\sum_{\beta} c^{B}_{j\beta}\, e_\beta^B.
\end{align}
Substituting these expansions into Eq.~\eqref{eq:MapP}
gives $\Pi_{\rm TL}$ as a linear combination of effect--state channels.

To identify the coefficients, define
\begin{align}
(\mathbb{M}^{\m A})_{\alpha i}
:= e_\alpha^A(\omega_i^A),
\qquad
(\mathbb{M}^{\m B})_{\beta j}
:= e_\beta^B(\omega_j^B).
\end{align}
Using duality and evaluating the above expansion of $t_i^A$ on
$\omega_{i'}^A$ yields
\begin{align}
\delta_{ii'}
&= t_i^A(\omega_{i'}^A)
= \sum_\alpha c^{A}_{i\alpha}\,
e_\alpha^A(\omega_{i'}^A)
= \sum_\alpha c^{A}_{i\alpha}\,
(\mathbb{M}^{\m A})_{\alpha i'}.
\end{align}
In matrix form this reads
\begin{align}
C^A \mathbb{M}^{\m A} = I,
\end{align}
hence $C^A=(\mathbb{M}^{\m A})^{-1}$, i.e.
\begin{align}
c^{A}_{i\alpha}
=
\big[(\mathbb{M}^{\m A})^{-1}\big]_{i\alpha},
\end{align}
and similarly
$c^{B}_{j\beta}
=
\big[(\mathbb{M}^{\m B})^{-1}\big]_{j\beta}$.
This proves Eq.~\eqref{eq:Pi_TL_coeff_decomposition_2}.
\end{proof}

\medskip

Notice that, as a linear map from $AB$ to $AB$, $\Pi_{\rm TL}$
admits different tensor decompositions.
For instance, the map
\begin{align}
P'(\,\cdot\,)
:=\sum_i (v_i^A\circ t_i^A)\boxtimes \mathcal{I}^B(\,\cdot\,)
\end{align}
coincides with $P$ as a linear operator on $AB$.
However, when extended to a tripartite system,
$P\boxtimes \mathcal{I}^C$ and
$P'\boxtimes \mathcal{I}^C$ generally act differently:
the former erases all holistic components of that subsystems $\m{A}$ and $\m{B}$ might have, while the latter need not annihilate holistic components
between $B$ and $C$.
Since we restrict attention here to the bipartite case,
this distinction is unimportant, but it may become relevant
in multipartite extensions.

\begin{proposition}\label{prop16}The projection $\Pi_{\rm TL}$ onto $(AB)_\otimes$ has the following properties:
\begin{enumerate}
\item $\Pi_{\rm TL}$ acts trivially on $AB_{\otimes}$ and (by pre-composition) on $(AB)^*_\otimes$. In particular, $\Pi_{\rm TL}(\omega^{\m{AB}})=\omega^{\m{AB}}$ for all states $\omega^{\m{AB}}\in AB_\otimes$ and $e^{\m{AB}}\circ\Pi_{\rm TL}=e^{\m{AB}}$ for all effects $e\in (AB)^*_\otimes$.

\item  $\Pi_{\rm TL}$ preserves normalization: $u^{\m{AB}}\circ \Pi_{\rm TL} = u^{\m{AB}}$.

\item The only state $\omega^{\m{AB}}\in \Omega_{\m{AB}}$ for which $\Pi_{\rm TL}(\omega^{\m{AB}})=0^{AB}$ is $\omega^{\m{AB}}=0^{AB}$. 

\item If $e^{\m{AB}}\left[\Pi_{\rm TL}\left(\omega^{\m{AB}}\right)\right]\neq e^{\m{AB}}[\omega^{\m{AB}}]$, then both the effect $e^{\m{AB}}$ and the state $\omega^{\m{AB}}$ have holistic components; i.e., $e^{\m{AB}}\not\in (AB)^*_\otimes$ and  $\omega^{\m{AB}}\not\in AB_\otimes$. 
\end{enumerate}
\end{proposition}
\begin{proof}
Let us start with the first item. Recall that any vector in $v\in AB$ can be written as $v=v_{\rm TL}+ h$ where $v_{\rm TL}\in (AB)^*_\otimes$ and $h\in H_S$ are unique (consequence of Lemma~\ref{lemma:AtensorB+Holistic}). For all $v\in AB_\otimes$, we have that $h=0$. Therefore, Def.~\ref{def:PiTL} implies $\Pi_{\rm TL}(v)=v_{\rm TL}=v$ for all $v\in AB_{\otimes}$. In particular, $\Pi_{\rm TL}(\omega^{\m{AB}})=\omega^{\m{AB}}_{\rm TL}=\omega^{\m{AB}}$ for all states $\omega^{\m{AB}}\in S_{\m{AB}}\bigcap AB_\otimes$. A similar argument for $w\in (AB)^*_\otimes$ follows, using Prop.~\ref{prop:DualPiTL} which ensures $w\circ\Pi_{\rm TL}=w$ for all $w\in (AB)^*_\otimes$ -- and again $e\in E_{\m{AB}}\bigcap (AB)^*_\otimes$ is a special case.

Now consider preservation of normalization. Since $u^{\m{AB}}=u^{\m{A}}\otimes u^B$, $u^{\m{AB}}\in (AB)^*_{\otimes}$ and, therefore, item $1$ immediately implies
        \begin{align}
        u^{\m{AB}}\circ\Pi_{\rm TL}= u^{\m{AB}}.
    \end{align}

Now consider a state $\omega^{\m{AB}}\in S_{\rm AB}$ such that $\Pi_{\rm TL}(\omega^{\m{AB}})=0^{AB}$. Then,
\begin{align}
    u^{\m{AB}}[\Pi_{\rm TL}(\omega^{\m{AB}})]=0 \implies (u^{\m{AB}}\circ\Pi_{\rm TL})[\omega^{\m{AB}}]=0\implies u^{\m{AB}}(\omega^{\m{AB}})=0;
\end{align}
where in the second implication we used that $\Pi_{\rm TL}$ preserves normalization. But $u^{\m{AB}}(\omega^{\m{AB}})=0$ implies $\omega^{\m{AB}}=0^{AB}$, which proves item 3.

Now, consider any pair of state and effect, $\omega^{\m{AB}}\in S_{\m{AB}}$ and $e^{\m{AB}}\in E_{\m{AB}}$. Suppose $\omega^{\m{AB}}$ has null component in the holistic subspace, ie $\omega\in AB_\otimes$. Then, item $1$ implies $\Pi_{\rm TL}(\omega^{\m{AB}})=\omega^{\m{AB}}$ which implies
\begin{align}
    e^{\m{AB}}(\Pi_{\rm TL}(\omega^{\m{AB}}))=e^{\m{AB}}(\omega^{\m{AB}})\,\,\forall e^{\m{AB}}\in E_{\m{AB}}.
\end{align}
If $e^{\m{AB}}\in (AB)^*_\otimes$, then again item $1$ implies $e^{\m{AB}}\circ\Pi_{TL}=e^{\m{AB}}$ and, therefore,
\begin{align}
    e^{\m{AB}}(\Pi_{TL}(\omega^{\m{AB}}))=(e^{\m{AB}}\circ\Pi_{\rm TL})[\omega^{\m{AB}}]=e^{\m{AB}}(\omega^{\m{AB}})\,\,\forall \omega^{\m{AB}}\in S_{\m{AB}},
\end{align} 
which proves item 4.
\end{proof}

\begin{proposition}\label{prop17} The projection $\Pi_{\rm TNL}$ has the following properties:
\begin{enumerate}
    \item $\Pi_{\rm TNL}=\mathbf{0}$ if and only if the GPT system is tomographically local;
    \item $\Pi_{\rm TNL}$ does not preserve normalization. Vectors in its image have normalization 0, since $u^{\m{AB}}\circ\Pi_{\rm TNL}=0^{(AB)^*}$;
    \item $\mathsf{ Im}(\Pi_{\rm TNL})\bigcap S_{\m{AB}}=\{0^{AB}\}$. 
    \item For a state $\omega$, $\Pi_{\rm TNL}(\omega)=0\iff \omega\in AB_{\otimes}$;  similarly for effects, $e\circ \Pi_{\rm TNL}=0\iff e\in AB^*_{\otimes}$. As a consequence of property $3$, for any normalized state $\omega$, $\Pi_{\rm TNL}(\omega)\not\in S_{\m{AB}}\iff \omega\not\in AB_{\otimes}$.
\end{enumerate}
\end{proposition}
\begin{proof}
    Recall that the action of $\Pi_{\rm TNL}(v)=h$, given any vector $v=v_{\rm TL}+h$ in $AB$, with $v_{TL}\in AB_\otimes$ and $h\in H_S$.  \blk If the GPT system has no holistic components, $H_S={0}$, which implies $\Pi_{\rm TNL}=\mathbf{0}$.

    Now consider item $2$. Recall that $u^{\m{AB}}\in AB^*_{\otimes}$, so it has no holistic component. Therefore, due to Prop.~\ref{prop:DualPinTL} $u^{\m{AB}}\circ\Pi_{\rm TNL}=0^{(AB)^*}$. 

    Since $\Pi_{\rm TNL}$ is a projection onto $H_S$, its image is $H_S$, and the only state which is contained in $H_S$ is $0^{AB}$ -- which follows from item $2$ and the fact that $\m{S_{AB}}=\mathsf{ConvHull}[0^{AB},\Omega_{\m{AB}}]$. This proves item $3$.

    For the proof of item $4$, recall that the action of the projection onto any state $v\in AB$ is given by  $\Pi_{\rm TNL}(v)=h$, where $h$ is the holistic component of $v$. Thus, $\Pi_{\rm TNL}(v)=0$ implies that $v$ has no holistic component, which is equivalent to $v\in AB_{\otimes}$.  This still holds true for the particular case of a state $\omega\in\m{S_{AB}}$. Conversely, consider a state $\omega\in AB_{\otimes}$; it has null holistic component, which implies $\Pi_{\rm TNL}(v)=0$. This proves the first part of item $4$. That a similar holds for effects is completely analogous, using Prop.~\ref{prop:DualPinTL}. Consider now that, for a given state $\omega\in\m{S_{AB}}$, $\Pi_{\rm TNL}(\omega)\not\in\m{S_AB}$; then, $\Pi_{\rm TNL}(\omega)\neq 0^{AB}$, which (given the first part of item $4$) implies $\omega\not\in AB_\otimes$. Conversely, suppose that the state is not contained in the tomographically local subspace, $\omega\not\in AB_{\otimes}$. Then, $\Pi_{\rm TNL}(\omega)\neq 0$, which, due to item $3$, implies $\Pi_{\rm TNL}(\omega)\not\in {\m{S_{AB}}}$.
\end{proof}

 \section{Proof of Proposition \texorpdfstring{~\ref{prop:TwoRebitLackingTLEntanglement}}{7}}
 \label{sec:ProofTwoRebitsLackingTLEProposition}

Here we present a proof of Prop.~\ref{prop:TwoRebitLackingTLEntanglement}, which let us recall it states the following: \\

Prop.~\ref{prop:TwoRebitLackingTLEntanglement}:  the following conditions hold for a two-rebit state $\rho^{\m{AB}}$:
\begin{enumerate}
    \item $\rho^{\m{AB}}$ lacks TL-entanglement if and only if
    $\Pi_{\rm TL}(\rho^{\m{AB}})$ is a valid two-rebit state, i.e.,
    $\Pi_{\rm TL}(\rho^{\m{AB}})\in\Omega_{\m{AB}}$.
    \item Let $\iota$ denote the natural inclusion of two-rebit states into
    two-qubit states in complex quantum theory. If $\iota(\rho^{\m{AB}})$
    is separable in complex quantum theory, then $\rho^{\m{AB}}$ lacks
    TL-entanglement.
\end{enumerate}

\begin{proof}
We begin with item~1.

\smallskip
\noindent\emph{($\Leftarrow$)} 
Assume that $\Pi_{\rm TL}(\rho^{\m{AB}})$ is a valid two-rebit state.
By definition, $\Pi_{\rm TL}(\rho^{\m{AB}})$ lies in the tomographically
local subspace $AB_{\otimes}$, and therefore it has no TNL entanglement.
For two-rebit systems, Proposition~\ref{prop:TwoRebitsNonSepIffTNLEntangled}
implies that every valid state in $AB_{\otimes}$ is separable.
Hence, $\Pi_{\rm TL}(\rho^{\m{AB}})$ is separable.

Since TL-entanglement depends only on the component of a state in
$AB_{\otimes}$, the separability of $\Pi_{\rm TL}(\rho^{\m{AB}})$ implies
that $\rho^{\m{AB}}$ lacks TL-entanglement.

\smallskip
\noindent\emph{($\Rightarrow$)} 
Conversely, if $\rho^{\m{AB}}$ lacks TL-entanglement, then by definition
its tomographically local component is separable, i.e.,
$\Pi_{\rm TL}(\rho^{\m{AB}})\in\mathsf{Sep}(\Omega_{\m{AB}})$.
In particular, $\Pi_{\rm TL}(\rho^{\m{AB}})$ is a valid two-rebit state.

\smallskip
This completes the proof of item~1.

\medskip
We now turn to item~2.

Let $\iota:\mathsf{Span}[\Omega_{\m{AB}}]\to\mathsf{Span}[\Omega_{\m{CD}}]$
denote the natural inclusion of two-rebit states into two-qubit states in
complex quantum theory, which embeds real density matrices as complex
density matrices with no imaginary components.

In complex quantum theory, the inclusion of the tomographically local
projection $\Pi_{\rm TL}(\rho^{\m{AB}})$ can be obtained by first embedding
$\rho^{\m{AB}}$ via $\iota$ and then applying the linear map 
\begin{align}
\label{eq:PiTLExtension}
    \mathcal{O}^{\m{CD}}
    :=\frac{1}{2}(\mathcal{I}^{\m{C}}+K^{\m{C}})\otimes \frac{1}{2}(\mathcal{I}^{\m{D}}+K^{\m{D}})
\end{align}
where $\mathcal{I}^{\m{C}}$ represents the identity transformation, $K^{\m{C}}$ denotes complex conjugation acting on subsystem $\m{C}$ and similarly for $\m{D}$.
Explicitly,\footnote{In fact, there are many possible extensions of $\Pi_{\rm TL}$ as operators acting on $CD$, and the proof goes through with any particular choice. We chose~\eqref{eq:PiTLExtension} because then $\m{O^{CD}}$ can be thought of as kind of local \emph{swirling} for two-qubit systems, which is closely related to a known way of defining RQT from standard QT~\cite{ying2025quantumtheoryneedscomplex}.}
\begin{align}
    \iota\!\left[\Pi_{\rm TL}(\rho^{\m{AB}})\right]
    =
    \mathcal{O}^{\m{CD}}\!\left[\iota(\rho^{\m{AB}})\right].
\end{align}
Operationally, $\mathcal{O}^{\m{CD}}$ removes all components that are odd
under complex conjugation, and in particular eliminates the
$\sigma_y\otimes\sigma_y$ term.

Now suppose that $\iota(\rho^{\m{AB}})$ is separable in complex quantum
theory. Since local complex conjugation maps separable states to separable
states,  $K^{\m{C}}\otimes\mathcal{I}^{\m{D}}[\iota(\rho^{\m{AB}})]$,
$\mathcal{I}^{\m{C}}\otimes K^{\m{D}}[\iota(\rho^{\m{AB}})]$ and $K^{\m{C}}\otimes K^{\m{D}}[\iota(\rho^{\m{AB}})]$ are separable states.
Therefore, their convex combination with $\iota(\rho^{\m{AB}})$ in the expression for
$\mathcal{O}^{\m{CD}}[\iota(\rho^{\m{AB}})]$ is also separable in complex
quantum theory.

Moreover, $\mathcal{O}^{\m{CD}}[\iota(\rho^{\m{AB}})]$ has no
$\sigma_y$ components and hence corresponds to a valid two-rebit state.
That is,
\begin{align}
    \Pi_{\rm TL}(\rho^{\m{AB}})\in\Omega_{\m{AB}}.
\end{align}
By item~1, this implies that $\rho^{\m{AB}}$ lacks TL-entanglement.
\end{proof}

\end{document}